\documentclass{article}
\usepackage[T1]{fontenc}
\newcommand{\lAngle}{\langle\!\langle}
\newcommand{\rAngle}{\rangle\!\rangle}

\usepackage{geometry}
\usepackage{array}
\usepackage{adjustbox}
\usepackage{tikz-cd}
\usepackage{tikz}
\usetikzlibrary{cd,arrows}
\usepackage{mathpartir}
\usepackage{caption}
\usepackage{subcaption}
\usepackage{amsfonts}
\usepackage{amsmath}
\usepackage[renew-matrix,renew-dots]{nicematrix}
\usepackage{float}
\usepackage{enumitem}
\usepackage{xcolor}
\usepackage{booktabs}
\usepackage{rotating}
\usepackage{multirow}
\usepackage{comment}
\usepackage{graphicx}
\usepackage{braket}
\usepackage{enumitem}
\usepackage{xparse}
\usepackage{amsmath,amsthm}
\usepackage{mathtools}
\usepackage{xifthen}
\usepackage{makecell}
\usepackage{txfonts}
\usepackage{hyperref}
\usepackage{cleveref}
\hypersetup{
    colorlinks,
    linkcolor={red!50!black},
    citecolor={blue!50!black},
    urlcolor={blue!80!black}
}
\usepackage{mathtools}
\usepackage{changepage}

\usepackage{quiver}

\usepackage{todonotes}

\usepackage{authblk}
\colorlet{tapeBg}{red!30}
\colorlet{tapeBorder}{red!60}

\DeclareMathSymbol{\mathinvertedexclamationmark}{\mathclose}{operators}{'074}
\DeclareMathSymbol{\mathexclamationmark}{\mathclose}{operators}{'041}

\makeatletter
\newcommand{\raisedmathinvertedexclamationmark}{%
  \mathclose{\mathpalette\raised@mathinvertedexclamationmark\relax}%
}
\newcommand{\raised@mathinvertedexclamationmark}[2]{%
  \raisebox{\depth}{$\m@th#1\mathinvertedexclamationmark$}%
}
\begingroup\lccode`~=`! \lowercase{\endgroup
  \def~}{\@ifnextchar`{\raisedmathinvertedexclamationmark\@gobble}{\mathexclamationmark}}
\mathcode`!="8000
\makeatother

\usepackage{etoolbox}

\makeatletter
\patchcmd{\endalign}{\restorealignstate@}{\global\let\df@label\@empty\restorealignstate@}{}{}
\makeatother

\newlength{\leftstackrelawd}
\newlength{\leftstackrelbwd}
\def\leftstackrel#1#2{\settowidth{\leftstackrelawd}%
{${{}^{#1}}$}\settowidth{\leftstackrelbwd}{$#2$}%
\addtolength{\leftstackrelawd}{-\leftstackrelbwd}%
\leavevmode\ifthenelse{\lengthtest{\leftstackrelawd>0pt}}%
{\kern-.5\leftstackrelawd}{}\mathrel{\mathop{#2}\limits^{#1}}}

\makeatletter
\newcommand{\mylabel}[2]{\def\@currentlabel{#2}\label{#1}}
\makeatother

\makeatletter
\newcommand\newtag[2]{#1\def\@currentlabel{#1}\label{#2}}
\makeatother

\newcommand{\op}[1]{#1^{\dagger}}
\newcommand{\co}[1]{\text{\reflectbox{\(#1\)}}}
\newcommand{\opcat}[1]{#1^{\mathsf{op}}}

\newcommand{\N}{\mathbb{N}}

\newcommand{\dcomp}{\mathbin{;}}

\newcommand{\lbbd}{\left \llbracket}
\newcommand{\rbbd}{\right \rrbracket}
\newcommand{\dsem}[1]{\lbbd #1 \rbbd}

\newcommand{\CBdsem}[1]{\dsem{#1}}
\newcommand{\TCBdsem}[1]{\dsem{#1}^{\sharp}}

\newcommand{\id}[1]{{id}_{#1}}

\newcommand{\perG}{\odot} %
\newcommand{\unoG}{I} %

\newcommand{\piu}{\oplus} %
\newcommand{\per}{\otimes} %

\newcommand{\CatTrTape}{\Cat{Tr}_\sign}

\NewDocumentCommand\Piu{oom}{\bigoplus_{\IfNoValueTF{#1}{}{#1}}^{\IfNoValueTF{#2}{}{#2}}{#3}} %
\NewDocumentCommand\Per{oom}{\bigotimes_{\IfNoValueTF{#1}{}{#1}}^{\IfNoValueTF{#2}{}{#2}}{#3}} %

\NewDocumentCommand\PiuPar{oom}{\Piu[#1][#2] {(#3)}} %
\NewDocumentCommand\PerPar{omm}{\Per[#1][#2] {(#3)}} %

\NewDocumentCommand\PiuL{oom}{\bigoplus\limits_{\IfNoValueTF{#1}{}{#1}}^{\IfNoValueTF{#2}{}{#2}}{#3}} %
\NewDocumentCommand\PerL{oom}{\bigotimes\limits_{\IfNoValueTF{#1}{}{#1}}^{\IfNoValueTF{#2}{}{#2}}{#3}} %

\NewDocumentCommand\PiuParL{oom}{\PiuL[#1][#2] {(#3)}} %
\NewDocumentCommand\PerParL{omm}{\PerL[#1][#2] {(#3)}} %

\newcolumntype{C}[1]{>{\centering\let\newline\\\arraybackslash\hspace{0pt}}p{#1}}
\newcolumntype{L}[1]{>{\raggedright\let\newline\\\arraybackslash\hspace{0pt}}p{#1}}
\newcolumntype{R}[1]{>{\raggedleft\let\newline\\\arraybackslash\hspace{0pt}}p{#1}}

\newcommand{\sort}{\mathcal{S}}
\newcommand{\sign}{\Sigma}

\newcommand{\gen}{s}

\newcommand{\Cat}[1]{\mathbf{#1}} %
\newcommand{\fun}[1]{\mathsf{#1}} %
\newcommand{\CatString}{\Cat{C}_\sign} %
\newcommand{\CatTape}{\Cat{T}_\sign} %

\newcommand{\CatKTape}{\Cat{KT}_\sign} %
\newcommand{\CatKTapeC}{\Cat{KT}_{\sign+\Gamma}} %

\newcommand{\CMon}{\Cat{CMon}} %
\newcommand{\CAT}{\Cat{Cat}} %
\newcommand{\CMonCat}{\CMon\CAT} %
\newcommand{\fbCat}{\Cat{FBC}} %
\newcommand{\FBC}{\Cat{FBC}} %
\newcommand{\CMC}{\Cat{CMC}} %

\newcommand{\Mat}[1]{{\Cat{Mat}(#1)}} %
\newcommand{\MatFun}{\fun{Mat}} %
\newcommand{\SMC}{\Cat{SMC}}

\newcommand{\Rel}{\Cat{Rel}}

\newcommand{\CB}{\Cat{CB}_{\sign}} %

\newcommand{\KTCB}{\Cat{KCT}_{\sign}}
\newcommand{\KTCBI}{\Cat{KCT}_{\sign,\basicR}}

\newcommand{\KTCBP}{\Cat{KCT}_{\sign,\mathbb{P}}}

\newcommand{\KBicat}{\Cat{KBicat}} %
\newcommand{\TKAlg}{\Cat{TKAlg}} %
\newcommand{\UTr}{\fun{UTr}} %
\newcommand{\freeTr}{\fun{Tr}} %
\newcommand{\Unif}{\fun{Unif}} %
\newcommand{\unifeq}{\approx} %

\newcommand{\TSMC}{\Cat{TrSMC}} %
\newcommand{\UTSMC}{\Cat{UTSMC}} %
\newcommand{\UTCMC}{\Cat{UTCMC}} %

\newcommand{\RIG}{\Cat{Rig}} %
\newcommand{\RIGUTr}{\Cat{UTRig}} %

\newcommand{\FBUTr}{\Cat{UTFBC}} %

\newcommand{\FBRIG}{\Cat{FBRig}} %
\newcommand{\FBRIGUtr}{\Cat{UTFBRig}} %

\newcommand{\assoc}[3]{\alpha_{#1,#2,#3}} %
\newcommand{\lunit}[1]{\lambda_{#1}} %
\newcommand{\Iassoc}[3]{\alpha^{-}_{#1,#2,#3}} %
\newcommand{\Ilunit}[1]{\lambda^{-}_{#1}} %
\NewDocumentCommand\symm{gg}{\sigma\IfNoValueTF{#1}{}{_{#1,#2}}} %

\newcommand{\assocp}[3]{\alpha^\piu_{#1,#2,#3}} %
\newcommand{\lunitp}[1]{\lambda^\piu_{#1}} %
\newcommand{\runitp}[1]{\rho^\piu_{#1}} %
\newcommand{\zero}{\ensuremath{0}} %
\NewDocumentCommand\symmp{gg}{\sigma^\piu\IfNoValueTF{#1}{}{_{#1,#2}}} %

\newcommand{\Iassocp}[3]{\alpha^{-\piu}_{#1,#2,#3}} %

\newcommand{\assoct}[3]{\alpha^\per_{#1,#2,#3}} %
\newcommand{\runitt}[1]{\rho^\per_{#1}} %
\newcommand{\uno}{1} %
\NewDocumentCommand\symmt{gg}{\sigma^\per\IfNoValueTF{#1}{}{_{#1,#2}}} %

\newcommand{\dl}[3]{\delta^l_{#1,#2,#3}} %
\newcommand{\dr}[3]{\delta^r_{#1,#2,#3}} %

\newcommand{\annl}[1]{\lambda^\bullet_{#1}} %
\newcommand{\annr}[1]{\rho^\bullet_{#1}} %

\newcommand{\Idl}[3]{\delta^{-l}_{#1,#2,#3}} %

\newcommand{\trace}{\mathsf{tr}}
\newcommand{\kstar}[1]{#1^{\ast}}

\NewDocumentCommand\Dl{gg}{\Delta^l\IfNoValueTF{#1}{}{_{#1,#2}}}
\NewDocumentCommand\Dr{gg}{\Delta^r\IfNoValueTF{#1}{}{_{#1,#2}}}
\NewDocumentCommand\IDl{gg}{\Delta^{-l}\IfNoValueTF{#1}{}{_{#1,#2}}}
\NewDocumentCommand\IDr{gg}{\Delta^{-r}\IfNoValueTF{#1}{}{_{#1,#2}}}
\NewDocumentCommand\D{gg}{\Delta\IfNoValueTF{#1}{}{_{#1,#2}}}
\NewDocumentCommand\ID{gg}{\Delta^{-}\IfNoValueTF{#1}{}{_{#1,#2}}}

\newcommand{\copier}[1]{\blacktriangleleft_{#1}} %
\newcommand{\discharger}[1]{\mbox{!}_{#1}} %

\newcommand{\cocopier}[1]{\blacktriangleright_{#1}} %
\newcommand{\codischarger}[1]{\mbox{!`}_{#1}} %

\newcommand{\diag}[1]{\lhd_{#1}}
\newcommand{\bang}[1]{\rotatebox{90}{$\multimapinv$}_{#1}}

\newcommand{\codiag}[1]{\rhd_{#1}}
\newcommand{\cobang}[1]{\rotatebox{90}{$\multimap$}_{#1}}

\newcommand{\WTrel}[1]{\mathsf{WTRel}_{#1}}
\newcommand{\WTrelC}{\WTrel{\Cat{C}}}

\newcommand{\precongR}[1]{\,{\leq}_{#1}\,} 
\newcommand{\congR}[1]{\,{\approx}_{#1}\,} 

\newcommand{\precongRA}[1]{\,{\lesssim}_{#1}\,}
\newcommand{\wiskR}[1]{\hat{#1}}
\newcommand{\basicR}{\mathbb{I}}
\newcommand{\basicK}{\mathbb{K}}
\newcommand{\basicCB}{\mathbb{CB}}
\newcommand{\basicKC}{\mathbb{KC}}

\newcommand{\wiskbasicR}{\wiskR{\basicR}}

\newcommand{\WprecongBA}{\precongRA{\wiskR{\basicR}}}

\newcommand{\precongB}{\precongR{{\basicR}}}
\newcommand{\precongK}{\precongR{{\basicK}}}

\newcommand{\precongKC}{\precongR{{\basicKC}}}

\newcommand{\congB}{\congR{{\basicR}}}

\newcommand{\interpretation}{\mathcal{I}} %
\newcommand{\interpret}[1]{\interpretation(#1)}

\newcommand{\LW}[2]{\fun{L}_{#1}({#2})} %
\newcommand{\RW}[2]{\fun{R}_{#1}({#2})} %
\newcommand{\LWI}[2]{\fun{L}_{#1}'({#2})} %
\newcommand{\RWI}[2]{\fun{R}_{#1}'({#2})} %
\newcommand\defeq{\stackrel{\text{\tiny def}}{=}}

\newcommand{\axeq}[1]{\stackrel{({#1})}{=}}
\newcommand{\axsubeq}[1]{\stackrel{({#1})}{\leq}}

\newcommand{\tape}[1]{\overline{\underline{\; #1 \vphantom{d_A} \;}}} %
\newcommand{\tapeFunct}[1]{\underline{\smash{\overline{\; #1 \vphantom{t} \;}}}}

\newcommand{\tapesymm}[2]{\tape{\sigma_{#1,#2}}} %
\renewcommand{\t}{\mathfrak{t}}
\newcommand{\s}{\mathfrak{s}}

\newcommand{\ar}{\mathit{ar}}
\newcommand{\coar}{\mathit{coar}}

\theoremstyle{plain}
\newtheorem{theorem}{Theorem}[section]
\newtheorem{lemma}[theorem]{Lemma}
\newtheorem{corollary}[theorem]{Corollary}

\theoremstyle{definition}
\newtheorem{definition}[theorem]{Definition}
\newtheorem{proposition}[theorem]{Proposition}
\newtheorem{remark}[theorem]{Remark}
\newtheorem{example}[theorem]{Example}
\newtheorem*{notation}{Notation}

\newcommand{\encoding}[1]{\mathcal{E}(#1)}

\usepackage{tikz}
\usetikzlibrary{cd,arrows,backgrounds,shapes.geometric,shapes.misc,matrix,arrows.meta,decorations.markings}
\pgfdeclarelayer{background}
\pgfdeclarelayer{edgelayer}
\pgfdeclarelayer{nodelayer}
\pgfsetlayers{background,edgelayer,nodelayer,main}

\tikzset{baseline=-0.5ex,>=stealth'}
\tikzset{every picture/.append style={scale=0.3}}
\tikzstyle{none}=[inner sep=0pt]
\tikzstyle{label}=[fill=none, draw=none, shape=circle, scale=0.7]
\tikzstyle{dots}=[fill=none, draw=none, shape=circle, scale=0.5]
\tikzstyle{black}=[circle, draw=black, fill=black, inner sep=0pt, minimum size=2.5pt]
\tikzstyle{white}=[circle, draw=black, fill=white, inner sep=0pt, minimum size=2.5pt]
\tikzstyle{reg}=[draw, fill=white, rounded rectangle, rounded rectangle left arc=none, minimum height=0.8em, minimum width=1em, node font={\tiny}]
\tikzstyle{coreg}=[draw, fill=white, rounded rectangle, rounded rectangle right arc=none, minimum height=0.8em, minimum width=1em, node font={\tiny}]
\tikzstyle{box}=[draw, fill=white, rectangle, minimum height=1em, minimum width=1em, node font={\tiny}]
\tikzstyle{stringbox}=[draw, fill=white, rectangle, minimum height=1em, minimum width=1em]
\tikzstyle{stringlongbox}=[draw, fill=white, rectangle, minimum height=2.2em, minimum width=1em]
\tikzstyle{triang}=[draw, fill=white, isosceles triangle, isosceles triangle apex angle=60, minimum size=0.5em, node font={\tiny}]

\tikzstyle{bbox}=[fill=white, draw=black, chamfered rectangle, chamfered rectangle corners={north east}, chamfered rectangle sep=2pt, inner sep=0pt, minimum height=1.4em, minimum width=1.4em, scale=0.7]
\tikzstyle{bboxOp}=[fill=white, draw=black, chamfered rectangle, chamfered rectangle corners={north west}, chamfered rectangle sep=2pt, inner sep=0pt, minimum height=1.4em, minimum width=1.4em, scale=0.7]
\tikzstyle{corefl}=[fill=white, draw=black, circle, inner sep=2pt, scale=0.7]

\tikzstyle{tapeFill} = [fill=tapeBg]
\tikzstyle{tapeNoFill} = [draw=tapeBorder, line width=0.5pt]
\tikzstyle{tape} = [fill=tapeBg, draw=tapeBorder, line width=0.5pt]

\tikzcdset{arrow style=tikz}

\newcommand{\comonoid}[1]{
    \begin{tikzpicture}
        \begin{pgfonlayer}{nodelayer}
            \node [style=none] (81) at (-0.75, 0) {};
            \node [style=black] (107) at (0, 0) {};
            \node [style=none] (108) at (0.5, -0.375) {};
            \node [style=none] (109) at (0.5, 0.375) {};
            \node [style=none] (110) at (0.75, -0.375) {};
            \node [style=none] (111) at (0.75, 0.375) {};
            \ifthenelse{\isempty{#1}}{}{
            \node [style=label] (112) at (1.075, -0.375) {$#1$};
            \node [style=label] (113) at (1.075, 0.375) {$#1$};
            \node [style=label] (114) at (-1.075, 0) {$#1$};
            }
        \end{pgfonlayer}
        \begin{pgfonlayer}{edgelayer}
            \draw [bend right] (109.center) to (107);
            \draw [bend right] (107) to (108.center);
            \draw (81.center) to (107);
            \draw (111.center) to (109.center);
            \draw (108.center) to (110.center);
        \end{pgfonlayer}
    \end{tikzpicture}           
 }

 \newcommand{\counit}[1]{
    \begin{tikzpicture}
        \begin{pgfonlayer}{nodelayer}
            \node [style=none] (81) at (-0.75, 0) {};
            \node [style=black] (107) at (0, 0) {};
            \ifthenelse{\isempty{#1}}{}{
            \node [style=label] (114) at (-1.075, 0) {$#1$};
            }
        \end{pgfonlayer}
        \begin{pgfonlayer}{edgelayer}
            \draw (81.center) to (107);
        \end{pgfonlayer}
    \end{tikzpicture}      
}

 \newcommand{\monoid}[1]{
    \begin{tikzpicture}
        \begin{pgfonlayer}{nodelayer}
            \node [style=none] (81) at (0.75, 0) {};
            \node [style=black] (107) at (0, 0) {};
            \node [style=none] (108) at (-0.5, -0.375) {};
            \node [style=none] (109) at (-0.5, 0.375) {};
            \node [style=none] (110) at (-0.75, -0.375) {};
            \node [style=none] (111) at (-0.75, 0.375) {};
            \ifthenelse{\isempty{#1}}{}{
            \node [style=label] (112) at (-1.075, -0.375) {$#1$};
            \node [style=label] (113) at (-1.075, 0.375) {$#1$};
            \node [style=label] (114) at (1.075, 0) {$#1$};
            }
        \end{pgfonlayer}
        \begin{pgfonlayer}{edgelayer}
            \draw [bend left] (109.center) to (107);
            \draw [bend left] (107) to (108.center);
            \draw (81.center) to (107);
            \draw (111.center) to (109.center);
            \draw (108.center) to (110.center);
        \end{pgfonlayer}
    \end{tikzpicture}                     
 }

\newcommand{\unit}[1]{
    \begin{tikzpicture}
        \begin{pgfonlayer}{nodelayer}
            \node [style=none] (81) at (-0.325, 0) {};
            \node [style=black] (107) at (-1.075, 0) {};
            \ifthenelse{\isempty{#1}}{}{
            \node [style=label] (114) at (0, 0) {$#1$};
            }
        \end{pgfonlayer}
        \begin{pgfonlayer}{edgelayer}
            \draw (81.center) to (107);
        \end{pgfonlayer}
    \end{tikzpicture}         
}

\newcommand{\wire}[1]{
    \begin{tikzpicture}
        \begin{pgfonlayer}{nodelayer}
            \node [style=label] (8) at (1.5, 0) {$#1$};
            \node [style=label] (11) at (-1.5, 0) {$#1$};
        \end{pgfonlayer}
        \begin{pgfonlayer}{edgelayer}
            \draw (11) to (8);
        \end{pgfonlayer}
    \end{tikzpicture}    
}

\newcommand{\Tcomonoid}[1]{
    \ifthenelse{\equal{#1}{\uno}}{
        \begin{tikzpicture}
            \begin{pgfonlayer}{nodelayer}
                \node [style=none] (1) at (1, 1) {};
                \node [style=none] (2) at (1, -1) {};
                \node [style=none] (6) at (-1.5, 0.5) {};
                \node [style=none] (7) at (-1.5, -0.5) {};
                \node [style=none] (8) at (1, -1.5) {};
                \node [style=none] (9) at (1, 1.5) {};
                \node [style=none] (10) at (1, 0.5) {};
                \node [style=none] (11) at (0.5, 0) {};
                \node [style=none] (12) at (1, -0.5) {};
                \node [style=none] (13) at (-0.5, -0.5) {};
                \node [style=none] (14) at (-0.5, 0.5) {};
                \node [style=none] (15) at (-1.5, 0) {};
            \end{pgfonlayer}
            \begin{pgfonlayer}{edgelayer}
                \path [style=tape] (9.center)
                     to [bend right] (14.center)
                     to (6.center)
                     to (7.center)
                     to (13.center)
                     to [bend right] (8.center)
                     to (12.center)
                     to [bend left=45] (11.center)
                     to [bend left=45] (10.center)
                     to cycle;
            \end{pgfonlayer}
        \end{tikzpicture}        
    }
    {
        \begin{tikzpicture}
            \begin{pgfonlayer}{nodelayer}
                \node [style=none] (0) at (0, 0) {};
                \node [style=none] (1) at (1, 1) {};
                \node [style=none] (2) at (1, -1) {};
                \ifthenelse{\isempty{#1}}{}{
                    \node [style=label] (3) at (-2, 0) {$#1$};
                    \node [style=label] (4) at (1.5, 1) {$#1$};
                    \node [style=label] (5) at (1.5, -1) {$#1$};
                }
                \node [style=none] (6) at (-1.5, 0.5) {};
                \node [style=none] (7) at (-1.5, -0.5) {};
                \node [style=none] (8) at (1, -1.5) {};
                \node [style=none] (9) at (1, 1.5) {};
                \node [style=none] (10) at (1, 0.5) {};
                \node [style=none] (11) at (0.5, 0) {};
                \node [style=none] (12) at (1, -0.5) {};
                \node [style=none] (13) at (-0.5, -0.5) {};
                \node [style=none] (14) at (-0.5, 0.5) {};
                \node [style=none] (15) at (-1.5, 0) {};
            \end{pgfonlayer}
            \begin{pgfonlayer}{edgelayer}
                \path [style=tape] (9.center)
                     to [bend right] (14.center)
                     to (6.center)
                     to (7.center)
                     to (13.center)
                     to [bend right] (8.center)
                     to (12.center)
                     to [bend left=45] (11.center)
                     to [bend left=45] (10.center)
                     to cycle;
                \draw [bend right=45] (0.center) to (2.center);
                \draw [bend left=45] (0.center) to (1.center);
                \draw (0.center) to (15.center);
            \end{pgfonlayer}
        \end{tikzpicture}           
    }             
}

\newcommand{\TLcomonoid}[2]{
    \begin{tikzpicture}
        \begin{pgfonlayer}{nodelayer}
            \node [style=none] (0) at (-0.5, 0.25) {};
            \node [style=none] (1) at (0.5, 1.15) {};
            \node [style=none] (2) at (0.5, -0.65) {};
            \node [style=label] (4) at (1.575, 1.4) {$#2$};
            \node [style=label] (5) at (1.575, 0.65) {$#1$};
            \node [style=none] (6) at (-1.75, 0.65) {};
            \node [style=none] (7) at (-1.75, -0.65) {};
            \node [style=none] (8) at (1, -1.55) {};
            \node [style=none] (9) at (1, 1.55) {};
            \node [style=none] (10) at (1, 0.25) {};
            \node [style=none] (12) at (1, -0.25) {};
            \node [style=none] (13) at (-0.75, -0.65) {};
            \node [style=none] (14) at (-0.75, 0.65) {};
            \node [style=none] (15) at (-1.75, 0.25) {};
            \node [style=none] (16) at (-0.5, -0.25) {};
            \node [style=none] (17) at (0.5, 0.65) {};
            \node [style=none] (18) at (0.5, -1.15) {};
            \node [style=none] (19) at (-1.75, -0.25) {};
            \node [style=label] (20) at (1.575, -0.65) {$#2$};
            \node [style=label] (21) at (1.575, -1.4) {$#1$};
            \node [style=label] (22) at (-2.325, 0.5) {$#2$};
            \node [style=label] (23) at (-2.325, -0.25) {$#1$};
            \node [style=none] (24) at (0.25, 0.25) {};
            \node [style=none] (25) at (0.25, -0.25) {};
            \node [style=none] (26) at (-0.025, 0) {};
            \node [style=none] (27) at (1, 1.15) {};
            \node [style=none] (28) at (1, -0.65) {};
            \node [style=none] (29) at (1, 0.65) {};
            \node [style=none] (30) at (1, -1.15) {};
            \node [style=none] (31) at (0.475, 1.55) {};
            \node [style=none] (32) at (0.575, -1.55) {};
        \end{pgfonlayer}
        \begin{pgfonlayer}{edgelayer}
            \draw [style=tape] (6.center)
                 to (7.center)
                 to (13.center)
                 to [bend right] (32.center)
                 to (8.center)
                 to (12.center)
                 to (25.center)
                 to [bend left=45] (26.center)
                 to [bend left=45] (24.center)
                 to (10.center)
                 to (9.center)
                 to (31.center)
                 to [bend right] (14.center)
                 to cycle;
            \draw [bend right=45] (0.center) to (2.center);
            \draw [bend left=45] (0.center) to (1.center);
            \draw (0.center) to (15.center);
            \draw [bend right=45] (16.center) to (18.center);
            \draw [bend left=45] (16.center) to (17.center);
            \draw (16.center) to (19.center);
            \draw (27.center) to (1.center);
            \draw (17.center) to (29.center);
            \draw (28.center) to (2.center);
            \draw (18.center) to (30.center);
        \end{pgfonlayer}
    \end{tikzpicture}                    
}

\newcommand{\Tcounit}[1]{
    \begin{tikzpicture}
        \begin{pgfonlayer}{nodelayer}
            \node [style=none] (0) at (1.25, 0) {};
            \ifthenelse{\isempty{#1}}{}{
                \node [style=label] (3) at (-0.75, 0) {$#1$};
            }
            \node [style=none] (6) at (-0.25, 0.5) {};
            \node [style=none] (7) at (-0.25, -0.5) {};
            \node [style=none] (13) at (0.75, -0.5) {};
            \node [style=none] (14) at (0.75, 0.5) {};
            \node [style=none] (15) at (-0.25, 0) {};
        \end{pgfonlayer}
        \begin{pgfonlayer}{edgelayer}
            \path [style=tape] (13.center)
                 to [bend right=45] (0.center)
                 to [bend right=45] (14.center)
                 to (6.center)
                 to (7.center)
                 to cycle;
            \draw (0.center) to (15.center);
        \end{pgfonlayer}
    \end{tikzpicture}      
}

\newcommand{\Tmonoid}[1]{
    \ifthenelse{\equal{#1}{\uno}}{
        \begin{tikzpicture}
            \begin{pgfonlayer}{nodelayer}
                \node [style=none] (1) at (-1, 1) {};
                \node [style=none] (2) at (-1, -1) {};
                \node [style=none] (6) at (1.5, 0.5) {};
                \node [style=none] (7) at (1.5, -0.5) {};
                \node [style=none] (8) at (-1, -1.5) {};
                \node [style=none] (9) at (-1, 1.5) {};
                \node [style=none] (10) at (-1, 0.5) {};
                \node [style=none] (11) at (-0.5, 0) {};
                \node [style=none] (12) at (-1, -0.5) {};
                \node [style=none] (13) at (0.5, -0.5) {};
                \node [style=none] (14) at (0.5, 0.5) {};
                \node [style=none] (15) at (1.5, 0) {};
            \end{pgfonlayer}
            \begin{pgfonlayer}{edgelayer}
                \path [style=tape] (9.center)
                     to [bend left] (14.center)
                     to (6.center)
                     to (7.center)
                     to (13.center)
                     to [bend left] (8.center)
                     to (12.center)
                     to [bend right=45] (11.center)
                     to [bend right=45] (10.center)
                     to cycle;
            \end{pgfonlayer}
        \end{tikzpicture}         
    }
    {
        \begin{tikzpicture}
            \begin{pgfonlayer}{nodelayer}
                \node [style=none] (0) at (0, 0) {};
                \node [style=none] (1) at (-1, 1) {};
                \node [style=none] (2) at (-1, -1) {};
                \ifthenelse{\isempty{#1}}{}{
                    \node [style=label] (3) at (2, 0) {$#1$};
                    \node [style=label] (4) at (-1.5, 1) {$#1$};
                    \node [style=label] (5) at (-1.5, -1) {$#1$};
                }
                \node [style=none] (6) at (1.5, 0.5) {};
                \node [style=none] (7) at (1.5, -0.5) {};
                \node [style=none] (8) at (-1, -1.5) {};
                \node [style=none] (9) at (-1, 1.5) {};
                \node [style=none] (10) at (-1, 0.5) {};
                \node [style=none] (11) at (-0.5, 0) {};
                \node [style=none] (12) at (-1, -0.5) {};
                \node [style=none] (13) at (0.5, -0.5) {};
                \node [style=none] (14) at (0.5, 0.5) {};
                \node [style=none] (15) at (1.5, 0) {};
            \end{pgfonlayer}
            \begin{pgfonlayer}{edgelayer}
                \path [style=tape] (9.center)
                     to [bend left] (14.center)
                     to (6.center)
                     to (7.center)
                     to (13.center)
                     to [bend left] (8.center)
                     to (12.center)
                     to [bend right=45] (11.center)
                     to [bend right=45] (10.center)
                     to cycle;
                \draw [bend left=45] (0.center) to (2.center);
                \draw [bend right=45] (0.center) to (1.center);
                \draw (0.center) to (15.center);
            \end{pgfonlayer}
        \end{tikzpicture}          
    }  
}

\newcommand{\Tunit}[1]{
    \begin{tikzpicture}
        \begin{pgfonlayer}{nodelayer}
            \node [style=none] (0) at (-1.25, 0) {};
            \ifthenelse{\isempty{#1}}{}{
                \node [style=label] (3) at (0.75, 0) {$#1$};
            }
            \node [style=none] (6) at (0.25, 0.5) {};
            \node [style=none] (7) at (0.25, -0.5) {};
            \node [style=none] (13) at (-0.75, -0.5) {};
            \node [style=none] (14) at (-0.75, 0.5) {};
            \node [style=none] (15) at (0.25, 0) {};
        \end{pgfonlayer}
        \begin{pgfonlayer}{edgelayer}
            \path [style=tape] (13.center)
                 to [bend left=45] (0.center)
                 to [bend left=45] (14.center)
                 to (6.center)
                 to (7.center)
                 to cycle;
            \draw (0.center) to (15.center);
        \end{pgfonlayer}
    \end{tikzpicture}      
}

\newcommand{\Tsymmp}[2]{
    \begin{tikzpicture}
        \begin{pgfonlayer}{nodelayer}
            \node [style=label] (3) at (-1.5, 1) {$#1$};
            \node [style=none] (6) at (-1, -0.5) {};
            \node [style=none] (7) at (-1, -1.5) {};
            \node [style=none] (15) at (-1, -1) {};
            \node [style=none] (17) at (-1, 1.5) {};
            \node [style=none] (18) at (-1, 0.5) {};
            \node [style=none] (21) at (-1, 1) {};
            \node [style=none] (22) at (2.5, 1) {};
            \node [style=none] (25) at (2.5, 0.5) {};
            \node [style=none] (26) at (2.5, 1.5) {};
            \node [style=none] (28) at (2.5, -1) {};
            \node [style=none] (31) at (2.5, -1.5) {};
            \node [style=none] (32) at (2.5, -0.5) {};
            \node [style=label] (34) at (3, -1) {$#1$};
            \node [style=label] (35) at (3, 1) {$#2$};
            \node [style=label] (36) at (-1.5, -1) {$#2$};
        \end{pgfonlayer}
        \begin{pgfonlayer}{edgelayer}
            \path [style=tapeFill] (26.center)
                 to [in=0, out=-180, looseness=0.75] (6.center)
                 to (7.center)
                 to [in=-180, out=0, looseness=0.75] (25.center)
                 to cycle;
            \path [style=tapeFill] (18.center)
                 to (17.center)
                 to [in=180, out=0, looseness=0.75] (32.center)
                 to (31.center)
                 to [in=0, out=-180, looseness=0.75] cycle;
            \path [style=tapeNoFill] (26.center)
                 to [in=0, out=-180, looseness=0.75] (6.center)
                 to (7.center)
                 to [in=-180, out=0, looseness=0.75] (25.center)
                 to cycle;
            \path [style=tapeNoFill] (18.center)
                 to (17.center)
                 to [in=180, out=0, looseness=0.75] (32.center)
                 to (31.center)
                 to [in=0, out=-180, looseness=0.75] cycle;
            \draw [in=-180, out=0, looseness=0.75] (15.center) to (22.center);
            \draw [in=180, out=0, looseness=0.75] (21.center) to (28.center);
        \end{pgfonlayer}
    \end{tikzpicture}    
}

\newcommand{\Twire}[1]{
    \ifthenelse{\equal{#1}{\uno}}{
        \begin{tikzpicture}
            \begin{pgfonlayer}{nodelayer}
                \node [style=none] (0) at (-1.5, 0.5) {};
                \node [style=none] (3) at (-1.5, -0.5) {};
                \node [style=none] (4) at (1.5, -0.5) {};
                \node [style=none] (5) at (1.5, 0.5) {};
                \node [style=none] (6) at (-1.5, 0) {};
                \node [style=none] (7) at (1.5, 0) {};
            \end{pgfonlayer}
            \begin{pgfonlayer}{edgelayer}
                \path [style=tape] (5.center)
                     to (4.center)
                     to (3.center)
                     to (0.center)
                     to cycle;
            \end{pgfonlayer}
        \end{tikzpicture}        
    }
    {
        \begin{tikzpicture}
            \begin{pgfonlayer}{nodelayer}
                \node [style=none] (0) at (-1.5, 0.5) {};
                \node [style=none] (3) at (-1.5, -0.5) {};
                \node [style=none] (4) at (1.5, -0.5) {};
                \node [style=none] (5) at (1.5, 0.5) {};
                \node [style=none] (6) at (-1.5, 0) {};
                \node [style=none] (7) at (1.5, 0) {};
                \node [style=label] (8) at (-2, 0) {$#1$};
                \node [style=label] (9) at (2, 0) {$#1$};
            \end{pgfonlayer}
            \begin{pgfonlayer}{edgelayer}
                \path [style=tape] (5.center)
                        to (4.center)
                        to (3.center)
                        to (0.center)
                        to cycle;
                \draw (6.center) to (7.center);
            \end{pgfonlayer}
        \end{tikzpicture}        
    }   
}

\newcommand{\TLwire}[2]{
    \begin{tikzpicture}
        \begin{pgfonlayer}{nodelayer}
            \node [style=label] (3) at (-2, -0.25) {$#1$};
            \node [style=none] (17) at (-1.5, 0.75) {};
            \node [style=none] (18) at (-1.5, -0.75) {};
            \node [style=none] (21) at (-1.5, -0.25) {};
            \node [style=none] (40) at (1.5, 0.75) {};
            \node [style=none] (41) at (1.5, -0.75) {};
            \node [style=none] (42) at (1.5, -0.25) {};
            \node [style=none] (46) at (0, -0.25) {};
            \node [style=none] (47) at (0, -0.75) {};
            \node [style=none] (48) at (0, 0.75) {};
            \node [style=label] (49) at (2, -0.25) {$#1$};
            \node [style=label] (50) at (-2, 0.5) {$#2$};
            \node [style=none] (51) at (-1.5, 0.25) {};
            \node [style=none] (52) at (1.5, 0.25) {};
            \node [style=none] (53) at (0, 0.25) {};
            \node [style=label] (54) at (2, 0.5) {$#2$};
        \end{pgfonlayer}
        \begin{pgfonlayer}{edgelayer}
            \path [style=tape] (41.center)
                 to [in=0, out=180, looseness=0.75] (47.center)
                 to [in=0, out=-180, looseness=0.75] (18.center)
                 to (17.center)
                 to [in=180, out=0, looseness=0.75] (48.center)
                 to [in=180, out=0, looseness=0.75] (40.center)
                 to cycle;
            \draw (42.center) to (46.center);
            \draw (21.center) to (46.center);
            \draw (52.center) to (53.center);
            \draw (51.center) to (53.center);
        \end{pgfonlayer}
    \end{tikzpicture}    
}

\newcommand{\TRwire}[2]{
    \TLwire{#2}{#1}
}

\newcommand{\Tcirc}[3]{
    \begin{tikzpicture}
        \begin{pgfonlayer}{nodelayer}
            \node [style=none] (6) at (-1.5, 0.75) {};
            \node [style=none] (7) at (-1.5, -0.75) {};
            \node [style=none] (15) at (-1.5, 0) {};
            \node [style=none] (22) at (1.5, 0) {};
            \node [style=none] (25) at (1.5, -0.75) {};
            \node [style=none] (26) at (1.5, 0.75) {};
            \node [style=label] (35) at (2, 0) {$#3$};
            \node [style=label] (36) at (-2, 0) {$#2$};
            \node [style=none] (37) at (-0.5, 0.5) {};
            \node [style=none] (38) at (-0.5, -0.5) {};
            \node [style=none] (39) at (-0.5, 0) {};
            \node [style=none] (40) at (0.5, 0) {};
            \node [style=none] (41) at (0.5, -0.5) {};
            \node [style=none] (42) at (0.5, 0.5) {};
            \node [style=none] (43) at (0, 0) {$#1$};
        \end{pgfonlayer}
        \begin{pgfonlayer}{edgelayer}
            \path [style=tape] (26.center)
                 to (6.center)
                 to (7.center)
                 to (25.center)
                 to cycle;
            \draw [fill=white] (42.center)
                 to (37.center)
                 to (38.center)
                 to (41.center)
                 to cycle;
            \draw (15.center) to (39.center);
            \draw (40.center) to (22.center);
        \end{pgfonlayer}
    \end{tikzpicture}      
}

\newcommand{\Cgen}[3]{
    \begin{tikzpicture}
        \begin{pgfonlayer}{nodelayer}
            \ifthenelse{\isempty{#2}}{
                \node [style=none] (15) at (-1.25, 0) {};
                \node [style=none] (22) at (1.25, 0) {};
            }{
                \node [style=none] (15) at (-1.5, 0) {};
                \node [style=none] (22) at (1.5, 0) {};
                \node [style=label] (35) at (2, 0) {$#3$};
                \node [style=label] (36) at (-2, 0) {$#2$};
            }     
            \node [style=none] (37) at (-0.5, 0.5) {};
            \node [style=none] (38) at (-0.5, -0.5) {};
            \node [style=none] (39) at (-0.5, 0) {};
            \node [style=none] (40) at (0.5, 0) {};
            \node [style=none] (41) at (0.5, -0.5) {};
            \node [style=none] (42) at (0.5, 0.5) {};
            \node [style=none] (43) at (0, 0) {$#1$};
        \end{pgfonlayer}
        \begin{pgfonlayer}{edgelayer}
            \draw [fill=white] (42.center)
                 to (37.center)
                 to (38.center)
                 to (41.center)
                 to cycle;
            \draw (15.center) to (39.center);
            \draw (40.center) to (22.center);
        \end{pgfonlayer}
    \end{tikzpicture}      
}

\newcommand{\Csymm}[2]{
    \begin{tikzpicture}
        \begin{pgfonlayer}{nodelayer}
            \node [style=label] (8) at (1.75, 0.5) {$#2$};
            \node [style=label] (11) at (-1.75, 0.5) {$#1$};
            \node [style=label] (12) at (-1.75, -0.5) {$#2$};
            \node [style=label] (13) at (1.75, -0.5) {$#1$};
        \end{pgfonlayer}
        \begin{pgfonlayer}{edgelayer}
            \draw [in=-180, out=0, looseness=1.25] (11) to (13);
            \draw [in=-180, out=0, looseness=1.25] (12) to (8);
        \end{pgfonlayer}
    \end{tikzpicture}             
}

\newcommand{\Creg}[2]{
    \begin{tikzpicture}
        \begin{pgfonlayer}{nodelayer}
            \node [style=reg] (0) at (0, 0) {$#1$};
            \node [style=none] (1) at (1.5, 0) {};
            \node [style=none] (2) at (-1.5, 0) {};
            \ifthenelse{\equal{#2}{\uno}}{}{
            \node [style=label] (3) at (-2.5, 0) {$#2$};
            \node [style=label] (4) at (2.5, 0) {$#2$};
            }
        \end{pgfonlayer}
        \begin{pgfonlayer}{edgelayer}
            \draw (2.center) to (0);
            \draw (0) to (1.center);
        \end{pgfonlayer}
    \end{tikzpicture}       
}

\newcommand{\Ccoreg}[2]{
    \begin{tikzpicture}
        \begin{pgfonlayer}{nodelayer}
            \node [style=coreg] (0) at (0, 0) {$#1$};
            \node [style=none] (1) at (1.5, 0) {};
            \node [style=none] (2) at (-1.5, 0) {};
            \ifthenelse{\equal{#2}{\uno}}{}{
            \node [style=label] (3) at (-2.5, 0) {$#2$};
            \node [style=label] (4) at (2.5, 0) {$#2$};
            }
        \end{pgfonlayer}
        \begin{pgfonlayer}{edgelayer}
            \draw (2.center) to (0);
            \draw (0) to (1.center);
        \end{pgfonlayer}
    \end{tikzpicture}
}

\newcommand{\Cunit}[1]{
    \begin{tikzpicture}
        \begin{pgfonlayer}{nodelayer}
            \ifthenelse{\equal{#1}{\uno}}{}{
            \node [style=label] (4) at (1, 0) {$#1$};
            }
            \node [style=none] (11) at (0, 0) {};
            \node [style=black] (12) at (-1, 0) {};
        \end{pgfonlayer}
        \begin{pgfonlayer}{edgelayer}
            \draw (12) to (11.center);
        \end{pgfonlayer}
    \end{tikzpicture}    
}

\newcommand{\Ccounit}[1]{
    \begin{tikzpicture}
        \begin{pgfonlayer}{nodelayer}
            \ifthenelse{\equal{#1}{\uno}}{}{
            \node [style=label] (4) at (-1, 0) {$#1$};
            }
            \node [style=none] (11) at (0, 0) {};
            \node [style=black] (12) at (1, 0) {};
        \end{pgfonlayer}
        \begin{pgfonlayer}{edgelayer}
            \draw (12) to (11.center);
        \end{pgfonlayer}
    \end{tikzpicture}    
}

\newcommand{\Tcodiagdiag}[1]{
    \ifthenelse{\equal{#1}{\uno}}{
        \begin{tikzpicture}
            \begin{pgfonlayer}{nodelayer}
                \node [style=none] (1) at (-0.75, 1.25) {};
                \node [style=none] (8) at (-0.75, -1.75) {};
                \node [style=none] (9) at (-0.75, 1.75) {};
                \node [style=none] (10) at (-0.75, 0.75) {};
                \node [style=none] (11) at (-1.5, 0) {};
                \node [style=none] (12) at (-0.75, -0.75) {};
                \node [style=none] (13) at (-2.5, -0.5) {};
                \node [style=none] (14) at (-2.5, 0.5) {};
                \node [style=none] (35) at (-3.75, -0.5) {};
                \node [style=none] (36) at (-3.75, 0.5) {};
                \node [style=none] (81) at (0.75, 1.25) {};
                \node [style=none] (83) at (0.75, -1.75) {};
                \node [style=none] (84) at (0.75, 1.75) {};
                \node [style=none] (85) at (0.75, 0.75) {};
                \node [style=none] (86) at (1.5, 0) {};
                \node [style=none] (87) at (0.75, -0.75) {};
                \node [style=none] (88) at (2.5, -0.5) {};
                \node [style=none] (89) at (2.5, 0.5) {};
                \node [style=none] (91) at (3.75, -0.5) {};
                \node [style=none] (92) at (3.75, 0.5) {};
            \end{pgfonlayer}
            \begin{pgfonlayer}{edgelayer}
                \draw [style=tape] (83.center)
                     to (8.center)
                     to [bend left] (13.center)
                     to (35.center)
                     to (36.center)
                     to (14.center)
                     to [bend left] (9.center)
                     to (84.center)
                     to [bend left] (89.center)
                     to (92.center)
                     to (91.center)
                     to (88.center)
                     to [bend left] cycle;
                \path [style=tapeNoFill, fill=white] (87.center)
                     to (12.center)
                     to [bend left=45] (11.center)
                     to [bend left=45] (10.center)
                     to (85.center)
                     to [bend left=45] (86.center)
                     to [bend left=45] cycle;
            \end{pgfonlayer}
        \end{tikzpicture}                 
    }
    {
        \begin{tikzpicture}
            \begin{pgfonlayer}{nodelayer}
                \node [style=none] (0) at (-2, 0) {};
                \node [style=none] (1) at (-0.75, 1.25) {};
                \node [style=none] (2) at (-0.75, -1.25) {};
                \node [style=none] (8) at (-0.75, -1.75) {};
                \node [style=none] (9) at (-0.75, 1.75) {};
                \node [style=none] (10) at (-0.75, 0.75) {};
                \node [style=none] (11) at (-1.5, 0) {};
                \node [style=none] (12) at (-0.75, -0.75) {};
                \node [style=none] (13) at (-2.5, -0.5) {};
                \node [style=none] (14) at (-2.5, 0.5) {};
                \node [style=none] (24) at (-3.75, 0) {};
                \node [style=none] (35) at (-3.75, -0.5) {};
                \node [style=none] (36) at (-3.75, 0.5) {};
                \node [style=none] (50) at (-0.75, 1.25) {};
                \node [style=none] (80) at (2, 0) {};
                \node [style=none] (81) at (0.75, 1.25) {};
                \node [style=none] (82) at (0.75, -1.25) {};
                \node [style=none] (83) at (0.75, -1.75) {};
                \node [style=none] (84) at (0.75, 1.75) {};
                \node [style=none] (85) at (0.75, 0.75) {};
                \node [style=none] (86) at (1.5, 0) {};
                \node [style=none] (87) at (0.75, -0.75) {};
                \node [style=none] (88) at (2.5, -0.5) {};
                \node [style=none] (89) at (2.5, 0.5) {};
                \node [style=none] (90) at (3.75, 0) {};
                \node [style=none] (91) at (3.75, -0.5) {};
                \node [style=none] (92) at (3.75, 0.5) {};
                \node [style=none] (93) at (0.75, 1.25) {};
                \node [style=label] (94) at (-4.25, 0) {$#1$};
                \node [style=label] (95) at (4.25, 0) {$#1$};
            \end{pgfonlayer}
            \begin{pgfonlayer}{edgelayer}
                \path [style=tape] (83.center)
                     to (8.center)
                     to [bend left] (13.center)
                     to (35.center)
                     to (36.center)
                     to (14.center)
                     to [bend left] (9.center)
                     to (84.center)
                     to [bend left] (89.center)
                     to (92.center)
                     to (91.center)
                     to (88.center)
                     to [bend left] cycle;
                \path [style=tapeNoFill, fill=white] (87.center)
                     to (12.center)
                     to [bend left=45] (11.center)
                     to [bend left=45] (10.center)
                     to (85.center)
                     to [bend left=45] (86.center)
                     to [bend left=45] cycle;
                \draw [bend right=45] (0.center) to (2.center);
                \draw [bend left=45] (0.center) to (1.center);
                \draw (24.center) to (0.center);
                \draw [bend left=45] (80.center) to (82.center);
                \draw [bend right=45] (80.center) to (81.center);
                \draw (90.center) to (80.center);
                \draw (50.center) to (93.center);
                \draw (82.center) to (2.center);
            \end{pgfonlayer}
        \end{tikzpicture}                   
    }     
}

\newcommand{\Tbox}[3]{
    \begin{tikzpicture}
        \begin{pgfonlayer}{nodelayer}
            \node [style=none] (35) at (-2, -0.5) {};
            \node [style=none] (36) at (-2, 0.5) {};
            \node [style=dots] (90) at (1.5, 0.025) {$\vdots$};
            \node [style=none] (91) at (2, -0.5) {};
            \node [style=none] (92) at (2, 0.5) {};
            \node [style=label] (94) at (-1.5, 1) {$#2$};
            \node [style=label] (95) at (1.5, 1) {$#3$};
            \node [style=none] (96) at (-0.75, -1) {};
            \node [style=none] (97) at (0.75, -1) {};
            \node [style=none] (98) at (0.75, 1) {};
            \node [style=none] (99) at (-0.75, 1) {};
            \node [style=none] (104) at (0, 0) {$#1$};
            \node [style=dots] (105) at (-1.5, 0.025) {$\vdots$};
            \node [style=none] (106) at (-0.75, -0.5) {};
            \node [style=none] (107) at (0.75, -0.5) {};
            \node [style=none] (108) at (0.75, 0.5) {};
            \node [style=none] (109) at (-0.75, 0.5) {};
        \end{pgfonlayer}
        \begin{pgfonlayer}{edgelayer}
            \draw [style=tape] (97.center)
                 to (96.center)
                 to (106.center)
                 to (35.center)
                 to (36.center)
                 to (109.center)
                 to (99.center)
                 to (98.center)
                 to (108.center)
                 to (92.center)
                 to (91.center)
                 to (107.center)
                 to cycle;
        \end{pgfonlayer}
    \end{tikzpicture}      
}

\newcommand{\Czero}[1]{
    \begin{tikzpicture}
        \begin{pgfonlayer}{nodelayer}
            \ifthenelse{\equal{#1}{\uno}}{}{
            \node [style=label] (4) at (1, 0) {$#1$};
            }
            \node [style=none] (11) at (0, 0) {};
            \node [style=white] (12) at (-1, 0) {};
        \end{pgfonlayer}
        \begin{pgfonlayer}{edgelayer}
            \draw (12) to (11.center);
        \end{pgfonlayer}
    \end{tikzpicture}    
}

\newcommand{\Cuno}[1]{
    \begin{tikzpicture}
        \begin{pgfonlayer}{nodelayer}
            \ifthenelse{\equal{#1}{\uno}}{}{
            \node [style=label] (4) at (-1, 0) {$#1$};
            }
            \node [style=none] (11) at (0, 0) {};
            \node [style=none] (12) at (1, 0) {};
            \node [style=none] (13) at (1, 0.25) {};
            \node [style=none] (14) at (1, -0.25) {};
        \end{pgfonlayer}
        \begin{pgfonlayer}{edgelayer}
            \draw (12.center) to (11.center);
            \draw (14.center) to (13.center);
        \end{pgfonlayer}
    \end{tikzpicture}    
}

\newcommand{\TPolyCounit}[1]{
    \begin{tikzpicture}
        \begin{pgfonlayer}{nodelayer}
            \node [style=label] (3) at (0, 1) {$#1$};
            \node [style=none] (23) at (-0.5, 0.5) {};
            \node [style=none] (24) at (-0.5, -0.5) {};
            \node [style=none] (25) at (0.75, -0.5) {};
            \node [style=none] (26) at (0.75, 0.5) {};
            \node [style=dots] (27) at (0, 0.025) {$\vdots$};
        \end{pgfonlayer}
        \begin{pgfonlayer}{edgelayer}
            \path [style=tape] (24.center)
                 to (25.center)
                 to [bend right=90, looseness=1.75] (26.center)
                 to (23.center)
                 to cycle;
        \end{pgfonlayer}
    \end{tikzpicture}    
}

\newcommand{\TPolyUnit}[1]{
    \begin{tikzpicture}
        \begin{pgfonlayer}{nodelayer}
            \node [style=label] (3) at (0.25, 1) {$#1$};
            \node [style=none] (23) at (0.75, 0.5) {};
            \node [style=none] (24) at (0.75, -0.5) {};
            \node [style=none] (25) at (-0.5, -0.5) {};
            \node [style=none] (26) at (-0.5, 0.5) {};
            \node [style=dots] (27) at (0.25, 0.025) {$\vdots$};
        \end{pgfonlayer}
        \begin{pgfonlayer}{edgelayer}
            \path [style=tape] (24.center)
                 to (25.center)
                 to [bend left=90, looseness=1.75] (26.center)
                 to (23.center)
                 to cycle;
        \end{pgfonlayer}
    \end{tikzpicture}        
}

\newcommand{\TPolyWire}[1]{
    \begin{tikzpicture}
        \begin{pgfonlayer}{nodelayer}
            \node [style=none] (0) at (-1.5, 0.5) {};
            \node [style=none] (3) at (-1.5, -0.5) {};
            \node [style=none] (4) at (1.5, -0.5) {};
            \node [style=none] (5) at (1.5, 0.5) {};
            \node [style=dots] (8) at (-1, 0.025) {$\vdots$};
            \node [style=dots] (9) at (1, 0.025) {$\vdots$};
            \node [style=label] (10) at (1, 1) {$#1$};
            \node [style=label] (11) at (-1, 1) {$#1$};
        \end{pgfonlayer}
        \begin{pgfonlayer}{edgelayer}
            \path [style=tape] (5.center)
                 to (4.center)
                 to (3.center)
                 to (0.center)
                 to cycle;
        \end{pgfonlayer}
    \end{tikzpicture}          
}

\newcommand{\TPolyDiag}[1]{
    \begin{tikzpicture}
        \begin{pgfonlayer}{nodelayer}
            \node [style=none] (8) at (0.75, -1.75) {};
            \node [style=none] (9) at (0.75, 1.75) {};
            \node [style=none] (10) at (0.75, 0.75) {};
            \node [style=none] (11) at (0, 0) {};
            \node [style=none] (12) at (0.75, -0.75) {};
            \node [style=none] (13) at (-1, -0.5) {};
            \node [style=none] (14) at (-1, 0.5) {};
            \node [style=none] (35) at (-2.25, -0.5) {};
            \node [style=none] (36) at (-2.25, 0.5) {};
            \node [style=label] (93) at (-1.5, 1) {$#1$};
            \node [style=dots] (96) at (-1.5, 0.025) {$\vdots$};
        \end{pgfonlayer}
        \begin{pgfonlayer}{edgelayer}
            \path[style=tape] (12.center)
                 to (8.center)
                 to [bend left] (13.center)
                 to (35.center)
                 to (36.center)
                 to (14.center)
                 to [bend left] (9.center)
                 to (10.center)
                 to [bend right=45] (11.center)
                 to [bend right=45] cycle;
        \end{pgfonlayer}
    \end{tikzpicture}    
}

\newcommand{\TPolyCodiag}[1]{
    \begin{tikzpicture}
        \begin{pgfonlayer}{nodelayer}
            \node [style=none] (8) at (-0.75, -1.75) {};
            \node [style=none] (9) at (-0.75, 1.75) {};
            \node [style=none] (10) at (-0.75, 0.75) {};
            \node [style=none] (11) at (0, 0) {};
            \node [style=none] (12) at (-0.75, -0.75) {};
            \node [style=none] (13) at (1, -0.5) {};
            \node [style=none] (14) at (1, 0.5) {};
            \node [style=none] (35) at (2.25, -0.5) {};
            \node [style=none] (36) at (2.25, 0.5) {};
            \node [style=label] (93) at (1.5, 1) {$#1$};
            \node [style=dots] (96) at (1.5, 0.025) {$\vdots$};
        \end{pgfonlayer}
        \begin{pgfonlayer}{edgelayer}
            \path[style=tape] (12.center)
                 to (8.center)
                 to [bend right] (13.center)
                 to (35.center)
                 to (36.center)
                 to (14.center)
                 to [bend right] (9.center)
                 to (10.center)
                 to [bend left=45] (11.center)
                 to [bend left=45] cycle;
        \end{pgfonlayer}
    \end{tikzpicture}    
}

\newcommand{\TLcounit}[2]{
    \begin{tikzpicture}
        \begin{pgfonlayer}{nodelayer}
            \node [style=label] (5) at (-1.75, -0.25) {$#1$};
            \node [style=none] (43) at (0.25, 0.65) {};
            \node [style=none] (45) at (0.25, -0.65) {};
            \node [style=none] (46) at (-1, 0.65) {};
            \node [style=none] (48) at (-1, -0.65) {};
            \node [style=label] (52) at (-1.75, 0.5) {$#2$};
            \node [style=none] (53) at (1, 0) {};
            \node [style=none] (54) at (-1, -0.25) {};
            \node [style=none] (55) at (-1, 0.25) {};
            \node [style=none] (56) at (0.925, 0.25) {};
            \node [style=none] (57) at (0.925, -0.25) {};
        \end{pgfonlayer}
        \begin{pgfonlayer}{edgelayer}
            \draw [style=tape] (53.center)
                 to [bend right=45] (43.center)
                 to (46.center)
                 to (48.center)
                 to (45.center)
                 to [bend right=45] cycle;
            \draw (55.center) to (56.center);
            \draw (57.center) to (54.center);
        \end{pgfonlayer}
    \end{tikzpicture}            
}

\newcommand{\TPolySymmp}[2]{
    \begin{tikzpicture}
        \begin{pgfonlayer}{nodelayer}
            \node [style=label] (3) at (-1.5, 2.25) {$#1$};
            \node [style=none] (6) at (-2, -0.75) {};
            \node [style=none] (7) at (-2, -1.75) {};
            \node [style=none] (17) at (-2, 1.75) {};
            \node [style=none] (18) at (-2, 0.75) {};
            \node [style=none] (25) at (2, 0.75) {};
            \node [style=none] (26) at (2, 1.75) {};
            \node [style=none] (31) at (2, -1.75) {};
            \node [style=none] (32) at (2, -0.75) {};
            \node [style=label] (34) at (1.5, -0.1) {$#1$};
            \node [style=label] (35) at (1.5, 2.25) {$#2$};
            \node [style=label] (36) at (-1.5, -0.1) {$#2$};
            \node [style=dots] (37) at (-1.5, 1.125) {$\vdots$};
            \node [style=dots] (38) at (1.5, 1.125) {$\vdots$};
            \node [style=dots] (39) at (1.5, -1.1) {$\vdots$};
            \node [style=dots] (40) at (-1.5, -1.1) {$\vdots$};
        \end{pgfonlayer}
        \begin{pgfonlayer}{edgelayer}
            \draw [style=tapeFill] (26.center)
                 to [in=0, out=-180, looseness=0.75] (6.center)
                 to (7.center)
                 to [in=-180, out=0, looseness=0.75] (25.center)
                 to cycle;
            \draw [style=tapeFill] (18.center)
                 to (17.center)
                 to [in=180, out=0, looseness=0.75] (32.center)
                 to (31.center)
                 to [in=0, out=-180, looseness=0.75] cycle;
            \draw [style=tapeNoFill] (26.center)
                 to [in=0, out=-180, looseness=0.75] (6.center)
                 to (7.center)
                 to [in=-180, out=0, looseness=0.75] (25.center)
                 to cycle;
            \draw [style=tapeNoFill] (18.center)
                 to (17.center)
                 to [in=180, out=0, looseness=0.75] (32.center)
                 to (31.center)
                 to [in=0, out=-180, looseness=0.75] cycle;
        \end{pgfonlayer}
    \end{tikzpicture}        
}

\newcommand{\Tcopier}[1]{\copier{#1}
}

\newcommand{\Tcocopier}[1]{
    \cocopier{#1}
}

\newcommand{\Tdischarger}[1]{\discharger{#1}
}

\newcommand{\Tcodischarger}[1]{\codischarger{#1}
}

\newcommand{\TRelCodiagDiag}[1]{\begin{tikzpicture}
	\begin{pgfonlayer}{nodelayer}
		\node [style=none] (0) at (-1.5, 0) {};
		\node [style=none] (1) at (-0.25, 1.25) {};
		\node [style=none] (2) at (-0.25, -1.25) {};
		\node [style=none] (8) at (-0.25, -1.75) {};
		\node [style=none] (9) at (-0.25, 1.75) {};
		\node [style=none] (10) at (-0.25, 0.75) {};
		\node [style=none] (11) at (-1, 0) {};
		\node [style=none] (12) at (-0.25, -0.75) {};
		\node [style=none] (13) at (-2, -0.5) {};
		\node [style=none] (14) at (-2, 0.5) {};
		\node [style=none] (24) at (-3, 0) {};
		\node [style=none] (35) at (-3, -0.5) {};
		\node [style=none] (36) at (-3, 0.5) {};
		\node [style=none] (50) at (-0.25, 1.25) {};
		\node [style=none] (80) at (1.5, 0) {};
		\node [style=none] (81) at (0.25, 1.25) {};
		\node [style=none] (82) at (0.25, -1.25) {};
		\node [style=none] (83) at (0.25, -1.75) {};
		\node [style=none] (84) at (0.25, 1.75) {};
		\node [style=none] (85) at (0.25, 0.75) {};
		\node [style=none] (86) at (1, 0) {};
		\node [style=none] (87) at (0.25, -0.75) {};
		\node [style=none] (88) at (2, -0.5) {};
		\node [style=none] (89) at (2, 0.5) {};
		\node [style=none] (90) at (3, 0) {};
		\node [style=none] (91) at (3, -0.5) {};
		\node [style=none] (92) at (3, 0.5) {};
		\node [style=none] (93) at (0.25, 1.25) {};
		\node [style=label] (94) at (-3.5, 0) {$#1$};
		\node [style=label] (95) at (3.5, 0) {$#1$};
	\end{pgfonlayer}
	\begin{pgfonlayer}{edgelayer}
		\draw [style=tape] (83.center)
			 to (8.center)
			 to [bend left] (13.center)
			 to (35.center)
			 to (36.center)
			 to (14.center)
			 to [bend left] (9.center)
			 to (84.center)
			 to [bend left] (89.center)
			 to (92.center)
			 to (91.center)
			 to (88.center)
			 to [bend left] cycle;
		\draw [style=tapeNoFill, fill=white] (87.center)
			 to (12.center)
			 to [bend left=45] (11.center)
			 to [bend left=45] (10.center)
			 to (85.center)
			 to [bend left=45] (86.center)
			 to [bend left=45] cycle;
		\draw [bend right=45] (0.center) to (2.center);
		\draw [bend left=45] (0.center) to (1.center);
		\draw (24.center) to (0.center);
		\draw [bend left=45] (80.center) to (82.center);
		\draw [bend right=45] (80.center) to (81.center);
		\draw (90.center) to (80.center);
		\draw (50.center) to (93.center);
		\draw (82.center) to (2.center);
	\end{pgfonlayer}
\end{tikzpicture}
}

\newcommand{\st}[2]{(#1 \mid #2)}

\NewDocumentCommand{\TTrace}{gooo}{
    
}

\title{A Diagrammatic Algebra for Program Logics}

\author[1]{Filippo Bonchi}
\author[2]{Alessandro Di Giorgio}
\author[1]{Elena Di Lavore}

\affil[1]{\small University of Pisa, Italy}
\affil[2]{\small University College London, United Kingdom}

\date{}

\begin{document}

\maketitle

\begin{abstract}
        Tape diagrams provide a convenient  notation for arrows of rig categories, i.e., categories equipped with two monoidal products, $\piu$ and $\per$, where $\per$ distributes over $\piu$.
        In this work, we extend  tape diagrams with traces over $\piu$ in order to deal with iteration in imperative programming languages.
        More precisely, we introduce Kleene-Cartesian bicategories, namely rig categories where the monoidal structure provided by $\per$ is a cartesian bicategory, while  the one provided by $\piu$ is  what we name a Kleene bicategory.
        We show that the associated language of tape diagrams is expressive enough to deal with imperative programs and the corresponding laws provide a proof system that is at least as powerful as the one of Hoare logic.        
\end{abstract}

\section{Introduction}\label{sec:intro}
In recent years, there has been a growing interest in using monoidal categories to model various types of systems \cite{coecke2011interacting,Backens-ZXcompleteness1,Fritz_stochasticmatrices,bruni2011connector,BaezCoya-propsnetworktheory,bonchi2015full,DBLP:journals/pacmpl/BonchiHPSZ19,BonchiPSZ19,Piedeleu2021}. However, rig categories \cite{laplaza_coherence_1972} --categories equipped with two monoidal products, $\piu$ and $\per$, where $\per$ distributes over $\piu$-- have been far less studied.

In this paper, we propose using rig categories as a foundation for programming languages, particularly for imperative programs and their associated program logics \cite{hoare1969axiomatic,kozen00hoarekleene,DBLP:conf/vmcai/CousotCFL13,o2019incorrectness,DBLP:journals/corr/abs-2310-18156}. The key insight is that $\per$ provides the necessary structure for \emph{data flow}, while $\piu$ is suited for \emph{control flow}.

This observation has been recognised at least since \cite{bainbridge1976feedback}, but the idea of capturing the interaction between data and control flow through the laws of rig categories has not been widely explored. This is likely because rig categories do not offer a straightforward framework as monoidal categories: coherence and strictification are far more complex \cite{laplaza_coherence_1972,johnson2021bimonoidal} and, unlike monoidal categories, which benefit from \emph{string diagrams} that completely embody their laws \cite{joyal1991geometry}, analogous diagrammatic notation for rig categories have been proposed only recently \cite{comfort2020sheet,bonchi2023deconstructing}.

In this paper, we adopt the diagrammatic notation introduced in \cite{bonchi2023deconstructing}: \emph{tape diagrams}. Unlike sheet diagrams \cite{comfort2020sheet}, which use three dimensions to represent the three compositions ($\piu$, $\per$, and $;$), tape diagrams are drawn in two dimensions, making them more intuitive and easier to visualise. This notation captures the laws of rig categories when $\piu$ represents a product, coproduct, or both, i.e., a \emph{biproduct}. Specifically, when $\piu$ is a biproduct, tape diagrams offer a universal language, meaning that the category of tape diagrams is the one freely generated from an \emph{arbitrary} rig signature.

\medskip

Our first contribution is the extension of tape diagrams with \emph{traces} \cite{Joyal_tracedcategories} over the monoidal product $\piu$. Such traces are essential for modelling \emph{iteration} in imperative programming languages.

Of particular interest is the fact that, to achieve this result, the trace must be assumed to be \emph{uniform} \cite{cuazuanescu1994feedback,hasegawa2003uniformity}, a property that we need for technical reasons but, as we will show in the rest of the paper, it is crucial for recovering the complete axiomatisation of Kleene algebras \cite{Kozen94acompleteness}, proving the induction law of Peano's axiomatisation of natural numbers and support proofs by invariants on imperative programs.

Once we have developed a comfortable diagrammatic notation, we move toward the modelling of imperative programming languages and their logics. Inspired by an early work on program logics \cite{pratt1976semantical} that exploits the calculus of relations \cite{tarski1941calculus}, we 
fund our approach on the rig category of sets and relations $\Rel$, where $\per$ is the cartesian product of sets and $\piu$ their disjoint union.  In other words, our effort can be described as identifying the categorical structure of $\Rel$ that is sufficient for dealing with program logics. Such structure can be succinctly described as a rig category where the monoidal category of $\per$ is a \emph{cartesian bicategory} \cite{Carboni1987}, while the monoidal structure of $\piu$ is what we name a \emph{Kleene bicategory}.

\medskip

In a nutshell, a Kleene bicategory is a poset-enriched traced monoidal category where the monoidal product $\piu$ is a biproduct, and the induced \cite{fox1976coalgebras} natural monoid is \emph{left adjoint} to the natural comonoid. As expected, the trace must be uniform, but uniformity now has to be strengthened to take care of the poset-enrichment. The name ``Kleene'' is justified by the fact that every Kleene bicategory is a (typed) Kleene algebra in the sense of Kozen \cite{Kozen94acompleteness,kozen98typedkleene}, while any Kleene algebra gives rise, through the matrix construction (also known as biproduct completion \cite{mclane}), to a Kleene bicategory. While uniform traces over biproduct categories have been widely studied (see e.g. \cite{cuazuanescu1994feedback}), to the best of our knowledge, the adjointness condition of (co)monoids in this contexts is novel, as well as the correspondence with Kozen's axiomatisation \cite{Kozen94acompleteness}. This is the second contribution of our work.

\medskip

We specialise tape diagrams with traces to Kleene-Cartesian rig categories ($\piu$ forms a Kleene bicategory, while $\per$ a cartesian bicategory) and we introduce the notion of \emph{Kleene-Cartesian theory}, shortly, a signature and a set of axioms amongst Kleene-Cartesian tapes. Analogously to Lawvere's functorial semantics \cite{LawvereOriginalPaper}, \emph{models} are morphisms from the corresponding category of tape diagrams to an arbitrary Kleene-Cartesian rig category. As an example, we illustrate the Kleene-Cartesian theory of Peano's natural numbers: all models in $\Rel$ are isomorphic to the one of natural numbers and, the associated tape diagrams are expressive enough to deal with Turing equivalent imperative programs. We conclude by illustrating how a simple imperative programming language can be encoded within tapes and that the laws of tapes provide a proof system that is at least as powerful as the rules of Hoare logic. This is our last contribution.

\paragraph{Structure of the paper.} We commence in Section \ref{sec:monoidal} by recalling monoidal categories and string diagrams; moreover, we show that $\Rel$ carries two monoidal categories, one  is a finite biproduct category (Definition \ref{def:biproduct category}) and one is a cartesian bicategory (Definition \ref{def:cartesian bicategory}). In Section \ref{sec:rigcategories}, we recall rig categories from \cite{bonchi2023deconstructing}, finite biproduct rig categories and the associated language of tape diagrams. In Section \ref{sec:traced}, we recall traced monoidal categories and the notion of uniform trace. We show that, from a monoidal category, one can always freely generate a uniformly traced one (Theorem \ref{th:free-uniform-trace}) and that such construction restricts to finite biproduct rig categories (Proposition \ref{prop:free-uniform-trace-fb-rig}). The latter result is crucial, in Section \ref{sec:tapes}, to extend tape diagrams with uniform traces and to prove that their category is a freely generated one (Theorem \ref{thm:freeut-fb}). Our second key contribution --Klenee bicategories-- is illustrated in  Section \ref{sec:kleene}: we show that any Kleene bicategory is a typed Kleene algebra (Corollary \ref{cor:kleeneareka}) and that, from a typed Kleene algebra, one can freely generate a Kleene bicategory by means of the finite biproduct construction (Corollary \ref{cor:adjKleene}). In Section \ref{sec:kleene-tapes}, we introduce Kleene tapes for rig categories, where $\piu$ carries the structure of a Kleene bicategory. In Section \ref{sec:cb}, we introduce Kleene-Cartesian rig categories (Definition \ref{def:kcrig}) and the corresponding tape diagrams, and we prove that they form a freely generated Kleene-Cartesian rig category (Theorem \ref{theorem:KTCB is kleene-cartesian}). We also introduce the notion of Kleene-Cartesian theory and we show that models are in one to one correspondence with functors (Proposition \ref{funct:sem}).
As an example, we introduce the Kleene-Cartesian theory of Peano: we show that three simple axioms amongst tapes are equivalent to Peano's axiomatisation of natural numbers (Theorem \ref{thm:induction} and Lemma \ref{lemma:Peano}). Finally, in Section \ref{sec:hoare} we show how to encode imperative programs into tape diagrams and that any Hoare triple provable through the usual proof system of Hoare logic is also provable by means of the laws of Kleene-Cartesian bicategories (Proposition \ref{prop:Hoare}). The appendices contain the missing proofs and some coherence conditions for the various encountered structures.

\paragraph{Acknowledgement} The authors would like to acknowledge Alessio Santamaria, Chad Nester and the students of the ACT school 2022 for several useful discussions at early stage of this project. Gheorghe Stefanescu and Dexter Kozen provided some wise feedback and offered some guidance through the rather wide literature.

\section{Monoidal Categories and String Diagrams}\label{sec:monoidal}
We begin our exposition by regarding string diagrams \cite{joyal1991geometry,selinger2010survey} as terms of a typed language. Given a set $\sort$ of basic \emph{sorts}, hereafter denoted by $A,B\dots$, types are elements of $\sort^\star$, i.e.\ words over $\sort$. Terms are defined by the following context free grammar
\begin{equation}\label{eq:syntaxsymmetricstrict}
\begin{array}{rcl}
f & ::=& \; \id{A} \; \mid \; \id{I} \; \mid \; \gen \; \mid \; \sigma_{A,B}^{\perG} \; \mid \;   f ; f   \; \mid \;  f \perG f \\
\end{array}
\end{equation}  
where $s$ belongs to a fixed set $\sign$ of \emph{generators} and $I$ is the empty word. Each $s\in \sign$ comes with two types: arity $\ar(s)$ and coarity $\coar(s)$. The tuple $(\sort, \sign, \ar, \coar)$, $\sign$ for short, forms a \emph{monoidal signature}. Amongst the terms generated by  \eqref{eq:syntaxsymmetricstrict}, we consider only those that can be typed according to the inference rules in Table \ref{fig:freestricmmoncatax}. String diagrams are such terms modulo the axioms in Table \ref{fig:freestricmmoncatax} where,  for any $X,Y\in \sort^\star$,  $\id{X}$ and $\sigma_{X,Y}^{\perG}$ can be easily built using $\id I$, $\id{A}$, $\sigma_{A,B}^{\perG}$, $\perG$ and $;$ (see e.g.~\cite{ZanasiThesis}).
\begin{table}[t]
\scriptsize{
\begin{center}
\begin{tabular}{c  c}
\begin{tabular}{c}
    \toprule
Objects ($A\in \sort$)\\
\midrule
$X \; ::=\; \; A \; \mid \; \unoG \; \mid \;  X \perG X \vphantom{\sigma_{A,B}^{\perG}}$\\
\midrule
\makecell{
    \\[-2pt] $(X\perG Y)\perG Z=X \perG (Y \perG Z)$ \\ $X \perG \unoG = X $ \\ $\unoG \perG X = X$ \\[2pt]
} \\[13pt]
\bottomrule
\end{tabular}
&
\begin{tabular}{cc}
\toprule
\multicolumn{2}{c}{Arrows ($A\in \sort$, $s\in \sign$)} \\
\midrule
\multicolumn{2}{c}{$f \; ::=\; \; \id{A} \; \mid \; \id{\unoG} \; \mid \; \gen  \; \mid \;   f ; f   \; \mid \;  f \perG f \; \mid \; \sigma_{A,B}^{\perG}$} \\
\midrule
$(f;g);h=f;(g;h)$ & $id_X;f=f=f;id_Y$\\

\multicolumn{2}{c}{$(f_1\perG f_2) ; (g_1 \perG g_2) = (f_1;g_1) \perG (f_2;g_2)$} \\

$id_{\unoG}\perG f = f = f \perG id_{\unoG}$ & $(f \perG g)\, \perG h = f \perG \,(g \perG h)$ \\

$\sigma_{A, B}^{\perG}; \sigma_{B, A}^{\perG}= id_{A \perG B}$ & $(\gen \perG id_Z) ; \sigma_{Y, Z}^{\perG} = \sigma_{X,Z}^{\perG} ; (id_Z \perG \gen)$ \\
\bottomrule
\end{tabular}
\end{tabular}

\vspace{1em}

\begin{tabular}{c}
    \toprule
    Typing rules \\
    \midrule
    $
    {id_A \colon A \!\to\! A} \qquad  {id_\unoG \colon \unoG \!\to\! \unoG} \qquad {\sigma_{A, B}^{\perG} \colon A \perG B \!\to\! B \perG A} \qquad 
        \inferrule{\gen \colon \ar(s) \!\to\! \coar(s) \in \sign}{\gen \colon \ar(s) \!\to\! \coar(s)} \qquad
        \inferrule{f \colon X_1 \!\to\! Y_1 \and g \colon X_2 \!\to\! Y_2}{f \perG g \colon X_1 \perG X_2 \!\to\! Y_1 \perG Y_2}  \qquad
        \inferrule{f \colon X \!\to\! Y \and g \colon Y \!\to\! Z}{f ; g \colon X \!\to\! Z}
    $\\    
\bottomrule  
\end{tabular}
\end{center}
 \caption{Axioms for $\CatString$}
    \label{fig:freestricmmoncatax}
    }
\end{table}

String diagrams enjoy an elegant graphical representation: a generator $\gen$ in $\sign$ with arity $X$ and coarity $Y$ is depicted as  a \emph{box} having \emph{labelled wires} on the left and on the right representing, respectively, the words $X$ and $Y$. For instance $\gen \colon AB \to C$ in $\sign$ is depicted as the leftmost diagram below. Moreover, $\id{A}$ is displayed as one wire,  $id_{\unoG} $ as the empty diagram and $\sigma_{A,B}^{\perG}$ as a crossing:
\[
    \InputIfFileExists{generator.tikz}{}{\input{./tikz/generator.tikz}}
 \qquad \qquad 
    \InputIfFileExists{id.tikz}{}{\input{./tikz/id.tikz}}
 \qquad  \qquad     
    \InputIfFileExists{empty.tikz}{}{\input{./tikz/empty.tikz}}
 \qquad  \qquad   
    \InputIfFileExists{symm.tikz}{}{\input{./tikz/symm.tikz}}
\]
Finally, composition $f;g$ is represented by connecting the right wires 
of $f$ with the left wires of $g$ when their labels match, 
while the monoidal product $f \perG g$ is depicted by stacking the corresponding 
diagrams on top of each other: \[
    \InputIfFileExists{seq_comp.tikz}{}{\input{./tikz/seq_comp.tikz}}
 \qquad \qquad \qquad  
    \InputIfFileExists{par_comp.tikz}{}{\input{./tikz/par_comp.tikz}}
 \]
The first three rows of axioms for arrows in Table~\ref{fig:freestricmmoncatax}
are implicit in the 
graphical representation while the axioms in the last row  are displayed as 
\[ 
    \InputIfFileExists{stringdiag_ax1_left.tikz}{}{\input{./tikz/stringdiag_ax1_left.tikz}}
 = 
    \begin{tikzpicture}
	\begin{pgfonlayer}{nodelayer}
		\node [style=label] (8) at (1.5, 0.5) {$A$};
		\node [style=label] (11) at (-1.5, 0.5) {$A$};
		\node [style=label] (12) at (1.5, -0.5) {$B$};
		\node [style=label] (13) at (-1.5, -0.5) {$B$};
	\end{pgfonlayer}
	\begin{pgfonlayer}{edgelayer}
		\draw (11) to (8);
		\draw (13) to (12);
	\end{pgfonlayer}
\end{tikzpicture}
}
 \quad\qquad 
    \InputIfFileExists{stringdiag_ax2_left.tikz}{}{\input{./tikz/stringdiag_ax2_left.tikz}}
 = 
    \InputIfFileExists{stringdiag_ax2_right.tikz}{}{\input{./tikz/stringdiag_ax2_right.tikz}}
 \]

Hereafter, we call  $\CatString$ the category having as objects words in $\sort^\star$ and as arrows string diagrams. Theorem 2.3 in~\cite{joyal1991geometry} states that  $\Cat{C}_\sign$ is a \emph{symmetric strict monoidal category freely generated} by $\sign$.
\begin{definition} A \emph{symmetric monoidal category}  consists 
of a category $\Cat{C}$, a functor $\perG \colon \Cat{C} \times \Cat{C} \to \Cat{C}$,
an object $\unoG$  and natural isomorphisms \[ \alpha_{X, Y, Z} \colon (X \perG Y) \perG Z \to X \perG (Y \perG Z) \qquad \lambda_X \colon \unoG \perG X \to X \qquad \rho_X \colon X \perG \unoG \to X \qquad \sigma_{X, Y}^{\perG} \colon X \perG Y \to Y \perG X \]
satisfying some coherence axioms (in Figures~\ref{fig:moncatax} and~\ref{fig:symmmoncatax}).
A monoidal category is said to be \emph{strict} when $\alpha$, $\lambda$ and $\rho$ are all identity natural isomorphisms. A \emph{strict symmetric monoidal functor} is a functor $F \colon \Cat{C} \to \Cat{D}$  preserving $\perG$, $\unoG$ and $\sigma^{\perG}$. We write $\SMC$ for the category of ssm categories and functors. 

\end{definition}
\begin{remark}\label{rmk:symstrict}
In \emph{strict}  symmetric monoidal (ssm) categories  the symmetry $\sigma$ is not forced to be the identity, since this would raise some problems: for instance, $(f_1;g_1) \perG (f_2;g_2) = (f_1;g_2) \perG (f_2;g_1)$ for all $f_1,f_2\colon A \to B$ and $g_1,g_2\colon B \to C$. As we will see in  Section~\ref{sec:rigcategories}, this fact will make the issue of strictness for rig categories rather subtle.
\end{remark}

To illustrate in which sense $\Cat{C}_\sign$ is freely generated, it is convenient to introduce \emph{interpretations} in a fashion similar to~\cite{selinger2010survey}: an interpretation $\interpretation$ of $\sign$ into an ssm category $\Cat{D}$ consists of two functions $\alpha_{\sort} \colon \sort \to Ob(\Cat{D})$ and $\alpha_{\sign}\colon \sign \to Ar(\Cat{D})$ such that, for all $s\in \sign$, $\alpha_{\sign}(s)$ is an arrow having as domain $\alpha_{\sort}^\sharp(\ar(s))$ and codomain $\alpha_{\sort}^\sharp(\coar(s))$, for $\alpha_{\sort}^\sharp\colon \sort^\star \to Ob(\Cat{D})$ the inductive extension of  $\alpha_{\sort}$.  $\CatString$ is freely generated by $\sign$ in the sense that, for all symmetric strict monoidal categories $\Cat{D}$ and  all interpretations $\interpretation$ of $\sign$ in $\Cat{D}$, there exists a unique ssm-functor $\dsem{-}_{\interpretation}\colon \Cat{C}_\sign \to \Cat{D}$ extending $\interpretation$ (i.e. $\dsem{s}_\interpretation=\alpha_{\sign}(s)$ for all $s\in \sign$).

One can easily extend the notion of interpretation of $\sign$ into a symmetric monoidal category $\Cat D$ that is not necessarily strict. In this case we set $\alpha^\sharp_\sort \colon \sort^\star \to Ob(\Cat D)$ to be the \emph{right bracketing} of the inductive extension of $\alpha_\sort$. For instance, $\alpha^\sharp_\sort(ABC) = \alpha_\sort(A) \perG (\alpha_\sort(B) \perG \alpha_\sort(C))$.

\subsection{The Two Monoidal Structures of $\Rel$}\label{sec:2monREL} It is often the case that the same category carries more than one monoidal product. An example relevant to this work is $\Rel$, which exhibits two monoidal structures:  $(\Rel, \per, \uno)$ and $(\Rel, \piu, \zero)$.
In the former, $\per$ is given by the cartesian product, i.e.  $R \per S \defeq \{(\,(x_1,x_2),\, (y_1,y_2)\,) \mid (x_1,y_1)\in R \text{ and } (x_2,y_2)\in S  \}$ for all relations $R,S$, and the monoidal unit is the singleton set $1=\{\bullet\}$.
In the latter, $\piu$ is given by disjoint union, i.e.\ $R \piu S \defeq \{(\,(x,0),\,(y,0)\,) \mid (x,y)\in R \} \cup \{(\,(x,1),\,(y,1)\,) \mid (x,y)\in S \}$, and the monoidal unit $0$ is the empty set. It is worth recalling that in
$\Rel$ the empty set is both an initial and final object, i.e.\ a zero object, and that the disjoint union is both a coproduct and product, i.e.\ a biproduct. Indeed, $(\Rel, \piu, \zero)$ is our first example of a finite biproduct category.

\begin{definition}\label{def:biproduct category}
    A \emph{finite biproduct category} is a symmetric monoidal category $(\Cat{C}, \perG, \unoG)$ where, for every object $X$, there are morphisms $\codiag{X} \colon X \perG X \!\to\! X, \;\; \cobang{X} \colon \unoG \!\to\! X, \;\; \diag{X} \colon X \!\to\!  X \perG X, \;\; \bang{X} \colon X \!\to\! \unoG$ such that
    \begin{enumerate}
\item $(\codiag{X}, \cobang{X})$ is a commutative monoid and $(\diag{X}, \bang{X})$ is a cocommutative comonoid, satisfying the coherence axioms in Figure~\ref{fig:fbcoherence},
        \item every arrow $f \colon X \to Y$ is both a monoid and a comonoid homomorphism, i.e.\
        \[
        (f\perG f) ; \codiag{Y} = \codiag{X}; f, \quad \cobang{X}; f = \cobang{Y}, \quad f; \diag{Y}= \diag{X}; (f\perG f) \quad \text{and} \quad f;\bang{Y}= \bang{X}.
        \] 
\end{enumerate}
    A \emph{morphism of finite biproduct categories} is a symmetric monoidal functor preserving  $\codiag{X},\cobang{X},\diag{X},\cobang{X}$. 
    We write $\FBC$ for the category of strict finite biproduct categories and their morphisms. 
    \end{definition}
Observe that the second condition simply amounts to naturality of monoids and comonoids. More generally, the reader who does not recognise the familiar definition of finite biproduct (fb) category may have a look at~\cite[Appendix D]{bonchi2023deconstructing}.  Monoids and comonoids in the monoidal category $(\Rel, \piu, \zero)$ are illustrated in the first column below: \begin{equation}\label{eq:comonoidsREL}
		\begin{tabular}{rcl c rcl}
			$\codiag{X}$ & $\!\!\defeq\!\!$ & $\{((x, 0), \; x) \mid x\in X\} \cup \{((x, 1), \; x) \mid x\in X\}$ && $\cocopier{X}$ & $\!\!\defeq\!\!$ & $\op{\copier{X}}$ \\[0.5em]
			$\cobang{X}$ & $\!\!\defeq\!\!$ & $\{\}$ && $\codischarger{X}$ & $\!\!\defeq\!\!$ &  $\op{\discharger{X}} $ \\[0.5em]
			$\diag{X}$ & $\!\!\defeq\!\!$ & $\op{\codiag{X}}$ && $\copier{X}$ & $\!\!\defeq\!\!$ & $\{(x, \; (x,x)) \mid x\in X\} $ \\[0.5em]
	     	$\bang{X}$ & $\!\!\defeq\!\!$ & $\op{\cobang{X}}$ && $\discharger{X}$ & $\!\!\defeq\!\!$ & $\{(x, \bullet) \mid x\in X\}\subseteq X \times 1$
		\end{tabular}
\end{equation}
Also $(\Rel, \per, \uno)$ has monoids and comonoids, illustrated in the second column above. However, they fail to be natural and, for this reason, $(\Rel, \per, \uno)$ is not an fb category. It is instead the archetypal example of a cartesian bicategory.

\begin{definition}\label{def:cartesian bicategory}
    A \emph{cartesian bicategory}, in the sense of~\cite{Carboni1987}, is a symmetric monoidal category $(\Cat{C}, \perG, \unoG)$ enriched over the category of posets where for every object $X$ there are morphisms $\cocopier{X} \colon X \perG X \to X, \;\; \codischarger{X} \colon \unoG \to X, \;\; \copier{X} \colon X \to X \perG X, \;\; \discharger{X} \colon X \to \unoG$ such that
    \begin{enumerate}
        \item $(\cocopier{X}, \codischarger{X})$ is a commutative monoid and $(\copier{X}, \discharger{X})$ is a cocommutative comonoid, satisfying the coherence axioms in Figure~\ref{fig:fbcoherence},
        \item every arrow $f \colon X \to Y$ is a lax comonoid homomorphism, i.e.\
        \[
        f;\copier{Y}\leq \copier{X}; (f \perG f) \quad \text{and} \quad f;\discharger{Y} \leq \discharger{X},
        \] 
\item monoids and comonoids form special Frobenius bimonoids (see e.g.\ \cite{Lack2004a}), 
        \item the comonoid $(\copier{X}, \discharger{X})$ is left adjoint to the monoid $(\cocopier{X}, \codischarger{X})$, i.e.\ :
        \[ \codischarger{X} ; \discharger{X} \leq \id{\unoG} \qquad \cocopier{X};\copier{X} \leq \id X \perG \id X \qquad \id X \leq \discharger{X} ; \codischarger{X} \qquad \id{X} \leq \copier{X};\cocopier{X} \]
    \end{enumerate}
A \emph{morphism of cartesian bicategories} is a poset enriched symmetric monoidal functor preserving monoids and comonoids.
\end{definition}

A string diagrammatic language, named $\CB$, expressing the cartesian bicategory structure of  $(\Rel, \per, \uno)$ is introduced in~\cite{GCQ}. One can similarly define a language for $(\Rel, \piu, \zero)$, but combining them would require a diagrammatic language that is able to express two different monoidal products at once. The appropriate categorical structure for this are rig categories, discussed in the next section.

\section{Rig Categories and Tape Diagrams}\label{sec:rigcategories}
Rig categories, also known as \emph{bimonoidal categories}, involve two (symmetric) monoidal structures where one distributes over the other. They were first studied by Laplaza~\cite{laplaza_coherence_1972}, who discovered two coherence results establishing which diagrams necessarily commute as a consequence of the axioms given in their definition. An extensive treatment was recently given by Johnson and Yau~\cite{johnson2021bimonoidal}, from which we borrow most of the notation in this paper.
\begin{definition}\label{def:rig}
    A \emph{rig category} is a category $\Cat{C}$ with 
    two symmetric monoidal structures $(\Cat{C}, \per, \uno, \symmt)$ and 
    $(\Cat{C}, \piu, \zero, \symmp)$ and natural isomorphisms 
    \[ \dl{X}{Y}{Z} \colon X \per (Y \piu Z) \to (X \per Y) \piu (X \per Z) \qquad  \annl{X} \colon \zero \per X \to \zero \]
    \[ \dr{X}{Y}{Z} \colon (X \piu Y) \per Z \to (X \per Z) \piu (Y \per Z) \qquad \annr{X} \colon X \per \zero \to \zero \]
satisfying the coherence axioms in Figure~\ref{fig:rigax}.    A rig category is said to be \emph{right} (respectively \emph{left}) \emph{strict} when both its monoidal structures are 
    strict and $\lambda^\bullet, \rho^\bullet$ and $\delta^r$ (respectively $\delta^l$) are all identity 
    natural isomorphisms. A \emph{right strict rig functor} is a strict symmetric monoidal functor for both $\per$ and $\piu$ preserving $\delta^l$. We write $\RIG$ for the category of right strict rig categories and functors.
\end{definition}

All rig categories considered in this paper are assumed to be right strict. This is harmless since any rig category is equivalent to a right strict one (see Theorem 5.4.6 in~\cite{johnson2021bimonoidal}).
The reader may wonder why only one of the two distributors is forced to be the identity within a strict rig category.
This can shortly be explained as follows: if both distributors would be identities then,  for all objects $A,B,C,D$, %
\begin{align*}
 ((A \per C)\piu (B\per C)) \piu ((A \per D)\piu (B \per D)) 
 = ((A \per C)\piu (A\per D)) \piu ((B \per C)\piu (B \per D))
\end{align*}
raising the same problems of strictification of symmetries (see Remark \ref{rmk:symstrict}). 

\subsection{Freely Generated Sesquistrict Rig Categories}

The traditional approach to strictness is however unsatisfactory when studying freely generated categories. To illustrate our concerns, consider a right strict rig category freely generated by a signature $\sign$ with sorts $\sort$. The objects of this category are terms generated by the grammar in Table~\ref{table:eq objects fsr} modulo the equations in the first three rows of the same table. These equivalence classes of terms do not come with a very handy form, unlike, for instance, the objects of a strict monoidal category, which are words. %

\begin{table}
	\begin{center}
	\footnotesize{	
		\begin{subtable}{0.52\textwidth}
            \begin{tabular}{c}
				\toprule
				$X \; ::=\; \; A \; \mid \; \uno \; \mid \; \zero \; \mid \;  X \per X \; \mid \;  X \piu X$\\
				\midrule
				\begin{tabular}{ccc}
					$ (X \per Y) \per Z = X \per (Y \per Z)$ & $\uno \per X = X$ &  $X \per \uno = X $\\
					$(X \piu Y) \piu Z = X \piu (Y \piu Z)$ & $\zero \piu X = X$  & $X \piu \zero =X  $\\
					$(X \piu Y) \per Z = (X \per Z) \piu (Y \per Z)$ &  $\!\zero \per X = \zero$ & $X \per \zero=X $\\
				\end{tabular}\\
				$A \per (Y \piu Z) = (A \per Y) \piu (A \per Z)$\\
				\bottomrule
			\end{tabular}
            \caption{}
            \label{table:eq objects fsr}
        \end{subtable}
		\begin{subtable}{0.45\textwidth}
            \begin{tabular}{c}
				\toprule
				$n$-ary sums and products $\vphantom{\mid}$\\
				\midrule
				\makecell{
					\\[-5pt]
					$\Piu[i=1][0]{X_i} = \zero \; \Piu[i=1][1]{X_i}= X_1 \; \Piu[i=1][n+1]{X_i} = X_1 \piu (\Piu[i=1][n]{X_{i+1}} )$ \\[1em]
					$\Per[i=1][0]{X_i} = \uno \; \Per[i=1][1]{X_i}= X_1 \; \Per[i=1][n+1]{X_i} = X_1 \per (\Per[i=1][n]{X_{i+1}} )$ \\[3pt]
				} \\
				\bottomrule
			\end{tabular}
            \caption{}
            \label{table:n-ary sums and prod}
        \end{subtable}
	}
	\end{center}
	\caption{Equations for the objects of a free sesquistrict rig category}\label{tab:equationsonobject}
\end{table}

An alternative solution is proposed in \cite{bonchi2023deconstructing}: the focus is  on  freely generated rig categories that are \emph{sesquistrict}, i.e.\ right strict but only partially left strict: namely the left distributor  $\dl{X}{Y}{Z} \colon X \per (Y \piu Z) \to (X \per Y) \piu (X \per Z)$ is the identity only when $X$ is a basic sort $A\in \sort$. In terms of the equations to impose on objects, this amounts to the one in the fourth row in Table~\ref{table:eq objects fsr} for each $A\in \sort$. It is useful to observe that the addition of these equations avoids the problem of using left and right strictness at the same time. Indeed $(A\piu B) \per (C \piu D)$ turns out to be equal to $(A \per C) \piu (A \per D) \piu (B \per C) \piu (B \per D)$ but not to $(A \per C) \piu (B \per C) \piu (A \per D)  \piu (B \per D)$.

\begin{definition}\label{def:sesquistrict rig category}
	A \emph{sesquistrict rig category} is a functor $H \colon \Cat S \to \Cat C$, where $\Cat S$ is a discrete category and $\Cat C$ is a strict rig category, such that for all $A \in \Cat S$
	\[
	\dl{H(A)}{X}{Y} \colon H(A) \per (X \piu Y) \to (H(A) \per X) \piu (H(A) \per Y)
	\]
	is an identity morphism. We will also say, in this case, that $\Cat C$ is a $\Cat S$-sesquistrict rig category.
	
	Given $H \!\colon\! \Cat S \!\to\! \Cat C$ and $H' \!\colon\! \Cat S' \!\to\! \Cat C'$ two sesquistrict rig categories, a \emph{sesquistrict rig functor} from $H$ to $H'$ is a pair $(\alpha \!\colon\! \Cat S \!\to\! \Cat S', \beta \!\colon\! \Cat C \!\to\! \Cat C')$, with $\alpha$ a functor and $\beta$ a strict rig functor, such that $\alpha; H' = H; \beta$.
\end{definition}
\begin{remark}
In \cite{bonchi2023deconstructing}, it was shown that for any rig category $\Cat{C}$, one can construct its strictification $\overline{\Cat{C}}$ as in \cite{johnson2021bimonoidal} and then consider the obvious embedding from  $ob(\Cat{C})$, the discrete category of the objects of $\Cat{C}$, into $\overline{\Cat{C}}$. The embedding $ob(\Cat{C}) \to \overline{\Cat{C}}$ forms a sesquistrict category and it is equivalent (as a rig category) to the original $\Cat{C}$ \cite[Corollary 4.5]{bonchi2023deconstructing}. Through the paper, when dealing with a rig category $\Cat{C}$, we will often implicitly refer to the equivalent sesquistrict $ob(\Cat{C}) \to \overline{\Cat{C}}$.
\end{remark}

Given a set of sorts $\sort$, a \emph{rig signature} is a tuple $(\sort,\sign,\ar,\coar)$ where $\ar$ and $\coar$ assign to each $\gen \in \sign$ an arity and a coarity respectively, which are terms in the grammar specified in Table~\ref{table:eq objects fsr} modulo the equations underneath it. (Notice that any monoidal signature is in particular a rig signature.) To define the notion of free sesquistrict rig category, we need to extend interpretations of monoidal signatures to the rig case. An \emph{interpretation} of a rig signature $(\sort,\sign,\ar,\coar)$ in a sesquistrict  rig category $H \colon \Cat M \to \Cat D$ is a pair of functions $(\alpha_{\sort} \colon \sort \to Ob(\Cat M), \alpha_\sign \colon \sign \to Ar(\Cat D))$ such that, for all $\gen \in \sign$, $\alpha_{\sign}(s)$ is an arrow having as domain and codomain $(\alpha_{\sort};H)^\sharp(\ar(s))$ and  $(\alpha_{\sort};H)^\sharp(\coar(s))$.

\begin{definition}\label{def:freesesqui}
Let $(\sort,\sign,\ar,\coar)$ (simply $\sign$ for short) be a rig signature. A sesquistrict  rig category $H \colon \Cat M \to \Cat D$ is said to be \emph{freely generated} by $\sign$ if there is an interpretation $(\alpha_S,\alpha_\sign)$ of $\sign$ in $H$ such that for every sesquistrict rig category $H' \colon \Cat M' \to \Cat D'$ and every interpretation $(\alpha_\sort' \colon \sort \to Ob(\Cat M'), \alpha_\sign' \colon \sign \to Ar(\Cat D'))$ there exists a unique sesquistrict rig functor $(\alpha \colon \Cat M \to \Cat M', \beta \colon \Cat D \to \Cat D')$ such that $\alpha_\sort ; \alpha = \alpha_\sort'$ and $\alpha_\sign ; \beta = \alpha_{\sign}'$. 
\end{definition}
This is the definition of free object on a generating one instantiated in the category of sesquistrict rig categories and the category of rig signatures.
Thus, sesquistrict rig categories generated by a given signature are isomorphic to each other and we may refer to ``the'' free sesquistrict rig category generated by a signature.

The objects of the free sesquistrict rig category generated by $(\sort, \sign)$ are the terms generated by the grammar in Table~\ref{table:eq objects fsr} modulo all the equations underneath it; by orienting the equations from left to right, one obtains a rewriting system that is confluent and terminating and, most importantly, the unique normal forms are exactly polynomials: a term $X$  is in \emph{polynomial} form if there exist $n$, $m_i$ and $A_{i,j}\in \sort$ for $i=1 \dots n$ and $j=1 \dots m_i$ such that $X=\Piu[i=1][n]{\Per[j=1][m_i]{A_{i,j}}}$ (for $n$-ary sums and products as in Table~\ref{table:n-ary sums and prod}).
We will always refer to terms in polynomial form as \emph{polynomials} and, for a polynomial like the aforementioned  $X$, we will call \emph{monomials} of $X$ the $n$ terms $\Per[j=1][m_i]{A_{i,j}}$. For instance the monomials of $(A \per B) \piu \uno$ are $A \per B$ and $1$. Note that, differently from the polynomials we are used to dealing with, here neither $\piu$ nor $\per$ is commutative so, for instance, $(A \per B) \piu \uno$ is different from both $\uno \piu (A \per B)$ and $(B \per A) \piu \uno$. Note that non-commutative polynomials are in one to one correspondence with \emph{words of words} over $\sort$, while monomials are words over $\sort$. 
\begin{notation}
Through the whole paper, we will denote by $A,B,C\dots$ the sorts in $\sort$, by $U,V,W \dots$ the words in $\sort^\star$ and by $P,Q,R,S \dots$ the words of words in $(\sort^\star)^\star$. Given two words $U,V\in \sort^\star$, we will write $UV$ for their concatenation and $1$ for the empty word. Given two words of words $P,Q\in (\sort^\star)^\star$, we will write $P\piu Q$ for their concatenation and $\zero$ for the empty word of words. Given a word of words $P$, we will write $\pi P$ for the corresponding term in polynomial form, for instance $\pi(A \piu BCD\piu 1 )$ is the term $A \piu ((B \per (C \per D)) \piu \uno)$. Throughout this paper  we  often omit $\pi$, thus we implicitly identify words of words with polynomials.
\end{notation}

Beyond concatenation ($\piu$), one can define a product operation $\per$ on $(\sort^\star)^\star$ by taking the unique normal form of $\pi(P) \per \pi(Q)$ for any $P,Q\in (\sort^\star)^\star$. More explicitly for
$P = \Piu[i]{U_i}$ and $Q = \Piu[j]{V_j}$, 
\begin{equation}\label{def:productPolynomials} P \per Q \defeq \Piu[i]{\Piu[j]{U_iV_j}}.
\end{equation}
Observe that, if both $P$ and $Q$ are monomials, namely, $P=U$ and $Q=V$ for some $U,V\in \sort^\star$, then $P\per Q = UV$. We can thus safely write $PQ$ in place of $P\per Q$ without the risk of any confusion.

\subsection{Finite Biproduct Rig Categories}
On many occasions, one is interested in rig categories where $\piu$ has some additional structure. For instance, distributive monoidal categories are rig categories where $\piu$ is a coproduct. In \cite{bonchi2023deconstructing}, the focus is on rig categories where $\piu$  is a biproduct, like the category of sets and relations $\Rel$ (see Section \ref{sec:2monREL}).

\begin{definition}
A \emph{finite biproduct (fb) rig category} is a rig category $(\Cat{C}, \piu, \zero, \per, \uno)$  such that $(\Cat{C}, \piu, \zero)$ is a finite biproduct category. A \emph{morphism of fb rig categories} is both a rig functor and a morphisms of fb categories. We write $\FBRIG$ for the category of fb rig categories and their morphisms.
\end{definition}
Sesquistrict finite biproduct rig categories and freely generated sesquistrig fb rig categories are defined analogously to the rig case. The interest in the finite biproducts is motivated by the following result (Theorem 4.9 in \cite{bonchi2023deconstructing}) stating that any rig signature can be safely reduced  to a monoidal one whenever $\piu$ is a biproduct.

\begin{theorem}
	For every rig signature $(\sort,\sign)$ there exists a monoidal signature $(\sort, \sign_M)$ such that the free sesquistrict fb rig categories generated by $(\sort,\sign)$ and by $(\sort, \sign_M)$ are isomorphic.
\end{theorem}

\subsection{Tape Diagrams for Rig Categories with Finite Biproducts}\label{sc:tape}

\begin{table}[t]
{
\begin{center}
\begin{tabular}{cccc}
    \toprule
    \multicolumn{2}{c}{$\diag{ A}; (\id{ A}\perG \diag{ A}) = \diag{ A};(\diag{ A}\perG \id{ A})$} & $\diag{ A} ; (\bang{ A}\perG \id{ A}) = \id{ A} $ & $ \diag{ A};\sigma_{ A, A}=\diag{ A}$ \\
    \multicolumn{2}{c}{$(\id{ A}\perG \codiag{ A}) ; \codiag{ A} = (\codiag{ A}\perG \id{ A}) ; \codiag{ A}$} & $(\cobang{ A}\perG \id{ A}) ; \codiag{ A}  = \id{ A} $ & $ \sigma_{ A, A};\codiag{ A}=\codiag{ A}$ \\
    $\codiag{ A};\diag{ A} = \diag{ A\perG  A} ; (\codiag{ A} \perG \codiag{ A}) $ & $ \cobang{ A}; \bang{ A}= \id{\unoG} $ & $ \cobang{ A}; \diag{ A}= \cobang{ A\perG  A} $ & $ \diag{ A}; \bang{ A}= \bang{ A\perG  A}$\\
    $\tapeFunct{c}; \bang{B}=\bang{A} $ & $ \tapeFunct{c}; \diag{B}=\diag{A}; (\tapeFunct{c} \perG \tapeFunct{c}) $ & $ \cobang{A};\tapeFunct{c} =\cobang{B} $ & $ \codiag{A};\tapeFunct{c} =(\tapeFunct{c} \perG \tapeFunct{c}); \codiag{B}$ \\
    \multicolumn{2}{c}{$\tape{\id{A}} = \id{A}$} & \multicolumn{2}{c}{$\tape{c ; d} = \tape{c} ; \tape{d}$} \\
    \bottomrule
    \end{tabular}
\end{center}
}
\caption{Additional axioms for $F_2(\Cat{C})$. Above, $c\colon A \to B$ is an arbitrary arrow of $\Cat{C}$}\label{fig:freestrictfbcat}
\end{table}

We have seen in Section~\ref{sec:monoidal} that string diagrams provide a convenient graphical language for strict monoidal categories.
In this section, we recall tape diagrams, a sound and complete graphical formalism for sesquistrict rig categories introduced in~\cite{bonchi2023deconstructing}.

The construction of tape diagrams goes through the adjunction in~\eqref{eq:adjunction}, where $\CAT$ is the category  of categories and functors, $\SMC$ is the category 
 of ssm categories and functors, and $\FBC$ is the category of strict finite biproduct categories and their morphisms. 
\begin{equation}\label{eq:adjunction} \begin{tikzcd}
		\SMC \ar[r,"U_1"] & \CAT \ar[r,bend left,"F_2"] \ar[r,draw=none,"\bot"description] & \FBC \ar[l,bend left,"U_2"]
\end{tikzcd}
\end{equation}
The functors $U_1$ and $U_2$ are the obvious forgetful functors. The functor $F_2$ is the left adjoint to $U_2$, and can be described as follows.

\begin{definition}\label{def:strict fb freely generated by C}
Let $\Cat{C}$ be a category. The  strict fb category freely generated by $\Cat{C}$, hereafter denoted by $F_2(\Cat{C})$, has as objects words of objects of $\Cat{C}$. Arrows are terms inductively generated by the following grammar, where $A,B$ and $c$ range over arbitrary objects and arrows of $\Cat{C}$,\begin{equation}
    \begin{array}{rcl}
        f & ::=& \; \id{A} \; \mid \; \id{I} \; \mid \; \tapeFunct{c} \; \mid \; \sigma_{A,B}^{\perG} \; \mid \;   f ; f   \; \mid \;  f \perG f  \; \mid \; \bang{A} \; \mid \; \diag{A} \;  \mid \; \cobang{A}\; \mid \; \codiag{A}\\
        \end{array}
\end{equation}  
modulo the axioms in  Tables~\ref{fig:freestricmmoncatax} and~\ref{fig:freestrictfbcat}. Notice in particular the last two from Table~\ref{fig:freestrictfbcat}:
\begin{equation}\label{ax:tape}\tag{Tape}
\tape{\id{A}} = \id{A} \qquad \tape{c;d}= \tape{c} ; \tape{d}
\end{equation}
\end{definition}
The assignment $\Cat{C} \mapsto F_2({\Cat{C}})$ easily extends to functors $H\colon \Cat{C} \to \Cat{D}$.
The unit of the adjunction $\eta \colon Id_{\CAT} \Rightarrow F_2U_2$ is defined for each category $\Cat{C}$ as the functor ${\tapeFunct{\cdot} \; \colon \Cat{C} \to U_2F_2(\Cat{C})}$ which is the identity on objects and maps every arrow $c$ in $\Cat{C}$ into the arrow $\tapeFunct{c}$ of $U_2F_2(\Cat{C})$. Observe that $\tapeFunct{\cdot}$ is indeed a functor, namely an arrow in $\CAT$, thanks to the axioms \eqref{ax:tape}. We will refer hereafter to this functor as the \emph{taping functor}.

\medskip

The sesquistrict fb rig category freely generated by a monoidal signature $\sign$ is $F_2U_1(\CatString)$, hereafter referred to as $\CatTape$, and it is presented as follows.

Recall that the set of objects of 
$\CatString$ is $\sort^\star$, i.e.\ words of sorts in $\sort$. The set of objects of $\CatTape$ is thus $(\sort^\star)^\star$, namely words of words of sorts in $\sort$. For arrows, consider the following two-layer grammar where $s \in \sign$, $A,B \in \sort$ and $U,V \in \sort^\star$.
\begin{equation}\label{tapesGrammar}
    \begin{tabular}{rc ccccccccccccccccccc}\setlength{\tabcolsep}{0.0pt}
        $c$  & ::= & $\id{A}$ & $\!\!\! \mid \!\!\!$ & $ \id{\uno} $ & $\!\!\! \mid \!\!\!$ & $ \gen $ & $\!\!\! \mid \!\!\!$ & $ \sigma_{A,B} $ & $\!\!\! \mid \!\!\!$ & $   c ; c   $ & $\!\!\! \mid \!\!\!$ & $  c \per c$ & \multicolumn{8}{c}{\;} \\
        $\t$ & ::= & $\id{U}$ & $\!\!\! \mid \!\!\!$ & $ \id{\zero} $ & $\!\!\! \mid \!\!\!$ & $ \tapeFunct{c} $ & $\!\!\! \mid \!\!\!$ & $ \sigma_{U,V}^{\piu} $ & $\!\!\! \mid \!\!\!$ & $   \t ; \t   $ & $\!\!\! \mid \!\!\!$ & $  \t \piu \t  $ & $\!\!\! \mid \!\!\!$ & $ \bang{U} $ & $\!\!\! \mid \!\!\!$ & $\diag{U}$ & $\!\!\! \mid \!\!\!$ & $\cobang{U}$ & $\!\!\! \mid \!\!\!$ & $\codiag{U}$    
    \end{tabular}
\end{equation}  
The terms of the first row, denoted by $c$, are called \emph{circuits}. Modulo the axioms in Table~\ref{fig:freestricmmoncatax} (after replacing $\perG$ with $\per$), these are exactly the arrows of $\CatString$.
The terms of the second row, denoted by $\t$, are called \emph{tapes}. Modulo the axioms in Tables \ref{fig:freestricmmoncatax} and \ref{fig:freestrictfbcat} (after replacing $\perG$ with $\piu$ and $A,B$ with $U,V$), these are exactly the arrows of $F_2U_1(\CatString)$, i.e.\  $\CatTape$. 

Since circuits are arrows of $\CatString$, these can be graphically represented as string diagrams. Also tapes can be represented as string diagrams, since they satisfy all of the axioms of ssmc. Note however that \emph{inside} tapes, there are string diagrams: this justifies the motto \emph{tape diagrams are string diagrams of string diagrams}.
We can render graphically and formally\footnote{A formalisation of the graphical language in terms of bimodular profunctors can be found in~\cite{braithwaite2023collages}.} the grammar in~\eqref{tapesGrammar}:
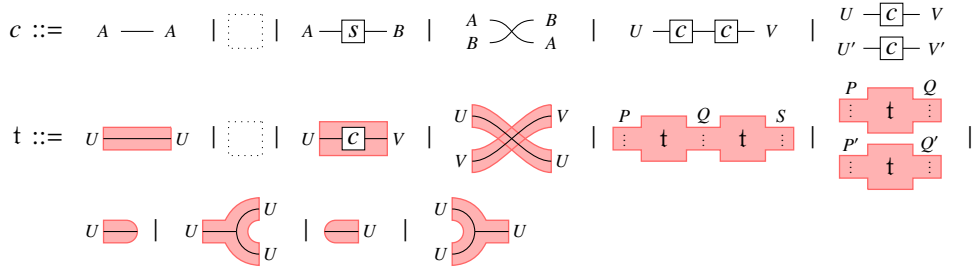
\begin{equation*}\label{tapesDiagGrammar}
    \setlength{\tabcolsep}{2pt}
    \begin{tabular}{rc cccccccccccc}
        $c$  & ::= &  $\wire{A}$ & $\mid$ & $ 
    \InputIfFileExists{empty.tikz}{}{\input{./tikz/empty.tikz}}
 $ & $\mid$ & $ \Cgen{\gen}{A}{B}  $ & $\mid$ & $ \Csymm{A}{B} $ & $\mid$ & $ 
    \InputIfFileExists{seq_compC.tikz}{}{\input{./tikz/seq_compC.tikz}}
   $ & $\mid$ & $  
    \InputIfFileExists{par_compC.tikz}{}{\input{./tikz/par_compC.tikz}}
$ & \\
        $\t$ & ::= & $\Twire{U}$ & $\mid$ & $ 
    \InputIfFileExists{empty.tikz}{}{\input{./tikz/empty.tikz}}
 $ & $\mid$ & $ \Tcirc{c}{U}{V}  $ & $\mid$ & $ \Tsymmp{U}{V} $ & $\mid$ & $ 
    \InputIfFileExists{tapes/seq_comp.tikz}{}{\input{./tikz/tapes/seq_comp.tikz}}
  $ & $\mid$ & $  
    \InputIfFileExists{tapes/par_comp.tikz}{}{\input{./tikz/tapes/par_comp.tikz}}
$ & $\mid$ \\
             &     & \multicolumn{12}{l}{$\Tcounit{U} \; \mid \; \Tcomonoid{U} \; \mid \; \Tunit{U} \; \mid \; \Tmonoid{U}$}
    \end{tabular}
\end{equation*}  
The identity $\id\zero$ is rendered as the empty tape $
    \InputIfFileExists{empty.tikz}{}{\input{./tikz/empty.tikz}}
$, while $\id\uno$ is $
    \InputIfFileExists{tapes/empty.tikz}{}{\input{./tikz/tapes/empty.tikz}}
$: a tape filled with the empty circuit. 
For a monomial $U \!=\! A_1\dots A_n$, $\id U$ is depicted as a tape containing  $n$ wires labelled by $A_i$. For instance, $\id{AB}$ is rendered as $\TRwire{A}{B}$. When clear from the context, we will simply represent it as a single wire  $\Twire{U}$ with the appropriate label.
Similarly, for a polynomial $P = \Piu[i=1][n]{U_i}$, $\id{P}$ is obtained as a vertical composition of tapes, as illustrated below on the left. \[ \id{AB \piu \uno \piu C} = \begin{aligned}\begin{gathered} \TRwire{A}{B} \\[-1.5mm] \Twire{\uno} \\[-2mm] \Twire{C} \end{gathered}\end{aligned} \qquad \qquad \tapesymm{AB}{C} = 
    \InputIfFileExists{tapes/examples/tapesymmABxC.tikz}{}{\input{./tikz/tapes/examples/tapesymmABxC.tikz}}
 \qquad\qquad \symmp{AB \piu \uno}{C} = 
    \InputIfFileExists{tapes/examples/symmpABp1pC.tikz}{}{\input{./tikz/tapes/examples/symmpABp1pC.tikz}}
 \]
We can render graphically the symmetries $\tapesymm{U}{V} \colon UV \!\to\! VU$ and $\symmp{P}{Q} \colon P \piu Q \!\to\! Q \piu P$ as crossings of wires and crossings of tapes,  see the two rightmost diagrams above.
The diagonal $\diag{U} \colon U \!\to\! U \piu U$ is represented as a splitting of tapes, while the bang $\bang{U} \colon U \!\to\! \zero$ is a tape closed on its right boundary. 
Codiagonals and cobangs are represented in the same way but mirrored along the y-axis. Exploiting the coherence axioms in Figure~\ref{fig:fbcoherence}, we can construct (co)diagonals and (co)bangs for arbitrary polynomials $P$.
For example, $\diag{AB}$, $\bang{CD}$, $\codiag{A\piu B \piu C}$ and $\cobang{AB \piu B \piu C}$ are depicted as:
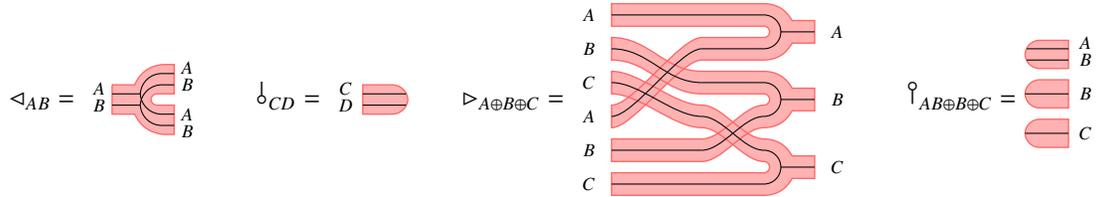
\[ 
    \diag{AB} = \TLcomonoid{B}{A} \qquad \bang{CD} = \TLcounit{D}{C} \qquad  \codiag{A\piu B \piu C} = 
    \InputIfFileExists{tapes/examples/codiagApBpC.tikz}{}{\input{./tikz/tapes/examples/codiagApBpC.tikz}}
 \qquad   \cobang{AB \piu B \piu C} = 
    \InputIfFileExists{tapes/examples/cobangABpBpC.tikz}{}{\input{./tikz/tapes/examples/cobangABpBpC.tikz}}
 
\]

When the structure inside a tape is not relevant the graphical language can be ``compressed'' in order to simplify the diagrammatic reasoning. For example, for arbitrary polynomials $P, Q$ we represent $\id{P}, \symmp{P}{Q}, \diag{P}, \bang{P}, \codiag{P}, \cobang{P}$ as follows:
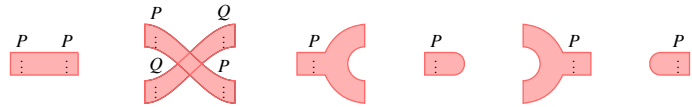
\[ \TPolyWire{P} \qquad \TPolySymmp{P}{Q} \qquad \TPolyDiag{P} \qquad \TPolyCounit{P} \qquad \TPolyCodiag{P} \qquad \TPolyUnit{P} \]
Moreover, for an arbitrary tape diagram $\t \colon P \to Q$ we write $\Tbox{\t}{P}{Q}$.

It is important to observe that the graphical representation takes care of the two axioms in \eqref{ax:tape}: both sides of the leftmost axiom are depicted as $\Twire{A}$ while both sides of the rightmost axiom as
$
    \InputIfFileExists{tapes/ax/c_poi_d.tikz}{}{\input{./tikz/tapes/ax/c_poi_d.tikz}}
$. The axioms of monoidal categories are also implicit in the graphical representation, while those for symmetries and the fb-structure (in Table~\ref{fig:freestrictfbcat}) have to be depicted explicitly as in Figure~\ref{fig:tapesax}. In particular, the diagrams in the first row express the inverse law and naturality of $\symmp$. In the second group there are the (co)monoid axioms and in the third group the bialgebra ones. Finally, the last group depicts naturality of the (co)diagonals and (co)bangs.

\begin{figure}[ht!]
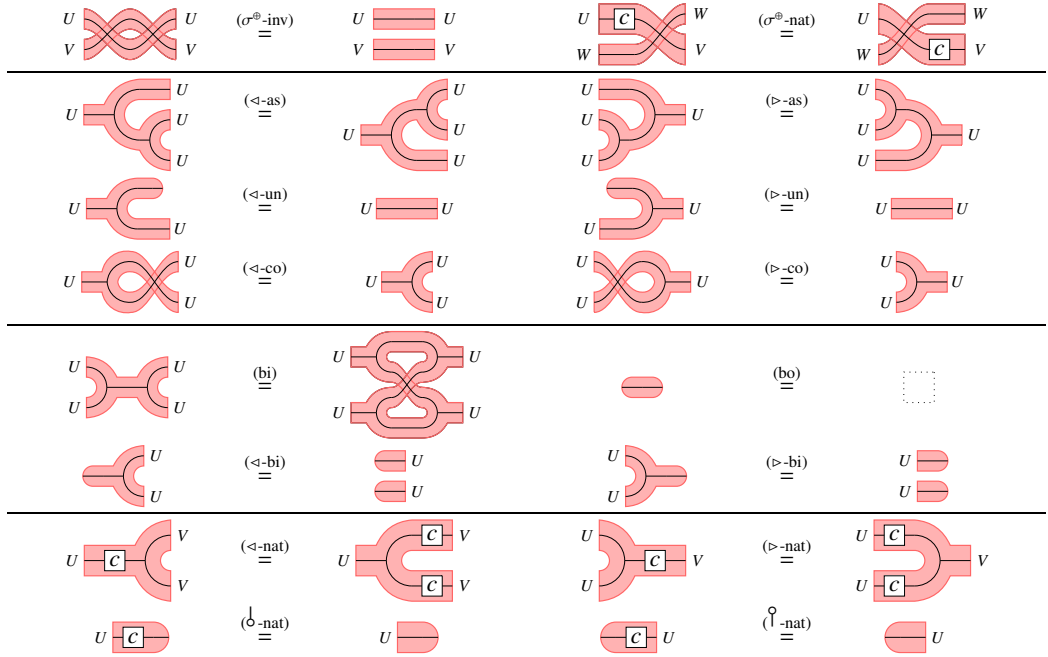

    \mylabel{ax:symmpinv}{$\symmp$-inv}
    \mylabel{ax:symmpnat}{$\symmp$-nat}
    \mylabel{ax:diagas}{$\diag{}$-as}
    \mylabel{ax:diagun}{$\diag{}$-un}
    \mylabel{ax:diagco}{$\diag{}$-co}
    \mylabel{ax:codiagas}{$\codiag{}$-as}
    \mylabel{ax:codiagun}{$\codiag{}$-un}
    \mylabel{ax:codiagco}{$\codiag{}$-co}
    \mylabel{ax:bi}{bi}
    \mylabel{ax:bo}{bo}
    \mylabel{ax:diagbi}{$\diag{}$-bi}
    \mylabel{ax:codiagbi}{$\codiag{}$-bi}
    \mylabel{ax:diagnat}{$\diag{}$-nat}
    \mylabel{ax:bangnat}{$\bang{}$-nat}
    \mylabel{ax:codiagnat}{$\codiag{}$-nat}
    \mylabel{ax:cobangnat}{$\cobang{}$-nat}
    \begin{center}
        \scalebox{0.9}{
    \begin{tabular}{C{1px} C{2.3cm} c C{2.5cm} C{5px} C{2.5cm} c C{2.3cm} C{1px}}
        &
    \InputIfFileExists{tapes/ax/symminv_left.tikz}{}{\input{./tikz/tapes/ax/symminv_left.tikz}}
 &$\stackrel{(\symmp\text{-inv})}{=}$& 
    \InputIfFileExists{tapes/ax/symminv_right.tikz}{}{\input{./tikz/tapes/ax/symminv_right.tikz}}
  && 
    \InputIfFileExists{tapes/ax/symmnat_left.tikz}{}{\input{./tikz/tapes/ax/symmnat_left.tikz}}
  &$\stackrel{(\symmp\text{-nat})}{=}$&  
    \InputIfFileExists{tapes/ax/symmnat_right.tikz}{}{\input{./tikz/tapes/ax/symmnat_right.tikz}}
 & \\[-1ex]
\hline
&
    \InputIfFileExists{tapes/whiskered_ax/comonoid_assoc_left.tikz}{}{\input{./tikz/tapes/whiskered_ax/comonoid_assoc_left.tikz}}
  &$\stackrel{(\diag{}\text{-as})}{=}$&  
    \InputIfFileExists{tapes/whiskered_ax/comonoid_assoc_right.tikz}{}{\input{./tikz/tapes/whiskered_ax/comonoid_assoc_right.tikz}}
 && 
    \InputIfFileExists{tapes/whiskered_ax/monoid_assoc_left.tikz}{}{\input{./tikz/tapes/whiskered_ax/monoid_assoc_left.tikz}}
 &$\stackrel{(\codiag{}\text{-as})}{=}$& 
    \InputIfFileExists{tapes/whiskered_ax/monoid_assoc_right.tikz}{}{\input{./tikz/tapes/whiskered_ax/monoid_assoc_right.tikz}}
 & \\
        &
    \InputIfFileExists{tapes/whiskered_ax/comonoid_unit_left.tikz}{}{\input{./tikz/tapes/whiskered_ax/comonoid_unit_left.tikz}}
 &$\stackrel{(\diag{}\text{-un})}{=}$& \Twire{U} && 
    \InputIfFileExists{tapes/whiskered_ax/monoid_unit_left.tikz}{}{\input{./tikz/tapes/whiskered_ax/monoid_unit_left.tikz}}
 &$\stackrel{(\codiag{}\text{-un})}{=}$& \Twire{U} & \\
        &
    \InputIfFileExists{tapes/whiskered_ax/comonoid_comm_left.tikz}{}{\input{./tikz/tapes/whiskered_ax/comonoid_comm_left.tikz}}
 &$\stackrel{(\diag{}\text{-co})}{=}$& \Tcomonoid{U} && 
    \InputIfFileExists{tapes/whiskered_ax/monoid_comm_left.tikz}{}{\input{./tikz/tapes/whiskered_ax/monoid_comm_left.tikz}}
 &$\stackrel{(\codiag{}\text{-co})}{=}$& \Tmonoid{U} &\\[-0.5ex]
\hline
        \\[-2.2ex]
        &
    \InputIfFileExists{tapes/whiskered_ax/bialg1_left.tikz}{}{\input{./tikz/tapes/whiskered_ax/bialg1_left.tikz}}
 &$\stackrel{(\text{bi})}{=}$& 
    \InputIfFileExists{tapes/whiskered_ax/bialg1_right.tikz}{}{\input{./tikz/tapes/whiskered_ax/bialg1_right.tikz}}
 && 
    \InputIfFileExists{tapes/whiskered_ax/bialg2_left.tikz}{}{\input{./tikz/tapes/whiskered_ax/bialg2_left.tikz}}
  &$\stackrel{(\text{bo})}{=}$&  
    \InputIfFileExists{empty.tikz}{}{\input{./tikz/empty.tikz}}
 & \\[-1ex]
        &
    \InputIfFileExists{tapes/whiskered_ax/bialg3_left.tikz}{}{\input{./tikz/tapes/whiskered_ax/bialg3_left.tikz}}
  &$\stackrel{(\diag{}\text{-bi})}{=}$&  
    \InputIfFileExists{tapes/whiskered_ax/bialg3_right.tikz}{}{\input{./tikz/tapes/whiskered_ax/bialg3_right.tikz}}
 && 
    \InputIfFileExists{tapes/whiskered_ax/bialg4_left.tikz}{}{\input{./tikz/tapes/whiskered_ax/bialg4_left.tikz}}
  &$\stackrel{(\codiag{}\text{-bi})}{=}$&  
    \InputIfFileExists{tapes/whiskered_ax/bialg4_right.tikz}{}{\input{./tikz/tapes/whiskered_ax/bialg4_right.tikz}}
 & \\
\hline
        \\[-2.2ex]
        &
    \InputIfFileExists{tapes/ax/comonoidnat_left.tikz}{}{\input{./tikz/tapes/ax/comonoidnat_left.tikz}}
 &$\stackrel{(\diag{}\text{-nat})}{=}$& 
    \InputIfFileExists{tapes/ax/comonoidnat_right.tikz}{}{\input{./tikz/tapes/ax/comonoidnat_right.tikz}}
 && 
    \InputIfFileExists{tapes/ax/monoidnat_left.tikz}{}{\input{./tikz/tapes/ax/monoidnat_left.tikz}}
 &$\stackrel{(\codiag{}\text{-nat})}{=}$& 
    \InputIfFileExists{tapes/ax/monoidnat_right.tikz}{}{\input{./tikz/tapes/ax/monoidnat_right.tikz}}
 & \\[-1ex]
        &
    \InputIfFileExists{tapes/ax/counitnat_left.tikz}{}{\input{./tikz/tapes/ax/counitnat_left.tikz}}
  &$\stackrel{(\bang{}\text{-nat})}{=}$& 
    \InputIfFileExists{tapes/ax/counitnat_right.tikz}{}{\input{./tikz/tapes/ax/counitnat_right.tikz}}
 && 
    \InputIfFileExists{tapes/ax/unitnat_left.tikz}{}{\input{./tikz/tapes/ax/unitnat_left.tikz}}
  &$\stackrel{(\cobang{}\text{-nat})}{=}$& 
    \InputIfFileExists{tapes/ax/unitnat_right.tikz}{}{\input{./tikz/tapes/ax/unitnat_right.tikz}}
 &
    \end{tabular}}
    \end{center}
    \caption{Axioms for tape diagrams}
    \label{fig:tapesax}
\end{figure}

\begin{theorem}\label{thm:Tapes is free sesquistrict generated by sigma}
	$\CatTape$ is the free sesquistrict fb rig category generated by the monoidal signature $(\sort, \sign)$.
\end{theorem}
The proof~\cite{bonchi2023deconstructing} of the above theorem mostly consists in illustrating that $\CatTape$ carries the structure of a rig category: $\per$ is defined on objects as in \eqref{def:productPolynomials}; the definition of symmetries for $\per$ and left distributors is given inductively in Table \ref{table:def dl symmt}; the definition of  $\per$ on tapes relies on the definition of \emph{left and right whiskerings}: see Table \ref{tab:producttape}. We will come back to whiskerings in Section \ref{ssec:tracerig}. For the time being, the reader can have a concrete grasp by means of the following example borrowed from \cite{bonchi2023deconstructing}.

\begin{example}\label{ex:diagPer}
    Consider $\t\colon U\piu V \to W\piu Z$ and $\s\colon U'\piu V' \to W'\piu Z'$ illustrated below on the left. Then $\t \per \s$ is simply the sequential composition of $\LW{U\piu V}{\s}$ and $\RW{W' \piu Z'}{\t}$:
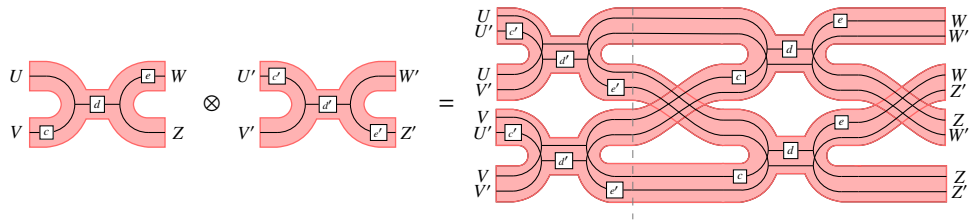
    \[  
    \InputIfFileExists{tapes/examples/t1.tikz}{}{\input{./tikz/tapes/examples/t1.tikz}}
 \per 
    \InputIfFileExists{tapes/examples/t2.tikz}{}{\input{./tikz/tapes/examples/t2.tikz}}
 = 
    \InputIfFileExists{tapes/examples/t1xt2_line.tikz}{}{\input{./tikz/tapes/examples/t1xt2_line.tikz}}
 \]
    The dashed line highlights the boundary between left and right polynomial whiskerings: $\LW{U \piu V}{\s}$, on the left, is simply the vertical composition of the monomial whiskerings $\LW{U}{\s}$ and $\LW{V}{\s}$ while, on the right, $\RW{W' \piu Z'}{\t}$ is rendered as the vertical composition of $\RW{W'}{\t}$ and $\RW{Z'}{\t}$, precomposed and postcomposed with left distributors.
\end{example}

\begin{table}[ht!]
    \begin{center}
        \begin{subtable}{0.60\textwidth}
            \begin{tabular}{l}
                \toprule
                $\dl{P}{Q}{R} \colon P \per (Q\piu R)  \to (P \per Q) \piu (P\per R) \vphantom{\symmt{P}{Q}}$ \\
                \midrule
                $\dl{\zero}{Q}{R} \defeq \id{\zero} \vphantom{\symmt{P}{\zero} \defeq \id{\zero}}$ \\
                $\dl{U \piu P'}{Q}{R} \defeq (\id{U\per (Q \piu R)} \piu \dl{P'}{Q}{R});(\id{U\per Q} \piu \symmp{U\per R}{P'\per Q} \piu \id{P'\per R}) \vphantom{\symmt{P}{V \piu Q'} \defeq \dl{P}{V}{Q'} ; (\Piu[i]{\tapesymm{U_i}{V}} \piu \symmt{P}{Q'})}$ \\
                \bottomrule
            \end{tabular}
            \caption{}
            \label{table:def dl}
        \end{subtable}
        \;
        \begin{subtable}{0.35\textwidth}
            \begin{tabular}{l}
                \toprule
                $\symmt{P}{Q} \colon P\per Q \to Q \per P$, with $P = \Piu[i]{U_i}$ \\
                \midrule
                $\symmt{P}{\zero} \defeq \id{\zero}$ \\
                $\symmt{P}{V \piu Q'} \defeq \dl{P}{V}{Q'} ; (\Piu[i]{\tapesymm{U_i}{V}} \piu \symmt{P}{Q'})$ \\
                \bottomrule
            \end{tabular}
            \caption{}
            \label{table:def symmt}
        \end{subtable}
        \caption{Inductive definition of $\delta^l$ and $\symmt$}
        \label{table:def dl symmt}
    \end{center}
\end{table}

\begin{table}[ht!]
    \begin{center}
    {
        \hfill
  \[\begin{array}{c}
        \def\arraystretch{1.5}
        \begin{array}{rcl|rcl}
            \toprule
            \LW U {\id\zero} &\defeq& \id\zero& \RW U {\id\zero} &\defeq& \id\zero \\
            \LW U {\tape{c}} &\defeq& \tape{\id U \per c}& \RW U {\tape{c}} &\defeq& \tape{c \per \id U} \\
            \LW U {\symmp{V}{W}} &\defeq& \symmp{UV}{UW}& \RW U {\symmp{V}{W}} &\defeq& \symmp{VU}{WU} \\
            \LW U {\diag V} &\defeq& \diag{UV}& \RW U {\diag V} &\defeq& \diag{VU} \\
            \LW U {\bang V} &\defeq& \bang{UV}& \RW U {\bang V} &\defeq& \bang{VU} \\
            \LW U {\codiag V} &\defeq& \codiag{UV}& \RW U {\codiag V} &\defeq& \codiag{VU} \\
            \LW U {\cobang V} &\defeq& \cobang{UV}& \RW U {\cobang V} &\defeq& \cobang{VU} \\
            \LW U {\t_1 ; \t_2} &\defeq& \LW U {\t_1} ; \LW U {\t_2}& \RW U {\t_1 ; \t_2} &\defeq& \RW U {\t_1} ; \RW U {\t_2} \\
            \LW U {\t_1 \piu \t_2} &\defeq& \LW U {\t_1} \piu \LW U {\t_2}& \RW U {\t_1 \piu \t_2} &\defeq& \RW U {\t_1} \piu \RW U {\t_2} \\
            \hline \hline
            \LW{\zero}{\t} &\defeq& \id{\zero} & \RW{\zero}{\t} &\defeq& \id{\zero} \\
            \LW{W\piu S'}{\t} &\defeq& \LW{W}{\t} \piu \LW{S'}{\t} & \RW{W \piu S'}{\t} &\defeq& \dl{P}{W}{S'} ; (\RW{W}{\t} \piu \RW{S'}{\t}) ; \Idl{Q}{W}{S'} \\
            \hline \hline
            \multicolumn{6}{c}{
                \t_1 \per \t_2 \defeq \LW{P}{\t_2} ; \RW{S}{\t_1}   \quad \text{ ( for }\t_1 \colon P \to Q, \t_2 \colon R \to S   \text{ )}
            }
            \\
            \bottomrule
        \end{array}
    \end{array}        
      \]
      \hfill
      \caption{Inductive definition of left and right monomial whiskerings (top); inductive definition of polynomial whiskerings (center); definition of $\per$ (bottom).}\label{tab:producttape}
       }
    \end{center}
\end{table}

\section{Uniformity in Traced Monoidal Categories}\label{sec:traced}
Tape diagrams add expressivity to languages of monoidal categories: the results in \cite{bonchi2023deconstructing} extend languages of quantum circuits \cite{coecke2011interacting} to express control gates, and provide a complete axiomatisation for the positive fragment of the calculus of relations \cite{tarski1941calculus,DBLP:conf/stacs/Pous18}.
Their expressivity, however, does not allow one to deal with the so called Kleene star, namely, the reflexive and transitive closure, that is often used to give semantics to while loops in imperative programming languages. To overcome this problem, we propose in this work to extend tape diagrams with \emph{traces} \cite{joyal1996traced}.

It turns out that the laws of traced monoidal categories are not sufficient to reuse the construction of tapes from \cite{bonchi2023deconstructing}, but one needs an additional condition on traces that is known as \emph{uniformity} \cite{cuazuanescu1994feedback}. In this section we recall uniformly traced monoidal categories and several adjunctions that will be crucial in the next section to introduce tape diagrams with traces. For the sake of brevity, hereafter all monoidal categories and functors are implicitly assumed to be symmetric and strict.

\begin{definition}\label{def:traced-category}
A  monoidal category $(\Cat{C}, \perG, \unoG)$  is \emph{traced} if it is endowed with an operator \(\trace_{S} \colon \Cat{C}(S \perG X, S \perG Y) \to \Cat{C}(X,Y)\), for all objects \(S\), \(X\) and \(Y\) of \(\Cat{C}\), that satisfies the axioms in \Cref{tab:trace-axioms} for all suitably typed \(f\), \(g\), \(u\) and \(v\).
 A \emph{morphism of traced  monoidal categories} is a  monoidal functor \(\fun{F} \colon \Cat{B} \to \Cat{C}\) 
  that preserves the trace, namely \(\fun{F}(\trace_{S}f) = \trace_{\fun{F}S}(\fun{F}f)\). We write $\TSMC$ for the category of traced monoidal categories and their morphisms.
\end{definition}

  \begin{table}[h!]
    \centering
    \mylabel{ax:trace:tightening}{tightening}
    \mylabel{ax:trace:strength}{strength}
    \mylabel{ax:trace:joining}{joining}
    \mylabel{ax:trace:vanishing}{vanishing}
    \mylabel{ax:trace:yanking}{yanking}
    \mylabel{ax:trace:sliding}{sliding}
    \begin{tabular}{l r}
      \toprule
      \(\trace_{S}((\id{} \perG u) \dcomp f \dcomp (\id{} \perG v)) = u \dcomp \trace_{S}f \dcomp v\) & (\ref*{ax:trace:tightening}) \\
      \(\trace_{S}(f \perG g) = \trace_{S}f \perG g\) & (\ref*{ax:trace:strength})\\
      \(\trace_{T} \trace_{S} f = \trace_{S \perG T} f\) & (\ref*{ax:trace:joining})\\
      \(\trace_{\unoG} f = f\) & (\ref*{ax:trace:vanishing})\\
      \(\trace_{T}(f \dcomp (u \perG \id{})) = \trace_{S}((u \perG \id{}) \dcomp f)\) & (\ref*{ax:trace:sliding})\\
      \(\trace_{S} \sigma^\perG_{S,S} = \id{S}\)& (\ref*{ax:trace:yanking})\\
      \bottomrule
    \end{tabular}
    \caption{Trace axioms.}\label{tab:trace-axioms}
  \end{table}
  
  String diagrams can be extended to deal with traces~\cite{joyal1996traced} (see e.g., \cite{selinger2010survey} for a survey). For a morphism $f\colon S \perG X \to S \perG Y$, we draw its trace as %
  
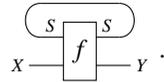
\[ 
    \begin{tikzpicture}
      \begin{pgfonlayer}{nodelayer}
        \node [style=label] (105) at (-2.75, -0.625) {$X$};
        \node [style=none] (117) at (-2.25, -0.625) {};
        \node [style=none] (118) at (-1.75, 0.625) {};
        \node [style=label] (120) at (2.75, -0.625) {$Y$};
        \node [style=none] (125) at (2.25, -0.625) {};
        \node [style=none] (126) at (1.75, 0.625) {};
        \node [style=stringlongbox] (128) at (0, 0) {$f$};
        \node [style=none] (129) at (-1.75, 2.125) {};
        \node [style=none] (130) at (1.75, 2.125) {};
        \node [style=label] (131) at (-1.25, 1.125) {$S$};
        \node [style=label] (132) at (1.25, 1.125) {$S$};
      \end{pgfonlayer}
      \begin{pgfonlayer}{edgelayer}
        \draw (125.center) to (117.center);
        \draw (118.center) to (126.center);
        \draw (130.center) to (129.center);
        \draw [bend right=90, looseness=1.50] (129.center) to (118.center);
        \draw [bend left=90, looseness=1.50] (130.center) to (126.center);
      \end{pgfonlayer}
    \end{tikzpicture}.     
  \]
  Using this convention, the axioms in Table \ref{tab:trace-axioms} acquire a more intuitive flavour: see  \Cref{fig:trace-axioms}.

  \begin{figure}[h!]
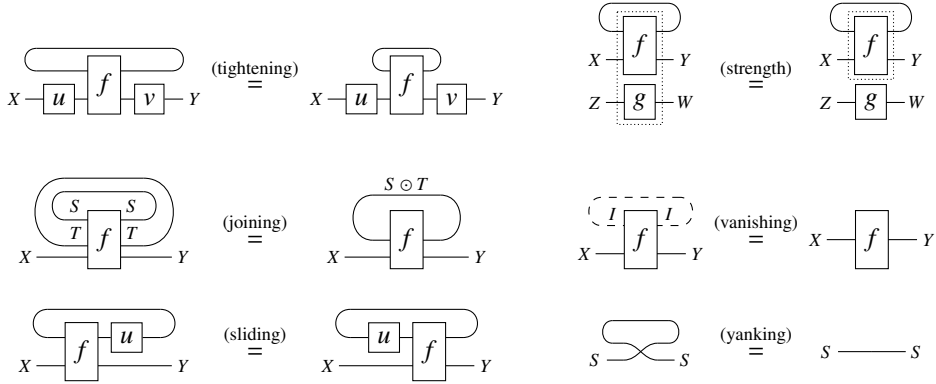

    \centering
    \[
      \begin{array}{c@{}c@{}c @{\qquad} c@{}c@{}c}
        
    \InputIfFileExists{traceAx/tight/lhs.tikz}{}{\input{./tikz/traceAx/tight/lhs.tikz}}
 &\stackrel{\text{(tightening)}}{=}& 
    \InputIfFileExists{traceAx/tight/rhs.tikz}{}{\input{./tikz/traceAx/tight/rhs.tikz}}

        &
        
    \InputIfFileExists{traceAx/strength/lhs.tikz}{}{\input{./tikz/traceAx/strength/lhs.tikz}}
 &\stackrel{\text{(strength)}}{=}& 
    \InputIfFileExists{traceAx/strength/rhs.tikz}{}{\input{./tikz/traceAx/strength/rhs.tikz}}

        \\[20pt]
        
    \InputIfFileExists{traceAx/join/lhs.tikz}{}{\input{./tikz/traceAx/join/lhs.tikz}}
 &\stackrel{\text{(joining)}}{=}& 
    \InputIfFileExists{traceAx/join/rhs.tikz}{}{\input{./tikz/traceAx/join/rhs.tikz}}

        &
        
    \InputIfFileExists{traceAx/vanish/lhs.tikz}{}{\input{./tikz/traceAx/vanish/lhs.tikz}}
 &\stackrel{\text{(vanishing)}}{=}& 
    \InputIfFileExists{traceAx/vanish/rhs.tikz}{}{\input{./tikz/traceAx/vanish/rhs.tikz}}

        \\[20pt]
        
    \InputIfFileExists{traceAx/slide/lhs.tikz}{}{\input{./tikz/traceAx/slide/lhs.tikz}}
 &\stackrel{\text{(sliding)}}{=}& 
    \InputIfFileExists{traceAx/slide/rhs.tikz}{}{\input{./tikz/traceAx/slide/rhs.tikz}}

        &
        
    \InputIfFileExists{traceAx/yank/lhs.tikz}{}{\input{./tikz/traceAx/yank/lhs.tikz}}
 &\stackrel{\text{(yanking)}}{=}& 
    \InputIfFileExists{traceAx/yank/rhs.tikz}{}{\input{./tikz/traceAx/yank/rhs.tikz}}

      \end{array}
    \]
    \caption{Trace axioms in string diagrams.\label{fig:trace-axioms}}
  \end{figure}

To extend tape diagrams with traces we need to require the trace to be uniform. This constraint, that arises from technical necessity, turns out to be the key to recover the axiomatisation of Kleene algebras in Section \ref{sec:kleene}, the induction proof principle in Section \ref{sec:peano} and the proof rules for while loops in Hoare logic in Section \ref{sec:hoare}.

\begin{definition}\label{def:utr-category}
  A traced monoidal category \(\Cat{C}\)
  is \emph{uniformly traced} if the trace operator satisfies the implication in \Cref{tab:uniformity} for all suitably typed \(f\), \(g\) and \(r\). A \emph{morphism of uniformly traced monoidal categories} is simply a morphism of traced monoidal categories. The category of uniformly traced monoidal categories and their morphisms is denoted by  \(\UTSMC\).
  \begin{table}[h!]
    \centering
    \mylabel{ax:trace:uniformity}{uniformity}
    \begin{tabular}{l r}
      \toprule
      if \(f \dcomp (r \perG \id{}) = (r \perG \id{}) \dcomp g\), then \(\trace_{S} f = \trace_{T} g\) & (\ref*{ax:trace:uniformity})\\
      \bottomrule
    \end{tabular}
    \caption{Uniformity axiom.}\label{tab:uniformity}
  \end{table}
\end{definition}

With string diagrams, the uniformity axiom is drawn as in Figure \ref{fig:uniformity}.
  \begin{figure}[h!]
    \centering
    \[
    \begin{array}{c}
      
    \InputIfFileExists{traceAx/uniformity/lhs.tikz}{}{\input{./tikz/traceAx/uniformity/lhs.tikz}}
 = 
    \InputIfFileExists{traceAx/uniformity/rhs.tikz}{}{\input{./tikz/traceAx/uniformity/rhs.tikz}}
 \stackrel{(\text{uniformity})}{\implies} 
    \InputIfFileExists{traceAx/uniformity/TR_lhs.tikz}{}{\input{./tikz/traceAx/uniformity/TR_lhs.tikz}}
 = 
    \InputIfFileExists{traceAx/uniformity/TR_rhs.tikz}{}{\input{./tikz/traceAx/uniformity/TR_rhs.tikz}}

    \end{array}
    \]
    \caption{Uniformity axiom in string diagrams.\label{fig:uniformity}}
  \end{figure}

\begin{remark}[Uniformity and sliding]\label{rem:uniformity-sliding}
  The sliding axiom is redundant as it follows from uniformity: 
    \[ 
    \InputIfFileExists{ufu.tikz}{}{\input{./tikz/ufu.tikz}}
 = 
    \InputIfFileExists{ufu.tikz}{}{\input{./tikz/ufu.tikz}}
 \implies 
    \begin{tikzpicture}
	\begin{pgfonlayer}{nodelayer}
		\node [style=label] (105) at (-2.5, -0.625) {$X$};
		\node [style=none] (117) at (-2, -0.625) {};
		\node [style=none] (118) at (-1.5, 0.625) {};
		\node [style=label] (120) at (4.5, -0.625) {$Y$};
		\node [style=none] (125) at (4, -0.625) {};
		\node [style=none] (126) at (3.5, 0.625) {};
		\node [style=stringlongbox] (128) at (0, 0) {$f$};
		\node [style=stringbox] (129) at (2, 0.625) {$u$};
		\node [style=none] (130) at (-1.5, 1.875) {};
		\node [style=none] (131) at (3.5, 1.875) {};
	\end{pgfonlayer}
	\begin{pgfonlayer}{edgelayer}
		\draw (118.center) to (129);
		\draw (129) to (126.center);
		\draw (125.center) to (117.center);
		\draw (131.center) to (130.center);
		\draw [bend right=90, looseness=1.75] (130.center) to (118.center);
		\draw [bend left=90, looseness=1.75] (131.center) to (126.center);
	\end{pgfonlayer}
\end{tikzpicture}
}
 = 
    \begin{tikzpicture}
	\begin{pgfonlayer}{nodelayer}
		\node [style=label] (105) at (4.5, -0.625) {$Y$};
		\node [style=none] (117) at (4, -0.625) {};
		\node [style=none] (118) at (3.5, 0.625) {};
		\node [style=label] (120) at (-2.5, -0.625) {$X$};
		\node [style=none] (125) at (-2, -0.625) {};
		\node [style=none] (126) at (-1.5, 0.625) {};
		\node [style=stringlongbox] (128) at (2, 0) {$f$};
		\node [style=stringbox] (129) at (0, 0.625) {$u$};
		\node [style=none] (130) at (3.5, 1.875) {};
		\node [style=none] (131) at (-1.5, 1.875) {};
	\end{pgfonlayer}
	\begin{pgfonlayer}{edgelayer}
		\draw (118.center) to (129);
		\draw (129) to (126.center);
		\draw (125.center) to (117.center);
		\draw (131.center) to (130.center);
		\draw [bend left=90, looseness=1.75] (130.center) to (118.center);
		\draw [bend right=90, looseness=1.75] (131.center) to (126.center);
	\end{pgfonlayer}
\end{tikzpicture}
}
 \]
 This fact will be useful for constructing the uniformly traced monoidal category freely generated by a monoidal category \(\Cat{C}\).
\end{remark}

\subsection{The Two Traced Monoidal Structures of $\Rel$}\label{sec:2trREL}
Recall from Section \ref{sec:2monREL} that the category of sets and relations $\Rel$ has two different monoidal structures:  $(\Rel, \per, \uno)$, intuitively representing data flow, and $(\Rel, \piu, \zero)$ representing control flow.
Since \cite{bainbridge1976feedback} (see also \cite{selinger1998note}), it is known that both monoidal categories are traced. For a relation $R\colon S \per X \to S \per Y$, its trace $\trace_{S}(R)$ in $(\Rel, \per, \uno)$ is defined as 
\[\trace_{S}(R) \defeq \{(x,y) \mid \exists s\in S. \; (\, (s,x) \,,\, (s,y) \, )\in R\}\subseteq X \times Y\]
To describe the trace in $(\Rel, \piu, \zero)$, we first need to observe that any relation $R\colon S \piu X \to S \piu Y$ can be decomposed as $R=R_{S,S} \cup R_{S,Y} \cup R_{X,S} \cup R_{X,Y}$ where 
\begin{equation}\label{eq:dec}
\begin{array}{rcl}
R_{S,S} &\defeq&\{(s,t)\mid (\,(s,0), (t,0) \,)\in R \} \subseteq S \times S \\ 
R_{S,Y} &\defeq&\{(s,y)\mid (\,(s,0), (y,1) \,)\in R \} \subseteq S \times Y \\ 
R_{X,S} &\defeq&\{(x,t)\mid (\,(x,1), (t,0) \,)\in R \} \subseteq X \times S \\ 
R_{X,Y} &\defeq&\{(x,y)\mid (\,(x,1), (y,1) \,)\in R \} \subseteq X \times Y 
\end{array}
\end{equation}
Then, the trace of $R\colon S \piu X \to S \piu Y$ in $(\Rel, \piu, \zero)$ is given by 
\begin{equation}%
  \trace_{S}(R) \defeq ( \, R_{X,S}; (R_{S,S})^\star ; R_{S,Y} \, ) \cup R_{X,Y}
\end{equation}
where, for any relation $T\subseteq S\times S$, $T^\star$ stands for the reflexive and transitive closure of $T$ (see e.g. \cite{joyal1996traced}).
The reader will see in Section \ref{sec:kleene}, that both the decomposition in \eqref{eq:dec} and the formula for the trace  in $(\Rel, \piu, \zero)$ come from its finite biproduct structure. In the same section, it will also become clear that such trace is uniform.
Instead, the trace in $(\Rel, \per, \uno)$ is not uniform: the readers can easily convince themselves by taking $r$ in Table \ref{tab:uniformity} to be the empty relation and observe that, for all arrows $f$ and $g$, the premise of the implication always holds. To properly tackle this kind of issues, in several works (see e.g. \cite{hasegawa2003uniformity}), uniformity is required on a restricted class of morphisms $r$, named \emph{strict} but we will not make similar restrictions here.

\subsection{The Free Uniform Trace}\label{sec:UTConstruction}
We now illustrate a construction that will play a key role in the rest of the paper: from a monoidal category, we freely generate a uniformly traced one.
Our first step consists in showing that, given a traced monoidal category $\Cat{C}$, one can always transform it into a uniformly traced one, named $\Unif(\Cat{C})$.

Let $\basicR$ be a set of pairs $(f, g)$ of arrows of $\Cat{C}$ with the same domain and codomain. We define $\congB$ to be the set generated by the following inference system (where $f \congB g$ is a shorthand for $(f,g)\in \congB$).
\begin{equation}\label{eq:uniformcong}
        \!\!\begin{array}{@{\qquad}c@{\qquad\qquad}c@{\qquad}c@{\qquad}}
            
            \inferrule*[right=($r$)]{-}{f \mathrel{\congB} f}
            &    
            \inferrule*[right=($t$)]{f \mathrel{\congB} g \quad g \mathrel{\congB} h}{f \mathrel{\congB} h}
            &
            \inferrule*[right=($s$)]{f \mathrel{\congB} g}{g \mathrel{\congB} f}
            \\[8pt]
                       \inferrule*[right=($\basicR$)]{f \mathbin{\basicR} g}{f \mathrel{\congB} g}
&
            \inferrule*[right=($;$)]{f \mathrel{\congB} f' \quad g \mathrel{\congB} g'}{f;g \mathrel{\congB} f';g'}
            &
            \inferrule*[right=($\perG$)]{f \mathrel{\congB} f' \quad g \mathrel{\congB} g'}{f\perG g \mathrel{\congB} f' \perG g'}
            \\[8pt]
            \multicolumn{3}{c}{
            \inferrule*[right=($ut$)]{u \mathrel{\congB} v \qquad f ; (u \piu \id{}) \mathrel{\congB} (v \piu \id{}) ; g}{\trace_{S}f \mathrel{\congB} \trace_{T}g}
            }
        \end{array}
\end{equation}
Observe that, for all $\basicR$, $\congB$ is an equivalence relation by (r), (t) and (s) and it is closed by composition $;$ and monoidal product $\perG$, thanks to the inference rules ($;$) and ($\perG$). By taking $u$ and $v$ in ($ut$) to be identities, $\congB$ is also closed by $\trace$: if $f\congB g$, then $\trace_{S}f \congB \trace_{S} g$.
When $\basicR$ is the empty set $\emptyset$, we just write $\approx$ in place of $\approx_{\emptyset}$. When we want to emphasise the underlying category $\Cat{C}$, we write $\approx_{\Cat{C}}$ in place of $\approx$.

We call $\Unif(\Cat{C})$ the quotient of $\Cat{C}$ by $\approx_{\Cat{C}}$. More explicitly, objects of $\Unif(\Cat{C})$ are the same as those of $\Cat{C}$ and arrows are $\approx_{\Cat{C}}$-equivalence classes $[f]\colon X \to Y$ of arrows $f\colon X\to Y$ in $\Cat{C}$.

\begin{proposition}\label{prop:uniformCat}
For a   traced monoidal category $\Cat{C}$, $\Unif(\Cat{C})$ is a uniformly traced monoidal category.
\end{proposition}
The assignment $\Cat{C} \mapsto \Unif(\Cat{C})$ extends to morphisms of traced monoidal categories: for $F\colon \Cat{B} \to \Cat{C}$, the functor $\Unif(F) \colon \Unif(\Cat{B}) \to \Unif(\Cat{C})$ is defined as
\[\Unif(F)(X)\defeq X \text{ for all objects }X\text{, and } \Unif(F)[f]\defeq [Ff] \text{ for all arrows }[f]\text{.}\]

\begin{lemma}\label{lemma:uniformity-quotient-functor}
  There is a functor \(\Unif \colon \TSMC \to \UTSMC\) defined as above.
\end{lemma}

Recall that $\TSMC$ and $\UTSMC$ are the categories of, respectively, traced monoidal categories and uniformly traced monoidal ones.

\begin{proposition}\label{prop:free-uniform}
Let $\fun{U}\colon \UTSMC \to \TSMC$ be the obvious embedding. Then $\Unif$ is left adjoint to $\fun{U}$. %
  \[\begin{tikzcd}
        \TSMC
        \arrow[r, "\Unif"{name=F}, bend left] &
        \UTSMC
        \arrow[l, "\fun{U}"{name=U}, bend left]
        \arrow[phantom, from=F, to=U, "\vdash" rotate=90]
    \end{tikzcd}\]
\end{proposition}

The next step consists in recalling from~\cite{katis2002feedback} the construction of the free traced monoidal category \(\freeTr (\Cat{C})\) on a symmetric monoidal category \(\Cat{C}\).
 \(\freeTr (\Cat{C})\) has  the same objects as \(\Cat{C}\), and morphisms \((f \mid S) \colon X \to Y\) are pairs of an object \(S\) and a morphism \(f \colon S \perG X \to S \perG Y\) of \(\Cat{C}\).  Morphisms are quotiented by \emph{yanking} and \emph{sliding}. Compositions, monoidal products and trace in \(\freeTr (\Cat{C})\) are recalled below.
  \begin{align}
    \st{f}{S} \dcomp \st{g}{T} &\defeq \st{\CStSeq{f}{S}{X}{g}{T}{Z}}{S \perG T} \label{eq:seq in Utr} \\
    \st{f}{S} \perG \st{g}{T} &\defeq \st{\CStPar{f}{S}{X}{Y}{g}{T}{Z}{W}}{S \perG T} \label{eq:per in Utr} \\
    \trace_T\st{f}{S} &\defeq \st{f}{S \perG T} \label{eq:trace in Utr}
  \end{align}
 The assignment $\Cat{C} \mapsto \freeTr(\Cat{C})$ extends to a functor $\freeTr \colon \SMC \to \TSMC$ which is the left adjoint to the obvious forgetful
 $\fun{U}\colon \TSMC \to \SMC$.
   \begin{equation}\label{eq:nonuniform}
   \begin{tikzcd}
        \SMC
        \arrow[r, "\freeTr"{name=F}, bend left] &
        \TSMC
        \arrow[l, "\fun{U}"{name=U}, bend left]
        \arrow[phantom, from=F, to=U, "\vdash" rotate=90]
    \end{tikzcd}\end{equation}
One can compose the  adjunction above with the one of Proposition \ref{prop:free-uniform}, to obtain the following result where  $\UTr \colon \SMC \to \UTSMC$ is the composition of $\freeTr \colon \SMC \to \TSMC$ with $\Unif \colon \TSMC \to \UTSMC$.

\begin{theorem}\label{th:free-uniform-trace}
Let $\fun{U}\colon \UTSMC \to \SMC$ be the obvious forgetful. Then $\UTr$ is left adjoint to $\fun{U}$:
   \[\begin{tikzcd}
        \SMC
        \arrow[r, "\UTr"{name=F}, bend left] &
        \UTSMC
        \arrow[l, "\fun{U}"{name=U}, bend left]
        \arrow[phantom, from=F, to=U, "\vdash" rotate=90]
    \end{tikzcd}\]
\end{theorem}

Since, by Remark \ref{rem:uniformity-sliding}, uniformity entails sliding, one can conveniently rephrase the construction of the freely generated uniformly traced monoidal category \(\UTr (\Cat{C})\) as follows.

\begin{definition}\label{def:uniform-state-construction}
  For a symmetric monoidal category \(\Cat{C}\), define \(\UTr (\Cat{C})\) with the same objects as \(\Cat{C}\), morphisms \((f \mid S) \colon X \to Y\) are pairs of an object \(S\) and a morphism \(f \colon S \perG X \to S \perG Y\) of \(\Cat{C}\).
  The morphisms are quotiented by the \emph{yanking}  and $\approx_{\Cat{C}}$.
\end{definition}

\subsection{The Free Uniform Trace preserves Products, Coproducts and Biproducts}\label{sec:UTConstruction-products}

A very important class of traced monoidal categories are those where the monoidal product is a product, a coproduct or a biproduct (see e.g. \cite{cstef?anescu1987flowchart,cstef?nescu1987flowchartII,esik1980identities}).
Recall that the latter have been introduced in Definition \ref{def:biproduct category}: finite products categories are defined similarly but without the monoids, finite coproducts categories without the comonoids.

\begin{proposition}\label{prop:free-uniform-trace-cartesian}
  The adjunction in \Cref{th:free-uniform-trace} restricts to finite product categories.
  In particular, if \(\Cat{C}\) is a finite product category, then so is \(\UTr(\Cat{C})\).
\end{proposition}
\begin{proof}
  Let \(f \colon X \to Y\) be a morphism in \(\UTr(\Cat{C})\).
  By the definition of \(\UTr\), there is a morphism \(g \colon S \perG X \to S \perG Y\) in \(\Cat{C}\) whose trace is \(f\), \(f = (g \mid S)\).
  By the universal property of products, \(g\) has two components: \(g = \diag{S \perG X} \dcomp (g_{1} \perG g_{2})\).
  The natural comonoid structure \((\diag{},\bang{})\) of \(\Cat{C}\) gives a comonoid structure \(((\diag{} \mid \unoG), (\bang{} \mid \unoG))\) in \(\UTr(\Cat{C})\) via the unit of the adjunction, \(\eta_{\Cat{C}}\).
  We show that this comonoid structure is natural in \(\UTr(\Cat{C})\).
  We can rewrite \(f \dcomp (\bang{Y} \mid \unoG)\) using naturality of \(\bang{Y}\) in \(\Cat{C}\).
  \[
    

}

  \]
  By uniformity, we rewrite the trace of \(g_{1}\).
  \[
    
    %
}

  \]
  This shows that \(f \dcomp (\bang{Y} \mid \unoG) = (\bang{X} \mid \unoG)\), i.e.\ that the counit is natural in \(\UTr(\Cat{C})\).

  Similarly, we can rewrite \(f \dcomp (\diag{Y} \mid \unoG)\) using naturality of \(\diag{Y}\) in \(\Cat{C}\).
  \[ 
    %
}
 \]
  By uniformity, we rewrite the trace of \(g_{1}\) with the comultiplication maps.
  \[ 
    %
}
 \]
  Then, \(f \dcomp (\diag{Y} \mid \unoG) = (\diag{X} \mid \unoG) \dcomp (f \times f)\), i.e.\ the comultiplication is also natural in \(\UTr(\Cat{C})\).
  \[ 
    %
}
 \]
  This shows that \(\UTr(\Cat{C})\) is a finite product category.
  Since a finite product functor is just a symmetric monoidal functor between finite product categories and a traced finite product functor is just a traced monoidal functor between traced finite product categories, the functor \(\UTr \colon \SMC \to \UTSMC\) restricts to a functor \(\UTr \colon \CMC \to \UTCMC\) from finite product categories to uniformly traced finite product categories.
  For the same reason, the unit and counit of the adjunction also restrict.
\end{proof}

\begin{remark}
Note how the above proof exploits the assumption of uniformity. In fact, it might be the case that $\Cat{C}$ has finite products but $\freeTr (\Cat{C})$ does not.
\end{remark}

\begin{definition}\label{def:utfbcategory}
A monoidal category $\Cat{C}$ is a \emph{uniformly traced finite biproduct} category (shortly ut-fb category) if it is uniformly traced and has finite biproducts.
Morphisms of ut-fb categories are monoidal functors that are both morphisms of traced monoidal categories and of finite biproduct categories.
We write $\FBUTr$  for the category of ut-fb categories and their morphisms.
\end{definition}

\begin{corollary}\label{cor:free-uniform-trace-fb}
  The adjunction in \Cref{th:free-uniform-trace} restricts to finite biproduct categories.
  \[\begin{tikzcd}
        \FBC
        \arrow[r, "\UTr"{name=F}, bend left] &
        \FBUTr
        \arrow[l, "\fun{U}"{name=U}, bend left]
        \arrow[phantom, from=F, to=U, "\vdash" rotate=90]
    \end{tikzcd}\]
\end{corollary}

\subsection{The Free Uniform Trace preserves the Rig Structure}\label{ssec:tracerig}

So far, we have illustrated that  $\UTr(\cdot)$ (\Cref{def:uniform-state-construction})  gives the free uniformly traced category over a symmetric monoidal one (Theorem \ref{th:free-uniform-trace}) and that such construction preserves the structure of finite biproduct category (\Cref{cor:free-uniform-trace-fb}). Here, we illustrate that the same construction additionally preserves the structure of rig categories, namely that the adjunction in  Theorem \ref{th:free-uniform-trace} restricts to an adjunction between the categories $\RIG$ and $\RIGUTr$, defined as follows.

\begin{definition}
A \emph{uniformly traced (ut) rig category} is a rig category $(\Cat{C}, \piu, \zero, \per, \uno)$  such that $(\Cat{C}, \piu, \zero)$ is a uniformly traced monoidal category. A \emph{morphism of uniformly traced rig categories} is both a rig functor and a morphisms of uniformly traced categories. We write $\RIGUTr$ for the category of uniformly traced rig categories and their morphisms.
\end{definition}

Our proof exploits the notion of whiskering that, as stated by the following proposition, enjoys useful properties in any rig category.

\begin{proposition}\label{prop:whisk}
    Let $\Cat{C}$ be a rig category, $X$ an object in $\Cat{C}$ and $\fun{L}_X, \fun{R}_X \colon \Cat{C} \to \Cat{C}$ two functors defined on objects as
    $\LW{X}{Y} \defeq X \per Y$ and $\RW{X}{Y} \defeq Y \per X$,
    and on arrows $f \colon Y \to Z$ as 
    \[ \LW{X}{f} \defeq \id{X} \per f \qquad \text{ and } \qquad   \RW{X}{f} \defeq f \per \id{X}. \]
    $\fun{L}_X$ and $\fun{R}_X$ are called, respectively, \emph{left} and \emph{right whiskering} and they satisfy the laws in Table~\ref{table:whisk}.
\end{proposition}

\begin{table}[t]
    \centering
    {\begin{tabular}{lc lc}
        \toprule
        \multicolumn{2}{l}{$1. \ \ \LW X {\id{Y}} = \id{X \per Y}$} & $2. \ \ \RW X {\id{Y}} = \id{Y \per X}$ & (\newtag{W1}{eq:whisk:id})\\[0.3em]
        \multicolumn{2}{l}{$1. \ \ \LW{X}{f ; g} = \LW{X}{f} ; \LW{X}{g}$} & $2. \ \ \RW{X}{f ; g} = \RW{X}{f} ; \RW{X}{g}$ & (\newtag{W2}{eq:whisk:funct})\\[0.3em]
        \multicolumn{2}{l}{$1. \ \ \LW{\uno}{f} = f$} & $2. \ \ \RW{\uno}{f} = f$ & (\newtag{W3}{eq:whisk:uno}) \\[0.3em]
        \multicolumn{2}{l}{$1. \ \ \LW{\zero}{f} = \id{\zero}$} & $2. \ \ \RW{\zero}{f} = \id{\zero}$ & (\newtag{W4}{eq:whisk:zero}) \\[0.3em]
        \multicolumn{2}{l}{$1. \ \ \LW{X}{f_1 \piu f_2} = \dl{X}{X_1}{X_2} ; (\LW{X}{f_1} \piu \LW{X}{f_2}) ; \Idl{X}{Y_1}{Y_2}$}  & $2. \ \ \RW{X}{f_1 \piu f_2} = \RW{X}{f_1} \piu \RW{X}{f_2}$ & (\newtag{W5}{eq:whisk:funct piu}) \\[0.3em]
        \multicolumn{2}{l}{$1. \ \ \LW{X \piu Y}{f} = \LW{X}{f} \piu \LW{Y}{f}$}  & $2. \ \ \RW{X \piu Y}{f} = \dl{Z}{X}{Y} ; ( \RW{X}{f} \piu \RW{Y}{f} ) ; \Idl{W}{X}{Y}$ & (\newtag{W6}{eq:whisk:sum}) \\[0.3em]
        \multicolumn{3}{c}{$\LW{X_1}{f_2} ; \RW{Y_2}{f_1} = \RW{X_2}{f_1} ; \LW{Y_1}{f_2}$} & (\newtag{W7}{eq:tape:LexchangeR}) \\[0.3em]
        $\RW X {\symmp{Y}{Z}} = \symmp{Y \per X}{Z \per X}$ & (\newtag{W8}{eq:whisk:symmp}) & $\symmt{X \per Y}{Z} = \LW{X}{ \symmt{Y}{Z}} ; \RW{Y}{\symmt{X}{Z}}$ & (\newtag{W9}{eq:symmper}) \\[0.3em]
        $\RW{X}{f} ; \symmt{Z}{X} = \symmt{Y}{X} ; \LW{X}{f}$ & (\newtag{W10}{eq:LRnatsym}) & $\LW{X}{\RW{Y}{f}} = \RW{Y}{\LW{X}{f}}$ & (\newtag{W11}{eq:tape:LR}) \\[0.3em]
        $\LW{X \per Y}{f} = \LW{X}{\LW{Y}{f}}$ & (\newtag{W12}{eq:tape:LL}) & $\RW{Y \per X}{f} = \RW{X}{\RW{Y}{f}}$ & (\newtag{W13}{eq:tape:RR}) \\[0.3em]
        $\RW X {\dl{Y}{Z}{W}} = \dl{Y}{Z \per X}{W \per X}$ & (\newtag{W14}{eq:whisk:dl}) & $\LW X {\dl{Y}{Z}{W}} = \dl{X \per Y}{Z}{W} ; \Idl{X}{Y \per Z}{Y \per W}$ & (\newtag{W15}{eq:whisk:Ldl}) \\
        \bottomrule
    \end{tabular}}
    \caption{The algebra of whiskerings}  
    \label{table:whisk}
\end{table}

In order to prove that $\UTr(\cdot)$ (Definition \ref{def:uniform-state-construction}) preserves the rig structure, we define below left and right whiskerings on $\UTr(\Cat{C})$ for an arbitrary rig category $\Cat{C}$.

\begin{definition}\label{def:utr-whisk}
    Let $(\Cat{C}, \piu, \zero, \per, \uno)$ be a rig category, $\UTr(\Cat{C})$ the uniformly traced category freely generated from $(\Cat{C}, \oplus, \uno)$, and $X$ an object of $\UTr(\Cat C)$. Then $\fun{L}_X, \fun{R}_X \colon \UTr(\Cat C) \to \UTr(\Cat C)$ are defined on objects as 
    $\LW{X}{Y} \defeq X \per Y$ and $\RW{X}{Y} \defeq Y \per X$,
    and on arrows $\st{f}{S} \colon Y \to Z$ as
    \[ \LW{X}{f \mid S} \defeq \st{\symmt{X}{Y}}{\zero} ; \RW{X}{f \mid S} ; \st{\symmt{Z}{X}}{\zero} \qquad \text{ and } \qquad \RW{X}{f \mid S} \defeq \st{\RW{X}{f}}{S \per X}.\] 
\end{definition}

\begin{proposition}\label{lemma:utr-whisk-algebra}
    $\fun{L}_X, \fun{R}_X \colon \UTr(\Cat{C}) \to \UTr(\Cat{C})$ satisfy the laws in Table~\ref{table:whisk}.
\end{proposition}

\begin{remark}
It is worth remarking that the proof of \eqref{eq:tape:LexchangeR} --that is equivalent to functoriality of $\per$-- crucially requires uniformity, once more.
\end{remark}

The left and right whiskerings in Definition~\ref{def:utr-whisk} allow us to define another monoidal product on $\UTr(\Cat C)$, which on objects coincides with $\per$ in $\Cat C$ and on arrows $\st{f_1}{S_1} \colon X_1 \to Y_1$, $\st{f_2}{S_2} \colon X_2 \to Y_2$ is defined as:
\begin{equation}\label{eq:utr-per}
    \st{f_1}{S_1} \per \st{f_2}{S_2} \defeq \fun{L}_{X_1}(f_2 \mid S_2) ; \fun{R}_{Y_2}(f_1 \mid S_1).
\end{equation}

One can check that $\per$ makes $\UTr(\Cat C)$ a symmetric monoidal category (see Lemma \ref{lemma:utr-per-monoidal}). Moreover the following key result holds.

\begin{theorem}\label{th:free-uniform-trace-rig}
    The adjunction in \Cref{th:free-uniform-trace} restricts to $\begin{tikzcd}
        \;\; \RIG \;\;
        \arrow[r, "\UTr"{name=F}, bend left] &
        \RIGUTr
        \arrow[l, "\fun{U}_3"{name=U}, bend left]
        \arrow[phantom, from=F, to=U, "\vdash" rotate=90]
    \end{tikzcd}$. 
    \end{theorem}

\subsection{Rigs, Biproducts and Uniform Traces}\label{ssec:rigbiut}
In the next section we are going to extend the language of tape diagrams for rig categories that are both finite biproduct and uniformly traced. We find thus convenient to introduce the following notion.

\begin{definition}
A \emph{uniformly traced finite biproduct (ut-fb) rig category} is a rig category $(\Cat{C}, \piu, \zero, \per, \uno)$  such that $(\Cat{C}, \piu, \zero)$ is a ut-fb category (see Definition \ref{def:utfbcategory}). A \emph{morphism of ut-fb rig categories} is both a rig functor and a morphisms of ut-fb categories. We write $\FBRIGUtr$ for the category of fb rig categories and their morphisms.
\end{definition}

By combining Corollary \ref{cor:free-uniform-trace-fb} and Theorem \ref{th:free-uniform-trace-rig}, one easily obtains the following.
\begin{proposition}\label{prop:free-uniform-trace-fb-rig}
    The adjunction in \Cref{th:free-uniform-trace} restricts to  $\begin{tikzcd}
        \FBRIG
        \arrow[r, "\UTr"{name=F}, bend left] &
        \FBRIGUtr
        \arrow[l, "\fun{U}_3"{name=U}, bend left]
        \arrow[phantom, from=F, to=U, "\vdash" rotate=90]
    \end{tikzcd}$. 
\end{proposition}

\section{Tape Diagrams with Uniform Traces}\label{sec:tapes}
In this section we introduce tape diagrams for rig categories with finite biproducts and uniform traces. 
The approach is similar to the one in Section~\ref{sc:tape} and it goes through the following adjunction.

\begin{equation}\label{eq:adjunction-utfb}
    \begin{tikzcd}
        \SMC \ar[r,"U_1"] &
        \CAT
        \arrow[r, "F_3"{name=F}, bend left] &
        \FBUTr
        \arrow[l, "U_3"{name=U}, bend left]
        \arrow[phantom, from=F, to=U, "\vdash" rotate=90]
    \end{tikzcd}
\end{equation}

The functors $U_1$ and $U_3$ are the obvious forgetful functors. The functor $F_3$ is the left adjoint to $U_3$ and can be described as follows.

\begin{definition}\label{def:strict ut-fb freely generated by C}
Let $\Cat{C}$ be a category. The  strict ut-fb category freely generated by $\Cat{C}$, hereafter denoted by $F_3(\Cat{C})$, has as objects words of objects of $\Cat{C}$. Arrows are terms inductively generated by the following grammar, where $A,B$ and $c$ range over arbitrary objects and arrows of $\Cat{C}$,
    \begin{equation}
        \begin{array}{rcl}
            f & ::=& \; \id{A} \; \mid \; \id{I} \; \mid \; \tapeFunct{c} \; \mid \; \sigma_{A,B}^{\perG} \; \mid \;   f ; f   \; \mid \;  f \perG f  \; \mid \; \bang{A} \; \mid \; \diag{A} \;  \mid \; \cobang{A}\; \mid \; \codiag{A} \; \mid \; \trace_{A} f\\
            \end{array}
    \end{equation}  
    modulo the axioms in  Tables~\ref{fig:freestricmmoncatax},~\ref{fig:freestrictfbcat},~\ref{tab:trace-axioms} and~\ref{tab:uniformity}. %
\end{definition}

Similarly to~\eqref{eq:adjunction}, the unit of the adjunction $\eta \colon Id_{\CAT} \Rightarrow F_3U_3$ is defined for each category $\Cat{C}$ as the identity-on-objects functor $G \colon \Cat{C} \to U_3F_3(\Cat{C})$ that maps each arrow $c$ in $\Cat{C}$ into the arrow $\tapeFunct{c}$ of $U_3 F_3(\Cat{C})$.

\begin{lemma}\label{lemma:tr-adj}
    $F_3 \colon \CAT \to \FBUTr$ is left adjoint to $U_3 \colon \FBUTr \to \CAT$.
\end{lemma}

Recall from Section \ref{sec:monoidal} the category of string diagrams $\CatString$ generated by a monoidal signature $\sign$. Hereafter, we focus on  $F_3U_1(\CatString)$, referred to as $\CatTrTape$.
The set of objects of $\CatTrTape$ is the same of $\CatTape$, i.e., words of words of sorts in $\sort$. For arrows, we extend the two-layer grammar in~\eqref{tapesGrammar} with one production accounting for the trace operation.
\begin{equation}\label{tracedTapesGrammar}
    \begin{tabular}{rc ccccccccccccccccccccc}\setlength{\tabcolsep}{0.0pt}
        $c$  & ::= & $\id{A}$ & $\!\!\! \mid \!\!\!$ & $ \id{\uno} $ & $\!\!\! \mid \!\!\!$ & $ \gen $ & $\!\!\! \mid \!\!\!$ & $ \sigma_{A,B} $ & $\!\!\! \mid \!\!\!$ & $   c ; c   $ & $\!\!\! \mid \!\!\!$ & $  c \per c$ & \multicolumn{8}{c}{\;} \\
        $\t$ & ::= & $\id{U}$ & $\!\!\! \mid \!\!\!$ & $ \id{\zero} $ & $\!\!\! \mid \!\!\!$ & $ \tapeFunct{c} $ & $\!\!\! \mid \!\!\!$ & $ \sigma_{U,V}^{\piu} $ & $\!\!\! \mid \!\!\!$ & $   \t ; \t   $ & $\!\!\! \mid \!\!\!$ & $  \t \piu \t  $ & $\!\!\! \mid \!\!\!$ & $ \bang{U} $ & $\!\!\! \mid \!\!\!$ & $\diag{U}$ & $\!\!\! \mid \!\!\!$ & $\cobang{U}$ & $\!\!\! \mid \!\!\!$ & $\codiag{U}$  &  $\!\!\! \mid \!\!\!$ & $\trace_{U}\t$    
    \end{tabular}
\end{equation}  

The terms of the first row are taken modulo the axioms in Table~\ref{fig:freestricmmoncatax} (after replacing $\perG$ with $\per$).
The terms of the second row are taken modulo the axioms in Tables \ref{fig:freestricmmoncatax},~\ref{fig:freestrictfbcat},~\ref{tab:trace-axioms} and~\ref{tab:uniformity} (after replacing $\perG$ with $\piu$ and $A,B$ with $U,V$).

As for $\CatTape$, the grammar in \eqref{tracedTapesGrammar} can be rendered diagrammatically as follows.
\begin{equation*}\label{tracedTapesDiagGrammar}
    \setlength{\tabcolsep}{2pt}
    \begin{tabular}{rc cccccccccccc}
        $c$  & ::= &  $\wire{A}$ & $\mid$ & $ 
    \InputIfFileExists{empty.tikz}{}{\input{./tikz/empty.tikz}}
 $ & $\mid$ & $ \Cgen{\gen}{A}{B}  $ & $\mid$ & $ \Csymm{A}{B} $ & $\mid$ & $ 
    \InputIfFileExists{seq_compC.tikz}{}{\input{./tikz/seq_compC.tikz}}
   $ & $\mid$ & $  
    \InputIfFileExists{par_compC.tikz}{}{\input{./tikz/par_compC.tikz}}
$ & \\
        $\t$ & ::= & $\Twire{U}$ & $\mid$ & $ 
    \InputIfFileExists{empty.tikz}{}{\input{./tikz/empty.tikz}}
 $ & $\mid$ & $ \Tcirc{c}{U}{V}  $ & $\mid$ & $ \Tsymmp{U}{V} $ & $\mid$ & $ 
    \InputIfFileExists{tapes/seq_comp.tikz}{}{\input{./tikz/tapes/seq_comp.tikz}}
  $ & $\mid$ & $  
    \InputIfFileExists{tapes/par_comp.tikz}{}{\input{./tikz/tapes/par_comp.tikz}}
$ & $\mid$ \\
             &     & \multicolumn{12}{l}{$\Tcounit{U} \; \mid \; \Tcomonoid{U} \; \mid \; \Tunit{U} \; \mid \; \Tmonoid{U} \; \mid \; \TTraceMon{\t}[U][P][Q]$}
    \end{tabular}
\end{equation*}

\begin{table}
\[\begin{array}{ccc}
\begin{array}{rcl}
 \bang{\zero} & \defeq & \id{\zero}  \\ 
 \bang{U \piu P} & \defeq & \bang{U} \piu \bang{P}  \\ 
\end{array}
&
\begin{array}{rcl}
\diag{\zero} & \defeq &  \id{\zero}  \\ 
\diag{U \piu P} & \defeq & (\diag{U} \piu \diag{P}) ; (\id{U} \piu \symm{U}{P} \piu \id{P})\\ 
\end{array}
\\ \\
\begin{array}{rcl}
 \cobang{\zero} & \defeq & \id{\zero}  \\ 
 \cobang{U \piu P} & \defeq & \cobang{U} \piu \cobang{P}  \\ 
\end{array}
&
\begin{array}{rcl}
\codiag{\zero} & \defeq &  \id{\zero}  \\ 
\codiag{U \piu P} & \defeq & (\id{U} \piu \symm{U}{P} \piu \id{P}) ; (\codiag{U} \piu \codiag{P}) \\ 
\end{array}
&
\begin{array}{rcl}
\trace_{\zero}(\t) & \defeq & \t \\
\trace_{U \piu P}(\t) & \defeq & \trace_{P} \trace_U (\t) 
\end{array}
\end{array}\]
\caption{Inductive definitions of $\bang{P}$, $\diag{P}$, $\cobang{P}$,  $\codiag{P}$ and  $\trace_{P}\t$ for all polynomial $P$}\label{tab:inddefutfb}
\end{table}

\begin{figure}[h!]
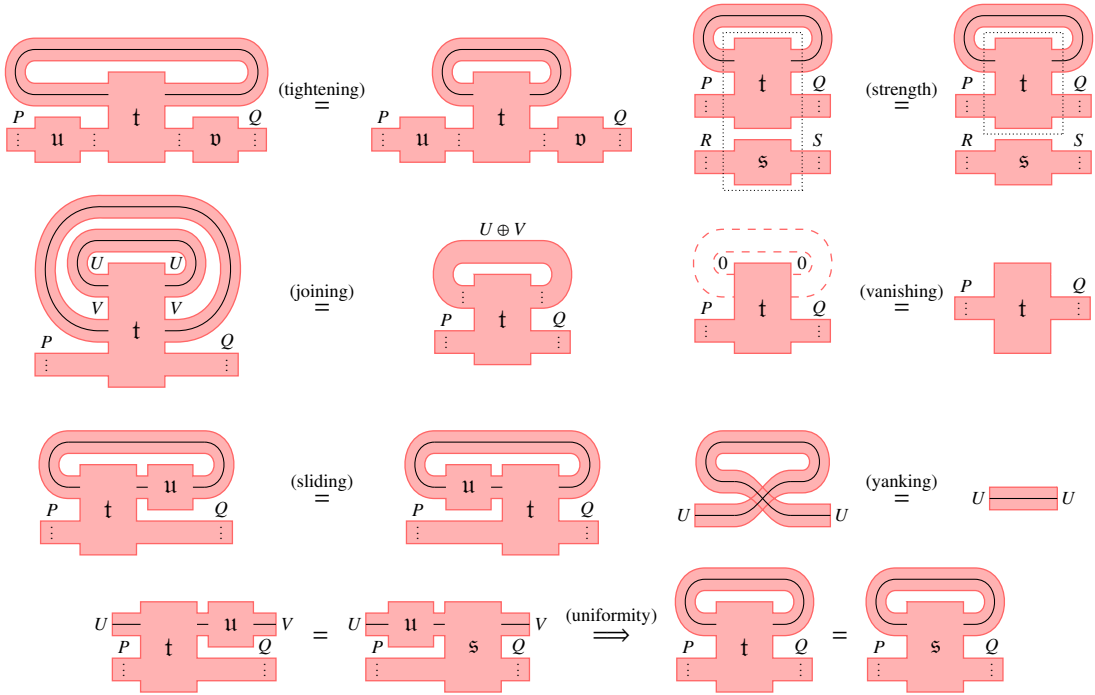

    \mylabel{ax:trace:tape:sliding}{sliding}
    \centering
    \[
        \begin{array}{c@{}c@{}c @{\quad} c@{}c@{}c}
            
    \InputIfFileExists{traceTapeAx/tight/lhs.tikz}{}{\input{./tikz/traceTapeAx/tight/lhs.tikz}}
 &\stackrel{\text{(tightening)}}{=}& 
    \InputIfFileExists{traceTapeAx/tight/rhs.tikz}{}{\input{./tikz/traceTapeAx/tight/rhs.tikz}}

            &
            
    \InputIfFileExists{traceTapeAx/strength/lhs.tikz}{}{\input{./tikz/traceTapeAx/strength/lhs.tikz}}
 &\stackrel{\text{(strength)}}{=}& 
    \InputIfFileExists{traceTapeAx/strength/rhs.tikz}{}{\input{./tikz/traceTapeAx/strength/rhs.tikz}}

            \\[20pt]
            
    \InputIfFileExists{traceTapeAx/join/lhs.tikz}{}{\input{./tikz/traceTapeAx/join/lhs.tikz}}
 &\stackrel{\text{(joining)}}{=}& 
    \InputIfFileExists{traceTapeAx/join/rhs.tikz}{}{\input{./tikz/traceTapeAx/join/rhs.tikz}}

            &
            
    \InputIfFileExists{traceTapeAx/vanish/lhs.tikz}{}{\input{./tikz/traceTapeAx/vanish/lhs.tikz}}
 &\stackrel{\text{(vanishing)}}{=}& 
    \InputIfFileExists{traceTapeAx/vanish/rhs.tikz}{}{\input{./tikz/traceTapeAx/vanish/rhs.tikz}}

            \\[40pt]
            
    \InputIfFileExists{traceTapeAx/slide/lhs.tikz}{}{\input{./tikz/traceTapeAx/slide/lhs.tikz}}
 &\stackrel{\text{(sliding)}}{=}& 
    \InputIfFileExists{traceTapeAx/slide/rhs.tikz}{}{\input{./tikz/traceTapeAx/slide/rhs.tikz}}

            &
            
    \InputIfFileExists{traceTapeAx/yank/lhs.tikz}{}{\input{./tikz/traceTapeAx/yank/lhs.tikz}}
 &\stackrel{\text{(yanking)}}{=}& 
    \InputIfFileExists{traceTapeAx/yank/rhs.tikz}{}{\input{./tikz/traceTapeAx/yank/rhs.tikz}}

            \\[20pt]
            \multicolumn{6}{c}{
                
    \InputIfFileExists{traceTapeAx/uniformity/lhs.tikz}{}{\input{./tikz/traceTapeAx/uniformity/lhs.tikz}}
 = 
    \InputIfFileExists{traceTapeAx/uniformity/rhs.tikz}{}{\input{./tikz/traceTapeAx/uniformity/rhs.tikz}}
 \stackrel{(\text{uniformity})}{\implies} 
    \InputIfFileExists{traceTapeAx/uniformity/TR_lhs.tikz}{}{\input{./tikz/traceTapeAx/uniformity/TR_lhs.tikz}}
 = 
    \InputIfFileExists{traceTapeAx/uniformity/TR_rhs.tikz}{}{\input{./tikz/traceTapeAx/uniformity/TR_rhs.tikz}}

            }
          \end{array}
    \]
    \caption{Uniform trace axioms in tape diagrams.}
    \label{fig:trace-axiom-tapes}
\end{figure}

Observe that (co)monoids and traces are defined for arbitrary monomials $U$, but not for all polynomials $P$. They can easily be defined inductively by means of the coherence axioms for (co)monoids and joining and vanishing for traces: see  
Table \ref{tab:inddefutfb}.

In the same way in which $\CatTape$ is the free $\sort$-sesquistrict fb rig category, $\CatTrTape$ is the free uniformly traced one.
\begin{theorem}\label{thm:freeut-fb}
$\CatTrTape$ is the free sesquistrict ut-fb rig category generated by the monoidal signature $(\sort, \sign)$.
\end{theorem}
One can prove the above theorem by extending the inductive definitions of whiskerings for tapes in Table \ref{tab:producttape} with the cases for traces given in Table \ref{tab:wisktraces} and then extends the proof of \Cref{thm:Tapes is free sesquistrict generated by sigma} by considering this additional inductive case. Hereafter we illustrate a more modular proof that allows to reuse Theorem \ref{thm:Tapes is free sesquistrict generated by sigma} and the free uniform state construction discussed in Section \ref{sec:traced}.

\begin{table}
    \begin{center}
        \hfill
  \[ \begin{array}{rcl|rcl}
    \LW U {\trace_V \t } &\defeq& \trace_{UV}{ \LW{U}{\t}} & \RW U {\trace_V \t} &\defeq& \trace_{VU}{\RW{U}\t} \\
    \end{array} 
     \]
      \hfill
      \caption{Extension of the definition of left and right whiskerings in Table \ref{tab:producttape} with the case of trace}\label{tab:wisktraces}
    \end{center}
\end{table}

\subsection{Proof of Theorem \ref{thm:freeut-fb}}

The adjunction in \eqref{eq:adjunction-utfb} can be decomposed in the following two adjunctions, where the leftmost is the one in~\eqref{eq:adjunction} and the rightmost is the one given by Corollary \ref{cor:free-uniform-trace-fb}.

\begin{equation}
    \begin{tikzcd}
        \SMC \ar[r,"U_1"] &
        \CAT
        \arrow[r, "F_2"{name=F}, bend left] &
        \FBC
        \arrow[r, "\UTr"{name=H}, bend left] 
        \arrow[l, "U_2"{name=U}, bend left]
        \arrow[phantom, from=F, to=U, "\vdash" rotate=90]&
        \FBUTr
        \arrow[l, "U"{name=G}, bend left]
        \arrow[phantom, from=H, to=G, "\vdash" rotate=90]
    \end{tikzcd}
\end{equation}

\begin{proposition}\label{prop:iso}
For all categories $\Cat{C}$, $\UTr F_2 (\Cat{C})$ and $F_3 (\Cat{C})$ are isomorphic as ut-fb-categories.
\end{proposition}
\begin{proof}
Observe that $U U_2 = U_3$. Since adjoints compose, then $\UTr F_2$ is left-adjoint to $U_3$. By uniqueness of adjoints, $\UTr F_2 (\Cat{C})$ is isomorphic to $F_3 (\Cat{C})$.
\end{proof}
\begin{corollary}
 $\CatTrTape$ and $\UTr(\CatTape)$ are isomorphic as ut-fb categories.
\end{corollary}

The above result suggests that to prove that $\CatTrTape$ is the free sesquistrict ut-fb rig category, one could rather prove that $\UTr(\CatTape)$ is the free one. This can be easily achieved by relying on Theorem \ref{thm:Tapes is free sesquistrict generated by sigma} and Proposition \ref{prop:free-uniform-trace-fb-rig}.

\begin{proposition}\label{prop:tracedtacesutfb}
$\UTr(\CatTape)$ is a $\sort$-sesquistrict ut-fb rig category.
\end{proposition}
\begin{proof}
By Proposition \ref{prop:free-uniform-trace-fb-rig}, $\UTr(\CatTape)$ is a ut-fb rig category. One only needs to show that the inclusion  functor $\sort \to \UTr(\CatTape)$ makes $\UTr(\CatTape)$ a $\sort$-sesquistrict rig category according to Definition \ref{def:sesquistrict rig category}. This means that we have to show that for all $A\in \sort$, $\dl{A}{Q}{R}=  \id{(A\per Q)\piu (A\per R)}$. The latter equivalence holds in $\CatTape$ (see e.g. the end of the proof of Theorem 5.10 in \cite{bonchi2023deconstructing}) and thus it also holds in $\UTr(\CatTape)$. %
\end{proof}

\begin{theorem}\label{thm:uniformfree}
    $\UTr(\CatTape)$ is the free sesquistrict ut-fb rig category generated by the monoidal signature $(\sort, \sign)$.
\end{theorem}
\begin{proof}
The obvious interpretation of  $(\sign, \sort)$ into $\UTr(\CatTape)$ is $(\id{\sort}, \tape{\cdot} ; \eta)$ where $\eta$ is the unit of the adjunction provided by Proposition \ref{prop:free-uniform-trace-fb-rig} mapping any tape $\t$ in $\CatTape$ into $\st{\t}{\zero}$.

Now, suppose that $M \to D$ is a $\sort$-sesquistrict ut-fb rig category with an interpretation $(\alpha_\sort, \alpha_\sign)$. 

Since $D$ is, in particular, a fb rig category then by Theorem \ref{thm:Tapes is free sesquistrict generated by sigma}, there exists an $\sort$-sesqustrict fb rig functor $(\alpha,\beta)$ with $\alpha\colon \sort \to M$ and $\beta \colon \CatTape \to D$ such that
\begin{equation}\label{eq:one}
\id{\sort}; \alpha = \alpha_\sort \qquad \text{ and } \qquad \tape{\cdot} ; \beta = \alpha_\sign\text{.}
\end{equation}
Since $D$ is a ut-fb rig category, by the adjunction in Proposition \ref{prop:free-uniform-trace-fb-rig}, there exists a unique ut-fb rig functor $\beta^\sharp \colon \UTr(\CatTape) \to D$ such that 
\begin{equation}\label{eq:two}
\eta ; \beta^\sharp = \beta
\end{equation}
From \eqref{eq:one} and \eqref{eq:two}, it immediately follows that $\tape{\cdot}; \eta ; \beta^\sharp = \tape{\cdot}; \beta = \alpha_\sign$. In summary, we have a sesquistrict ut-fb rig functor
$(\alpha\colon \sort \to M, \beta^\sharp \colon \UTr(\CatTape) \to D)$ such that 
\begin{equation}\label{eq:one}
\id{\sort}; \alpha = \alpha_\sort \qquad \text{ and } \qquad \tape{\cdot} ; \eta ; \beta^\sharp = \alpha_\sign\text{.}
\end{equation}
\end{proof}

Thanks to the isomorphism in Proposition \ref{prop:iso}, $\CatTrTape$ inherits the rig structures from $\UTr(\CatTape)$.
Let $H \colon \UTr(\CatTape) \to \CatTrTape$ and $K \colon \CatTrTape \to \UTr(\CatTape)$  be the functors witnessing the isomorphism. Then, one can define $\per$, distributors and symmetries on $\CatTrTape$ as follows:
\begin{equation}\label{eq:rigstruct}
        \t_1 \per \t_2 \defeq H(K(\t_1) \per K(\t_2)) \qquad \dl{P}{Q}{R} \defeq H(\dl{K(P)}{K(Q)}{K(R)} \mid \zero ) \qquad \symmt{P}{Q} \defeq H(\symmt{K(P)}{K(Q)}\mid \zero )
\end{equation}
The above definitions make $\CatTrTape$ and $\UTr(\CatTape)$ isomorphic as ut-fb rig categories. By Theorem \ref{thm:uniformfree}, it follows that  $\CatTrTape$ is the free sesquistrict ut-fb rig category.

\begin{remark}
Observe that the rig structure of $\CatTrTape$ in \eqref{eq:rigstruct} is defined differently than using the whiskerings in Tables \ref{tab:producttape} and \ref{tab:wisktraces} and symmetries and distributors in Table \ref{table:def dl symmt}. 
Since the definition in \eqref{eq:rigstruct} passes through the isomorphism, it is a bit unhandy. The reader can safely use those in Tables \ref{table:def dl symmt}, \ref{tab:producttape} and \ref{tab:wisktraces}, since the two constructions coincide: see Appendix \ref{app:rigtrTape} for a detailed proof.
\end{remark}

\section{Kleene Bicategories}\label{sec:kleene}

In this section we leave the rig structure aside and we consider a special type of categories with finite biproducts and traces that resembles more closely the monoidal category $(\Rel,\piu,\zero)$. 
In the next section, we will enrich such categories with the rig structure and study the corresponding tape diagrams.

\subsection{Finite Biproduct Categories with Idempotent Convolution}\label{ssec:fbic}
In any fb category $\Cat{C}$, the \emph{convolution monoid} is defined for all objects $X,Y$ and arrows $f,g\colon X \to Y$ as %
\begin{equation}\label{eq:covolution}
f+g \defeq \stringConvolution{f}{g}{X}{Y} \quad \text{(i.e., $\diag{X}; f \piu g ;\codiag{Y}$)} \quad\qquad 0 \defeq \stringBottom{X}{Y} \quad \text{(i.e., $\bang{X};\cobang{Y}$).}
\end{equation}
With this definition one can readily see that $\Cat{C}$ is enriched over $\Cat{CMon}$, the category of commutative monoids, namely each homset carries a commutative monoid
\begin{equation}\label{eq:cmon laws}
(f+g)+h = f+(g+h)
\quad
f+g=g+f
\quad
f+0=f
\end{equation}
and such monoid distributes over the composition $;$
\begin{equation}\label{eq:cmon enrichment}
(f+g);h = (f;h + g;h)
\quad
h;(f+g) = (h;f + h;g+)
\quad
f;0 =0= 0;f
\end{equation}

In this section we focus on a special kind of fb category, defined as follows. 
\begin{definition}\label{def:fbidempotent}
  A poset enriched monoidal category \(\Cat{C}\) is a \emph{finite biproduct category with idempotent convolution} iff $\Cat{C}$ has finite biproducts and the monoids $(\codiag{X},\cobang{X})$ are left adjoint to the comonoids $(\diag{X},\bang{X})$, i.e.,
  \[\id{X\piu X} \leq \codiag{X} ; \diag{X} \qquad  \diag{X};\codiag{X} \leq \id{X} \qquad \id{0} \leq \cobang{X};\bang{X} \qquad \bang{X};\cobang{X} \leq \id{X}\]
 \end{definition}
The axioms of adjunction for are illustrated by means of string diagrams in \Cref{fig:adjoint-biproducts}.
  \begin{figure}[h!]
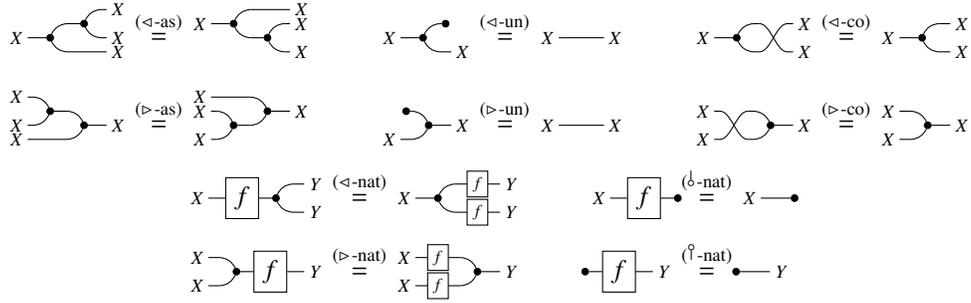

    \centering
    \mylabel{ax:comonoid:assoc}{$\diag{}$-as}
    \mylabel{ax:comonoid:unit}{$\diag{}$-un}
    \mylabel{ax:comonoid:comm}{$\diag{}$-co}
    \mylabel{ax:comonoid:nat:copy}{$\diag{}$-nat}
    \mylabel{ax:comonoid:nat:discard}{$\bang{}$-nat}
    \mylabel{ax:monoid:assoc}{$\codiag{}$-as}
    \mylabel{ax:monoid:unit}{$\codiag{}$-un}
    \mylabel{ax:monoid:comm}{$\codiag{}$-co}
    \mylabel{ax:monoid:nat:copy}{$\codiag{}$-nat}
    \mylabel{ax:monoid:nat:discard}{$\cobang{}$-nat}
    \[
    \begin{array}{c@{}c@{}c @{\qquad} c@{}c@{}c @{\qquad} c@{}c@{}c}
      
    \InputIfFileExists{fbAx/comonoid/assoc/lhs.tikz}{}{\input{./tikz/fbAx/comonoid/assoc/lhs.tikz}}
 & \stackrel{\text{($\diag{}$-as)}}{=} & 
    \InputIfFileExists{fbAx/comonoid/assoc/rhs.tikz}{}{\input{./tikz/fbAx/comonoid/assoc/rhs.tikz}}

      &
      
    \InputIfFileExists{fbAx/comonoid/unit/lhs.tikz}{}{\input{./tikz/fbAx/comonoid/unit/lhs.tikz}}
 & \stackrel{\text{($\diag{}$-un)}}{=} & 
    \InputIfFileExists{fbAx/comonoid/unit/rhs.tikz}{}{\input{./tikz/fbAx/comonoid/unit/rhs.tikz}}

      &
      
    \InputIfFileExists{fbAx/comonoid/comm/lhs.tikz}{}{\input{./tikz/fbAx/comonoid/comm/lhs.tikz}}
 & \stackrel{\text{($\diag{}$-co)}}{=} & 
    \InputIfFileExists{fbAx/comonoid/comm/rhs.tikz}{}{\input{./tikz/fbAx/comonoid/comm/rhs.tikz}}

      \\[10pt]
      
    \InputIfFileExists{fbAx/monoid/assoc/lhs.tikz}{}{\input{./tikz/fbAx/monoid/assoc/lhs.tikz}}
 & \stackrel{\text{($\codiag{}$-as)}}{=} & 
    \InputIfFileExists{fbAx/monoid/assoc/rhs.tikz}{}{\input{./tikz/fbAx/monoid/assoc/rhs.tikz}}

      &
      
    \InputIfFileExists{fbAx/monoid/unit/lhs.tikz}{}{\input{./tikz/fbAx/monoid/unit/lhs.tikz}}
 & \stackrel{\text{($\codiag{}$-un)}}{=}& 
    \InputIfFileExists{fbAx/monoid/unit/rhs.tikz}{}{\input{./tikz/fbAx/monoid/unit/rhs.tikz}}

      &
      
    \InputIfFileExists{fbAx/monoid/comm/lhs.tikz}{}{\input{./tikz/fbAx/monoid/comm/lhs.tikz}}
 & \stackrel{\text{($\codiag{}$-co)}}{=} & 
    \InputIfFileExists{fbAx/monoid/comm/rhs.tikz}{}{\input{./tikz/fbAx/monoid/comm/rhs.tikz}}

      \\[10pt]
      \multicolumn{9}{c}{
        \begin{array}{c@{}c@{}c @{\qquad} c@{}c@{}c}
          
    \InputIfFileExists{fbAx/comonoid/nat/copy/lhs.tikz}{}{\input{./tikz/fbAx/comonoid/nat/copy/lhs.tikz}}
  & \stackrel{\text{($\diag{}$-nat)}}{=} & 
    \InputIfFileExists{fbAx/comonoid/nat/copy/rhs.tikz}{}{\input{./tikz/fbAx/comonoid/nat/copy/rhs.tikz}}

          &
          
    \InputIfFileExists{fbAx/comonoid/nat/discard/lhs.tikz}{}{\input{./tikz/fbAx/comonoid/nat/discard/lhs.tikz}}
  & \stackrel{\text{($\bang{}$-nat)}}{=} & 
    \InputIfFileExists{fbAx/comonoid/nat/discard/rhs.tikz}{}{\input{./tikz/fbAx/comonoid/nat/discard/rhs.tikz}}

          \\[10pt]
          
    \InputIfFileExists{fbAx/monoid/nat/copy/lhs.tikz}{}{\input{./tikz/fbAx/monoid/nat/copy/lhs.tikz}}
  & \stackrel{\text{($\codiag{}$-nat)}}{=} & 
    \InputIfFileExists{fbAx/monoid/nat/copy/rhs.tikz}{}{\input{./tikz/fbAx/monoid/nat/copy/rhs.tikz}}

          &
          
    \InputIfFileExists{fbAx/monoid/nat/discard/lhs.tikz}{}{\input{./tikz/fbAx/monoid/nat/discard/lhs.tikz}}
  & \stackrel{\text{($\cobang{}$-nat)}}{=} & 
    \InputIfFileExists{fbAx/monoid/nat/discard/rhs.tikz}{}{\input{./tikz/fbAx/monoid/nat/discard/rhs.tikz}}

        \end{array}
      }
    \end{array}
    \]
    \caption{Axioms of fb categories in string diagrams.}
    \label{fig:ax-fb-string}
  \end{figure}
  \begin{figure}[h!]
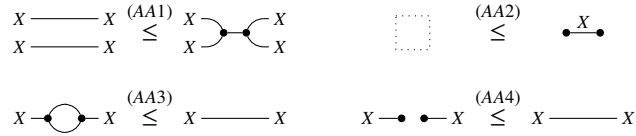

    \centering
    \mylabel{ax:adjbiprod:1}{AA1}
    \mylabel{ax:adjbiprod:2}{AA2}
    \mylabel{ax:adjbiprod:3}{AA3}
    \mylabel{ax:adjbiprod:4}{AA4}
    \[
    \begin{array}{c@{}c@{}c @{\qquad} c@{}c@{}c}
      
    \InputIfFileExists{adjbiprod/AA1_lhs.tikz}{}{\input{./tikz/adjbiprod/AA1_lhs.tikz}}
 & \axsubeq{\ref*{ax:adjbiprod:1}} & 
    \InputIfFileExists{adjbiprod/AA1_rhs.tikz}{}{\input{./tikz/adjbiprod/AA1_rhs.tikz}}
 & 
    \InputIfFileExists{adjbiprod/AA2_lhs.tikz}{}{\input{./tikz/adjbiprod/AA2_lhs.tikz}}
 & \axsubeq{\ref*{ax:adjbiprod:2}} & 
    \InputIfFileExists{adjbiprod/AA2_rhs.tikz}{}{\input{./tikz/adjbiprod/AA2_rhs.tikz}}
 \\[15pt]
      
    \InputIfFileExists{adjbiprod/AA3_lhs.tikz}{}{\input{./tikz/adjbiprod/AA3_lhs.tikz}}
 & \axsubeq{\ref*{ax:adjbiprod:3}} & 
    \InputIfFileExists{adjbiprod/AA3_rhs.tikz}{}{\input{./tikz/adjbiprod/AA3_rhs.tikz}}
 & 
    \InputIfFileExists{adjbiprod/AA4_lhs.tikz}{}{\input{./tikz/adjbiprod/AA4_lhs.tikz}}
 & \axsubeq{\ref*{ax:adjbiprod:4}} & 
    \InputIfFileExists{adjbiprod/AA4_rhs.tikz}{}{\input{./tikz/adjbiprod/AA4_rhs.tikz}}

    \end{array}
    \]
    \caption{Duality between the monoid and comonoid structures.\label{fig:adjoint-biproducts}}
  \end{figure}

As expected, the name refers to the fact that the convolution monoid in \eqref{eq:covolution} turns out to be idempotent, 
\[f+f=f\]
and thus any fb category with idempotent convolution turns out to be enriched over $\Cat{Jsl}$, the category of join semilattices. In particular, the posetal enrichement in the definition above coincides with the one induced by the semilattice structure.
\begin{lemma}\label{lemma:order-adjointness}
  In a fb category with idempotent convolution, $f \leq g$ iff $f+g=g$ for all $f,g \colon X \to Y$.
\end{lemma}
Viceversa, one can also show that in an arbitrary fb category $\Cat{C}$, if $+$ is idempotent then $\Cat{C}$ is poset enriched and the axioms in \Cref{fig:adjoint-biproducts} holds. %
A more useful fact, it is the following normal form.

\begin{proposition}[Matrix normal form]\label{prop:matrixform}
In a fb category $\Cat{C}$, any arrow $f\colon S \piu X \to T \piu Y$ has a normal form %
\[f = \stringMatrix{f}{S}{T}{X}{Y} \]
where $f_{ST}\colon S \to T$,  $f_{SY}\colon S \to Y$, $f_{XT}\colon X \to T$ and  $f_{XY}\colon X \to Y$ are defined as follows.
\begin{equation}\label{eq:components}
\begin{array}{cc}
f_{ST} \defeq (\id{S} \piu \cobang{X}); f ; (\id{T} \piu \cobang{Y}) &  f_{SY} \defeq (\id{S} \piu \cobang{X}); f ; (\cobang{T} \piu \id{Y} )\\
f_{XT} \defeq (\cobang{S} \piu \id{X} ); f ; (\id{T} \piu \cobang{Y}) &  f_{XY} \defeq(\cobang{S} \piu \id{X} ); f ; (\cobang{T} \piu \id{Y} )
\end{array}
\end{equation}
Moreover, if $\Cat{C}$ has idempotent convolution, for all $f,g\colon S \piu X \to T \piu Y$, it holds that $f\leq g$ iff
\[\begin{array}{cccc}
f_{ST} \leq g_{ST}, &  f_{SY} \leq g_{SY}, &
f_{XT} \leq g_{XT} ,&  f_{XY} \leq g_{XY}.
\end{array}
\]
\end{proposition}
The reader can easily check that $(\Rel,\piu,\zero)$ (Section \ref{sec:2monREL}) is a finite biproduct category with idempotent convolution by checking that the four inequalities in Definition \ref{def:fbidempotent} hold using the definition of monoids and comonoids from \eqref{eq:comonoidsREL}. Moreover one can easily see that the four morphisms defined by \eqref{eq:components} instantiate, in the case of $(\Rel,\piu,\zero)$, to those in \eqref{eq:dec}.

\subsection{Kleene Bicategories are Typed Kleene Algebras}

We can now introduce the main structures of this section: Kleene bicategories. These are fb categories with idempotent convolution equipped with a trace that, intuitively, behaves well w.r.t. the poset enrichement.

\begin{definition}\label{def:kleenebicategory}
    A \emph{Kleene bicategory} is a fb category with idempotent convolution that is traced monoidal such that %
    \begin{enumerate}
    \item the trace satisfies the laws in Figure \ref{fig:ineq-uniformity}: for all $f\colon S\piu X \to S \piu Y$ and $g \colon T\piu X \to T \piu Y$,
    \begin{itemize}
    \item[(AU1)] if $\exists r\colon S \to T$ such that  $f ; (r \piu \id{Y}) \leq (r \piu \id{X}) ; g$, then $\trace_{S}f \leq \trace_{S}g$.
    \item[(AU2)] if $\exists r\colon T \to S$ such that  $(r \piu \id{X})  ; f \leq  g ; (r \piu \id{Y})$, then $\trace_{S}f \leq \trace_{S}g$.
    \end{itemize}
    \item the trace satisfies the axiom in Figure \ref{fig:happy-trace}: $\trace_{X}(\codiag{X};\diag{X}) \leq \id{X}$
    \end{enumerate}
A \emph{morphism of Kleene bicategories} is a poset enriched symmetric monoidal functor preserving (co)monoids and traces. Kleene bicategories and their morphisms form a category \(\KBicat\).
\end{definition}

  \begin{figure}[h!]
    \centering
    \mylabel{ax:posetunif:1}{AU1}
    \mylabel{ax:posetunif:2}{AU2}
    \[
    \begin{array}{c}
      
    \InputIfFileExists{posetunif/AU1_lhs.tikz}{}{\input{./tikz/posetunif/AU1_lhs.tikz}}
 \leq 
    \InputIfFileExists{posetunif/AU1_rhs.tikz}{}{\input{./tikz/posetunif/AU1_rhs.tikz}}
 \stackrel{(\ref*{ax:posetunif:1})}{\implies} 
    \InputIfFileExists{posetunif/AU1TR_lhs.tikz}{}{\input{./tikz/posetunif/AU1TR_lhs.tikz}}
 \leq 
    \InputIfFileExists{posetunif/AU1TR_rhs.tikz}{}{\input{./tikz/posetunif/AU1TR_rhs.tikz}}

      \\[20pt]
      
    \InputIfFileExists{posetunif/AU2_lhs.tikz}{}{\input{./tikz/posetunif/AU2_lhs.tikz}}
 \leq 
    \InputIfFileExists{posetunif/AU2_rhs.tikz}{}{\input{./tikz/posetunif/AU2_rhs.tikz}}
 \stackrel{(\ref*{ax:posetunif:2})}{\implies} 
    \InputIfFileExists{posetunif/AU2TR_lhs.tikz}{}{\input{./tikz/posetunif/AU2TR_lhs.tikz}}
 \leq 
    \InputIfFileExists{posetunif/AU2TR_rhs.tikz}{}{\input{./tikz/posetunif/AU2TR_rhs.tikz}}

    \end{array}
    \]
    \caption{Uniformity axioms for posetal bicategories.\label{fig:ineq-uniformity}}
  \end{figure}
  \begin{figure}[h!]
    \centering
    \mylabel{ax:kb:traceid}{AT1}
    \[ 
    \InputIfFileExists{kb/traceid_lhs.tikz}{}{\input{./tikz/kb/traceid_lhs.tikz}}
 \axsubeq{\ref*{ax:kb:traceid}} 
    \InputIfFileExists{kb/traceid_rhs.tikz}{}{\input{./tikz/kb/traceid_rhs.tikz}}
 \]
    \caption{Repeating the identity.\label{fig:happy-trace}}
  \end{figure}

The axioms in Figure \ref{fig:ineq-uniformity} can be understood as the posetal extension of the uniformity axioms defined in Section \ref{sec:traced}. Note that, by antisymmetry of $\leq$, the axioms in Figure \ref{fig:ineq-uniformity} entail those in Figure \ref{fig:uniformity}. Moreover, (see Lemma \ref{lemma:equivalentUnif} in Appendix \ref{app:Kleene}) the laws (AU1) and (AU2) can equivalently be expressed by the following one.

\begin{equation}\label{eq:equivalentuni1}\tag{AU1'}
\text{If }\exists r_1,r_2\colon S \to T\text{ such that } r_2 \leq r_1\text{ and }f ; (r_1 \piu \id{Y}) \leq (r_2 \piu \id{X}) ; g\text{, then }\trace_{S}f \leq \trace_{T}g\text{;}
\end{equation}
\begin{equation}\label{eq:equivalentuni2}\tag{AU2'}
\text{If }\exists r_1,r_2\colon T \to S\text{ such that } r_2 \leq r_1\text{ and }(r_1 \piu \id{X}) ; f   \leq   g; (r_2 \piu \id{Y})\text{, then }\trace_{S}f \leq \trace_{T}g\text{;}
\end{equation}

It is worth remarking that, while the axiom of uniformity has been widely studied (see e.g. \cite{hasegawa2003uniformity}), its posetal extension in Figure \ref{fig:ineq-uniformity} is, to the best of our knowledge, novel. Instead, the axioms in Figure \ref{fig:happy-trace} already appeared in the literature (see e.g. \cite{lmcs:10963}).

\begin{figure}[h!]
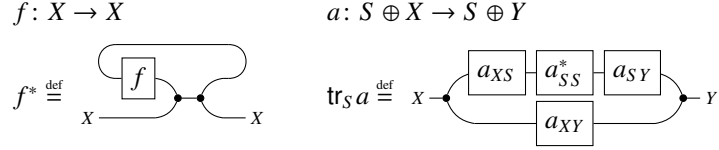

  \centering
  \[
  \begin{array}{l@{\qquad}l}
    f \colon X \to X & a \colon S \piu X \to S \piu Y \\[5pt]
    \kstar{f} \defeq \Crepetition{f}{X}{X} & \trace_S a \defeq \Ctracerep{a\vphantom{\kstar{a}_{SS}}}{S}{X}{Y}
  \end{array}
  \]
  \caption{Repetition from trace and trace from repetition in finite biproduct categories.\label{fig:star-trace}}\end{figure}

Like in any finite biproduct category with trace (see e.g. \cite{cuazuanescu1994feedback}), in a Kleene bicategory one can define for each endomorphism $f\colon X \to X$, a morphism $\kstar{f}\colon X \to X$ as in Figure \ref{fig:star-trace}.
The distinguishing property of Kleene bicategories is that $\kstar{(\cdot)}$ satisfies the laws of Kleene star as axiomatised by Kozen in \cite{Kozen94acompleteness}.

\begin{definition}
A \emph{Kleene star operator} on a category $\Cat{C}$ enriched over join-semi lattices consists of a family of operations \(\kstar{(\cdot)} \colon \Cat{C}(X,X) \to  \Cat{C}(X,X)\)  such that for all $f\colon X\to X$, $r\colon X \to Y$ and $l\colon Y \to X$:\begin{equation}\label{eq:Kllenelaw}
    \begin{array}{c@{\qquad}c}
      \id{X} + f \dcomp \kstar{f} \leq \kstar{f} & f \dcomp r \leq r \implies  \kstar{f} \dcomp r \leq r  \\
      \id{X} + \kstar{f} \dcomp f \leq \kstar{f} & l \dcomp f  \leq l \implies   l \dcomp \kstar{f} \leq l
    \end{array}
\end{equation}
A \emph{typed Kleene algebra} is a category enriched over join-semi lattices that has a Kleene star operator. A \emph{morphism of typed Kleene algebras} is a functor preserving both the structures of join semilattice and Kleene star.
Typed Kleene algebras and their morphism form a category referred as $\TKAlg$.
\end{definition}
\begin{remark}
The notion of typed Kleene algebra has been introduced by Kozen in \cite{kozen98typedkleene} in order to deal with Kleene algebras \cite{Kozen94acompleteness} with multiple sorts. In other words, a Kleene algebra is a typed Kleene algebra with a single object.
\end{remark}

On the one hand, the laws of Kleene bicategories are sufficient for defining a Kleene star operation. On the other, any Kleene star operation gives rise to a trace as in the right of Figure \ref{fig:star-trace} satisfying the laws of Kleene bicategories.

\begin{proposition}\label{prop:trace-star}
  Let $\Cat{C}$ be a fb category with idempotent convolution. $\Cat{C}$ is a Kleene bicategory iff $\Cat{C}$ has a Kleene-star operator.
\end{proposition}
Since Kleene bicategories are enriched over join semilattices, from the above result we have that

\begin{corollary}\label{cor:kleeneareka}
  All Kleene bicategories are typed Kleene algebras.
\end{corollary}
The opposite does not hold: not all typed Kleene algebras are monoidal categories. Nevertheless, from an arbitrary Kleene algebra, one can canonically build a Kleene bicategory by means of the matrix construction, illustrated in the next section.

\subsection{The Matrix Construction}\label{ssec:matrix}
Thanks to Corollary \ref{cor:kleeneareka}, one can easily construct a forgetful functor $U\colon \KBicat \to \TKAlg$: any Kleene bicategory is a typed Kleene algebra and any morphism of Kleene bicategories is a morphism of typed Kleene algebras. To see the latter, observe that preserving $\kstar{(\cdot)}$, as defined in \Cref{fig:star-trace}, and the join-semi lattice, as in \eqref{eq:covolution}, is enough to preserve traces, monoidal product and (co)monoids.

We now illustrate that $U\colon \KBicat \to \TKAlg$ has a left adjoint provided by the \emph{matrix construction}, also known as \emph{biproduct completion} \cite{coecke2017two,mac2013categories}. In \cite[Exercises VIII.2.5-6]{mac2013categories}, it is shown that there exists an adjunction in between $\CMonCat$, the category of $\Cat{CMon}$-enriched categories, and $\FBC$, the category of fb categories.
\begin{equation}\label{eq:matrixadj}
\begin{tikzcd}
        \CMonCat
        \arrow[r, "\MatFun"{name=F}, bend left] &
        \FBC
        \arrow[l, "\fun{U}"{name=U}, bend left]
        \arrow[phantom, from=F, to=U, "\vdash" rotate=90]
    \end{tikzcd}
\end{equation}
The functor $U$ is the obvious forgetful functor: as recalled in Section \ref{ssec:fbic}, every fb category is  $\Cat{CMon}$-enriched. 
Given a $\CMon$-enriched category $\Cat S$, one can form the biproduct completion of $\Cat S$, denoted as $\Mat{\Cat S}$. Its objects are formal sums of objects of $\Cat S$, while a morphism $M \colon \Piu[k=1][n]{A_k} \to \Piu[k=1][m]{B_k}$ is a $m \times n$ matrix where $M_{ji} \in \Cat S[A_i,B_j]$. Composition is given by matrix multiplication, with the addition being the plus operation on the homsets (provided by the enrichment) and multiplication being composition. The identity morphism of $\Piu[k=1][n]{A_k}$ is given by the $n \times n$ matrix $(\delta_{ji})$, where $\delta_{ji} = \id{A_j}$ if $i=j$, while if $i \neq j$, then $\delta_{ji}$ is the zero morphism of $\Cat S[A_i,A_j]$.

\begin{proposition}\label{prop:matrices-kleene-bicategory}
Let  \(\Cat{K}\) be a typed Kleene algebra. Then  \(\Mat{\Cat{K}}\) is a Kleene bicategory.
\end{proposition}

More generally, one can show that  the functor \(\MatFun \colon \CMonCat \to \fbCat\) restricts to typed Kleene algebras and Kleene bicategories
and that this gives rise to the left adjoint to $U\colon \KBicat \to \TKAlg$.

\begin{corollary}\label{cor:adjKleene}
The adjunction in \eqref{eq:matrixadj} restricts to 
\[
\begin{tikzcd}
        \TKAlg
        \arrow[r, "\MatFun"{name=F}, bend left] &
        \;\; \KBicat \;\;
        \arrow[l, "\fun{U}"{name=U}, bend left]
        \arrow[phantom, from=F, to=U, "\vdash" rotate=90]
    \end{tikzcd}
\]
\end{corollary}

\section{Kleene Tapes}\label{sec:kleene-tapes}
In this section, we combine the structure of Kleene bicategories from Section \ref{sec:kleene} with the one of rig categories from Section \ref{sec:rigcategories}.
We illustrate the corresponding tape diagrams, named Kleene tapes, and the corresponding notion of theories. We begin by introducing the structures of interest.

\begin{definition}\label{def:kcrig}
A poset enriched rig category $\Cat{C}$ is said to be a \emph{Kleene rig category}  if
$(\Cat{C}, \piu, \zero)$ is a Kleene bicategory. A \emph{morphism of Kleene rig categories} is a poset enriched rig functor that is also a Kleene morphism.
\end{definition}

In any Kleene rig category $\per$ distributes over the convolution monoid, or more precisely the join-semi lattice, in \eqref{eq:covolution}.

\begin{lemma}\label{lemma:ditributivityper}
For all $f_1,f_2\colon X \to Y$ and $g\colon S \to T$  in a Kleene rig category, it holds that
\[(f_1+f_2)\per g =(f_1\per g +f_2\per g) \qquad g \per (f_1+f_2) = (f_1\per g + f_2 \per g) \qquad 0 \per g = 0 = g \per 0\text{.}\] 
\end{lemma}

The following two results illustrate the interaction of the product $\per$  with the Kleene star and  trace.

\begin{proposition}\label{prop:star-per}
	For all $a \colon X \to X$ and $b \colon Y \to Y$ in a Kleene rig category, it holds that \[\kstar{(a \per b)} \leq \kstar{a} \per \kstar{b}\text{.}\]
\end{proposition}
\begin{proof}
	First observe that the following inequality holds:
	\begin{align*}
		(a \per b) ; (\kstar{a} \per \kstar{b})
		&= (a ; \kstar{a}) \per (b ; \kstar{b}) \tag{\Cref{fig:freestricmmoncatax}} \\ 
		&\leq (\kstar{a} \per \kstar{b}) \tag{\ref{eq:Kllenelaw}} \\
	\end{align*}

	Thus, by~\eqref{eq:Kllenelaw}, it follows that
	
	\begin{equation}\label{eq:A}
		\kstar{(a \per b)} ; (\kstar{a} \per \kstar{b}) \leq \kstar{a} \per \kstar{b}.
	\end{equation}

	To conclude, observe that the following holds:
	\begin{align*}
		\kstar{(a \per b)} 
		&= \kstar{(a \per b)} ; \id{X \per Y} \tag{\Cref{fig:freestricmmoncatax}} \\ 
		&= \kstar{(a \per b)} ; (\id{X} \per \id{Y}) \tag{\Cref{fig:freestricmmoncatax}} \\ 
		&\leq \kstar{(a \per b)} ; (\kstar{a} \per \kstar{b}) \tag{\ref{eq:Kllenelaw}} \\
		&\leq \kstar{a} \per \kstar{b} \tag{\ref{eq:A}}
	\end{align*}
	
\end{proof}

\begin{proposition}\label{prop:trace-per}
	For all $f \colon S \piu X \to S \piu Y$ and $f' \colon S' \piu X' \to S' \piu Y'$ in a Kleene rig category, it holds that 
	\[ 
		\trace_{S \per S'} \begin{psmallmatrix}
			f_{SS} \per f'_{S' S'} & f_{SY} \per f'_{S'Y'} \\
			f_{XS} \per f'_{X' S'} & f_{XY} \per f'_{X'Y'}
		\end{psmallmatrix}
		\; \leq \;
		\trace_{S} f \; \per \; \trace_{S'} f'
	 \]
	 where $\begin{psmallmatrix} f_{SS} & f_{SY} \\ f_{XS} & f_{XY} \end{psmallmatrix}$ and $\begin{psmallmatrix} f'_{S'S'} & f'_{S'Y'} \\ f'_{X'S'} & f'_{X'Y'} \end{psmallmatrix}$ are, respectively, the matrix normal forms of $f$ and $f'$.
\end{proposition}
\begin{proof}
	\begin{align*}
		\trace_{S} f \; \per \; \trace_{S'} f'
		&= (f_{XS} ; \kstar{f_{SS}} ; f_{SY} + f_{XY}) \per (f'_{X'S'} ; \kstar{(f'_{S'S'})} ; f'_{S'Y'} + f'_{X'Y'}) \tag{\Cref{fig:star-trace}} \\
		&= \begin{array}[t]{cl}
			\multicolumn{2}{c}{ ((f_{XS} ; \kstar{f_{SS}} ; f_{SY}) \per (f'_{X'S'} ; \kstar{(f'_{S'S'})} ; f'_{S'Y'})) } \\ 
			+ & ((f_{XS} ; \kstar{f_{SS}} ; f_{SY} + f_{XY}) \per f'_{X'Y'}) \\
			+ & (f_{XY} \per (f'_{X'S'} ; \kstar{(f'_{S'S'})} ; f'_{S'Y'})) \\
			+ & (f_{XY} \per f'_{X'Y'})
		\end{array} \tag{Lemma \ref{lemma:ditributivityper}}  \\
		&\geq { ((f_{XS} ; \kstar{f_{SS}} ; f_{SY}) \per (f'_{X'S'} ; \kstar{(f'_{S'S'})} ; f'_{S'Y'})) \; + \; (f_{XY} \per f'_{X'Y'}) }  \tag{\Cref{lemma:order-adjointness}} \\
		&= (f_{XS} \per f'_{X'S'}) ; (\kstar{f_{SS}} \per \kstar{(f'_{S'S'})}) ; (f_{SY} \per f'_{X'Y'}) + (f_{XY} \per f'_{X'Y'}) \tag{\Cref{fig:freestricmmoncatax}} \\
		&\geq (f_{XS} \per f'_{X'S'}) ; \kstar{(f_{SS} \per f'_{S'S'})} ; (f_{SY} \per f'_{X'Y'}) + (f_{XY} \per f'_{X'Y'}) \tag{\Cref{prop:star-per}} \\
		&= \trace_{S \per S'} \begin{psmallmatrix}
			f_{SS} \per f'_{S' S'} & f_{SY} \per f'_{S'Y'} \\
			f_{XS} \per f'_{X' S'} & f_{XY} \per f'_{X'Y'}
		\end{psmallmatrix} \tag{\Cref{fig:star-trace}}
	\end{align*}
\end{proof}

\subsection{From Traced Tapes to Kleene Tapes}

\begin{figure}
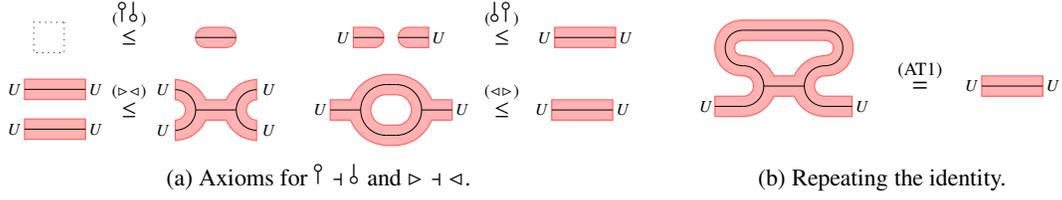
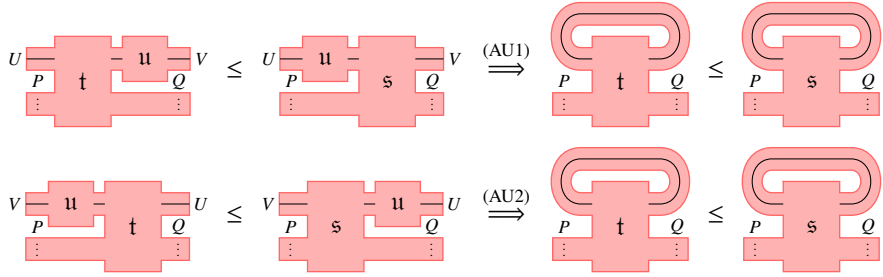

    \begin{subfigure}{0.6\linewidth}
        \mylabel{ax:relCobangBang}{$\cobang{}\bang{}$}
        \mylabel{ax:relBangCobang}{$\bang{}\cobang{}$}
        \mylabel{ax:relCodiagDiag}{$\codiag{}\diag{}$}
        \mylabel{ax:relDiagCodiag}{$\diag{}\codiag{}$}
        \scalebox{0.9}{
            \begin{tabular}{c@{}c@{}c @{\;} c@{}c@{}c}
                
    \InputIfFileExists{empty.tikz}{}{\input{./tikz/empty.tikz}}
 & $\axsubeq{\cobang{}\bang{}}$ & 
    \InputIfFileExists{tapes/whiskered_ax/bialg2_left.tikz}{}{\input{./tikz/tapes/whiskered_ax/bialg2_left.tikz}}
 & \Tcounit{U}\Tunit{U} & $\axsubeq{\bang{}\cobang{}}$ & \Twire{U} \\[2ex]
                $\begin{aligned}\begin{gathered}\Twire{U} \\ \Twire{U} \end{gathered}\end{aligned}$ & $\axsubeq{\codiag{}\diag{}}$ & 
    \InputIfFileExists{tapes/whiskered_ax/bialg1_left.tikz}{}{\input{./tikz/tapes/whiskered_ax/bialg1_left.tikz}}
 & \TRelCodiagDiag{U} & $\axsubeq{\diag{}\codiag{}}$ & \Twire{U}
            \end{tabular}
        }
        \caption{Axioms for  $\cobang{}  \dashv \bang{}$ and $\codiag{} \dashv \diag{}$.}
        \label{fig:rel axioms}
    \end{subfigure}
    \begin{subfigure}{0.38\linewidth}
        \mylabel{ax:trace:tape:kstar-id}{AT1}
        \[
            \scalebox{0.9}{
    \InputIfFileExists{traceTapeAx/rep/lhs.tikz}{}{\input{./tikz/traceTapeAx/rep/lhs.tikz}}
} \stackrel{\text{(AT1)}}{=} \scalebox{0.9}{
    \InputIfFileExists{traceTapeAx/rep/rhs.tikz}{}{\input{./tikz/traceTapeAx/rep/rhs.tikz}}
}
        \]
        \caption{Repeating the identity.}
        \label{fig:id* tape axiom}
    \end{subfigure}
    \\[15pt]
    \begin{subfigure}{\linewidth}
        \mylabel{ax:trace:tape:AU1}{AU1}
        \mylabel{ax:trace:tape:AU2}{AU2}
        \[ 
    \InputIfFileExists{traceTapeAx/posetUnif/1/lhs.tikz}{}{\input{./tikz/traceTapeAx/posetUnif/1/lhs.tikz}}
 \leq 
    \InputIfFileExists{traceTapeAx/posetUnif/1/rhs.tikz}{}{\input{./tikz/traceTapeAx/posetUnif/1/rhs.tikz}}
 \stackrel{(\text{AU1})}{\implies} 
    \InputIfFileExists{traceTapeAx/posetUnif/1/TR_lhs.tikz}{}{\input{./tikz/traceTapeAx/posetUnif/1/TR_lhs.tikz}}
 \leq 
    \InputIfFileExists{traceTapeAx/posetUnif/1/TR_rhs.tikz}{}{\input{./tikz/traceTapeAx/posetUnif/1/TR_rhs.tikz}}
 \]
        \[ 
    \InputIfFileExists{traceTapeAx/posetUnif/2/lhs.tikz}{}{\input{./tikz/traceTapeAx/posetUnif/2/lhs.tikz}}
 \leq 
    \InputIfFileExists{traceTapeAx/posetUnif/2/rhs.tikz}{}{\input{./tikz/traceTapeAx/posetUnif/2/rhs.tikz}}
 \stackrel{(\text{AU2})}{\implies} 
    \InputIfFileExists{traceTapeAx/posetUnif/2/TR_lhs.tikz}{}{\input{./tikz/traceTapeAx/posetUnif/2/TR_lhs.tikz}}
 \leq 
    \InputIfFileExists{traceTapeAx/posetUnif/2/TR_rhs.tikz}{}{\input{./tikz/traceTapeAx/posetUnif/2/TR_rhs.tikz}}
 \]
        \caption{Posetal uniformity axioms in tape diagrams.}
        \label{fif:poset unif tape axioms}
    \end{subfigure}
    \caption{Additional axioms for Kleene Tape Diagrams}\label{fig:add}
\end{figure}

We now introduce \emph{Kleene tapes}, in a nutshell tape diagrams for Kleene rig categories.

Kleene tapes are constructed as traced tape diagrams quotiented by the additional axioms of Kleene bicategories illustrated in the form of tapes in Figure \ref{fig:add}. Since these axioms include the posetal uniformity laws ((AU1) and (AU2) in Definition \ref{def:kleenebicategory}) that are not equational, such quotient needs to be performed with some care.

Let $\basicR$ be a a set of pairs $(\t_1, \t_2)$ of arrows of $\CatTrTape$ with the same domain and codomain. We define $\precongB$ to be the set generated by the following inference system (where $\t \precongB \s$ is a shorthand for $(\t,\s)\in \precongB$).
\begin{equation}\label{eq:uniformprecong}
        \!\!\begin{array}{@{\qquad}c@{\qquad\qquad}c@{\qquad}c@{\qquad}}
            \inferrule*[right=($\basicR$)]{\t_1 \mathbin{\basicR} \t_2}{\t_1 \mathrel{\precongB} \t_2}
            &
            \inferrule*[right=($r$)]{-}{\t \mathrel{\precongB} \t}
            &    
            \inferrule*[right=($t$)]{\t_1 \mathrel{\precongB} \t_2 \quad \t_2 \mathrel{\precongB} \t_3}{\t_1 \mathrel{\precongB} \t_3}
            \\[8pt]
            \inferrule*[right=($;$)]{\t_1 \mathrel{\precongB} \t_2 \quad \s_1 \mathrel{\precongB} \s_2}{\t_1;\s_1 \mathrel{\precongB} \t_2;\s_2}
            &
            \inferrule*[right=($\piu$)]{\t_1 \mathrel{\precongB} \t_2 \quad \s_1 \mathrel{\precongB} \s_2}{\t_1\piu\s_1 \mathrel{\precongB} \t_2 \piu \s_2}
            &
            \inferrule*[right=($\per $)]{\t_1 \mathrel{\precongB} \t_2 \quad \s_1 \mathrel{\precongB} \s_2}{\t_1\per \s_1 \mathrel{\precongB} \t_2 \per \s_2}
            \\[8pt]
            \multicolumn{3}{c}{
            \inferrule*[right=($ut$-1)]{\s_2 \mathrel{\precongB} \s_1 \qquad \t_1 ; (\s_1 \piu \id{}) \mathrel{\precongB} (\s_2 \piu \id{}) ; \t_2}{\trace_{S_1}\t_1 \mathrel{\precongB} \trace_{S_2}\t_2}
            \quad\;\;
            \inferrule*[right=($ut$-2)]{ \s_2 \mathrel{\precongB} \s_1 \qquad (\s_1 \piu \id{}) ; \t_1 \mathrel{\precongB} \t_2 ; (\s_2 \piu \id{})}{\trace_{S_1}\t_1 \mathrel{\precongB} \trace_{S_2}\t_2}
            }
        \end{array}
\end{equation}

We then take $\basicK$ to be the set of pairs of tapes containing those in Figure \ref{fig:rel axioms} and in Figure \ref{fig:id* tape axiom}. More explicitly,

\begin{align*}
 \basicK\defeq & \{(\id{P\piu P}  \, , \, \codiag{P};\diag{P} \mid P \in Ob(\CatTrTape) \} \cup \{( \diag{P}; \codiag{P} \,,\, \id{P} ) \mid P \in Ob(\CatTrTape)\} \cup \\
 & \{(\id{\zero}  \, , \, \cobang{P};\bang{P} \mid P \in Ob(\CatTrTape) \} \cup \{( \bang{P}; \cobang{P} \,,\, \id{P} ) \mid P \in Ob(\CatTrTape)\} \cup \\
 & \{(\trace_{P}(\codiag{P};\diag{P}) \,,\, \id{P})\mid P \in Ob(\CatTrTape)\}\text{.}
\end{align*}
We fix $\sim_\basicK \defeq \leq_{\basicK} \cap \leq_{\basicK}$.

With these definitions we can construct $\CatKTape$, the Kleene rig category of Kleene tapes. Objects are the same of $\CatTrTape$. Arrows are $\sim_\basicK$-equivalence classes of arrows of $\CatTrTape$. Every homset $\CatKTape[P,Q]$ is ordered by $\precongK$. One can easily check that the construction of $\CatKTape$ is well defined and that it gives rise to a sesquistrict Kleene rig category (see Proposition \ref{prop:tapeisKleenerig}). More importantly, $\CatKTape$ is the freely generated one.

\begin{theorem}\label{thm:Kleenetapesfree}
$\CatKTape$ is the free sesquistrict Kleene rig category generated by the monoidal signature $(\sort, \sign)$.
\end{theorem}

\section{Cartesian Bicategories}\label{sc:cb:background}

In Section~\ref{sec:2monREL} we gave the definition of cartesian bicategory (Definition~\ref{def:cartesian bicategory}). In this section we recall some of its properties that will be useful later on.

\begin{proposition}\label{prop:dagger}
    Let $\Cat{C}$ be a cartesian bicategory. There is an identity on objects isomorphism $\op{(\cdot)} \colon \Cat{C} \to \opcat{\Cat{C}}$ defined for all arrows $f \colon X \to Y$ as 
    \begin{equation}\label{def:dagger}\tag{\dag}
        \op{f} \defeq \stringOp{\scriptstyle f}{X}{Y}.
    \end{equation}
    Moreover, $\op{(\cdot)}$ is an isomorphism of cartesian bicategories, i.e. the laws in Table~\ref{table:re:daggerproperties} hold.
\end{proposition}
\begin{proof}
    See Theorem 2.4 in~\cite{Carboni1987}.
\end{proof}

\begin{table}[ht!]
    \centering
    $
    \begin{array}{@{}c@{\quad}c@{\quad}c@{\quad}c@{}}
        \toprule
        \multicolumn{2}{c}{
            \text{if } f \leq g \text{ then }\op{f} \leq \op{g}
        }
        &
        \multicolumn{2}{c}{
            \op{(\op{f})}= f
        }
        \\[4pt]
    \op{(f ; g)} = \op{g} ; \op{f}
    &\;\;\op{(\id{X})}=\id{X}
    &\;\;\op{(\cocopier{X})}= \copier{X}
    &\;\;\op{(\codischarger{X})}= \discharger{X}
    \\[4pt]
    \op{(f \perG g)} = \op{f} \perG \op{g}
    &\;\;\op{(\sigma^{\odot}_{X,Y})} = \sigma^{\odot}_{Y,X}
    &\;\;\op{(\copier{X})}= \cocopier{X}
    &\;\;\op{(\discharger{X})}= \codischarger{X}
    \\[4pt]
    \bottomrule
    \end{array}
    $
    \caption{Properties of $\op{(\cdot)} \colon \Cat{C} \to \opcat{\Cat{C}}$}\label{table:re:daggerproperties}
\end{table}

\begin{remark}
    From now on, we will depict a morphism $f \colon X \to Y$ as $\stringBox{f}{X}{Y}$, and use $\stringBoxOp{f}{X}{Y}$ as syntactic sugar for $\op{f}$.
\end{remark}

\begin{definition}
    In a cartesian bicategory $\Cat{C}$, an arrow $f \colon X \to Y$ is said to be \emph{single valued} iff satisfies~\eqref{eq:cb:sv}, \emph{total} iff satisfies~\eqref{eq:cb:tot}, \emph{injective} iff satisfies~\eqref{eq:cb:inj} and \emph{surjective} iff satisfies~\eqref{eq:cb:sur}. A \emph{map} is an arrow that is both single valued and total. Similarly, a \emph{comap} is an arrow that is both injective and surjective.
    
    \smallskip

    {
    \noindent\begin{minipage}{0.48\linewidth}
        \begin{equation}\label{eq:cb:sv}\tag{SV}
            \hspace*{-1.7cm}
    \begin{tikzpicture}[scale=1.5]
	\begin{pgfonlayer}{nodelayer}
		\node [style=none] (81) at (-0.75, 0) {};
		\node [style=label] (104) at (-1.25, 0) {$X$};
		\node [style=black] (112) at (0, 0) {};
		\node [style=label] (115) at (2.25, 0.375) {$Y$};
		\node [style=label] (116) at (2.25, -0.375) {$Y$};
		\node [style=none] (117) at (1.75, 0.375) {};
		\node [style=none] (118) at (1.75, -0.375) {};
		\node [style=none] (119) at (0.5, 0.375) {};
		\node [style=none] (120) at (0.5, -0.375) {};
		\node [style=bbox, scale=0.8] (121) at (1, 0.375) {$f$};
		\node [style=bbox, scale=0.8] (122) at (1, -0.375) {$f$};
	\end{pgfonlayer}
	\begin{pgfonlayer}{edgelayer}
		\draw (81.center) to (112);
		\draw [bend left] (112) to (119.center);
		\draw [bend right] (112) to (120.center);
		\draw (120.center) to (122);
		\draw (122) to (118.center);
		\draw (117.center) to (121);
		\draw (121) to (119.center);
	\end{pgfonlayer}
\end{tikzpicture}
}
 \leq 
    \begin{tikzpicture}[scale=1.5]
	\begin{pgfonlayer}{nodelayer}
		\node [style=bbox, scale=0.9] (107) at (0.25, 0) {$f$};
		\node [style=label] (110) at (-1, 0) {$X$};
		\node [style=none] (117) at (-0.5, 0) {};
		\node [style=label] (120) at (2.25, -0.375) {$Y$};
		\node [style=black] (121) at (1, 0) {};
		\node [style=none] (122) at (1.5, -0.375) {};
		\node [style=none] (123) at (1.5, 0.375) {};
		\node [style=label] (124) at (2.25, 0.375) {$Y$};
		\node [style=none] (125) at (1.75, -0.375) {};
		\node [style=none] (126) at (1.75, 0.375) {};
	\end{pgfonlayer}
	\begin{pgfonlayer}{edgelayer}
		\draw [bend right] (123.center) to (121);
		\draw [bend right] (121) to (122.center);
		\draw (107) to (121);
		\draw (123.center) to (126.center);
		\draw (125.center) to (122.center);
		\draw (117.center) to (107);
	\end{pgfonlayer}
\end{tikzpicture}
}

        \end{equation}
    \end{minipage}
    \begin{minipage}{0.48\linewidth}
        \begin{equation}\label{eq:cb:adj:sv}
            \stringSpan{f}{Y} \leq \stringId{Y}
        \end{equation}
    \end{minipage}
    
    \smallskip
    
    \noindent\begin{minipage}{0.48\linewidth}
        \begin{equation}\label{eq:cb:tot}\tag{TOT}
             
    \begin{tikzpicture}[scale=1.5]
	\begin{pgfonlayer}{nodelayer}
		\node [style=none] (114) at (-1, 0.75) {};
		\node [style=none] (115) at (-1, -0.75) {};
		\node [style=none] (117) at (-1, 0) {};
		\node [style=label] (118) at (-1.5, 0) {$X$};
		\node [style=black] (119) at (0.65, 0) {};
		\node [style=label] (121) at (1.25, 0) {};
	\end{pgfonlayer}
	\begin{pgfonlayer}{edgelayer}
		\draw (117.center) to (119);
	\end{pgfonlayer}
\end{tikzpicture}
}
 \leq 
    \begin{tikzpicture}[scale=1.5]
	\begin{pgfonlayer}{nodelayer}
		\node [style=none] (113) at (2, 0.75) {};
		\node [style=none] (114) at (-1, 0.75) {};
		\node [style=none] (115) at (-1, -0.75) {};
		\node [style=none] (116) at (2, -0.75) {};
		\node [style=none] (117) at (-1, 0) {};
		\node [style=label] (118) at (-1.5, 0) {$X$};
		\node [style=black] (119) at (1.15, 0) {};
		\node [style=bbox, scale=0.9] (120) at (0, 0) {$f$};
		\node [style=label] (121) at (2.75, 0) {};
	\end{pgfonlayer}
	\begin{pgfonlayer}{edgelayer}
		\draw (120) to (119);
		\draw (120) to (117.center);
	\end{pgfonlayer}
\end{tikzpicture}
}

        \end{equation}
    \end{minipage}
    \begin{minipage}{0.48\linewidth}
        \begin{equation}\label{eq:cb:adj:tot}
            \stringId{X} \leq \stringCospan{f}{X}
        \end{equation}
    \end{minipage}
    
    \smallskip
    
    \noindent\begin{minipage}{0.48\linewidth}
        \begin{equation}\label{eq:cb:inj}\tag{INJ}
            \hspace*{-1.7cm}
    \begin{tikzpicture}[scale=1.5]
	\begin{pgfonlayer}{nodelayer}
		\node [style=none] (123) at (1.75, 0) {};
		\node [style=label] (124) at (2.25, 0) {$Y$};
		\node [style=black] (125) at (1, 0) {};
		\node [style=label] (126) at (-1.25, 0.375) {$X$};
		\node [style=label] (127) at (-1.25, -0.375) {$X$};
		\node [style=none] (128) at (-0.75, 0.375) {};
		\node [style=none] (129) at (-0.75, -0.375) {};
		\node [style=none] (130) at (0.5, 0.375) {};
		\node [style=none] (131) at (0.5, -0.375) {};
		\node [style=bbox, scale=0.8] (132) at (0, 0.375) {$f$};
		\node [style=bbox, scale=0.8] (133) at (0, -0.375) {$f$};
	\end{pgfonlayer}
	\begin{pgfonlayer}{edgelayer}
		\draw (123.center) to (125);
		\draw [bend right] (125) to (130.center);
		\draw [bend left] (125) to (131.center);
		\draw (131.center) to (133);
		\draw (133) to (129.center);
		\draw (128.center) to (132);
		\draw (132) to (130.center);
	\end{pgfonlayer}
\end{tikzpicture}
}
 \leq 
    \begin{tikzpicture}[scale=1.5]
	\begin{pgfonlayer}{nodelayer}
		\node [style=bbox, scale=0.9] (107) at (0.75, 0) {$f$};
		\node [style=label] (110) at (2, 0) {$Y$};
		\node [style=none] (117) at (1.5, 0) {};
		\node [style=label] (120) at (-1.25, -0.375) {$X$};
		\node [style=black] (121) at (0, 0) {};
		\node [style=none] (122) at (-0.5, -0.375) {};
		\node [style=none] (123) at (-0.5, 0.375) {};
		\node [style=label] (124) at (-1.25, 0.375) {$X$};
		\node [style=none] (125) at (-0.75, -0.375) {};
		\node [style=none] (126) at (-0.75, 0.375) {};
	\end{pgfonlayer}
	\begin{pgfonlayer}{edgelayer}
		\draw [bend left] (123.center) to (121);
		\draw [bend left] (121) to (122.center);
		\draw (107) to (121);
		\draw (123.center) to (126.center);
		\draw (125.center) to (122.center);
		\draw (117.center) to (107);
	\end{pgfonlayer}
\end{tikzpicture}
}

        \end{equation}
    \end{minipage}
    \begin{minipage}{0.48\linewidth}
        \begin{equation}\label{eq:cb:adj:inj}
            \stringCospan{f}{X} \leq \stringId{X}
        \end{equation}
    \end{minipage}
    
    \smallskip
    
    \noindent\begin{minipage}{0.48\linewidth}
        \begin{equation}\label{eq:cb:sur}\tag{SUR}
            \hspace*{-1.4cm}
    \begin{tikzpicture}[scale=1.5]
	\begin{pgfonlayer}{nodelayer}
		\node [style=none] (114) at (2.25, 0.75) {};
		\node [style=none] (115) at (2.25, -0.75) {};
		\node [style=none] (117) at (2.25, 0) {};
		\node [style=label] (118) at (2.75, 0) {$Y$};
		\node [style=black] (119) at (0.6, 0) {};
		\node [style=label] (121) at (0, 0) {};
	\end{pgfonlayer}
	\begin{pgfonlayer}{edgelayer}
		\draw (117.center) to (119);
	\end{pgfonlayer}
\end{tikzpicture}
}
 \leq 
    \begin{tikzpicture}[scale=1.5]
	\begin{pgfonlayer}{nodelayer}
		\node [style=none] (114) at (2.25, 0.75) {};
		\node [style=none] (115) at (2.25, -0.75) {};
		\node [style=none] (117) at (2.25, 0) {};
		\node [style=label] (118) at (2.75, 0) {$Y$};
		\node [style=black] (119) at (0.1, 0) {};
		\node [style=bbox, scale=0.9] (120) at (1.25, 0) {$f$};
		\node [style=label] (121) at (-0.5, 0) {};
	\end{pgfonlayer}
	\begin{pgfonlayer}{edgelayer}
		\draw (120) to (119);
		\draw (120) to (117.center);
	\end{pgfonlayer}
\end{tikzpicture}
}

        \end{equation}
    \end{minipage}
    \begin{minipage}{0.48\linewidth}
        \begin{equation}\label{eq:cb:adj:sur}
            \stringId{Y} \leq \stringSpan{f}{Y}
        \end{equation}
    \end{minipage}
    }
\end{definition}

\begin{lemma}\label{lemma:cb:adjoints}
    In a cartesian bicategory $\Cat{C}$, an arrow $f \colon X \to Y$ is single valued iff~\eqref{eq:cb:adj:sv}, it is total iff~\eqref{eq:cb:adj:tot}, it is injective iff~\eqref{eq:cb:adj:inj} and it is surjective iff~\eqref{eq:cb:adj:sur}. In particular, an arrow is a map iff it has a right adjoint, namely $f \dashv \op{f}$; and it is a comap iff it has a left adjoint, namely $\op{f} \dashv f$.
\end{lemma}
\begin{proof}
    See Lemma 4.4 in~\cite{Bonchi2017c}.
\end{proof}

In any cartesian bicategory, one can define a convolution monoid for all objects $X,Y$ and arrows $f,g$ as
\begin{equation}\label{eq:cb:covolution}
    f \sqcap g \defeq \stringCBConvolution{f}{g}{X}{Y} \quad \text{(i.e., $\copier{X}; f \per g ;\cocopier{Y}$)} \quad\qquad \top \defeq \stringBottom{X}{Y} \quad \text{(i.e., $\discharger{X};\codischarger{Y}$).}
\end{equation}

However, unlike the case of fb categories with idempotent convolution, cartesian bicategories are not enriched over $\CMon$. In particular, each homset carries a commutative monoid structure, i.e. the laws in~\eqref{eq:cmon laws} hold; but the laws in \eqref{eq:cmon enrichment} hold only laxly, namely 
\begin{equation}
    (f \sqcap g);h \leq (f;h \sqcap g;h)
    \qquad
    h;(f \sqcap g) \leq (h;f \sqcap h;g)
    \qquad
    f ; \top \leq \top \geq \top ; f
\end{equation}

Given that the structure defined in~\eqref{eq:cb:covolution} is an idempotent monoid, and using the third inequality above, it is easy to see that each homset of a cartesian bicategory form a meet-semilattice with top.
In particular, the following holds for every arrow $f \colon X \to Y$:
\[
    \begin{tikzpicture}[scale=1.5]
        \begin{pgfonlayer}{nodelayer}
            \node [style=none] (117) at (-1, 0) {};
            \node [style=label] (118) at (-1.5, 0) {$X$};
            \node [style=bbox, scale=0.9] (120) at (0, 0) {$f$};
            \node [style=none] (121) at (1, 0) {};
            \node [style=label] (122) at (1.5, 0) {$Y$};
        \end{pgfonlayer}
        \begin{pgfonlayer}{edgelayer}
            \draw (120) to (117.center);
            \draw (120) to (121.center);
        \end{pgfonlayer}
    \end{tikzpicture}
    =
    \begin{tikzpicture}[scale=1.5]
        \begin{pgfonlayer}{nodelayer}
            \node [style=none] (117) at (-1, 0) {};
            \node [style=label] (118) at (-1.5, 0) {$X$};
            \node [style=bbox, scale=0.9] (120) at (0, 0) {$f$};
            \node [style=none] (121) at (2, 0) {};
            \node [style=label] (122) at (2.5, 0) {$Y$};
        \end{pgfonlayer}
        \begin{pgfonlayer}{edgelayer}
            \draw (120) to (117.center);
            \draw (120) to (121.center);
        \end{pgfonlayer}
    \end{tikzpicture}        
    \leq 
    \begin{tikzpicture}[scale=1.5]
        \begin{pgfonlayer}{nodelayer}
            \node [style=none] (117) at (-1.5, 0) {};
            \node [style=label] (118) at (-2, 0) {$X$};
            \node [style=bbox, scale=0.9] (120) at (0, 0) {$f$};
            \node [style=none] (121) at (2.5, 0) {};
            \node [style=label] (122) at (3, 0) {$Y$};
            \node [style=black] (123) at (1, 0) {};
            \node [style=black] (124) at (1.75, 0) {};
        \end{pgfonlayer}
        \begin{pgfonlayer}{edgelayer}
            \draw (120) to (117.center);
            \draw (123) to (120);
            \draw (124) to (121.center);
        \end{pgfonlayer}
    \end{tikzpicture}     
    \leq
    \begin{tikzpicture}[scale=1.5]
        \begin{pgfonlayer}{nodelayer}
            \node [style=none] (117) at (-1.25, 0) {};
            \node [style=label] (118) at (-1.75, 0) {$X$};
            \node [style=none] (121) at (1, 0) {};
            \node [style=label] (122) at (1.5, 0) {$Y$};
            \node [style=black] (128) at (-0.5, 0) {};
            \node [style=black] (129) at (0.25, 0) {};
        \end{pgfonlayer}
        \begin{pgfonlayer}{edgelayer}
            \draw (129) to (121.center);
            \draw (117.center) to (128);
        \end{pgfonlayer}
    \end{tikzpicture}    
\]
and the order on the homsets coincides with the one defined by the semilattice structure.

\begin{lemma}
    In a cartesian bicategory, $f \leq g$ iff $f \sqcap g = f$ for all $f,g \colon X \to Y$.
\end{lemma}
\begin{proof}
    See Lemma 4.13 in~\cite{Bonchi2017c}.
\end{proof}

\subsection{Coreflexives in Cartesian Bicategories}

In this section we recall the notion of \emph{coreflexive} morphisms in cartesian bicategories, along with some of their key properties. Recall that a relation $R \subseteq X \times X$ is called reflexive whenever $\id{X} \subseteq R$. Dually, $R$ is said to be coreflexive when $R \subseteq \id{X}$. The concept of coreflexive relation is abstracted in cartesian bicategories as expected.

\begin{definition}
    In a cartesian bicategory, a morphism $f \colon X \to X$ is a \emph{coreflexive} if $f \leq \id{X}$.
\end{definition}

\begin{lemma}\label{lemma:coreflexive properties}
    In a cartesian bicategory, the following hold for all coreflexives $f,g \colon X \to X$:
    \begin{enumerate}
        \item $
    \InputIfFileExists{coreflexive/copyR.tikz}{}{\input{./tikz/coreflexive/copyR.tikz}}
 = 
    \InputIfFileExists{coreflexive/Rcopy.tikz}{}{\input{./tikz/coreflexive/Rcopy.tikz}}
$,
        \item $f ; g = f \sqcap g$,
        \item $f$ is transitive, i.e. $f ; f = f$,
        \item $f$ is symmetric, i.e. $f = \op{f}$,
        \item $f$ is single valued,
        \item $f$ is injective.
    \end{enumerate}
     Moreover, the following hold in any cartesian bicategory \(\Cat{C}\):
     \begin{enumerate}
        \setcounter{enumi}{6}
        \item there is an isomorphism $\mathsf{Corefl}(\Cat{C})[X, X] \cong \Cat{C}[I, X]$, where $\mathsf{Corefl}(\Cat{C})$ is the subcategory of $\Cat{C}$ whose morphisms are all and only the coreflexives,
        \item for all morphims $f \colon I \to X$ and coreflexives $g \colon X \to X$, $f ; g = f \sqcap g'$, where $g' \colon I \to X$ is the morphism corresponding to $g$ via the isomorphism above,
        \item for all morphisms $f \colon X \to X$, if $f$ is transitive, symmetric and single valued, then $f$ is a coreflexive.
    \end{enumerate}
\end{lemma}
\begin{proof}
    \begin{enumerate}
        \item We prove the two inclusions separately: 
        \[ 
    \InputIfFileExists{coreflexive/Rcopy.tikz}{}{\input{./tikz/coreflexive/Rcopy.tikz}}
 \stackrel{\eqref{ax:copiernat}}{\leq} 

}
$. Thus $f$ is single valued by means of \Cref{lemma:cb:adjoints}.
        \item $
    %
}
 \stackrel{(4)}{=} 
    }
 \!\!\stackrel{(\text{coreflexive})}{\leq} \!\! 
    }
$. Thus $f$ is injective by means of \Cref{lemma:cb:adjoints}.
        \item Consider the functions $i \colon \mathsf{Corefl}(\Cat{C})[X, X] \to \Cat{C}[I, X]$ and $c \colon \Cat{C}[I, X] \to \mathsf{Corefl}(\Cat{C})[X, X]$ defined as follows: %
            \[ i(
    \InputIfFileExists{coreflexive/7/f.tikz}{}{\input{./tikz/coreflexive/7/f.tikz}}
) \defeq 
    %
}
) \defeq 
    \InputIfFileExists{coreflexive/7/Cg.tikz}{}{\input{./tikz/coreflexive/7/Cg.tikz}}
  \]
        and observe that $c(\, 
    %
}
 \!) = 
    \InputIfFileExists{coreflexive/7/Cg.tikz}{}{\input{./tikz/coreflexive/7/Cg.tikz}}
 \stackrel{\eqref{ax:dischargernat}}{\leq} 
    \InputIfFileExists{coreflexive/7/step1.tikz}{}{\input{./tikz/coreflexive/7/step1.tikz}}
 \stackrel{\eqref{ax:cocopierun}}{=} 
    }
.  \]
        To conclude, observe that $i$ and $c$ are inverse to each other:
        \[ i(c(\, 
    %
                
        $. %
        \item $
    }
 \stackrel{(\text{transitive})}{=} 
    }
 \stackrel{(\text{symmetric})}{=} 
    }
 \stackrel{(\text{single valued})}{\leq} 
    }
$.
    \end{enumerate}
\end{proof}

\begin{remark}\label{rem:notation-coreflexives}
    From now on we will use $\stringCorefl{f}{X}$ to depict coreflexive morphisms. This graphical representation is, in some sense, orientation agnostic, and it reflects the fact that coreflexives are symmetric, as stated by \Cref{lemma:coreflexive properties}.4.
\end{remark}

\section{Kleene-Cartesian Tape Diagrams}\label{sec:cb}

\begin{figure}[ht!]
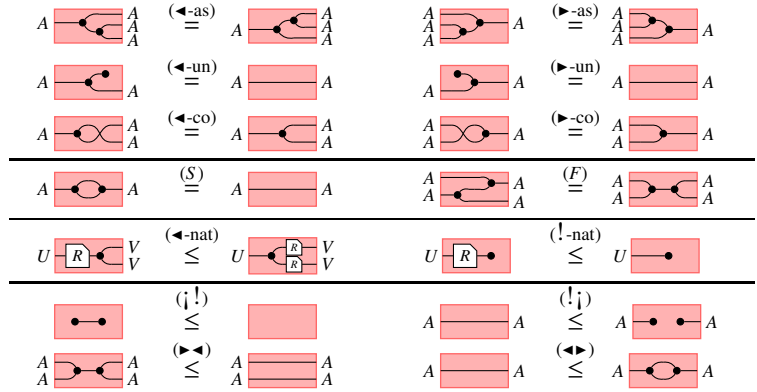

    \mylabel{ax:symmtinv}{$\sigma^{\per}$-inv}
    \mylabel{ax:symmtnat}{$\sigma^{\per}$-nat}
    \mylabel{ax:copieras}{$\copier{}$-as}
    \mylabel{ax:copierun}{$\copier{}$-un}
    \mylabel{ax:copierco}{$\copier{}$-co}
    \mylabel{ax:cocopieras}{$\cocopier{}$-as}
    \mylabel{ax:cocopierun}{$\cocopier{}$-un}
    \mylabel{ax:cocopierco}{$\cocopier{}$-co}
    \mylabel{ax:specfrob}{S}
    \mylabel{ax:frob}{F}
    \mylabel{ax:copiernat}{$\copier{}$-nat}
    \mylabel{ax:dischargernat}{$\discharger{}$-nat}
    \mylabel{ax:dischargeradj1}{$\codischarger{}\discharger{}$}
    \mylabel{ax:dischargeradj2}{$\discharger{}\codischarger{}$}
    \mylabel{ax:copieradj1}{$\cocopier{}\copier{}$}
    \mylabel{ax:copieradj2}{$\copier{}\cocopier{}$}
	\centering
	\begin{tabular}{c@{}c@{}c @{\qquad} c@{}c@{}c}
        
    \InputIfFileExists{cb/monoid_assoc_left.tikz}{}{\input{./tikz/cb/monoid_assoc_left.tikz}}
 & $\axeq{\copier{}\text{-as}}$ & 
    \InputIfFileExists{cb/monoid_assoc_right.tikz}{}{\input{./tikz/cb/monoid_assoc_right.tikz}}
 & 
    \InputIfFileExists{cb/comonoid_assoc_left.tikz}{}{\input{./tikz/cb/comonoid_assoc_left.tikz}}
 & $\axeq{\cocopier{}\text{-as}}$ & 
    \InputIfFileExists{cb/comonoid_assoc_right.tikz}{}{\input{./tikz/cb/comonoid_assoc_right.tikz}}
 \\
        
    \InputIfFileExists{cb/monoid_unit_left.tikz}{}{\input{./tikz/cb/monoid_unit_left.tikz}}
 & $\axeq{\copier{}\text{-un}}$ & 
    \InputIfFileExists{cb/monoid_unit_right.tikz}{}{\input{./tikz/cb/monoid_unit_right.tikz}}
 & 
    \InputIfFileExists{cb/comonoid_unit_left.tikz}{}{\input{./tikz/cb/comonoid_unit_left.tikz}}
 & $\axeq{\cocopier{}\text{-un}}$ & 
    \InputIfFileExists{cb/comonoid_unit_right.tikz}{}{\input{./tikz/cb/comonoid_unit_right.tikz}}
 \\
        
    \InputIfFileExists{cb/monoid_comm_left.tikz}{}{\input{./tikz/cb/monoid_comm_left.tikz}}
 & $\axeq{\copier{}\text{-co}}$ & 
    \InputIfFileExists{cb/monoid_comm_right.tikz}{}{\input{./tikz/cb/monoid_comm_right.tikz}}
 & 
    \InputIfFileExists{cb/comonoid_comm_left.tikz}{}{\input{./tikz/cb/comonoid_comm_left.tikz}}
 & $\axeq{\cocopier{}\text{-co}}$ & 
    \InputIfFileExists{cb/comonoid_comm_right.tikz}{}{\input{./tikz/cb/comonoid_comm_right.tikz}}
 \\
        \hline
        
    \InputIfFileExists{cb/spec_frob_left.tikz}{}{\input{./tikz/cb/spec_frob_left.tikz}}
 & $\axeq{S}$ & 
    \InputIfFileExists{cb/monoid_unit_right.tikz}{}{\input{./tikz/cb/monoid_unit_right.tikz}}
 & 
    \InputIfFileExists{cb/frob_left.tikz}{}{\input{./tikz/cb/frob_left.tikz}}
 & $\axeq{F}$ & 
    \InputIfFileExists{cb/frob_center.tikz}{}{\input{./tikz/cb/frob_center.tikz}}
 \\
        \hline
        
    \InputIfFileExists{cb/copier_nat_left.tikz}{}{\input{./tikz/cb/copier_nat_left.tikz}}
 & $\axsubeq{\copier{}\text{-nat}}$ & 
    \InputIfFileExists{cb/copier_nat_right.tikz}{}{\input{./tikz/cb/copier_nat_right.tikz}}
 & 
    \InputIfFileExists{cb/discharger_nat_left.tikz}{}{\input{./tikz/cb/discharger_nat_left.tikz}}
 & $\axsubeq{\discharger{}\text{-nat}}$ & 
    \InputIfFileExists{cb/discharger_nat_right.tikz}{}{\input{./tikz/cb/discharger_nat_right.tikz}}
 \\
        \hline
        
    \InputIfFileExists{cb/adjoint1_bangs_left.tikz}{}{\input{./tikz/cb/adjoint1_bangs_left.tikz}}
 & $\axsubeq{\codischarger{}\discharger{}}$ & 
    \InputIfFileExists{cb/empty.tikz}{}{\input{./tikz/cb/empty.tikz}}
 & 
    \InputIfFileExists{cb/monoid_unit_right.tikz}{}{\input{./tikz/cb/monoid_unit_right.tikz}}
 & $\axsubeq{\discharger{}\codischarger{}}$ & 
    \InputIfFileExists{cb/adjoint2_bangs_right.tikz}{}{\input{./tikz/cb/adjoint2_bangs_right.tikz}}
\\
        
    \InputIfFileExists{cb/frob_center.tikz}{}{\input{./tikz/cb/frob_center.tikz}}
 & $\axsubeq{\cocopier{}\copier{}}$ & 
    \InputIfFileExists{cb/adjoint1_diags_right.tikz}{}{\input{./tikz/cb/adjoint1_diags_right.tikz}}
 & 
    \InputIfFileExists{cb/monoid_unit_right.tikz}{}{\input{./tikz/cb/monoid_unit_right.tikz}}
 & $\axsubeq{\copier{}\cocopier{}}$ & 
    \InputIfFileExists{cb/spec_frob_left.tikz}{}{\input{./tikz/cb/spec_frob_left.tikz}}

	\end{tabular}
    \caption{Axioms of cartesian bicategories}
    \label{fig:cb axioms}
\end{figure}

In Section \ref{sec:rigcategories}  we recalled from \cite{bonchi2023deconstructing} tape diagrams for rig categories with finite biproducts.
In Section \ref{sec:tapes}, we extended tape diagrams with uniform trace and then, in Section \ref{sec:kleene-tapes}, we imposed to such diagrams the laws of Kleene bicategories.
In this section we illustrate our last step: we illustrate tape diagrams for rig categories where $(\Cat{C}, \piu, \zero)$ is a Kleene bicategory and $(\Cat{C}, \per, \uno)$ is a cartesian bicategory.
More precisely, we are interested in the following structures.

\begin{definition}\label{def:kcbrig}
A \emph{Kleene-Cartesian rig category} (shortly, kc rig) is a poset enriched rig category $\Cat{C}$ such that 
\begin{enumerate}
\item $(\Cat{C}, \piu, \zero)$ is a Kleene bicategory;
\item $(\Cat{C}, \per, \uno)$ is a cartesian bicategory;
\item the (co)monoids of both  $(\Cat{C}, \piu, \zero)$ and $(\Cat{C}, \per, \uno)$ satisfy the following coherence conditions. %
\end{enumerate}
    \begin{equation}\label{eq:fbcbcoherence}
        \begin{array}{rcl C{0.5cm} rcl}
            \copier{X \piu Y} &=& (\Tcopier{X} \piu  \copier{Y}) ; (\id{XX} \piu \cobang{XY} \piu \cobang{YX} \piu \id{YY}) ; (\Idl{X}{X}{Y} \piu \Idl{Y}{X}{Y}) && \discharger{X \piu Y} &=& (\Tdischarger{X} \piu \discharger{Y}) ; \codiag{\uno} \\
            \cocopier{X \piu Y} &=& (\Tcocopier{X} \piu  \cocopier{Y}) ; (\id{XX} \piu \bang{XY} \piu \bang{YX} \piu \id{YY}) ; (\dl{X}{X}{Y} \piu \dl{Y}{X}{Y}) && \codischarger{X \piu Y} &=& \diag{\uno} ; (\Tcodischarger{X} \piu \codischarger{Y})
        \end{array}
\end{equation}

A \emph{morphism of Kleene-Cartesian rig-categories} is a poset enriched rig functor that is a morphism of both Kleene and cartesian bicategories.
\end{definition}
We have already seen in Section~\ref{sec:2monREL} that $(\Rel,\piu, \zero)$ is an Kleene bicategory and $(\Rel,\per, \uno)$ is a cartesian bicategory. To conclude that $\Rel$ is an kc rig category is enough to check that coherence conditions: this is trivial by using the definitions of the two (co)monoids of $\Rel$ in~\eqref{eq:comonoidsREL}.

\begin{table}
	\[
		\begin{array}{@{}c@{}}
			\toprule
			\begin{array}{@{}c|c@{}}
				\begin{array}{@{}cc@{}}
					\multicolumn{2}{c}{
						\text{Interaction of $\op{(\cdot)}$ with $(\piu, \diag{}, \bang{}, \codiag{}, \cobang{})$}
					}
					\\
					\midrule
					\op{(f \piu g)} = \op{f} \piu \op{g} & \op{(\sigma^{\piu}_{X,Y})} = \sigma^{\piu}_{Y,X}
					\\
					\op{(\diag{X})} = \codiag{X} & \op{(\bang{X})} = \cobang{X}
					\\
					\op{(\codiag{X})} = \diag{X} & \op{(\cobang{X})} = \bang{X}
				\end{array}
				&
				\begin{array}{@{}cc@{}}
					\multicolumn{2}{@{}c@{}}{
						\text{Interaction of $\op{(\cdot)}$ and $\kstar{(\cdot)}$ with $(\sqcap, \top)$ and $(+, \bot)$}
					}
					\\
					\midrule
					\begin{array}{@{}cc@{}}
						\op{(f \sqcap g)} = \op{f} \sqcap \op{g} & \op{\top} = \top
						\\
						\op{(f + g)} = \op{f} + \op{g} & \op{\bot} = \bot
					\end{array}
					&
					\begin{array}{@{}cc@{}}
						\kstar{(f \sqcap g)} \leq \kstar{f} \sqcap \kstar{g} & \kstar{\top} = \top
						\\
						 \kstar{f} + \kstar{g} \leq \kstar{(f + g)} & \kstar{\bot} = \id{} 
					\end{array}
					\\
					\multicolumn{2}{@{}c@{}}{
						\kstar{(\op{f})} = \op{(\kstar{f})}
					}
				\end{array}
			\end{array}
			\\
			\midrule
			\text{Interaction of $(\sqcap, \top)$ and $(+, \bot)$}
			\\
			\midrule
			\begin{array}{@{}c @{\qquad} c@{}}
				f \sqcap (g + h) = (\, f \sqcap g \, ) + (\, f \sqcap h  \,)
				&
				f \sqcap \bot = \bot
				\\
				f + (g \sqcap h) = (\, f + g \, ) \sqcap (\, f + h  \,)
				&
				f + \top = \top
			\end{array}
			\\
			\bottomrule
		\end{array}
	\]
	\caption{Derived laws in kc rig categories.}
	\label{table:kc rig derived laws}
\end{table}

\begin{proposition}\label{prop:kc rig laws}
	The laws in \Cref{table:kc rig derived laws} hold in any kc rig category. Moreover, the distributivity of $\op{(\cdot)}$ over $+$, together with the commutativity of $\op{(\cdot)}$ with $\kstar{(\cdot)}$ yield that a kc rig category is also a typed Kleene algebra with converse~\cite{brunet2014kleene}; while the laws at the bottom state that the homsets of a kc rig category are distributive lattices.
\end{proposition}

\subsection{From Kleene to Kleene-Cartesian Tapes}
Now we are going to construct the tape diagrams for kc rig categories.

For  a monoidal signature $(\sort, \sign)$, we fix 
\[\Gamma \defeq \{ \cocopier{A} \colon A \perG A \to A, \;\; \codischarger{A} \colon \unoG \to A, \;\; \copier{A} \colon A \to A \perG A, \;\; \discharger{A} \colon A \to \unoG \;\;\mid \;\; A \in \sort \}\]
and consider the signature obtained as the disjoint union of $\sign$ and $\Gamma$, that is $(\sort, \sign + \Gamma)$. Then consider the corresponding category of Kleene tapes: $\CatKTapeC$.
We now define a preorder on this category using the same recipe of $\precongK$ in Section \ref{sec:kleene-tapes}: we take $\basicCB$ to be the set of pairs of tapes containing all and only the pairs in  Figure \ref{fig:cb axioms}.
We fix $\basicKC \defeq \basicCB \cup \basicK$ and define $\leq_{\basicKC}$ according to the rules in \eqref{eq:uniformprecong}. Analogously to Section \ref{sec:kleene-tapes}, $\sim_{\basicKC} \defeq \leq_{\basicKC} \cap \geq_{\basicKC}$.

With these definitions we can construct the category of Kleene cartesian tapes $\KTCB$: Objects are the same of $\CatKTapeC$. Arrows are $\sim_\basicKC$-equivalence classes of arrows of $\CatKTapeC$. Every homset $\KTCB[P,Q]$ is ordered by $\precongKC$. In a nustshell, objects of $\KTCB$ are polynomials in $(\sort^\star)^\star$. Arrows are  $\sim_{\basicKC}$-equivalence classes of the tape generated by the following grammar where $A\in \sort$, $U\in \sort^\star$ and $s \in \sign$.
\begin{equation}\label{tracedTapesGrammar}
    \begin{tabular}{rc ccccccccccccccccccccc}\setlength{\tabcolsep}{0.0pt}
        $c$  & ::= & $\id{A}$ & $\!\!\! \mid \!\!\!$ & $ \id{\uno} $ & $\!\!\! \mid \!\!\!$ & $ \gen $ & $\!\!\! \mid \!\!\!$ & $ \sigma_{A,B} $ & $\!\!\! \mid \!\!\!$ & $   c ; c   $ & $\!\!\! \mid \!\!\!$ & $  c \per c$ & $\!\!\! \mid \!\!\!$ & $\discharger{A}$ & $\!\!\! \mid \!\!\!$ & $\copier{A}$ & $\!\!\! \mid \!\!\!$ & $\codischarger{A}$ & $\!\!\! \mid \!\!\!$ & $\cocopier{A} $ \\
        $\t$ & ::= & $\id{U}$ & $\!\!\! \mid \!\!\!$ & $ \id{\zero} $ & $\!\!\! \mid \!\!\!$ & $ \tapeFunct{c} $ & $\!\!\! \mid \!\!\!$ & $ \sigma_{U,V}^{\piu} $ & $\!\!\! \mid \!\!\!$ & $   \t ; \t   $ & $\!\!\! \mid \!\!\!$ & $  \t \piu \t  $ & $\!\!\! \mid \!\!\!$ & $ \bang{U} $ & $\!\!\! \mid \!\!\!$ & $\diag{U}$ & $\!\!\! \mid \!\!\!$ & $\cobang{U}$ & $\!\!\! \mid \!\!\!$ & $\codiag{U}$  &  $\!\!\! \mid \!\!\!$ & $\trace_{U}\t$    
    \end{tabular}
\end{equation}  

Recall from Section \ref{sec:tapes}, that traces and $\piu$-(co)monoids for arbitrary polynomials are defined as in Table \ref{tab:inddefutfb}; the monoidal product $\per$ is defined for arbitrary tapes as in Tables \ref{tab:producttape} and \ref{tab:wisktraces};     left distributors $\delta^l$ and symmetries $\symmt$ as in Table \ref{table:def dl symmt}. In  $\KTCB$,  by means of the coherence conditions in \eqref{eq:fbcbcoherence}, one can inductively define  $\per$-(co)monoids for arbitrary polynomials: see Table \ref{table:defmonoidper}.  For instance, $\copier{A \piu B} \colon A \piu B \to (A \piu B)\per(A \piu B) = AA \piu AB \piu BA \piu BB$ and $\discharger{A \piu B} \colon A \piu B \to \uno$ are
    \[\copier{A \piu B} = 
    \begin{tikzpicture}
	\begin{pgfonlayer}{nodelayer}
		\node [style=label] (76) at (-2.25, 2.625) {$A$};
		\node [style=none] (77) at (-1.75, 1.875) {};
		\node [style=none] (78) at (-1.75, 3.375) {};
		\node [style=black] (80) at (-0.5, 2.625) {};
		\node [style=none] (82) at (0.225, 3.025) {};
		\node [style=none] (83) at (0.225, 2.225) {};
		\node [style=none] (84) at (-1.75, 2.625) {};
		\node [style=none] (87) at (0.75, 1.875) {};
		\node [style=none] (88) at (0.75, 3.375) {};
		\node [style=label] (146) at (1.25, 3.025) {$A$};
		\node [style=label] (147) at (1.25, 2.225) {$A$};
		\node [style=none] (164) at (0.75, 1.175) {};
		\node [style=none] (165) at (0.75, 0.575) {};
		\node [style=none] (168) at (-0.825, 1.175) {};
		\node [style=none] (169) at (-0.825, 0.575) {};
		\node [style=none] (181) at (0.75, 1.625) {};
		\node [style=none] (182) at (0, 1.625) {};
		\node [style=none] (183) at (0, 0.125) {};
		\node [style=none] (184) at (0.75, 0.125) {};
		\node [style=none] (189) at (0.75, 3.025) {};
		\node [style=none] (190) at (0.75, 2.225) {};
		\node [style=label] (191) at (-2.25, -2.625) {$B$};
		\node [style=none] (192) at (-1.75, -1.875) {};
		\node [style=none] (193) at (-1.75, -3.375) {};
		\node [style=none] (198) at (0.75, -1.875) {};
		\node [style=none] (199) at (0.75, -3.375) {};
		\node [style=label] (200) at (1.25, -3.025) {$B$};
		\node [style=label] (201) at (1.25, -2.225) {$B$};
		\node [style=none] (202) at (0.75, -1.175) {};
		\node [style=none] (203) at (0.75, -0.575) {};
		\node [style=none] (206) at (-0.825, -1.175) {};
		\node [style=none] (207) at (-0.825, -0.575) {};
		\node [style=none] (208) at (0.75, -1.625) {};
		\node [style=none] (209) at (0, -1.625) {};
		\node [style=none] (210) at (0, -0.125) {};
		\node [style=none] (211) at (0.75, -0.125) {};
		\node [style=black] (212) at (-0.5, -2.625) {};
		\node [style=none] (213) at (0.225, -2.225) {};
		\node [style=none] (214) at (0.225, -3.025) {};
		\node [style=none] (215) at (-1.75, -2.625) {};
		\node [style=none] (216) at (0.75, -2.225) {};
		\node [style=none] (217) at (0.75, -3.025) {};
		\node [style=label] (218) at (1.25, 1.275) {$A$};
		\node [style=label] (219) at (1.25, 0.475) {$B$};
		\node [style=label] (220) at (1.25, -0.475) {$B$};
		\node [style=label] (221) at (1.25, -1.275) {$A$};
	\end{pgfonlayer}
	\begin{pgfonlayer}{edgelayer}
		\draw [style=tape] (88.center)
			 to (87.center)
			 to (77.center)
			 to (78.center)
			 to cycle;
		\draw [style=tape] (184.center)
			 to (181.center)
			 to (182.center)
			 to [bend right=90, looseness=2.00] (183.center)
			 to cycle;
		\draw (84.center) to (80);
		\draw [bend left] (80) to (82.center);
		\draw [bend right=330] (83.center) to (80);
		\draw (168.center) to (164.center);
		\draw (165.center) to (169.center);
		\draw (82.center) to (189.center);
		\draw (190.center) to (83.center);
		\draw [style=tape] (199.center)
			 to (198.center)
			 to (192.center)
			 to (193.center)
			 to cycle;
		\draw [style=tape] (211.center)
			 to (208.center)
			 to (209.center)
			 to [bend left=90, looseness=2.00] (210.center)
			 to cycle;
		\draw (206.center) to (202.center);
		\draw (203.center) to (207.center);
		\draw (215.center) to (212);
		\draw [bend left] (212) to (213.center);
		\draw [bend right=330] (214.center) to (212);
		\draw (213.center) to (216.center);
		\draw (217.center) to (214.center);
	\end{pgfonlayer}
\end{tikzpicture}
}
 \qquad \discharger{A \piu B} = 
    \begin{tikzpicture}
	\begin{pgfonlayer}{nodelayer}
		\node [style=none] (224) at (3.5, 0.75) {};
		\node [style=none] (225) at (3.5, -0.75) {};
		\node [style=none] (231) at (2.5, -0.75) {};
		\node [style=none] (232) at (2.5, 0.75) {};
		\node [style=label] (233) at (-1, 1.25) {$A$};
		\node [style=none] (234) at (-0.5, 2) {};
		\node [style=none] (235) at (-0.5, 0.5) {};
		\node [style=none] (236) at (0.75, 0.5) {};
		\node [style=none] (237) at (0.75, 2) {};
		\node [style=label] (238) at (-1, -1.25) {$B$};
		\node [style=none] (239) at (-0.5, -0.5) {};
		\node [style=none] (240) at (-0.5, -2) {};
		\node [style=none] (241) at (0.75, -2) {};
		\node [style=none] (242) at (0.75, -0.5) {};
		\node [style=black] (243) at (0.75, 1.25) {};
		\node [style=none] (244) at (-0.5, 1.25) {};
		\node [style=black] (245) at (0.75, -1.25) {};
		\node [style=none] (246) at (-0.5, -1.25) {};
	\end{pgfonlayer}
	\begin{pgfonlayer}{edgelayer}
		\draw [tape] (237.center)
			 to [bend left] (232.center)
			 to (224.center)
			 to (225.center)
			 to (231.center)
			 to [bend left] (241.center)
			 to (240.center)
			 to (239.center)
			 to (242.center)
			 to [bend right=90, looseness=2.00] (236.center)
			 to (235.center)
			 to (234.center)
			 to cycle;
		\draw (244.center) to (243);
		\draw (246.center) to (245);
	\end{pgfonlayer}
\end{tikzpicture}
}
\]

\begin{table}
    \begin{equation}\label{eq:copierind}
        \begin{array}{rcl C{0.5cm} rcl}
            \copier{\zero} &\defeq& \id{\zero} && \discharger{\zero} &\defeq& \cobang{\uno}  \\
            \copier{U \piu P'} &\defeq& \Tcopier{U} \piu \cobang{UP'} \piu ((\cobang{P'U} \piu \copier{P'}) ; \Idl{P'}{U}{P'}) && \discharger{U \piu P'} &\defeq& (\Tdischarger{U} \piu \discharger{P'}) ; \codiag{\uno}
        \end{array}
    \end{equation}
    \begin{equation}\label{eq:cocopierind}
        \begin{array}{rcl C{0.5cm} rcl}
            \cocopier{\zero} &\defeq& \id{\zero} && \codischarger{\zero} &\defeq& \bang{\uno}  \\
            \cocopier{U \piu P'} &\defeq& \Tcocopier{U} \piu \bang{UP'} \piu (\dl{P'}{U}{P'} ; (\bang{P'U} \piu \cocopier{P'})) && \codischarger{U \piu P'} &\defeq& \diag{\uno} ; (\Tcodischarger{U} \piu \codischarger{P'})
        \end{array}
    \end{equation}
        \caption{Inductive definitions of $\discharger{P}$, $\copier{P}$, $\codischarger{P}$ and $\cocopier{P}$}
        \label{table:defmonoidper}
\end{table}
These structures make $\KTCB$  a kc rig category.

\begin{theorem}\label{theorem:KTCB is kleene-cartesian}
	$\KTCB$ is a Kleene-Cartesian rig category.
\end{theorem}
\begin{proof}
By construction $\KTCB$ is a Kleene rig category. In order to prove that $(\KTCB, \per, \uno)$ is a cartesian bicategory we widely rely on the proof of \cite[Theorem 7.3]{bonchi2023deconstructing}.
By \cite[Theorem 7.3]{bonchi2023deconstructing},  $\copier{P}, \discharger{P}, \cocopier{P}$ and $\codischarger{P}$ satisfy the axioms of special Frobenius algebras and the comonoid $(\copier{P}, \discharger{P})$ is left adjoint to the monoid $(\cocopier{P}, \discharger{P})$. Moreover, every trace-free tape diagram $\t \colon P \to Q$ is a lax comonoid homomorphism, i.e. 
	\begin{equation}\label{eq:trace free lax naturality}
		\t ; \copier{Q} \; \leq \; \copier{P} ; (\t \per \t) 	\quad \text{ and } \quad \t ; \discharger{Q} \; \leq \; \discharger{P}.
	\end{equation}
	To conclude, we need to show that the inequalities above hold for \emph{every} tape diagram $\t \colon P \to Q$ in $\KTCB$. 

	By the normal form of traced monoidal categories, there exists a trace-free tape diagram $\t' \colon S \piu P \to S \piu Q$, such that $\trace_{S} \t' = \t$. Now, let $\begin{psmallmatrix} \t'_{SS} & \t'_{SQ} \\ \t'_{PS} & \t'_{PQ} \end{psmallmatrix}$ be the matrix normal form of $\t'$ and observe that the following holds.

	\begingroup
	\allowdisplaybreaks
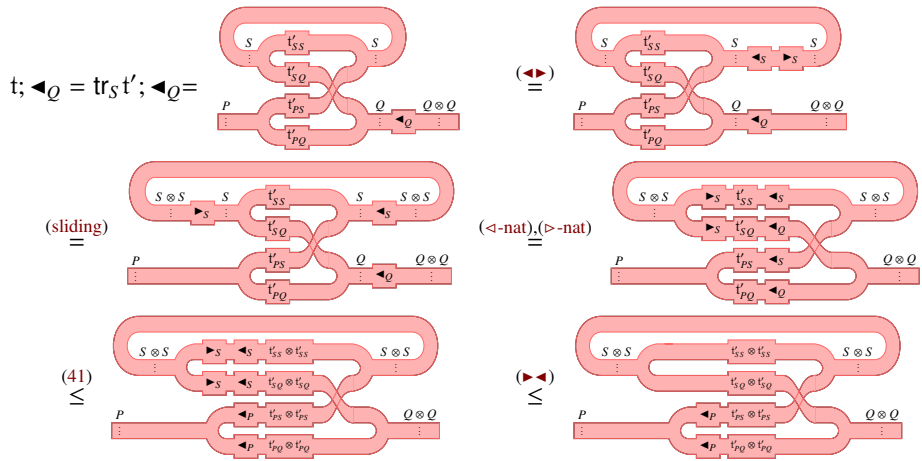
		\begin{align*}
			\t ; \copier{Q} 
			&= 
			\trace_{S} \t' ; \copier{Q}  
			= 
			\scalebox{0.7}{
    \InputIfFileExists{laxnaturality/copier/step1.tikz}{}{\input{./tikz/laxnaturality/copier/step1.tikz}}
} 
			&&\!\!\!\!\!\!\!\!\!\!\!\!\!\leftstackrel{\eqref{ax:copieradj2}}{=}
			\scalebox{0.7}{
    \InputIfFileExists{laxnaturality/copier/step2.tikz}{}{\input{./tikz/laxnaturality/copier/step2.tikz}}
} \\
			&\leftstackrel{\eqref{ax:trace:sliding}}{=}
			\scalebox{0.7}{
    \InputIfFileExists{laxnaturality/copier/step3.tikz}{}{\input{./tikz/laxnaturality/copier/step3.tikz}}
}
			&&\!\!\!\!\!\!\!\!\!\!\!\!\!\leftstackrel{\eqref{ax:diagnat}, \eqref{ax:codiagnat}}{=}
			\scalebox{0.7}{
    \InputIfFileExists{laxnaturality/copier/step4.tikz}{}{\input{./tikz/laxnaturality/copier/step4.tikz}}
} \\
			&\leftstackrel{\eqref{eq:trace free lax naturality}}{\leq}
			\scalebox{0.7}{
    \InputIfFileExists{laxnaturality/copier/step5.tikz}{}{\input{./tikz/laxnaturality/copier/step5.tikz}}
} 
			&&\!\!\!\!\!\!\!\!\!\!\!\!\!\leftstackrel{\eqref{ax:copieradj1}}{\leq}
			\scalebox{0.7}{
    \InputIfFileExists{laxnaturality/copier/step6.tikz}{}{\input{./tikz/laxnaturality/copier/step6.tikz}}
} \\
			&\leftstackrel{\eqref{ax:diagnat}}{=}
			\scalebox{0.7}{
    \InputIfFileExists{laxnaturality/copier/step7.tikz}{}{\input{./tikz/laxnaturality/copier/step7.tikz}}
}
			&&\!\!\!\!\!\!\!\!\!\!\!\!\!\leftstackrel{\text{(\Cref{prop:trace-per})}}{\leq}
			\copier{P} ; (\trace_{S} \t' \per \trace_{S} \t')
			= 
			\copier{P} ; (\t \per \t)
		\end{align*}
	\endgroup

	To prove the other inequality we exploit again the matrix normal form.

	\begingroup
	\allowdisplaybreaks
		\begin{align*}
			\t ; \discharger{Q} 
			= 
			\trace_{S} \t' ; \discharger{Q}
			&= 
			\scalebox{0.7}{
    \InputIfFileExists{laxnaturality/discard/step1.tikz}{}{\input{./tikz/laxnaturality/discard/step1.tikz}}
}
			\stackrel{\eqref{ax:dischargeradj2}}{\leq}
			\scalebox{0.7}{
    \InputIfFileExists{laxnaturality/discard/step2.tikz}{}{\input{./tikz/laxnaturality/discard/step2.tikz}}
}
			\stackrel{\eqref{ax:trace:sliding}}{=}
			\scalebox{0.7}{
    \InputIfFileExists{laxnaturality/discard/step3.tikz}{}{\input{./tikz/laxnaturality/discard/step3.tikz}}
} \\
			&\leftstackrel{\eqref{ax:diagnat}, \eqref{ax:codiagnat}}{=}
			\scalebox{0.7}{
    \InputIfFileExists{laxnaturality/discard/step4.tikz}{}{\input{./tikz/laxnaturality/discard/step4.tikz}}
} 
			\stackrel{\eqref{eq:trace free lax naturality}}{\leq}
			\scalebox{0.7}{
    \InputIfFileExists{laxnaturality/discard/step5.tikz}{}{\input{./tikz/laxnaturality/discard/step5.tikz}}
}
			\stackrel{\eqref{ax:dischargeradj1}}{\leq}
			\scalebox{0.7}{
    \InputIfFileExists{laxnaturality/discard/step6.tikz}{}{\input{./tikz/laxnaturality/discard/step6.tikz}}
} \\
			&\leftstackrel{\eqref{ax:diagnat}}{=}
			\scalebox{0.7}{
    \InputIfFileExists{laxnaturality/discard/step7.tikz}{}{\input{./tikz/laxnaturality/discard/step7.tikz}}
}
			\stackrel{\eqref{ax:bi}}{=}
			\scalebox{0.7}{
    \InputIfFileExists{laxnaturality/discard/step8.tikz}{}{\input{./tikz/laxnaturality/discard/step8.tikz}}
}
			\stackrel{\eqref{ax:trace:tape:kstar-id}}{=}
			\scalebox{0.7}{
    \InputIfFileExists{laxnaturality/discard/step9.tikz}{}{\input{./tikz/laxnaturality/discard/step9.tikz}}
}
		\end{align*}
	\endgroup

\end{proof}

Given a monoidal signature $(\sort,\sign)$, a (sesquistrict) kc rig category $\Cat{C}$ and an interpretation $\mathcal{I}=(\alpha_\sort,\alpha_\sign)$ of $\sign$ in $\Cat{C}$, 
the inductive extension of $\mathcal{I}$, hereafter referred as $\CBdsem{\cdot}_{\interpretation} \colon \KTCB \to \Cat{C}$, is defined as follows. %
\small{
\renewcommand{\arraystretch}{1.5}
\begin{tabular}{lllll}
$\CBdsem{s}_{\interpretation} = \alpha_{\sign}(s) $&$ \CBdsem{\copier{A}}_{\interpretation} = \copier{\alpha_\sort(A)}  $&$ \CBdsem{\discharger{A}}_{\interpretation} = \discharger{\alpha_\sort(A)}  $&$ \CBdsem{\cocopier{A}}_{\interpretation} = \cocopier{\alpha_\sort(A)}  $&$ \CBdsem{\codischarger{A}}_{\interpretation} = \codischarger{\alpha_\sort(A)}$\\
$\CBdsem{\id{A}}_{\interpretation}= \id{\alpha_\sort(A)} $&$ \CBdsem{\id{1}}_{\interpretation}= \id{1} $&$ \CBdsem{\symmt{A}{B}}_{\interpretation} = \symmt{\alpha_\sort(A)}{\alpha_\sort(B)}  $&$ \CBdsem{c;d}_{\interpretation} = \CBdsem{c}_{\interpretation}; \CBdsem{d}_{\interpretation}  $ & $ \CBdsem{c\per d}_{\interpretation} = \CBdsem{c}_{\interpretation} \per \CBdsem{d}_{\interpretation}$\\
$\CBdsem{\, \tapeFunct{c} \,}_{\interpretation}= \CBdsem{c}_{\interpretation} $&$ \CBdsem{\diag{U}}_{\interpretation}= \diag{\alpha^\sharp_\sort(U)} $&$ \CBdsem{\bang{U}}_{\interpretation} = \bang{\alpha^\sharp_\sort(U)}$ & $\CBdsem{\codiag{U}}_{\interpretation}= \codiag{\alpha^\sharp_\sort(U)} $&$ \CBdsem{\cobang{U}}_{\interpretation} = \cobang{\alpha^\sharp_\sort(U)}$  \\
$\CBdsem{\id{U}}_{\interpretation}= \id{\alpha^\sharp_\sort(U)} $&$ \CBdsem{\id{\zero}}_{\interpretation} = \id{\zero}  $&$ \CBdsem{\symmp{U}{V}}_{\interpretation}= \symmp{\alpha^\sharp_\sort(U)}{\alpha^\sharp_\sort(V)} $&$ \CBdsem{\s;\t}_{\interpretation} = \CBdsem{\s}_{\interpretation} ; \CBdsem{\t}_{\interpretation} $&$ \CBdsem{\s \piu \t}_{\interpretation} = \CBdsem{\s}_{\interpretation} \piu \CBdsem{\t}_{\interpretation} $\\
\multicolumn{5}{c}{
	$\CBdsem{\trace_{U}\t}_{\interpretation} = \trace_{\alpha^\sharp_\sort(U)} \CBdsem{\t}_{\interpretation}$
}
\end{tabular}
}

It turns out that  $\CBdsem{\cdot}_{\interpretation}$ is a kc rig morphism and it is actually the unique one  respecting the interpretation $\mathcal{I}$.

\begin{theorem}\label{thm:KleeneCartesiantapesfree}
$\KTCB$ is the free sesquistrict kc rig category generated by the monoidal signature $(\sort, \sign)$.
\end{theorem}

\subsection{Functorial Semantics}\label{ssec:funsem}

The usual way of reasoning through string diagrams is based on monoidal theories, namely a signature plus a set of axioms: either equations or inequations. Similarly a \emph{kc tape theory} is a pair $(\sign, \basicR)$ where $\sign$ is a monoidal signature and $\basicR$ is a set of pairs $(\t_1, \t_2)$ of arrows in $\KTCB$ with same domain and codomain. Hereafter, we think of each pair $(\t_1, \t_2)$ as an inequation $\t_1 \leq \t_2$, but the results that we develop in this section trivially hold  also for equations: it is enough to add in $\basicR$ a pair $(\t_2, \t_1)$ for each $(\t_1, \t_2) \in \basicR$.

Given a kc tape theory $(\sign, \basicR)$, we will write $\precongB$ for $\precongR{\basicKC \cup \basicR}$ where the latter is generated from $\basicKC \cup \basicR$ by the rules  in \eqref{eq:uniformprecong}. Since the pairs of arrows in $\basicKC$
express the laws of kc rig categories, it is safe to always keep $\basicKC$ implicit.

\begin{remark}
Derivations using $\precongB$ might ends up to be non entirely graphical because they may rely on  the decomposition via $\per$: see, e.g., Example 5.14 in \cite{bonchi2023deconstructing}.
The solution devised in \cite{bonchi2023deconstructing} also works for kc tapes: take the kc tape theory $(\sign, \wiskbasicR)$, where $\wiskbasicR = \{ (\t_1 \per \id{U}, \t_2 \per \id{U}) \mid (\t_1, \t_2) \in \basicR \text{ and } U\in \sort^\star \}$ and define $\WprecongBA$ as in \eqref{eq:uniformprecong} but without the rules for $\per$, i.e., with the following rules. 
\[
    \begin{array}{ccccc}
        \inferrule*[right=($\wiskbasicR$)]{\t_1 \mathbin{\wiskbasicR} \t_2}{\t_1 \mathrel{\WprecongBA} \t_2}
        &
        \inferrule*[right=($r$)]{-}{\t \mathrel{\WprecongBA} \t}
        &    
        \inferrule*[right=($t$)]{\t_1 \mathrel{\WprecongBA} \t_2 \quad \t_2 \mathrel{\WprecongBA} \t_3}{\t_1 \mathrel{\WprecongBA} \t_3}
        &
        \inferrule*[right=($;$)]{\t_1 \mathrel{\WprecongBA} \t_2 \quad \s_1 \mathrel{\WprecongBA} \s_2}{\t_1;\s_1 \mathrel{\WprecongBA} \t_2;\s_2}
        &
        \inferrule*[right=($\piu$)]{\t_1 \mathrel{\WprecongBA} \t_2 \quad \s_1 \mathrel{\WprecongBA} \s_2}{\t_1\piu\s_1 \mathrel{\WprecongBA} \t_2 \piu \s_2}
        \\[8pt]
        \multicolumn{5}{c}{
        \inferrule*[right=($ut$-1)]{\s_2 \mathrel{\WprecongBA} \s_1 \qquad \t_1 ; (\s_1 \piu \id{}) \mathrel{\WprecongBA} (\s_2 \piu \id{}) ; \t_2}{\trace_{S_1}\t_1 \mathrel{\WprecongBA} \trace_{S_2}\t_2}
        \qquad
        \inferrule*[right=($ut$-2)]{ \s_2 \mathrel{\WprecongBA} \s_1 \qquad (\s_1 \piu \id{}) ; \t_1 \mathrel{\WprecongBA} \t_2 ; (\s_2 \piu \id{})}{\trace_{S_1}\t_1 \mathrel{\WprecongBA} \trace_{S_2}\t_2}
        }
    \end{array}
\]
Then, by  the same proof of Theorem 5.15 in~\cite{bonchi2023deconstructing}, it holds that  $\t_1 \precongB \t_2$ if and only if $\t_1 \WprecongBA \t_2$.
\end{remark}

Recall that an interpretation $\mathcal{I}=(\alpha_{\sort}, \alpha_{\sign})$ of a monoidal signature $(\sort,\sign)$ in a sesquistrict rig category  $\Cat{C}$ consists of $\alpha_{\sort}\colon \sort \to Ob(\Cat{C})$ and $\alpha_{\sign}\colon \sign \to Ar(\Cat{C})$ preserving (co)arities of symbols $s\in \sign$. Whenever $\Cat{C}$ is a kc rig category, $\mathcal{I}$ gives rises uniquely, by freeness of $\KTCB$, to the morphisms of kc rig categories $\TCBdsem{\cdot}_{\interpretation} \colon \KTCB \to \Cat{C}$. We say that an intepretation $\interpretation$ of $\sign$ is \emph{a model of the theory} $(\sign, \basicR)$ whenever $\TCBdsem{\cdot}_{\interpretation}$ preserves $\precongB$: if $\t_1 \precongB \t_2$, then $\TCBdsem{\t_1}_{\interpretation} $ is below $\TCBdsem{\t_2}_{\interpretation}$ in $\Cat{C}$.

\begin{example}\label{ex:totalorder}
	Consider the signature $(\sort, \sign)$ where $\sort$ contains a single sort $A$ and $\sign = \{ R \colon A \to A \}$. Take as $\basicR$ the set consisting of the following inequalities:
	\mylabel{R:refl}{refl}
	\mylabel{R:tr}{tr}
	\mylabel{R:anti}{anti}
	\mylabel{R:lin}{lin}
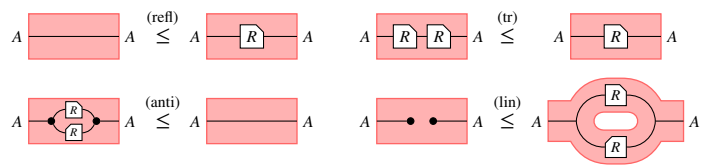
	\[
		\begin{array}{c@{}c@{}c c@{}c@{}c}
			
    \InputIfFileExists{linOrd/id.tikz}{}{\input{./tikz/linOrd/id.tikz}}
 &\stackrel{(\text{refl})}{\leq}& 
    \InputIfFileExists{linOrd/R.tikz}{}{\input{./tikz/linOrd/R.tikz}}
 & 
    \InputIfFileExists{linOrd/RR.tikz}{}{\input{./tikz/linOrd/RR.tikz}}
 &\stackrel{(\text{tr})}{\leq}& 
    \InputIfFileExists{linOrd/R.tikz}{}{\input{./tikz/linOrd/R.tikz}}

			\\[10pt]
			
    \InputIfFileExists{linOrd/and.tikz}{}{\input{./tikz/linOrd/and.tikz}}
 &\stackrel{(\text{anti})}{\leq}& 
    \InputIfFileExists{linOrd/id.tikz}{}{\input{./tikz/linOrd/id.tikz}}
 & 
    \InputIfFileExists{linOrd/top.tikz}{}{\input{./tikz/linOrd/top.tikz}}
 &\stackrel{(\text{lin})}{\leq}& 
    \InputIfFileExists{linOrd/or.tikz}{}{\input{./tikz/linOrd/or.tikz}}
		
		\end{array}
	\]
	An interpretation of $(\sign, \basicR)$ in the kc rig category $\Rel$ is a set $X$, together with a relation $R \subseteq X \times X$. It is a model iff $R$ is a linear order, i.e. it is reflexive, transitive, antisymmetric and linear.
\end{example}

Models enjoy a beautiful characterisation provided by Proposition \ref{funct:sem} below. Let $\KTCBI$ be the category having the same objects as $\KTCB$ and arrows $\sim_{\basicR}$-equivalence classes of arrows of $\KTCB$ ordered by $\precongB$.
Since $\KTCB$ is a kc rig category, thus so is also $\KTCBI$.

\begin{proposition}\label{funct:sem}
	Let $(\sign, \basicR)$ be a kc tape theory and $\Cat{C}$ a sesquistrict kc rig category.  Models of $(\sign, \basicR)$ are in bijective correspondence with  morphisms of sesquistrict kc rig categories from $\KTCBI$ to $\Cat{C}$.
\end{proposition}

\section{The Kleene-Cartesian Tape Theory of Peano}\label{sec:peano}
In Section \ref{ssec:funsem}, we introduced kc tape theories and in Example \ref{ex:totalorder} we illustrated the theory of linear orders.
In this section we illustrate a further example of theory that fully exploits the expressive power of Kleene-Cartesian rig categories: Peano's axiomatisation of natural numbers. This  theory is not expressible in first order logic \cite{harsanyi1983mathematics}
but it can be succintly expressed by a kc tape theory.

As expected we begin by fixing  the  signature:
\[\sort\defeq \{A\} \qquad \text{and} \qquad \sign \defeq \{ \; 
    \begin{tikzpicture}
	\begin{pgfonlayer}{nodelayer}
		\node [style=bbox] (69) at (0, 0) {$0$};
		\node [style=none] (71) at (1.5, 0) {};
		\node [style=label] (72) at (2, 0) {$A$};
	\end{pgfonlayer}
	\begin{pgfonlayer}{edgelayer}
		\draw (69) to (71.center);
	\end{pgfonlayer}
\end{tikzpicture}
}
, 
    \begin{tikzpicture}
	\begin{pgfonlayer}{nodelayer}
		\node [style=bbox] (69) at (0, 0) {$s$};
		\node [style=none] (70) at (-1.5, 0) {};
		\node [style=none] (71) at (1.5, 0) {};
		\node [style=label] (72) at (2, 0) {$A$};
		\node [style=label] (73) at (-2, 0) {$A$};
	\end{pgfonlayer}
	\begin{pgfonlayer}{edgelayer}
		\draw (70.center) to (69);
		\draw (69) to (71.center);
	\end{pgfonlayer}
\end{tikzpicture}
}
  \}\text{.} \]
An interpretation of $\sign$ in $\Rel$ consists of 
\begin{center}
\begin{tabular}{rl}
a set $X$ &(i.e., $\alpha_\sort(A)$), \\
a relation $0 \subseteq 1 \times X$ & (i.e., $\alpha_\sign( \, 
    \begin{tikzpicture}
	\begin{pgfonlayer}{nodelayer}
		\node [style=bbox] (69) at (0, 0) {$0$};
		\node [style=none] (71) at (1.5, 0) {};
		\node [style=label] (72) at (2, 0) {$A$};
	\end{pgfonlayer}
	\begin{pgfonlayer}{edgelayer}
		\draw (69) to (71.center);
	\end{pgfonlayer}
\end{tikzpicture}
}
 )$) and \\
a relation $s\subseteq  X \times X$ & (i.e., $\alpha_\sign(  
    \begin{tikzpicture}
	\begin{pgfonlayer}{nodelayer}
		\node [style=bbox] (69) at (0, 0) {$s$};
		\node [style=none] (70) at (-1.5, 0) {};
		\node [style=none] (71) at (1.5, 0) {};
		\node [style=label] (72) at (2, 0) {$A$};
		\node [style=label] (73) at (-2, 0) {$A$};
	\end{pgfonlayer}
	\begin{pgfonlayer}{edgelayer}
		\draw (70.center) to (69);
		\draw (69) to (71.center);
	\end{pgfonlayer}
\end{tikzpicture}
}
 )$)\text{.}
\end{tabular}
\end{center}

The set of axioms $\mathbb{P}$ contains those illustrated in Figure \ref{fig:tape-theory-natural-numbers}.
The axioms at the top force the copairing of $0$ and $s$, that is $[0,s]  \defeq (0 \piu s) ; \codiag{A}$ to be an isomorphism of type $1+X \to X$.
The axiom on the left states that $\op{[0,s]}; [0,s] = \id{A}$; the one on the right that $[0,s];\op{[0,s]}= \id{A \piu A}$.
The axiom at the bottom of Figure \ref{fig:tape-theory-natural-numbers} is the induction axiom. Intuitively, it states that every element of $X$ should be contained in $0;s^* \colon 1\to X$.
Thus, overall an interpretation in $\Rel$ is a model of the theory $(\sign, \mathbb{P})$ iff $X$ is isomorphic to $X+1$ and equal to $0;s^*$.

An example of a model is obviously given by the set of natural numbers $\N$, equipped with the element zero $0\colon 1 \to \N$ and the successor function $s\colon \N \to \N$. 
We will see soon that  this is the unique --up-to isomorphism-- model of $(\sign, \mathbb{P})$.

\begin{figure}[H]
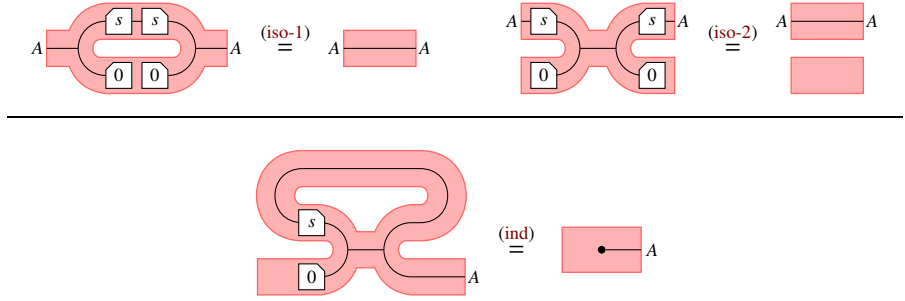

    \mylabel{ax:peanoT:zero:sv}{$0$-sv}
    \mylabel{ax:peanoT:zero:tot}{$0$-tot}
    \mylabel{ax:peanoT:succ:sv}{$s$-sv}
    \mylabel{ax:peanoT:succ:tot}{$s$-tot}
    \mylabel{ax:peanoT:ind}{ind}
    \mylabel{ax:peanoT:iso:1}{iso-1}
    \mylabel{ax:peanoT:iso:2}{iso-2}
    \[
        \begin{array}{c@{\qquad}c}
                
    \InputIfFileExists{peano/ax1/iso1/lhs.tikz}{}{\input{./tikz/peano/ax1/iso1/lhs.tikz}}
 \stackrel{\eqref{ax:peanoT:iso:1}}{=} 
    \InputIfFileExists{peano/ax1/iso1/rhs.tikz}{}{\input{./tikz/peano/ax1/iso1/rhs.tikz}}

                &
                
    \InputIfFileExists{peano/ax1/iso2/lhs.tikz}{}{\input{./tikz/peano/ax1/iso2/lhs.tikz}}
 \stackrel{\eqref{ax:peanoT:iso:2}}{=} 
    \InputIfFileExists{peano/ax1/iso2/rhs.tikz}{}{\input{./tikz/peano/ax1/iso2/rhs.tikz}}

                \\[20pt]
                \hline
                \multicolumn{2}{c}{
                    \begin{array}{c}
                        \\
                        
    \InputIfFileExists{peano/ax1/ind/lhs.tikz}{}{\input{./tikz/peano/ax1/ind/lhs.tikz}}
 \stackrel{\eqref{ax:peanoT:ind}}{=} 
    \InputIfFileExists{peano/ax1/ind/rhs.tikz}{}{\input{./tikz/peano/ax1/ind/rhs.tikz}}

                    \end{array}
                }
        \end{array}
    \]
    \caption{Tape theory of the natural numbers.}
    \label{fig:tape-theory-natural-numbers}
\end{figure}

First we show that $(\sign, \mathbb{P})$ is equivalent to Peano's axiomatisation of natural numbers. Possibly, the most interesting axiom is the principle of induction that, as illustrated below, follows easily from uniformity and \eqref{ax:peanoT:ind}.

\begin{theorem}[Principle of Induction]\label{thm:induction}
    For all morphisms $P \colon \uno \to A$ in $\KTCBP$,
    \[ 
        \text{if } \quad
        
    \begin{tikzpicture}
	\begin{pgfonlayer}{nodelayer}
		\node [style=bbox] (69) at (0, 0) {$0$};
		\node [style=label] (72) at (2.25, 0) {$A$};
		\node [style=none] (76) at (1.75, 0) {};
		\node [style=none] (77) at (1.75, 1) {};
		\node [style=none] (78) at (-1.75, 1) {};
		\node [style=none] (79) at (-1.75, -1) {};
		\node [style=none] (80) at (1.75, -1) {};
	\end{pgfonlayer}
	\begin{pgfonlayer}{edgelayer}
		\draw [tape] (79.center)
			 to (80.center)
			 to (77.center)
			 to (78.center)
			 to cycle;
		\draw (69) to (76.center);
	\end{pgfonlayer}
\end{tikzpicture}
}
 \leq \;
    \begin{tikzpicture}
	\begin{pgfonlayer}{nodelayer}
		\node [style=bbox] (69) at (0, 0) {$P$};
		\node [style=label] (72) at (2.25, 0) {$A$};
		\node [style=none] (76) at (1.75, 0) {};
		\node [style=none] (77) at (1.75, 1) {};
		\node [style=none] (78) at (-1.75, 1) {};
		\node [style=none] (79) at (-1.75, -1) {};
		\node [style=none] (80) at (1.75, -1) {};
	\end{pgfonlayer}
	\begin{pgfonlayer}{edgelayer}
		\draw [tape] (79.center)
			 to (80.center)
			 to (77.center)
			 to (78.center)
			 to cycle;
		\draw (69) to (76.center);
	\end{pgfonlayer}
\end{tikzpicture}
}
 
        \quad \text{ and } \quad 
        
    \begin{tikzpicture}
	\begin{pgfonlayer}{nodelayer}
		\node [style=bbox] (69) at (0.75, 0) {$s$};
		\node [style=label] (72) at (2.25, 0) {$A$};
		\node [style=none] (76) at (1.75, 0) {};
		\node [style=none] (77) at (1.75, 1) {};
		\node [style=none] (78) at (-1.75, 1) {};
		\node [style=none] (79) at (-1.75, -1) {};
		\node [style=none] (80) at (1.75, -1) {};
		\node [style=bbox] (81) at (-0.75, 0) {$P$};
	\end{pgfonlayer}
	\begin{pgfonlayer}{edgelayer}
		\draw [tape] (79.center)
			 to (80.center)
			 to (77.center)
			 to (78.center)
			 to cycle;
		\draw (69) to (76.center);
		\draw (69) to (81);
	\end{pgfonlayer}
\end{tikzpicture}
}
 \leq \;
    \begin{tikzpicture}
	\begin{pgfonlayer}{nodelayer}
		\node [style=bbox] (69) at (0, 0) {$P$};
		\node [style=label] (72) at (2.25, 0) {$A$};
		\node [style=none] (76) at (1.75, 0) {};
		\node [style=none] (77) at (1.75, 1) {};
		\node [style=none] (78) at (-1.75, 1) {};
		\node [style=none] (79) at (-1.75, -1) {};
		\node [style=none] (80) at (1.75, -1) {};
	\end{pgfonlayer}
	\begin{pgfonlayer}{edgelayer}
		\draw [tape] (79.center)
			 to (80.center)
			 to (77.center)
			 to (78.center)
			 to cycle;
		\draw (69) to (76.center);
	\end{pgfonlayer}
\end{tikzpicture}
}
,
        \quad \text{ then } \quad
        
    \begin{tikzpicture}
	\begin{pgfonlayer}{nodelayer}
		\node [style=bbox] (69) at (0, 0) {$P$};
		\node [style=label] (72) at (2.25, 0) {$A$};
		\node [style=none] (76) at (1.75, 0) {};
		\node [style=none] (77) at (1.75, 1) {};
		\node [style=none] (78) at (-1.75, 1) {};
		\node [style=none] (79) at (-1.75, -1) {};
		\node [style=none] (80) at (1.75, -1) {};
	\end{pgfonlayer}
	\begin{pgfonlayer}{edgelayer}
		\draw [tape] (79.center)
			 to (80.center)
			 to (77.center)
			 to (78.center)
			 to cycle;
		\draw (69) to (76.center);
	\end{pgfonlayer}
\end{tikzpicture}
}
 = 
    \begin{tikzpicture}
	\begin{pgfonlayer}{nodelayer}
		\node [style=black] (71) at (-0.5, 0) {};
		\node [style=none] (77) at (-2.25, 1) {};
		\node [style=none] (78) at (1.25, 1) {};
		\node [style=none] (79) at (1.25, -1) {};
		\node [style=none] (80) at (-2.25, -1) {};
		\node [style=none] (82) at (1.25, 0) {};
		\node [style=label] (83) at (1.75, 0) {$A$};
		\node [style=label] (84) at (-2.75, 0) {};
	\end{pgfonlayer}
	\begin{pgfonlayer}{edgelayer}
		\draw [tape] (79.center)
			 to (80.center)
			 to (77.center)
			 to (78.center)
			 to cycle;
		\draw (82.center) to (71);
	\end{pgfonlayer}
\end{tikzpicture}
}
 
    \]
\end{theorem}
\begin{proof} %
    Observe that the following holds:
    \[
        
    \begin{tikzpicture}[scale=0.8]
	\begin{pgfonlayer}{nodelayer}
		\node [style=none] (71) at (-2.5, 1.5) {};
		\node [style=none] (77) at (-2.5, 2.5) {};
		\node [style=none] (80) at (-2.25, 0.5) {};
		\node [style=none] (82) at (2.5, 1.5) {};
		\node [style=none] (85) at (-2.5, -1.5) {};
		\node [style=none] (86) at (-2.25, -0.5) {};
		\node [style=none] (89) at (-2.5, -2.5) {};
		\node [style=none] (90) at (2.5, -1.5) {};
		\node [style=label] (91) at (3.5, 1.5) {$A$};
		\node [style=none] (95) at (-0.5, 1) {};
		\node [style=none] (96) at (-0.5, -1) {};
		\node [style=none] (99) at (2.5, 2.5) {};
		\node [style=none] (100) at (2.25, 0.5) {};
		\node [style=none] (102) at (2.25, -0.5) {};
		\node [style=none] (103) at (2.5, -2.5) {};
		\node [style=none] (106) at (0.5, 1) {};
		\node [style=none] (107) at (0.5, -1) {};
		\node [style=none] (109) at (-1, 0) {};
		\node [style=none] (110) at (1, 0) {};
		\node [style=bbox] (112) at (-3, -1.5) {$0$};
		\node [style=none] (113) at (3, -0.5) {};
		\node [style=none] (114) at (3, -2.5) {};
		\node [style=none] (115) at (3, 2.5) {};
		\node [style=none] (116) at (3, 0.5) {};
		\node [style=none] (117) at (-6.25, -0.5) {};
		\node [style=none] (118) at (-6.25, -2.5) {};
		\node [style=none] (119) at (-6.25, 2.5) {};
		\node [style=none] (120) at (-6.25, 0.5) {};
		\node [style=none] (122) at (3, 1.5) {};
		\node [style=none] (123) at (3, -1.5) {};
		\node [style=label] (124) at (3.5, -1.5) {$A$};
		\node [style=bbox] (133) at (-3, 1.5) {$s$};
		\node [style=bbox] (134) at (-5, 1.5) {$P$};
	\end{pgfonlayer}
	\begin{pgfonlayer}{edgelayer}
		\draw [tape] (80.center)
			 to [bend left=90, looseness=1.50] (86.center)
			 to (117.center)
			 to (118.center)
			 to (89.center)
			 to [bend right] (96.center)
			 to (107.center)
			 to [bend right] (103.center)
			 to (114.center)
			 to (113.center)
			 to (102.center)
			 to [bend left=90, looseness=1.50] (100.center)
			 to (116.center)
			 to (115.center)
			 to (99.center)
			 to [bend right] (106.center)
			 to (95.center)
			 to [bend right] (77.center)
			 to (119.center)
			 to (120.center)
			 to cycle;
		\draw [bend left=45] (110.center) to (82.center);
		\draw [bend right=45] (110.center) to (90.center);
		\draw [bend right=45] (85.center) to (109.center);
		\draw [bend right=45] (109.center) to (71.center);
		\draw (85.center) to (112);
		\draw (109.center) to (110.center);
		\draw (122.center) to (82.center);
		\draw (90.center) to (123.center);
		\draw (133) to (134);
		\draw (133) to (71.center);
	\end{pgfonlayer}
\end{tikzpicture}
}
 \stackrel{\text{(Hypothesis)}}{\leq} \; 
    \begin{tikzpicture}[scale=0.8]
	\begin{pgfonlayer}{nodelayer}
		\node [style=none] (71) at (-2.5, 1.5) {};
		\node [style=none] (77) at (-2.5, 2.5) {};
		\node [style=none] (80) at (-2.25, 0.5) {};
		\node [style=none] (82) at (2.5, 1.5) {};
		\node [style=none] (85) at (-2.5, -1.5) {};
		\node [style=none] (86) at (-2.25, -0.5) {};
		\node [style=none] (89) at (-2.5, -2.5) {};
		\node [style=none] (90) at (2.5, -1.5) {};
		\node [style=label] (91) at (3.5, 1.5) {$A$};
		\node [style=none] (95) at (-0.5, 1) {};
		\node [style=none] (96) at (-0.5, -1) {};
		\node [style=none] (99) at (2.5, 2.5) {};
		\node [style=none] (100) at (2.25, 0.5) {};
		\node [style=none] (102) at (2.25, -0.5) {};
		\node [style=none] (103) at (2.5, -2.5) {};
		\node [style=none] (106) at (0.5, 1) {};
		\node [style=none] (107) at (0.5, -1) {};
		\node [style=none] (109) at (-1, 0) {};
		\node [style=none] (110) at (1, 0) {};
		\node [style=bbox] (112) at (-3, -1.5) {$P$};
		\node [style=none] (113) at (3, -0.5) {};
		\node [style=none] (114) at (3, -2.5) {};
		\node [style=none] (115) at (3, 2.5) {};
		\node [style=none] (116) at (3, 0.5) {};
		\node [style=none] (117) at (-4.25, -0.5) {};
		\node [style=none] (118) at (-4.25, -2.5) {};
		\node [style=none] (119) at (-4.25, 2.5) {};
		\node [style=none] (120) at (-4.25, 0.5) {};
		\node [style=none] (122) at (3, 1.5) {};
		\node [style=none] (123) at (3, -1.5) {};
		\node [style=label] (124) at (3.5, -1.5) {$A$};
		\node [style=bbox] (134) at (-3, 1.5) {$P$};
	\end{pgfonlayer}
	\begin{pgfonlayer}{edgelayer}
		\draw [tape] (80.center)
			 to [bend left=90, looseness=1.50] (86.center)
			 to (117.center)
			 to (118.center)
			 to (89.center)
			 to [bend right] (96.center)
			 to (107.center)
			 to [bend right] (103.center)
			 to (114.center)
			 to (113.center)
			 to (102.center)
			 to [bend left=90, looseness=1.50] (100.center)
			 to (116.center)
			 to (115.center)
			 to (99.center)
			 to [bend right] (106.center)
			 to (95.center)
			 to [bend right] (77.center)
			 to (119.center)
			 to (120.center)
			 to cycle;
		\draw [bend left=45] (110.center) to (82.center);
		\draw [bend right=45] (110.center) to (90.center);
		\draw [bend right=45] (85.center) to (109.center);
		\draw [bend right=45] (109.center) to (71.center);
		\draw (85.center) to (112);
		\draw (109.center) to (110.center);
		\draw (122.center) to (82.center);
		\draw (90.center) to (123.center);
		\draw (134) to (71.center);
	\end{pgfonlayer}
\end{tikzpicture}
}
 \stackrel{\eqref{ax:diagnat}, \eqref{ax:codiagnat}}{=} \; 
    \begin{tikzpicture}[scale=0.8]
	\begin{pgfonlayer}{nodelayer}
		\node [style=none] (77) at (-2.5, 2.5) {};
		\node [style=none] (80) at (-2.25, 0.5) {};
		\node [style=none] (86) at (-2.25, -0.5) {};
		\node [style=none] (89) at (-2.5, -2.5) {};
		\node [style=label] (91) at (5.25, 1.5) {$A$};
		\node [style=none] (95) at (-0.5, 1) {};
		\node [style=none] (96) at (-0.5, -1) {};
		\node [style=none] (99) at (2.5, 2.5) {};
		\node [style=none] (100) at (2.25, 0.5) {};
		\node [style=none] (102) at (2.25, -0.5) {};
		\node [style=none] (103) at (2.5, -2.5) {};
		\node [style=none] (106) at (0.5, 1) {};
		\node [style=none] (107) at (0.5, -1) {};
		\node [style=none] (113) at (4.75, -0.5) {};
		\node [style=none] (114) at (4.75, -2.5) {};
		\node [style=none] (115) at (4.75, 2.5) {};
		\node [style=none] (116) at (4.75, 0.5) {};
		\node [style=none] (117) at (-4, -0.5) {};
		\node [style=none] (118) at (-4, -2.5) {};
		\node [style=none] (119) at (-4, 2.5) {};
		\node [style=none] (120) at (-4, 0.5) {};
		\node [style=none] (122) at (4.75, 1.5) {};
		\node [style=none] (123) at (4.75, -1.5) {};
		\node [style=label] (124) at (5.25, -1.5) {$A$};
		\node [style=bbox] (135) at (3.25, -1.5) {$P$};
		\node [style=bbox] (136) at (3.25, 1.5) {$P$};
	\end{pgfonlayer}
	\begin{pgfonlayer}{edgelayer}
		\draw [tape] (80.center)
			 to [bend left=90, looseness=1.50] (86.center)
			 to (117.center)
			 to (118.center)
			 to (89.center)
			 to [bend right] (96.center)
			 to (107.center)
			 to [bend right] (103.center)
			 to (114.center)
			 to (113.center)
			 to (102.center)
			 to [bend left=90, looseness=1.50] (100.center)
			 to (116.center)
			 to (115.center)
			 to (99.center)
			 to [bend right] (106.center)
			 to (95.center)
			 to [bend right] (77.center)
			 to (119.center)
			 to (120.center)
			 to cycle;
		\draw (136) to (122.center);
		\draw (123.center) to (135);
	\end{pgfonlayer}
\end{tikzpicture}
}
.
    \]
    Thus, by~\eqref{ax:trace:tape:AU2} the inclusion below holds and the derivation concludes the proof.
    
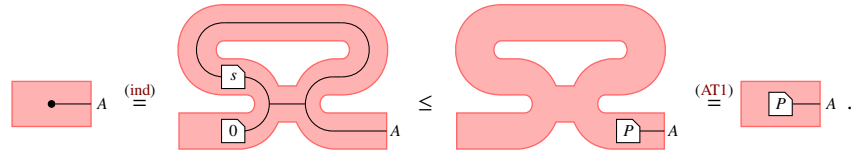
\[
        
    \begin{tikzpicture}
	\begin{pgfonlayer}{nodelayer}
		\node [style=black] (71) at (-0.5, 0) {};
		\node [style=none] (77) at (-2.25, 1) {};
		\node [style=none] (78) at (1.25, 1) {};
		\node [style=none] (79) at (1.25, -1) {};
		\node [style=none] (80) at (-2.25, -1) {};
		\node [style=none] (82) at (1.25, 0) {};
		\node [style=label] (83) at (1.75, 0) {$A$};
		\node [style=label] (84) at (-2.75, 0) {};
	\end{pgfonlayer}
	\begin{pgfonlayer}{edgelayer}
		\draw [tape] (79.center)
			 to (80.center)
			 to (77.center)
			 to (78.center)
			 to cycle;
		\draw (82.center) to (71);
	\end{pgfonlayer}
\end{tikzpicture}
}
 \stackrel{\eqref{ax:peanoT:ind}}{=} 
    \begin{tikzpicture}[scale=0.8]
	\begin{pgfonlayer}{nodelayer}
		\node [style=none] (71) at (-2.5, 1.5) {};
		\node [style=none] (77) at (-2.5, 2.5) {};
		\node [style=none] (80) at (-2.25, 0.5) {};
		\node [style=none] (82) at (2.5, 1.5) {};
		\node [style=none] (85) at (-2.5, -1.5) {};
		\node [style=none] (86) at (-2.25, -0.5) {};
		\node [style=none] (89) at (-2.5, -2.5) {};
		\node [style=none] (90) at (2.5, -1.5) {};
		\node [style=label] (91) at (6, -1.5) {$A$};
		\node [style=none] (95) at (-0.5, 1) {};
		\node [style=none] (96) at (-0.5, -1) {};
		\node [style=none] (99) at (2.5, 2.5) {};
		\node [style=none] (100) at (2.25, 0.5) {};
		\node [style=none] (102) at (2.25, -0.5) {};
		\node [style=none] (103) at (2.5, -2.5) {};
		\node [style=none] (106) at (0.5, 1) {};
		\node [style=none] (107) at (0.5, -1) {};
		\node [style=none] (109) at (-1, 0) {};
		\node [style=none] (110) at (1, 0) {};
		\node [style=bbox] (111) at (-3, 1.5) {$s$};
		\node [style=bbox] (112) at (-3, -1.5) {$0$};
		\node [style=none] (113) at (5.5, -0.5) {};
		\node [style=none] (114) at (5.5, -2.5) {};
		\node [style=none] (115) at (3, 2.5) {};
		\node [style=none] (116) at (3, 0.5) {};
		\node [style=none] (117) at (-6, -0.5) {};
		\node [style=none] (118) at (-6, -2.5) {};
		\node [style=none] (119) at (-3.5, 2.5) {};
		\node [style=none] (120) at (-3.5, 0.5) {};
		\node [style=none] (121) at (-3.5, 1.5) {};
		\node [style=none] (122) at (3, 1.5) {};
		\node [style=none] (123) at (5.5, -1.5) {};
		\node [style=none] (124) at (3, 5.5) {};
		\node [style=none] (125) at (3, 3.5) {};
		\node [style=none] (126) at (3, 4.5) {};
		\node [style=none] (127) at (-3.5, 5.5) {};
		\node [style=none] (128) at (-3.5, 3.5) {};
		\node [style=none] (129) at (-3.5, 4.5) {};
	\end{pgfonlayer}
	\begin{pgfonlayer}{edgelayer}
		\draw [tape] (116.center)
			 to [bend right=90, looseness=1.75] (124.center)
			 to (127.center)
			 to [bend left=270, looseness=1.75] (120.center)
			 to (80.center)
			 to [bend left=90, looseness=1.50] (86.center)
			 to (117.center)
			 to (118.center)
			 to (89.center)
			 to [bend right] (96.center)
			 to (107.center)
			 to [bend right] (103.center)
			 to (114.center)
			 to (113.center)
			 to (102.center)
			 to [bend left=90, looseness=1.50] (100.center)
			 to cycle;
		\draw [tape, fill=white] (99.center)
			 to [bend right] (106.center)
			 to (95.center)
			 to [bend right] (77.center)
			 to (119.center)
			 to [bend left=90, looseness=1.50] (128.center)
			 to (125.center)
			 to [bend left=90, looseness=1.50] (115.center)
			 to cycle;
		\draw [bend left=45] (110.center) to (82.center);
		\draw [bend right=45] (110.center) to (90.center);
		\draw [bend right=45] (85.center) to (109.center);
		\draw [bend right=45] (109.center) to (71.center);
		\draw (85.center) to (112);
		\draw (71.center) to (111);
		\draw (121.center) to (111);
		\draw (109.center) to (110.center);
		\draw (82.center) to (122.center);
		\draw (90.center) to (123.center);
		\draw [bend right=90, looseness=1.75] (129.center) to (121.center);
		\draw [bend right=90, looseness=1.75] (122.center) to (126.center);
		\draw (129.center) to (126.center);
	\end{pgfonlayer}
\end{tikzpicture}
}
 \leq 
    \begin{tikzpicture}[scale=0.8]
	\begin{pgfonlayer}{nodelayer}
		\node [style=none] (77) at (-2.5, 2.5) {};
		\node [style=none] (80) at (-2.25, 0.5) {};
		\node [style=none] (86) at (-2.25, -0.5) {};
		\node [style=none] (89) at (-2.5, -2.5) {};
		\node [style=none] (90) at (3.5, -1.5) {};
		\node [style=label] (91) at (6, -1.5) {$A$};
		\node [style=none] (95) at (-0.5, 1) {};
		\node [style=none] (96) at (-0.5, -1) {};
		\node [style=none] (99) at (2.5, 2.5) {};
		\node [style=none] (100) at (2.25, 0.5) {};
		\node [style=none] (102) at (2.25, -0.5) {};
		\node [style=none] (103) at (2.5, -2.5) {};
		\node [style=none] (106) at (0.5, 1) {};
		\node [style=none] (107) at (0.5, -1) {};
		\node [style=bbox] (112) at (3.5, -1.5) {$P$};
		\node [style=none] (113) at (5.5, -0.5) {};
		\node [style=none] (114) at (5.5, -2.5) {};
		\node [style=none] (115) at (3, 2.5) {};
		\node [style=none] (116) at (3, 0.5) {};
		\node [style=none] (117) at (-6, -0.5) {};
		\node [style=none] (118) at (-6, -2.5) {};
		\node [style=none] (119) at (-3.5, 2.5) {};
		\node [style=none] (120) at (-3.5, 0.5) {};
		\node [style=none] (123) at (5.5, -1.5) {};
		\node [style=none] (124) at (3, 5.5) {};
		\node [style=none] (125) at (3, 3.5) {};
		\node [style=none] (127) at (-3.5, 5.5) {};
		\node [style=none] (128) at (-3.5, 3.5) {};
	\end{pgfonlayer}
	\begin{pgfonlayer}{edgelayer}
		\draw [tape] (116.center)
			 to [bend right=90, looseness=1.75] (124.center)
			 to (127.center)
			 to [bend left=270, looseness=1.75] (120.center)
			 to (80.center)
			 to [bend left=90, looseness=1.50] (86.center)
			 to (117.center)
			 to (118.center)
			 to (89.center)
			 to [bend right] (96.center)
			 to (107.center)
			 to [bend right] (103.center)
			 to (114.center)
			 to (113.center)
			 to (102.center)
			 to [bend left=90, looseness=1.50] (100.center)
			 to cycle;
		\draw [tape, fill=white] (99.center)
			 to [bend right] (106.center)
			 to (95.center)
			 to [bend right] (77.center)
			 to (119.center)
			 to [bend left=90, looseness=1.50] (128.center)
			 to (125.center)
			 to [bend left=90, looseness=1.50] (115.center)
			 to cycle;
		\draw (90.center) to (123.center);
	\end{pgfonlayer}
\end{tikzpicture}
}
 \stackrel{\eqref{ax:trace:tape:kstar-id}}{=} 
    }
.
    \]
\end{proof}

The common form of Peano's axioms state that $0$ is a natural number, $s$ is an injective function and that $0$ is not the successor of any natural number. These are illustrated by means of tape diagrams in Figure \ref{fig:peano-theory-natural-numbers}, where we use the characterisation of total, single valued and injective relations provided by \Cref{lemma:cb:adjoints} in Section \ref{sc:cb:background}. Observe that the axiom \eqref{ax:peano:bottom} states that $\{x\in X \mid (x,0)\in s \} \subseteq \emptyset$.
At the bottom of Figure \ref{fig:peano-theory-natural-numbers}, there is the induction principle as expressed by Peano. Note that, since \eqref{ax:peano:ind} is an implication this is not a kc tape theory. Nevertheless, one can see that this set of laws is equivalent to $\mathbb{P}$.

\begin{figure}[H]
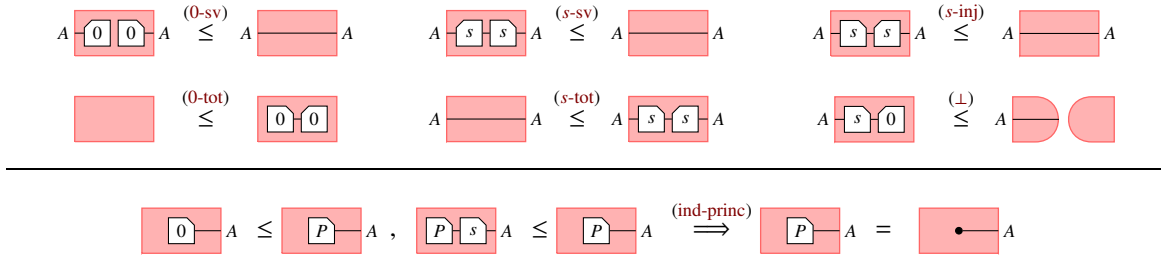

    \mylabel{ax:peano:zero:sv}{$0$-sv}
    \mylabel{ax:peano:zero:tot}{$0$-tot}
    \mylabel{ax:peano:succ:sv}{$s$-sv}
    \mylabel{ax:peano:succ:tot}{$s$-tot}
    \mylabel{ax:peano:succ:inj}{$s$-inj}
    \mylabel{ax:peano:bottom}{$\bot$}
    \mylabel{ax:peano:ind}{ind-princ}
    \[
        \begin{array}{cc}
            \begin{array}{cc}
                \begin{array}[c]{c@{}c@{}c}
                    
    \InputIfFileExists{peano/ax1/zeroSV/lhs.tikz}{}{\input{./tikz/peano/ax1/zeroSV/lhs.tikz}}
 &\stackrel{\eqref{ax:peano:zero:sv}}{\leq}& 
    \InputIfFileExists{peano/ax1/zeroSV/rhs.tikz}{}{\input{./tikz/peano/ax1/zeroSV/rhs.tikz}}
 
                    \\[15pt]
                    
    \InputIfFileExists{peano/ax1/zeroTOT/lhs.tikz}{}{\input{./tikz/peano/ax1/zeroTOT/lhs.tikz}}
 &\stackrel{\eqref{ax:peano:zero:tot}}{\leq}& 
    \InputIfFileExists{peano/ax1/zeroTOT/rhs.tikz}{}{\input{./tikz/peano/ax1/zeroTOT/rhs.tikz}}
  
                \end{array}
                &
                \begin{array}[c]{c@{}c@{}c}
                    
    \InputIfFileExists{peano/ax1/succSV/lhs.tikz}{}{\input{./tikz/peano/ax1/succSV/lhs.tikz}}
 &\stackrel{\eqref{ax:peano:succ:sv}}{\leq}& 
    \InputIfFileExists{peano/ax1/succSV/rhs.tikz}{}{\input{./tikz/peano/ax1/succSV/rhs.tikz}}
 
                    \\[15pt]
                    
    \InputIfFileExists{peano/ax1/succTOT/lhs.tikz}{}{\input{./tikz/peano/ax1/succTOT/lhs.tikz}}
 &\stackrel{\eqref{ax:peano:succ:tot}}{\leq}& 
    \InputIfFileExists{peano/ax1/succTOT/rhs.tikz}{}{\input{./tikz/peano/ax1/succTOT/rhs.tikz}}

                \end{array}        
            \end{array}
            &
            \begin{array}[c]{c@{}c@{}c}
                
    \InputIfFileExists{peano/ax2/succINJ/lhs.tikz}{}{\input{./tikz/peano/ax2/succINJ/lhs.tikz}}
 &\stackrel{\eqref{ax:peano:succ:inj}}{\leq}& 
    \InputIfFileExists{peano/ax2/succINJ/rhs.tikz}{}{\input{./tikz/peano/ax2/succINJ/rhs.tikz}}

                \\[15pt]
                
    \InputIfFileExists{peano/ax2/bottom/lhs.tikz}{}{\input{./tikz/peano/ax2/bottom/lhs.tikz}}
 &\stackrel{\eqref{ax:peano:bottom}}{\leq}& 
    \InputIfFileExists{peano/ax2/bottom/rhs.tikz}{}{\input{./tikz/peano/ax2/bottom/rhs.tikz}}

            \end{array}
            \\[30pt]
            \hline
            \multicolumn{2}{c}{
                \begin{array}{c}
                    \\
                    
    \InputIfFileExists{peano/ax2/ind/base/lhs.tikz}{}{\input{./tikz/peano/ax2/ind/base/lhs.tikz}}
 \leq 
    \InputIfFileExists{peano/ax2/ind/base/rhs.tikz}{}{\input{./tikz/peano/ax2/ind/base/rhs.tikz}}
, \;\; 
    \InputIfFileExists{peano/ax2/ind/hp/lhs.tikz}{}{\input{./tikz/peano/ax2/ind/hp/lhs.tikz}}
 \leq 
    \InputIfFileExists{peano/ax2/ind/hp/rhs.tikz}{}{\input{./tikz/peano/ax2/ind/hp/rhs.tikz}}
 
                    \stackrel{\eqref{ax:peano:ind}}{\implies }
                    
    \InputIfFileExists{peano/ax2/ind/concl/lhs.tikz}{}{\input{./tikz/peano/ax2/ind/concl/lhs.tikz}}
 = 
    \InputIfFileExists{peano/ax2/ind/concl/rhs.tikz}{}{\input{./tikz/peano/ax2/ind/concl/rhs.tikz}}

                \end{array}
            }
        \end{array}
    \]
    \caption{Peano's theory of the natural numbers.}
    \label{fig:peano-theory-natural-numbers}
\end{figure}

\begin{lemma}\label{lemma:Peano}
    The laws in Figures~\ref{fig:tape-theory-natural-numbers} and \ref{fig:peano-theory-natural-numbers} are equivalent. 
\end{lemma}
\begin{proof}
    First, we prove that the axioms in Figure~\ref{fig:tape-theory-natural-numbers} entail those in Figure~\ref{fig:peano-theory-natural-numbers}.
    \begin{itemize}
        \item \eqref{ax:peano:succ:sv},\eqref{ax:peano:zero:sv} follow from \eqref{ax:peanoT:iso:1}, i.e.
        \[
            

}
 \leq 
    }
 \stackrel{\eqref{ax:peanoT:iso:1}}{=} 
    }
.
        \]
        \item \eqref{ax:peano:succ:tot}, \eqref{ax:peano:succ:inj}, \eqref{ax:peano:zero:tot}, \eqref{ax:peano:bottom} follow from the matrix normal form of~\eqref{ax:peanoT:iso:2} and Proposition~\ref{prop:matrixform}, i.e.
        \begin{equation}\label{eq:peano:iso2-equivalence}
            
    %
 \right.  
        \end{equation}
        Observe that the last equality on the right imposes injectivity of $0$ as well, however this is always the case for morphisms $1 \to A$ that are total.
        \item \eqref{ax:peano:ind} holds by Theorem~\ref{thm:induction}.
    \end{itemize}
    Similarly, we prove that the axioms in Figure~\ref{fig:peano-theory-natural-numbers} entail those in Figure~\ref{fig:tape-theory-natural-numbers}. It is convenient to prove them in the following order.
    \begin{itemize}
        \item For \eqref{ax:peanoT:ind}, let $P \colon 1 \to A$ be the following morphism
        \[
    }
 \defeq  
    %
}
\]
        and observe that the first condition of~\eqref{ax:peano:ind} holds:
        \[
            
    }
 \stackrel{\eqref{ax:relBangCobang}}{\geq} 
    %
}
 \stackrel{\eqref{ax:cobangnat}, \eqref{ax:codiagun}, \eqref{ax:diagun}}{=} 
    }
.
        \]
        For the second condition observe that $P$ is defined as $0 ; \kstar{s}$, and recall that $\kstar{s} ; s \leq \kstar s$ by \Cref{cor:bounds-star}. Thus $P ; s = 0 ; \kstar{s} ; s \leq 0 ; \kstar{s} = P$. We conclude using~\eqref{ax:peano:ind}.
        \item For \eqref{ax:peanoT:iso:1} we prove the two inclusions separately, starting with the one below.
        \[
            
    }
 \stackrel{\eqref{ax:peano:succ:sv}, \eqref{ax:peano:zero:sv}}{\leq} 
    %
}
 \stackrel{\eqref{ax:relDiagCodiag}}{\leq} 
    }
 
        \]
        For the other inclusion, observe that by~\Cref{lemma:cb:adjoints}
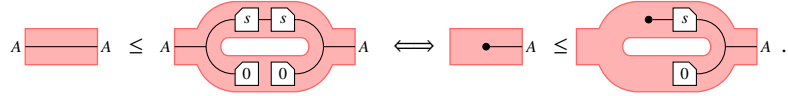
        \[ 
            
    }
 \leq 
    }
  \iff  
    %
}
.
        \]
        Thus, we prove the following:
        \begin{align*}
            
    }
 &\stackrel{\eqref{ax:peanoT:ind}}{=} 
    %
}
 \\
                                                            &\leftstackrel{\text{(Proposition~\ref{prop:star-fixpoint})}}{=} 
    }
 \stackrel{\eqref{ax:peanoT:ind}}{=} 
    }

        \end{align*}
        \item \eqref{ax:peanoT:iso:2} follow from \eqref{ax:peano:succ:tot}, \eqref{ax:peano:succ:inj}, \eqref{ax:peano:zero:tot}, \eqref{ax:peano:bottom} as shown in \eqref{eq:peano:iso2-equivalence}.
    \end{itemize}

\end{proof}

In \cite{dedekind1888nature}, Dedekind showed that any two models of Peano's axioms are isomorphic, and thus any model of $(\sign, \mathbb{P})$ is isomorphic to the one on natural numbers.

\subsection{First Steps in Arithmetic}
To give to the reader a taste of how one can program with tapes, we now illustrate how to start to encode arithmetic within  $(\sign, \mathbb{P})$.
The tape for addition is illustrated below.
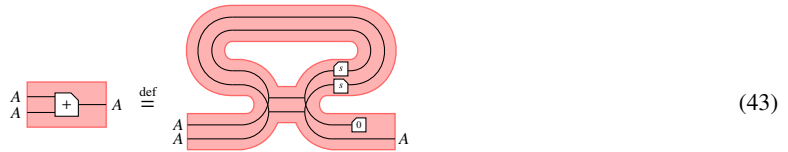
\begin{equation}\label{eq:add}
    
    \InputIfFileExists{addition/addCirc.tikz}{}{\input{./tikz/addition/addCirc.tikz}}
 \defeq 
    \InputIfFileExists{addition/additionDef.tikz}{}{\input{./tikz/addition/additionDef.tikz}}

\end{equation}
As it will be clearer later, this tape can be thought as a simple imperative program:
\begin{center}
\texttt{add(x,y) = while (x>0) \{ x:=x-1; y:=y+1 \}; return y}
\end{center}
The variable \texttt{x} corresponds to the top wire in \eqref{eq:add}, while \texttt{y} to the bottom one. At any iteration, the program checks whether 
 \texttt{x} is $0$, in which case it returns \texttt{x}, or the successor of some number, in which case \texttt{x} takes such number, while \texttt{y} takes its own successor.

One can easily prove that $
    \InputIfFileExists{addition/addCirc.tikz}{}{\input{./tikz/addition/addCirc.tikz}}
$ satisfies the usual inductive definition of addition in Peano's arithmetics.
\begin{lemma}
    The following hold in $\KTCBP$:
    \begin{enumerate}
        \item $
    \InputIfFileExists{addition/addZero.tikz}{}{\input{./tikz/addition/addZero.tikz}}
 = 
    \InputIfFileExists{addition/addZeroRhs.tikz}{}{\input{./tikz/addition/addZeroRhs.tikz}}
 \qquad \textnormal{(\texttt{ add(0,y) = y })}$ 
        \item $
    \InputIfFileExists{addition/addSuccLhs.tikz}{}{\input{./tikz/addition/addSuccLhs.tikz}}
 = 
    \InputIfFileExists{addition/addSuccRhs.tikz}{}{\input{./tikz/addition/addSuccRhs.tikz}}
 \qquad \textnormal{(\texttt{ add(succ(x),y) = succ(add(x,y)) })}$
    \end{enumerate}
\end{lemma}
\begin{proof}
    First, recall that by Proposition~\ref{prop:star-fixpoint},
    \begin{equation}\label{eq:add-fixpoint}
        
    \InputIfFileExists{addition/additionDef.tikz}{}{\input{./tikz/addition/additionDef.tikz}}
 = 
    \InputIfFileExists{addition/addStar.tikz}{}{\input{./tikz/addition/addStar.tikz}}
. 
    \end{equation}
    Then, observe that for $1.$ the following holds:
    \begin{align*}
        
    \InputIfFileExists{addition/addZero.tikz}{}{\input{./tikz/addition/addZero.tikz}}
 &\leftstackrel{\eqref{eq:add-fixpoint}}{=} 
    \InputIfFileExists{addition/partialEval/zero/step1.tikz}{}{\input{./tikz/addition/partialEval/zero/step1.tikz}}
 \stackrel{\eqref{ax:diagnat}, \eqref{ax:codiagnat}}{=} 
    \InputIfFileExists{addition/partialEval/zero/step2.tikz}{}{\input{./tikz/addition/partialEval/zero/step2.tikz}}
 \\
                                   &\leftstackrel{\eqref{ax:peano:bottom}}{=} 
    \InputIfFileExists{addition/partialEval/zero/step3.tikz}{}{\input{./tikz/addition/partialEval/zero/step3.tikz}}
 \stackrel{\eqref{ax:cobangnat}}{=} 
    \InputIfFileExists{addition/partialEval/zero/step4.tikz}{}{\input{./tikz/addition/partialEval/zero/step4.tikz}}
 \stackrel{\eqref{ax:diagun}, \eqref{ax:codiagun}}{=} 
    \InputIfFileExists{addition/partialEval/zero/step5.tikz}{}{\input{./tikz/addition/partialEval/zero/step5.tikz}}
 \\
                                   &\leftstackrel{\eqref{ax:peano:zero:tot}}{=} 
    \InputIfFileExists{addition/partialEval/zero/step6.tikz}{}{\input{./tikz/addition/partialEval/zero/step6.tikz}}
. \\
    \end{align*}
    And for $2.$ the following holds:
    \begingroup
    \allowdisplaybreaks
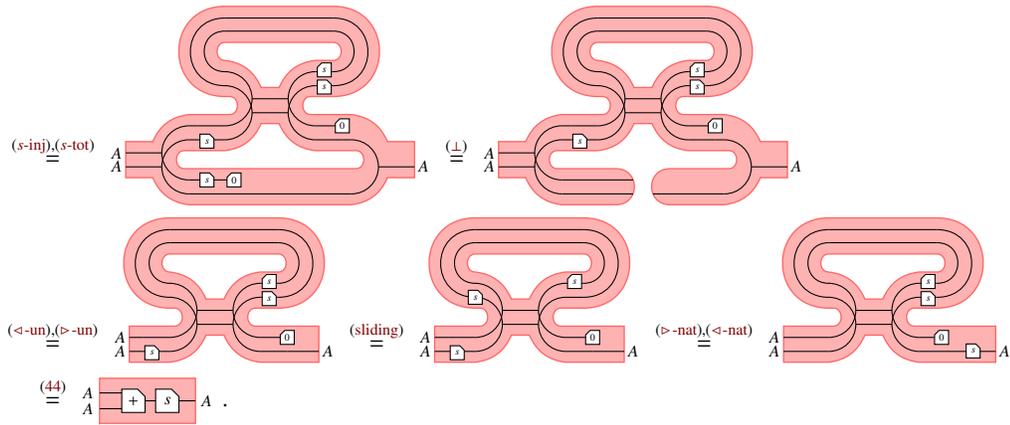
    \begin{align*}
        
    \InputIfFileExists{addition/addSuccLhs.tikz}{}{\input{./tikz/addition/addSuccLhs.tikz}}
 &\leftstackrel{\eqref{eq:add-fixpoint}}{=} 
    \InputIfFileExists{addition/partialEval/succ/step1.tikz}{}{\input{./tikz/addition/partialEval/succ/step1.tikz}}
 \stackrel{\eqref{ax:diagnat}, \eqref{ax:codiagnat}}{=} 
    \InputIfFileExists{addition/partialEval/succ/step2.tikz}{}{\input{./tikz/addition/partialEval/succ/step2.tikz}}
 \\
                                      &\leftstackrel{\eqref{ax:peano:succ:inj}, \eqref{ax:peano:succ:tot}}{=} 
    \InputIfFileExists{addition/partialEval/succ/step3.tikz}{}{\input{./tikz/addition/partialEval/succ/step3.tikz}}
 \stackrel{\eqref{ax:peano:bottom}}{=} 
    \InputIfFileExists{addition/partialEval/succ/step4.tikz}{}{\input{./tikz/addition/partialEval/succ/step4.tikz}}
 \\
                                      &\leftstackrel{\eqref{ax:diagun}, \eqref{ax:codiagun}}{=} 
    \InputIfFileExists{addition/partialEval/succ/step5.tikz}{}{\input{./tikz/addition/partialEval/succ/step5.tikz}}
 \stackrel{\eqref{ax:trace:tape:sliding}}{=} 
    \InputIfFileExists{addition/partialEval/succ/step6.tikz}{}{\input{./tikz/addition/partialEval/succ/step6.tikz}}
  \stackrel{\eqref{ax:codiagnat}, \eqref{ax:diagnat}}{=} 
    \InputIfFileExists{addition/partialEval/succ/step7.tikz}{}{\input{./tikz/addition/partialEval/succ/step7.tikz}}
 \\
                                      &\leftstackrel{\eqref{eq:add-fixpoint}}{=} 
    \InputIfFileExists{addition/addSuccRhs.tikz}{}{\input{./tikz/addition/addSuccRhs.tikz}}
.
    \end{align*}
    \endgroup
\end{proof}

While, it is straightforward that $
    \InputIfFileExists{addition/addCirc.tikz}{}{\input{./tikz/addition/addCirc.tikz}}
$ terminates with all possible inputs, it is interesting to see how this can be proved within the kc tape theory $\mathbb{P}$.

\begin{lemma}
    The function $
    \InputIfFileExists{addition/addCirc.tikz}{}{\input{./tikz/addition/addCirc.tikz}}
$ is total, i.e. the following equality holds:
    
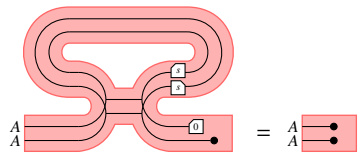
\[ 
    \begin{tikzpicture}[scale=0.8]
	\begin{pgfonlayer}{nodelayer}
		\node [style=none] (71) at (2.5, 3.375) {};
		\node [style=none] (77) at (2.5, 4) {};
		\node [style=none] (80) at (2.25, 2) {};
		\node [style=none] (82) at (-2.5, 3.375) {};
		\node [style=none] (85) at (2.5, 0.375) {};
		\node [style=none] (86) at (2.25, 1) {};
		\node [style=none] (89) at (2.5, -1) {};
		\node [style=none] (90) at (-2.5, 0.375) {};
		\node [style=label] (91) at (-6, 0.375) {$A$};
		\node [style=none] (95) at (0.5, 2.5) {};
		\node [style=none] (96) at (0.5, 0.5) {};
		\node [style=none] (99) at (-2.5, 4) {};
		\node [style=none] (100) at (-2.25, 2) {};
		\node [style=none] (102) at (-2.25, 1) {};
		\node [style=none] (103) at (-2.5, -1) {};
		\node [style=none] (106) at (-0.5, 2.5) {};
		\node [style=none] (107) at (-0.5, 0.5) {};
		\node [style=none] (109) at (1, 1.875) {};
		\node [style=none] (110) at (-1, 1.875) {};
		\node [style=bboxOp, scale=0.6] (112) at (4, 0.375) {$0$};
		\node [style=none] (113) at (-5.5, 1) {};
		\node [style=none] (114) at (-5.5, -1) {};
		\node [style=none] (115) at (-3, 4) {};
		\node [style=none] (116) at (-3, 2) {};
		\node [style=none] (117) at (6, 1) {};
		\node [style=none] (118) at (6, -1) {};
		\node [style=none] (119) at (3.5, 4) {};
		\node [style=none] (120) at (3.5, 2) {};
		\node [style=none] (121) at (3.5, 3.375) {};
		\node [style=none] (122) at (-3, 3.375) {};
		\node [style=none] (123) at (-5.5, 0.375) {};
		\node [style=none] (124) at (-3, 7) {};
		\node [style=none] (125) at (-3, 5) {};
		\node [style=none] (126) at (-3, 5.625) {};
		\node [style=none] (127) at (3.5, 7) {};
		\node [style=none] (128) at (3.5, 5) {};
		\node [style=none] (129) at (3.5, 5.625) {};
		\node [style=none] (130) at (2.5, 2.6) {};
		\node [style=none] (131) at (-2.5, 2.6) {};
		\node [style=none] (132) at (2.5, -0.4) {};
		\node [style=none] (133) at (-2.5, -0.4) {};
		\node [style=none] (134) at (1, 1.1) {};
		\node [style=none] (135) at (-1, 1.1) {};
		\node [style=black] (137) at (5, -0.4) {};
		\node [style=none] (138) at (3.5, 2.6) {};
		\node [style=none] (139) at (-3, 2.6) {};
		\node [style=none] (140) at (-5.5, -0.4) {};
		\node [style=none] (141) at (-3, 6.35) {};
		\node [style=none] (142) at (3.5, 6.35) {};
		\node [style=label] (143) at (-6, -0.375) {$A$};
		\node [style=bbox, scale=0.6] (145) at (3, 2.525) {$s$};
		\node [style=bboxOp, scale=0.6] (146) at (3, 3.475) {$s$};
	\end{pgfonlayer}
	\begin{pgfonlayer}{edgelayer}
		\draw [tape] (116.center)
			 to [bend left=90, looseness=1.75] (124.center)
			 to (127.center)
			 to [bend right=270, looseness=1.75] (120.center)
			 to (80.center)
			 to [bend right=90, looseness=1.50] (86.center)
			 to (117.center)
			 to (118.center)
			 to (89.center)
			 to [bend left] (96.center)
			 to (107.center)
			 to [bend left] (103.center)
			 to (114.center)
			 to (113.center)
			 to (102.center)
			 to [bend right=90, looseness=1.50] (100.center)
			 to cycle;
		\draw [tape, fill=white] (99.center)
			 to [bend left] (106.center)
			 to (95.center)
			 to [bend left] (77.center)
			 to (119.center)
			 to [bend right=90, looseness=1.50] (128.center)
			 to (125.center)
			 to [bend right=90, looseness=1.50] (115.center)
			 to cycle;
		\draw [bend right=45] (110.center) to (82.center);
		\draw [bend left=45] (110.center) to (90.center);
		\draw [bend left=45] (85.center) to (109.center);
		\draw [bend left=45] (109.center) to (71.center);
		\draw (85.center) to (112);
		\draw (109.center) to (110.center);
		\draw (82.center) to (122.center);
		\draw (90.center) to (123.center);
		\draw [bend left=90, looseness=1.75] (129.center) to (121.center);
		\draw [bend left=90, looseness=1.75] (122.center) to (126.center);
		\draw (129.center) to (126.center);
		\draw [bend right=45] (135.center) to (131.center);
		\draw [bend left=45] (135.center) to (133.center);
		\draw [bend left=45] (132.center) to (134.center);
		\draw [bend left=45] (134.center) to (130.center);
		\draw (132.center) to (137);
		\draw (134.center) to (135.center);
		\draw (131.center) to (139.center);
		\draw (133.center) to (140.center);
		\draw [bend left=90, looseness=1.75] (142.center) to (138.center);
		\draw [bend left=90, looseness=1.75] (139.center) to (141.center);
		\draw (142.center) to (141.center);
		\draw (71.center) to (121.center);
		\draw (138.center) to (130.center);
	\end{pgfonlayer}
\end{tikzpicture}
}
 = 
    \begin{tikzpicture}[scale=0.8]
	\begin{pgfonlayer}{nodelayer}
		\node [style=black] (90) at (0.25, 0.375) {};
		\node [style=label] (91) at (-2, 0.375) {$A$};
		\node [style=none] (102) at (1.5, 1) {};
		\node [style=none] (103) at (1.5, -1) {};
		\node [style=none] (113) at (-1.5, 1) {};
		\node [style=none] (114) at (-1.5, -1) {};
		\node [style=none] (123) at (-1.5, 0.375) {};
		\node [style=black] (133) at (0.25, -0.4) {};
		\node [style=none] (140) at (-1.5, -0.4) {};
		\node [style=label] (143) at (-2, -0.375) {$A$};
	\end{pgfonlayer}
	\begin{pgfonlayer}{edgelayer}
		\draw[tape] (102.center)
			 to (103.center)
			 to (114.center)
			 to (113.center)
			 to cycle;
		\draw (90) to (123.center);
		\draw (133) to (140.center);
	\end{pgfonlayer}
\end{tikzpicture}
}
 \]
\end{lemma}
\begin{proof}
    First observe that the following holds:
    \begin{align*}
        
    \InputIfFileExists{addition/terminates/derivation/step1.tikz}{}{\input{./tikz/addition/terminates/derivation/step1.tikz}}
 &\leftstackrel{\eqref{ax:dischargeradj2}}{\leq} 
    \InputIfFileExists{addition/terminates/derivation/step2.tikz}{}{\input{./tikz/addition/terminates/derivation/step2.tikz}}
 \stackrel{\eqref{ax:codiagnat}, \eqref{ax:diagnat}}{=} 
    \InputIfFileExists{addition/terminates/derivation/step3.tikz}{}{\input{./tikz/addition/terminates/derivation/step3.tikz}}
 \stackrel{\eqref{ax:dischargeradj1}}{\leq} 
    \InputIfFileExists{addition/terminates/derivation/step4.tikz}{}{\input{./tikz/addition/terminates/derivation/step4.tikz}}
 \\
                                                       &\leftstackrel{\eqref{ax:peano:succ:sv}, \eqref{ax:peano:succ:tot}}{=} 
    \InputIfFileExists{addition/terminates/derivation/step5.tikz}{}{\input{./tikz/addition/terminates/derivation/step5.tikz}}
.
    \end{align*}
    Then, by~\eqref{ax:trace:tape:AU1}, the inequality below holds and the derivation concludes the proof.
    \[
        
    }
  \stackrel{\eqref{ax:peanoT:ind}}{=} 
    \InputIfFileExists{addition/terminates/derivation/unif_lhs.tikz}{}{\input{./tikz/addition/terminates/derivation/unif_lhs.tikz}}
 \leq 
    \InputIfFileExists{addition/terminates/derivation/unif_rhs.tikz}{}{\input{./tikz/addition/terminates/derivation/unif_rhs.tikz}}

    \]
\end{proof}

\section{Diagrammatic Hoare Logic}\label{sec:hoare}

In this section we illustrate how Kleene Cartesian tapes can provide a comfortable setting to reason about imperative programs. For the sake of generality, we avoid to fix basic types and operations and, rather, we work parametrically with respect to a triple $(\sort,\mathcal{F},\mathcal{P})$. $\sort$ is a set of sorts, ranging over $A,B,C, \dots$, representing basic types; $\mathcal{F}$ is a set of function symbols, ranging over $f,g,h,\dots$, equipped with an arity in $\sort^\star$ and a coarity in $\sort$. As usual, we write $f\colon U \to A$ to mean that $f$ has arity $U$ and coarity $A$;  $\mathcal{P}$ is a set of predicate symbols, ranging over $P,Q,R,\dots$, equipped only with an arity in $\sort^\star$. The coarity, for all predicate symbols, is fixed to be $1$.

For any predicate symbol $P$, we consider an extra symbol $\bar{P}$, with the same arity of $P$, that will represent its negation. We fix $\bar{\mathcal{P}} \defeq \{\bar{P} \mid P\in \mathcal{P}\}$.
We consider the monoidal signature $\sign \defeq \mathcal{F} \cup \mathcal{P} \cup \bar{\mathcal{P}}$. %
The kc tape theory $\basicR$ contains the following equations for all $f\colon U \to A$ in $\mathcal{F}$
\begin{equation}\label{eq:deterministic-total}
  
    \begin{tikzpicture}[scale=1.5]
	\begin{pgfonlayer}{nodelayer}
		\node [style=bbox, scale=0.9] (107) at (0, 0) {$f$};
		\node [style=label] (110) at (-1.25, 0) {$U$};
		\node [style=none] (117) at (-0.75, 0) {};
		\node [style=label] (120) at (2, -0.375) {$A$};
		\node [style=black] (121) at (0.75, 0) {};
		\node [style=none] (122) at (1.25, -0.375) {};
		\node [style=none] (123) at (1.25, 0.375) {};
		\node [style=label] (124) at (2, 0.375) {$A$};
		\node [style=none] (125) at (1.5, -0.375) {};
		\node [style=none] (126) at (1.5, 0.375) {};
		\node [style=none] (127) at (-0.75, 0.75) {};
		\node [style=none] (128) at (1.5, 0.75) {};
		\node [style=none] (129) at (1.5, -0.75) {};
		\node [style=none] (130) at (-0.75, -0.75) {};
	\end{pgfonlayer}
	\begin{pgfonlayer}{edgelayer}
    \draw [tape] (129.center)
			 to (128.center)
			 to (127.center)
			 to (130.center)
			 to cycle;
		\draw [bend right] (123.center) to (121);
		\draw [bend right] (121) to (122.center);
		\draw (107) to (121);
		\draw (123.center) to (126.center);
		\draw (125.center) to (122.center);
		\draw (117.center) to (107);
	\end{pgfonlayer}
\end{tikzpicture}
}
 = 
    \begin{tikzpicture}[scale=1.5]
	\begin{pgfonlayer}{nodelayer}
		\node [style=label] (110) at (-1.25, 0) {$U$};
		\node [style=label] (120) at (2.25, -0.375) {$A$};
		\node [style=label] (124) at (2.25, 0.375) {$A$};
		\node [style=none] (127) at (-0.75, 0.75) {};
		\node [style=none] (128) at (1.75, 0.75) {};
		\node [style=none] (129) at (1.75, -0.75) {};
		\node [style=none] (130) at (-0.75, -0.75) {};
		\node [style=none] (131) at (-0.75, 0) {};
		\node [style=black] (133) at (0, 0) {};
		\node [style=none] (136) at (1.75, 0.375) {};
		\node [style=none] (137) at (1.75, -0.375) {};
		\node [style=none] (138) at (0.5, 0.375) {};
		\node [style=none] (139) at (0.5, -0.375) {};
		\node [style=bbox, scale=0.8] (140) at (1, 0.375) {$f$};
		\node [style=bbox, scale=0.8] (141) at (1, -0.375) {$f$};
	\end{pgfonlayer}
	\begin{pgfonlayer}{edgelayer}
		\draw [tape] (129.center)
			 to (128.center)
			 to (127.center)
			 to (130.center)
			 to cycle;
		\draw (131.center) to (133);
		\draw [bend left] (133) to (138.center);
		\draw [bend right] (133) to (139.center);
		\draw (139.center) to (141);
		\draw (141) to (137.center);
		\draw (136.center) to (140);
		\draw (140) to (138.center);
	\end{pgfonlayer}
\end{tikzpicture}
}
 \qquad 
    \begin{tikzpicture}[scale=1.5]
	\begin{pgfonlayer}{nodelayer}
		\node [style=bbox, scale=0.9] (107) at (0, 0) {$f$};
		\node [style=label] (110) at (-1.25, 0) {$U$};
		\node [style=none] (117) at (-0.75, 0) {};
		\node [style=black] (121) at (0.75, 0) {};
		\node [style=none] (127) at (-0.75, 0.75) {};
		\node [style=none] (128) at (1.5, 0.75) {};
		\node [style=none] (129) at (1.5, -0.75) {};
		\node [style=none] (130) at (-0.75, -0.75) {};
	\end{pgfonlayer}
	\begin{pgfonlayer}{edgelayer}
		\draw [tape] (129.center)
			 to (128.center)
			 to (127.center)
			 to (130.center)
			 to cycle;
		\draw (107) to (121);
		\draw (117.center) to (107);
	\end{pgfonlayer}
\end{tikzpicture}
}
 = 
    \begin{tikzpicture}[scale=1.5]
	\begin{pgfonlayer}{nodelayer}
		\node [style=label] (110) at (-1.25, 0) {$U$};
		\node [style=none] (127) at (-0.75, 0.75) {};
		\node [style=none] (128) at (1.75, 0.75) {};
		\node [style=none] (129) at (1.75, -0.75) {};
		\node [style=none] (130) at (-0.75, -0.75) {};
		\node [style=none] (131) at (-0.75, 0) {};
		\node [style=black] (133) at (0.5, 0) {};
	\end{pgfonlayer}
	\begin{pgfonlayer}{edgelayer}
		\draw [tape] (129.center)
			 to (128.center)
			 to (127.center)
			 to (130.center)
			 to cycle;
		\draw (131.center) to (133);
	\end{pgfonlayer}
\end{tikzpicture}
}

\end{equation}
and for all $P\in \mathcal{P}$,\\
\[
  \!\!\!\!\!\!\!\!\!
    \begin{tikzpicture}
	\begin{pgfonlayer}{nodelayer}
		\node [style=bbox] (69) at (0, 1.5) {$P$};
		\node [style=none] (77) at (0.5, 2.5) {};
		\node [style=none] (80) at (0.75, 0.5) {};
		\node [style=none] (82) at (-0.5, 1.5) {};
		\node [style=label] (83) at (-4.5, 0) {$U$};
		\node [style=none] (86) at (0.75, -0.5) {};
		\node [style=none] (89) at (0.5, -2.5) {};
		\node [style=none] (90) at (-0.5, -1.5) {};
		\node [style=none] (94) at (3.75, 1.15) {};
		\node [style=none] (95) at (2.5, 1.15) {};
		\node [style=none] (96) at (2.5, -1.15) {};
		\node [style=none] (97) at (3.75, -1.15) {};
		\node [style=none] (99) at (-0.5, 2.5) {};
		\node [style=none] (100) at (-0.75, 0.5) {};
		\node [style=none] (102) at (-0.75, -0.5) {};
		\node [style=none] (103) at (-0.5, -2.5) {};
		\node [style=none] (104) at (-3.75, 0) {};
		\node [style=none] (105) at (-3.75, 1.15) {};
		\node [style=none] (106) at (-2.5, 1.15) {};
		\node [style=none] (107) at (-2.5, -1.15) {};
		\node [style=none] (108) at (-3.75, -1.15) {};
		\node [style=none] (110) at (-2, 0) {};
		\node [style=bbox] (111) at (0, -1.5) {$\bar P$};
	\end{pgfonlayer}
	\begin{pgfonlayer}{edgelayer}
		\draw [tape] (106.center)
			 to (105.center)
			 to (108.center)
			 to (107.center)
			 to [bend right] (103.center)
			 to (89.center)
			 to [bend right] (96.center)
			 to (97.center)
			 to (94.center)
			 to (95.center)
			 to [bend right] (77.center)
			 to (99.center)
			 to [bend right] cycle;
		\draw [tape, fill=white] (100.center)
			 to (80.center)
			 to [bend left=90, looseness=1.50] (86.center)
			 to (102.center)
			 to [bend left=90, looseness=1.50] cycle;
		\draw (82.center) to (69);
		\draw (104.center) to (110.center);
		\draw [bend left=45] (110.center) to (82.center);
		\draw [bend right=45] (110.center) to (90.center);
		\draw (90.center) to (111);
	\end{pgfonlayer}
\end{tikzpicture}
}
 = 
    \begin{tikzpicture}[scale=1.5]
	\begin{pgfonlayer}{nodelayer}
		\node [style=label] (110) at (-1.25, 0) {$U$};
		\node [style=none] (127) at (-0.75, 0.75) {};
		\node [style=none] (128) at (1.75, 0.75) {};
		\node [style=none] (129) at (1.75, -0.75) {};
		\node [style=none] (130) at (-0.75, -0.75) {};
		\node [style=none] (131) at (-0.75, 0) {};
		\node [style=black] (133) at (0.5, 0) {};
	\end{pgfonlayer}
	\begin{pgfonlayer}{edgelayer}
		\draw [tape] (129.center)
			 to (128.center)
			 to (127.center)
			 to (130.center)
			 to cycle;
		\draw (131.center) to (133);
	\end{pgfonlayer}
\end{tikzpicture}
}
 \quad\qquad 
    \begin{tikzpicture}[scale=1.5]
	\begin{pgfonlayer}{nodelayer}
		\node [style=label] (110) at (-1.25, 0) {$U$};
		\node [style=none] (127) at (-0.75, 0.75) {};
		\node [style=none] (128) at (1.75, 0.75) {};
		\node [style=none] (129) at (1.75, -0.75) {};
		\node [style=none] (130) at (-0.75, -0.75) {};
		\node [style=none] (131) at (-0.75, 0) {};
		\node [style=black] (133) at (0, 0) {};
		\node [style=none] (138) at (0.5, 0.375) {};
		\node [style=none] (139) at (0.5, -0.375) {};
		\node [style=bbox, scale=0.8] (140) at (1, 0.375) {$P$};
		\node [style=bbox, scale=0.8] (141) at (1, -0.375) {$\bar{P}$};
	\end{pgfonlayer}
	\begin{pgfonlayer}{edgelayer}
		\draw [tape] (129.center)
			 to (128.center)
			 to (127.center)
			 to (130.center)
			 to cycle;
		\draw (131.center) to (133);
		\draw [bend left] (133) to (138.center);
		\draw [bend right] (133) to (139.center);
		\draw (139.center) to (141);
		\draw (140) to (138.center);
	\end{pgfonlayer}
\end{tikzpicture}
}
 = 
    \begin{tikzpicture}[scale=1.5]
	\begin{pgfonlayer}{nodelayer}
		\node [style=label] (110) at (-1.25, 0) {$U$};
		\node [style=none] (127) at (-0.75, 0.75) {};
		\node [style=none] (128) at (0, 0.75) {};
		\node [style=none] (129) at (0, -0.75) {};
		\node [style=none] (130) at (-0.75, -0.75) {};
		\node [style=none] (131) at (-0.75, 0) {};
		\node [style=none] (132) at (2.75, 0.75) {};
		\node [style=none] (133) at (2, 0.75) {};
		\node [style=none] (134) at (2, -0.75) {};
		\node [style=none] (135) at (2.75, -0.75) {};
		\node [style=none] (137) at (0.775, 0) {};
	\end{pgfonlayer}
	\begin{pgfonlayer}{edgelayer}
		\draw [tape] (129.center)
			 to [bend right=90, looseness=1.75] (128.center)
			 to (127.center)
			 to (130.center)
			 to cycle;
		\draw [tape] (134.center)
			 to [bend left=90, looseness=1.75] (133.center)
			 to (132.center)
			 to (135.center)
			 to cycle;
		\draw (131.center) to (137.center);
	\end{pgfonlayer}
\end{tikzpicture}
}

\]
This axioms force any model of the theory to interpret symbols in $\mathcal{F}$ as functions and $\bar{P}$ as the complement of $P$.

\begin{remark}
Notice that, in some program logics, function symbols are not forced to be total in order to deal with errors and exceptions like, for instance, division by $0$. One can easily allow symbols in $\mathcal{F}$ to be partial functions by dropping the rightmost axiom in \Cref{eq:deterministic-total}, but we would loose in elegance. In particular, the equality in Lemma \ref{lemma:encodingsubst} will become an inequality and, consequently, the proof of the assignment rule will require some extra work.
\end{remark}

\subsection{Expressions}
As usual, expressions are defined by the following grammar
\[e \Coloneqq x \mid f(e_{1}, \ldots, e_{n})\]
where $f$ is a function symbol in $\mathcal{F}$ and $x$ is a variable taken from some fixed set of variables.

In order to encode expressions into diagrams, we need to make copying and discarding of variables explicit; this is usually kept implicit by traditional syntax. For this reason we define an elementary type systems with judgement of the form
\[\Gamma \vdash e \colon A \]
where $e$ is an expression, $A$ is a sort in $\sort$ and $\Gamma$ is a \emph{typing context}, namely, an ordered sequence $x_1\colon A_1, \dots x_n \colon A_n$ where, all the $x_i$ are distinct variables and $A_i \in \sort$.
The type system consists of the following two rules
\[ \inferrule*[right=(var)]{-}{\Gamma, x \colon A, \Delta \vdash x \colon A} \qquad \inferrule*[right=(op)]{\Gamma \vdash e_{i} \colon A_{i} \quad f \colon A_{1} \per \cdots \per A_{n} \to A}{\Gamma \vdash f(e_{1}, \ldots, e_{n}) \colon A}\]
where $\Gamma$ and $\Delta$ are arbitrary typing contexts.

To encode well typed expressions into diagrams, we first deal with  typing contexts. We fix $\encoding{x_1\colon A_1, \dots , x_n \colon A_n} \defeq A_1 \per \dots  \per A_n$. For any well typed expressions $\Gamma \vdash e \colon A$, $\encoding{\cdot}$ is defined by induction on the two rules above:
\[ \encoding{\Gamma, x \colon A, \Delta \vdash x \colon A} \defeq \discharger{\encoding{\Gamma}} \per \id{A} \per \discharger{\encoding{\Delta}}\qquad \encoding{\Gamma \vdash f(e_{1}, \ldots, e_{n}) \colon A} \defeq \copier{\encoding{\Gamma}}^n ; (\encoding{\Gamma \vdash e_1} \per \dots \per \encoding{\Gamma \vdash e_n} ) ; f\]
where $\copier{U}^n \colon U \to U^n$ is the diagram defined inductively as expected: $\copier{\encoding{\Gamma}}^0 \defeq \discharger{U}$ and $\copier{\encoding{\Gamma}}^{n+1} \defeq \copier{U} ; (\copier{U}^n \per \id{U})$.

Later on, the notion of \emph{substitution} will be crucial. Given two expression $t$ and $e$ and a variable $x$, the expression $e[t / x]$ is defined inductively as follows, where $y$ is a variable different from $x$.
\[x[t / x] \defeq t \qquad y[ t / x]\defeq y \qquad f(e_1, \dots, e_n)[t/x] \defeq f(e_1[t/x], \dots, e_n[t/x] )\]

A simple inductive argument confirms that substitutions types well.

\begin{lemma}
Let $\Gamma ' = \Gamma, x\colon A, \Delta$ for some typing contexts $\Gamma$ and $\Delta$. If $\Gamma' \vdash e\colon B$ and $\Gamma ' \vdash t\colon A$, then $\Gamma' \vdash e[t /x] \colon B$.
\end{lemma}
\begin{proof}
We proceed by induction on $\Gamma' \vdash e\colon B$.

If $e$ is the variable $x$, then by definition of substitution $e[t / x] = t$ and by hypothesis we know that $\Gamma ' \vdash t\colon A$. Observe that by the typing rule (VAR), $B$ is forced to be the type of $x$, i.e., $A$. Thus, $\Gamma' \vdash x[t /x] \colon B$.

If $e$ is a variable $y$ different from $x$, then by definition of substitution $e[t / x] = y$ and by hypothesis we know that $\Gamma ' \vdash y\colon B$.

If $e=f(e_1, \dots, e_n)$, then by definition of substitution $e[t / x] = f(e_1 [t/x], \dots , e_n [t/x])$; by typing rule (OP), we know that $f\colon A_1 \per \dots A_n \to B$ and $\Gamma' \vdash e_i \colon A_i$. From the latter and  induction hypothesis, $\Gamma' \vdash e_i [t/x] \colon A_i$. Again by the rule (OP), we have that $\Gamma' \vdash f(e_1 [t/x], \dots , e_n [t/x])\colon B$.
\end{proof}

The following result will be useful for the assignment rule below.

\begin{lemma}\label{lemma:encodingsubst}
Let $\Gamma ' = \Gamma, x\colon A, \Delta$ for some typing contexts $\Gamma$ and $\Delta$. If $\Gamma' \vdash e\colon B$ and $\Gamma ' \vdash t\colon A$, then
\[\encoding{ \Gamma' \vdash e[t /x] \colon B} = 
    \InputIfFileExists{hoare/subs/step1.tikz}{}{\input{./tikz/hoare/subs/step1.tikz}}
\] %
\end{lemma}
\begin{proof}
The proof proceeds by induction on $\Gamma' \vdash e\colon B$.

If $e$ is the variable $x$, then by the rule (VAR) $A=B$. Moreover, by definition of $\encoding{\cdot}$, $\encoding{\Gamma, x\colon A, \Delta \vdash x \colon A} =  \discharger{\encoding{\Gamma}} \per \id{A} \per \discharger{\encoding{\Delta}}$. 
Thus
\[ 
    \InputIfFileExists{hoare/subs/step1.tikz}{}{\input{./tikz/hoare/subs/step1.tikz}}
 = 
    \InputIfFileExists{hoare/subs/step2.tikz}{}{\input{./tikz/hoare/subs/step2.tikz}}
 \stackrel{\eqref{ax:copierun}}{=} 
    \InputIfFileExists{hoare/subs/step3.tikz}{}{\input{./tikz/hoare/subs/step3.tikz}}
 = \encoding{\Gamma ' \vdash x [t / x]\colon A} \]

If $e$ is a variable $y$, different from $x$, then by the rule (VAR), there are two possible cases: either $\Gamma = \Gamma_1, y\colon B , \Gamma_2$ for some typing contexts $\Gamma_1$ and $\Gamma_2$ or  $\Delta = \Delta_1, y\colon B , \Delta_2$.
We consider the first case, the second is symmetrical. Observe that, by definition of $\encoding{\cdot}$,

\[\encoding{\Gamma_1, y\colon B, \Gamma_2, x \colon A, \Delta \vdash y \colon B} =  \discharger{\encoding{\Gamma_1}} \per \id{B} \per  \discharger{\encoding{\Gamma_2}} \per  \discharger{A} \per  \discharger{\encoding{\Delta}} \]

Thus,
\[ 
    \InputIfFileExists{hoare/subs/step1.tikz}{}{\input{./tikz/hoare/subs/step1.tikz}}
 = 
    \InputIfFileExists{hoare/subs/step4.tikz}{}{\input{./tikz/hoare/subs/step4.tikz}}
 \stackrel{\eqref{eq:deterministic-total}}{=} 
    \InputIfFileExists{hoare/subs/step5.tikz}{}{\input{./tikz/hoare/subs/step5.tikz}}
 \stackrel{\eqref{ax:copierun}}{=} 
    \InputIfFileExists{hoare/subs/step6.tikz}{}{\input{./tikz/hoare/subs/step6.tikz}}
 = \encoding{\Gamma ' \vdash y [t / x]\colon A} \]

If \(e\) is an application, $e=f(e_1, \dots , e_n)$, by definition of \(\encoding{\cdot}\) on operations, \(\encoding{\Gamma' \vdash f(e_{1}, \ldots, e_{n}) \colon A} \defeq \copier{\encoding{\Gamma'}}^n ; (\encoding{\Gamma' \vdash e_1} \per \dots \per \encoding{\Gamma' \vdash e_n} ) ; f\).
By naturality of copy, we obtain
\begin{align*}
  
    \InputIfFileExists{hoare/subs/step1.tikz}{}{\input{./tikz/hoare/subs/step1.tikz}}
 &= 
    \InputIfFileExists{hoare/subs/step7.tikz}{}{\input{./tikz/hoare/subs/step7.tikz}}
 \stackrel{\eqref{eq:deterministic-total}}{=} 
    \InputIfFileExists{hoare/subs/step8.tikz}{}{\input{./tikz/hoare/subs/step8.tikz}}
 \\
  &= \encoding{\Gamma' \vdash f(e_{1} [t / x],\ldots ,e_{n} [t / x])\colon A} = \encoding{\Gamma' \vdash f(e_{1},\ldots ,e_{n}) [t / x]\colon A}
\end{align*}
\end{proof}

\subsection{Predicates}
We now turn to predicates, defined by the following grammar 
\[P \Coloneqq R(e_1, \dots, e_n) \mid \bar{R}(e_1, \dots, e_n) \mid \top \mid \bot  \mid P \lor P \mid P \land P  \] %
where $R$ is a symbol in $\mathcal{P}$ and $e_i$ are expressions. Observe that negation $\lnot P$ can be easily defined as expected:
$$\begin{array}{c@{\qquad}c@{\qquad}c}
\neg R(e_1, \dots, e_n) \defeq \bar{R}(e_1, \dots, e_n) &   \neg \top \defeq \bot &  \neg(P \lor P) \defeq \neg P \land \neg Q \\
\neg \bar{R}(e_1, \dots, e_n) \defeq R(e_1, \dots, e_n) & \neg \bot \defeq \top & \neg(P \land P) \defeq \neg P \lor \neg Q
\end{array}
$$
As in the case of expressions, we consider a simple type system for predicates.
  \begin{equation*}
  \begin{array}{c@{\qquad}c@{\qquad}c@{\qquad}c}
    \inferrule*[right=(top)]{ }{\Gamma \vdash \top \colon 1} &  \inferrule*[right=(and)]{\Gamma \vdash P \colon 1 \quad \Gamma \vdash Q \colon 1}{\Gamma \vdash (P \land Q) \colon 1} & \inferrule*[right=($\bar{\textnormal{R}}$)]{\Gamma\vdash e_i\colon A_i \quad R \colon A_1 \per \dots \per A_n \to 1}{\Gamma \vdash \bar{R}(e_1, \dots e_n) \colon 1} \\[12pt]
    \inferrule*[right=(bot)]{ }{\Gamma \vdash \bot \colon 1} & \inferrule*[right=(or)]{\Gamma \vdash P \colon 1 \quad \Gamma \vdash Q \colon 1}{\Gamma \vdash (P \lor Q) \colon 1} & \inferrule*[right=(R)]{\Gamma\vdash e_i\colon A_i \quad R \colon A_1 \per \dots \per A_n \to 1}{\Gamma \vdash R(e_1, \dots e_n) \colon 1}
  \end{array}
  \end{equation*}
The encoding $\encoding{\cdot}$ maps well typed predicates $\Gamma \vdash P \colon 1$ into diagrams of type $\encoding{\Gamma} \to 1$.
$$\begin{array}{rcl}
\encoding{\Gamma \vdash R(e_1, \dots e_n)\colon 1} &\defeq& \copier{\encoding{\Gamma}}^n ; (\encoding{\Gamma \vdash e_1 \colon 1} \per \dots \per \encoding{\Gamma \vdash e_n \colon 1}) ; R \\
\encoding{\Gamma \vdash \bar{R}(e_1, \dots e_n)\colon 1} &\defeq& \copier{\encoding{\Gamma}}^n ; (\encoding{\Gamma \vdash e_1 \colon 1} \per \dots \per \encoding{\Gamma \vdash e_n \colon 1}) ; \bar{R} \\
\encoding{\Gamma \vdash \top\colon 1} &\defeq& \discharger{\encoding{\Gamma}}\\
\encoding{\Gamma \vdash \bot \colon 1} &\defeq& \bang{\encoding{\Gamma}} \dcomp \cobang{1}\\
\encoding{\Gamma \vdash P \lor  Q \colon 1} &\defeq& \diag{\encoding{\Gamma}} ; (\encoding{\Gamma \vdash P\colon 1} \piu \encoding{\Gamma \vdash P\colon 1} ) ; \codiag{1}\\ %
\encoding{\Gamma \vdash P \land Q \colon 1} &\defeq& \copier{\encoding{\Gamma}} ;  (\encoding{\Gamma \vdash P\colon 1} \per \encoding{\Gamma \vdash P\colon 1} )    %
\end{array}
$$
Similarly to the case of expressions, one defines substitution on a variable $x$ of a term $t$ in a predicate $P$ inductively:
\[ R(e_1, \dots, e_n) [t/x] \defeq R(e_1[t/x], \dots, e_n[t/x]) \qquad \bar{R}(e_1, \dots, e_n) [t/x] \defeq \bar{R}(e_1[t/x], \dots, e_n[t/x]) \]
\[ \top [t/x] \defeq \top \quad \bot [t/x] \defeq \bot \quad (P \lor Q)[t/x]\defeq P[t/x] \lor Q[t/x] \quad (P \land Q)[t/x]\defeq P[t/x] \land Q[t/x] \]
A simple inductive arguments confirm that substituion are well-typed. Moreover
\begin{lemma}\label{lemma:encodingsubstPed}
Let $\Gamma ' = \Gamma, x\colon A, \Delta$ for some typing contexts $\Gamma$ and $\Delta$. If $\Gamma' \vdash P\colon 1$ and $\Gamma ' \vdash t\colon A$, then
\[\encoding{ \Gamma' \vdash P[t /x] \colon 1} = 
    \InputIfFileExists{hoare/pred/step1.tikz}{}{\input{./tikz/hoare/pred/step1.tikz}}
 \]%
\end{lemma}
\begin{proof}
  Proceed by induction on the typing rules for predicates.

  If \(P\) is \(\top\), then
  \[
    
    \InputIfFileExists{hoare/pred/step1.tikz}{}{\input{./tikz/hoare/pred/step1.tikz}}
 =  
    \InputIfFileExists{hoare/pred/step2.tikz}{}{\input{./tikz/hoare/pred/step2.tikz}}
 \stackrel{\eqref{eq:deterministic-total}}{=} 
    \InputIfFileExists{hoare/pred/step3.tikz}{}{\input{./tikz/hoare/pred/step3.tikz}}
 \stackrel{\eqref{ax:copierun}}{=} 
    \InputIfFileExists{hoare/pred/step4.tikz}{}{\input{./tikz/hoare/pred/step4.tikz}}
  = \encoding{\Gamma' \vdash \top \colon 1} = \encoding{\Gamma' \vdash \top[t /x] \colon 1}.
  \]

  If \(P\) is \(\bot\), then
  \[  
    
    \InputIfFileExists{hoare/pred/step1.tikz}{}{\input{./tikz/hoare/pred/step1.tikz}}
 = 
    \InputIfFileExists{hoare/pred/step5.tikz}{}{\input{./tikz/hoare/pred/step5.tikz}}
  \stackrel{\eqref{ax:bangnat}}{=} 
    \InputIfFileExists{hoare/pred/step6.tikz}{}{\input{./tikz/hoare/pred/step6.tikz}}
 = \encoding{\Gamma' \vdash \bot \colon 1} = \encoding{\Gamma' \vdash \bot[t /x] \colon 1}.
  \]

  If \(P\) is a predicate symbol \(R\), then
  \begin{align*}
    
    \InputIfFileExists{hoare/pred/step1.tikz}{}{\input{./tikz/hoare/pred/step1.tikz}}
 &= 
    \InputIfFileExists{hoare/pred/step7.tikz}{}{\input{./tikz/hoare/pred/step7.tikz}}
 \stackrel{\eqref{eq:deterministic-total}}{=} 
    \InputIfFileExists{hoare/pred/step8.tikz}{}{\input{./tikz/hoare/pred/step8.tikz}}
 \\
    &= \encoding{\Gamma' \vdash R(e_{1} [t / x],\ldots ,e_{n} [t / x])\colon A} = \encoding{\Gamma' \vdash R(e_{1},\ldots ,e_{n}) [t / x]\colon A}.
  \end{align*}

  If \(P\) is a negated predicate symbol \(\bar{R}\), then
  \begin{align*}
    
    \InputIfFileExists{hoare/pred/step1.tikz}{}{\input{./tikz/hoare/pred/step1.tikz}}
 &= 
    \InputIfFileExists{hoare/pred/step9.tikz}{}{\input{./tikz/hoare/pred/step9.tikz}}
 \stackrel{\eqref{eq:deterministic-total}}{=} 
    \InputIfFileExists{hoare/pred/step10.tikz}{}{\input{./tikz/hoare/pred/step10.tikz}}
 \\
    &= \encoding{\Gamma' \vdash \bar{R}(e_{1} [t / x],\ldots ,e_{n} [t / x])\colon A} = \encoding{\Gamma' \vdash \bar{R}(e_{1},\ldots ,e_{n}) [t / x]\colon A}.
  \end{align*}

  For the conjunction case, \(P = Q \land R\),%
  \begin{align*}
    
    \InputIfFileExists{hoare/pred/step1.tikz}{}{\input{./tikz/hoare/pred/step1.tikz}}
 &= 
    \InputIfFileExists{hoare/pred/step11.tikz}{}{\input{./tikz/hoare/pred/step11.tikz}}
 \stackrel{\eqref{eq:deterministic-total}}{=} 
    \InputIfFileExists{hoare/pred/step12.tikz}{}{\input{./tikz/hoare/pred/step12.tikz}}
 \\
    &= \encoding{\Gamma' \vdash Q[t/x] \colon 1} \land \encoding{\Gamma' \vdash R[t/x] \colon 1} \\
    &= \encoding{\Gamma' \vdash Q[t/x] \land R[t/x] \colon 1} = \encoding{\Gamma' \vdash (Q \land R)[t/x] \colon 1}.
  \end{align*}

  For the disjunction case, \(P = Q \lor R\),%
  \begin{align*}
    
    \InputIfFileExists{hoare/pred/step1.tikz}{}{\input{./tikz/hoare/pred/step1.tikz}}
 &= 
    \InputIfFileExists{hoare/pred/step13.tikz}{}{\input{./tikz/hoare/pred/step13.tikz}}
 \stackrel{\eqref{ax:diagnat}}{=} 
    \InputIfFileExists{hoare/pred/step14.tikz}{}{\input{./tikz/hoare/pred/step14.tikz}}
 \\
    &= \encoding{\Gamma' \vdash Q[t/x] \colon 1} \lor \encoding{\Gamma' \vdash R[t/x] \colon 1} \\
    &= \encoding{\Gamma' \vdash Q[t/x] \lor R[t/x] \colon 1} = \encoding{\Gamma' \vdash (Q \lor R)[t/x] \colon 1}.
  \end{align*}
\end{proof}

\subsection{Commands}
Commands are defined by the following grammar
\[ C = \mathsf{abort} \mid \mathsf{skip} \mid \mathop{\mathsf{if}} P \mathop{\mathsf{then}} C \mathop{\mathsf{else}} D \mid \mathop{\mathsf{while}} P \mathop{\mathsf{do}} C \mid C ; D \mid x \coloneqq e\]
where $P$ is a predicate and $e$ an expression. The type system is rather simple: the only non straightforward rule is the one for the assignment where one has to ensure that the type of $e$ is the same to the one of $x$.
\begin{equation*}
\begin{array}{c}
  \begin{array}{ccc}
    \inferrule*[right=(abort)]{ }{\Gamma \vdash \mathsf{abort} }&
\inferrule*[right=(skip)]{ }{\Gamma \vdash \mathsf{skip} }&
    \inferrule*[right=(assn)]{\Gamma = \Gamma',x \colon A,\Delta' \qquad \Gamma \vdash e \colon A}{\Gamma \vdash x \coloneqq e }
    \end{array}\\
    \begin{array}{ccc}
    \inferrule*[right=(;)]{\Gamma \vdash C   \quad \Gamma \vdash D }{\Gamma \vdash C ; D }&
    \inferrule*[right=(if)]{\Gamma \vdash P \colon 1 \quad \Gamma \vdash C   \quad \Gamma \vdash D  }{\Gamma \vdash \mathop{\mathsf{if}} P \mathop{\mathsf{then}} C \mathop{\mathsf{else}} D }&
\inferrule*[right=(while)]{\Gamma \vdash P \colon 1 \quad \Gamma \vdash C  }{ \Gamma \vdash \mathop{\mathsf{while}} P \mathop{\mathsf{do}} C }
\end{array}
\end{array}
\end{equation*}

The predicates occurring in a command will be regarded as the corresponding coreflexive: we fix
\[ c(P) \defeq 
    \begin{tikzpicture}[scale=1.5]
	\begin{pgfonlayer}{nodelayer}
		\node [style=none] (127) at (1.75, 1) {};
		\node [style=none] (128) at (-1, 1) {};
		\node [style=none] (129) at (-1, -0.75) {};
		\node [style=none] (130) at (1.75, -0.75) {};
		\node [style=label] (131) at (-1.5, 0) {$U$};
		\node [style=none] (132) at (-1, 0) {};
		\node [style=label] (133) at (2.25, -0.375) {$U$};
		\node [style=black] (134) at (-0.25, 0) {};
		\node [style=none] (135) at (0.25, -0.375) {};
		\node [style=none] (136) at (0.25, 0.375) {};
		\node [style=none] (137) at (1.75, -0.375) {};
		\node [style=bbox, scale=0.9] (138) at (0.75, 0.375) {$\encoding{P}$};
	\end{pgfonlayer}
	\begin{pgfonlayer}{edgelayer}
		\draw [tape] (129.center)
			 to (128.center)
			 to (127.center)
			 to (130.center)
			 to cycle;
		\draw [bend right] (136.center) to (134);
		\draw [bend right] (134) to (135.center);
		\draw (137.center) to (135.center);
		\draw (134) to (132.center);
		\draw (136.center) to (138);
	\end{pgfonlayer}
\end{tikzpicture}
}
 = 
    \begin{tikzpicture}[scale=1.5]
	\begin{pgfonlayer}{nodelayer}
		\node [style=none] (127) at (-1.375, 0.875) {};
		\node [style=none] (128) at (1.375, 0.875) {};
		\node [style=none] (129) at (1.375, -0.875) {};
		\node [style=none] (130) at (-1.375, -0.875) {};
		\node [style=label] (131) at (1.875, 0) {$U$};
		\node [style=none] (141) at (-1.375, 0) {};
		\node [style=none] (142) at (1.375, 0) {};
		\node [style=corefl, scale=0.7] (143) at (0, 0) {$\encoding{P}$};
		\node [style=label] (144) at (-1.875, 0) {$U$};
	\end{pgfonlayer}
	\begin{pgfonlayer}{edgelayer}
		\draw [tape] (129.center)
			 to (128.center)
			 to (127.center)
			 to (130.center)
			 to cycle;
		\draw (142.center) to (141.center);
	\end{pgfonlayer}
\end{tikzpicture}
}
 \]
which, according to the notation in \Cref{rem:notation-coreflexives}, will be drawn as a circle.

The encoding maps any well typed command $\Gamma \vdash C $ into a diagram of type $\encoding{\Gamma} \to \encoding{\Gamma}$.
$$\begin{array}{rcl}
\encoding{\Gamma \vdash \mathsf{abort}}  & \defeq & \bang{\encoding{\Gamma}} ; \cobang{\encoding{\Gamma}}\\
\encoding{\Gamma \vdash  \mathsf{skip} }  & \defeq & \id{\encoding{\Gamma}} \\
\encoding{\Gamma \vdash  C;D }  & \defeq & \encoding{\Gamma \vdash C} ; \encoding{\Gamma \vdash D}  \\
\encoding{\Gamma \vdash \mathop{\mathsf{if}} P \mathop{\mathsf{then}} C \mathop{\mathsf{else}} D}  & \defeq & (c(P); \encoding{\Gamma \vdash C} ) + (c(\neg P); \encoding{\Gamma \vdash D}) \\
\encoding{\Gamma \vdash\mathop{\mathsf{while}} P \mathop{\mathsf{do}} C}  & \defeq & \kstar{(c(P); \encoding{\Gamma \vdash C })} ; c(\neg P)  \\
\interpret{\Gamma \vdash x \coloneqq e}  & \defeq & (\copier{\encoding{\Gamma'}} \per \id{A} \per \copier{\encoding{\Delta'}}) ; (\id{\encoding{\Gamma'}} \per \encoding{\Gamma \vdash e\colon A}  \per \id{\encoding{\Delta'}} )     %
\end{array}$$
Apart from the assignment the encodings of the other commands is pretty standard, see for instance Kleene Algebra with tests \cite{kozen1997kleene}. The assignment instead crucially exploits the structure of cartesian bicategories to properly model data flow.
For convenience of the reader we draw the corresponding tape diagram below. %
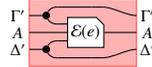
\[ 
    \begin{tikzpicture}
	\begin{pgfonlayer}{nodelayer}
		\node [style=label] (110) at (-1.75, -0.75) {$\Delta'$};
		\node [style=none] (127) at (-1.25, 1.5) {};
		\node [style=none] (128) at (3.5, 1.5) {};
		\node [style=none] (129) at (3.5, -1.5) {};
		\node [style=none] (130) at (-1.25, -1.5) {};
		\node [style=none] (131) at (-1.25, -0.75) {};
		\node [style=black] (133) at (-0.5, -0.75) {};
		\node [style=none] (138) at (0, -0.375) {};
		\node [style=none] (139) at (0, -1.125) {};
		\node [style=bbox, minimum height=1.8em] (142) at (1.25, 0) {$\encoding{e}$};
		\node [style=label] (143) at (-1.75, 0.75) {$\Gamma'$};
		\node [style=none] (144) at (-1.25, 0.75) {};
		\node [style=black] (145) at (-0.5, 0.75) {};
		\node [style=none] (146) at (0, 1.125) {};
		\node [style=none] (147) at (0, 0.375) {};
		\node [style=none] (148) at (-1.25, 0) {};
		\node [style=none] (150) at (3.5, 0) {};
		\node [style=none] (151) at (1.5, -1.125) {};
		\node [style=none] (152) at (1.5, 1.125) {};
		\node [style=none] (153) at (1.25, -0.375) {};
		\node [style=none] (154) at (1.25, 0.375) {};
		\node [style=none] (157) at (2, 1.125) {};
		\node [style=none] (158) at (2, -1.125) {};
		\node [style=label] (160) at (-1.75, 0) {$A$};
		\node [style=none] (161) at (3.5, 0.75) {};
		\node [style=none] (162) at (3.5, -0.75) {};
		\node [style=label] (163) at (4, -0.75) {$\Delta'$};
		\node [style=label] (164) at (4, 0.75) {$\Gamma'$};
		\node [style=label] (165) at (4, 0) {$A$};
	\end{pgfonlayer}
	\begin{pgfonlayer}{edgelayer}
		\draw [tape] (129.center)
			 to (128.center)
			 to (127.center)
			 to (130.center)
			 to cycle;
		\draw (131.center) to (133);
		\draw [bend left] (133) to (138.center);
		\draw [bend right] (133) to (139.center);
		\draw (144.center) to (145);
		\draw [bend left] (145) to (146.center);
		\draw [bend right] (145) to (147.center);
		\draw (148.center) to (142);
		\draw (142) to (150.center);
		\draw (147.center) to (154.center);
		\draw (138.center) to (153.center);
		\draw (139.center) to (151.center);
		\draw (146.center) to (152.center);
		\draw [in=180, out=0, looseness=1.25] (151.center) to (158.center);
		\draw [in=0, out=180, looseness=1.25] (157.center) to (152.center);
		\draw [in=360, out=-180, looseness=0.75] (161.center) to (157.center);
		\draw [in=0, out=-180, looseness=0.75] (162.center) to (158.center);
	\end{pgfonlayer}
\end{tikzpicture}
}
 \]

\subsection{Hoare Triples}
Hoare logic \cite{hoare1969axiomatic} is one of the most influential language to reason about imperative programs. Its rules --in the version appearing in \cite{winskel1993formal}-- are reported in Figure \ref{fig:hoare-rules}. 
In partial correctness, the triple \(\{P\}C\{Q\}\) asserts that if the command $C$ is executed from any state that satisfies the precondition $P$, and if the execution terminates, the resulting state will satisfy the postcondition $Q$.
The following result shows that if a triple \(\{P\}C\{Q\}\) is derivable within the Hoare logic, one can prove \(\op{\encoding{P}} \dcomp \encoding{C} \leq \op{\encoding{Q}}\) within tape diagrams.

\begin{figure}
  \begin{equation*}
  \begin{array}{c@{\quad}c@{\quad}c}
    \inferrule*[right=($\mathsf{skip}$)]{ }{\{P\}\mathsf{skip}\{P\}} & \inferrule*[right=(sub)]{ }{\{P[E/x]\}x \coloneqq E\{P\}} & \inferrule*[right=(\(\subseteq\))]{P_1 \subseteq P_2 \quad \{P_2\}C\{Q_2\} \quad Q_2 \subseteq Q_1}{\{P_1\}C\{Q_1\}} \\[12pt]
    \inferrule*[right=(seq)]{\{P\}C\{Q\} \quad \{Q\}D\{R\}}{\{P\}C ; D\{R\}} & \inferrule*[right=(\(\mathsf{if-else}\))]{\{P \land B\}C\{Q\} \quad \{P \land \lnot B \}D\{Q\}}{\{P\}\mathop{\mathsf{if}} B \mathop{\mathsf{then}} C \mathop{\mathsf{else}} D\{Q\}} & \inferrule*[right=(\(\mathsf{while}\))]{\{P \land B\}C\{P\}}{\{P\}\mathop{\mathsf{while}} B \mathop{\mathsf{do}} C\{P \land \lnot B\}} %
  \end{array}
  \end{equation*}
  \caption{Hoare derivation rules\label{fig:hoare-rules}}
\end{figure}

\begin{proposition}\label{prop:Hoare}
  If a Hoare triple \(\{P\}C\{Q\}\) is derivable with the rules in \Cref{fig:hoare-rules}, then \(\op{\encoding{P}} \dcomp \encoding{C} \leq \op{\encoding{Q}}\) in \(\KTCB\).
\end{proposition}
\begin{proof}
  By induction on the deduction rules in \Cref{fig:hoare-rules}.\\
  
  \medskip

  \noindent\((\mathsf{skip})\). \; \(\encoding{P} \dcomp \encoding{\mathsf{skip}} \stackrel{\text{(Def. of $\encoding{\cdot}$)}}{=} \encoding{P} \dcomp \id{} = \encoding{P}\).\\

  \medskip

  \noindent(SUB). \;
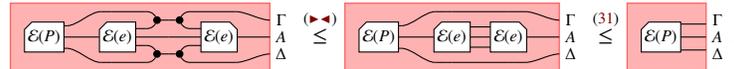
  \[ 
    \begin{tikzpicture}
	\begin{pgfonlayer}{nodelayer}
		\node [style=none] (127) at (-8, 1.5) {};
		\node [style=none] (128) at (3.5, 1.5) {};
		\node [style=none] (129) at (3.5, -1.5) {};
		\node [style=none] (130) at (-8, -1.5) {};
		\node [style=none] (131) at (-1, -0.75) {};
		\node [style=black] (133) at (-0.5, -0.75) {};
		\node [style=none] (138) at (0, -0.375) {};
		\node [style=none] (139) at (0, -1.125) {};
		\node [style=bbox, minimum height=1.8em] (142) at (1.25, 0) {$\encoding{e}$};
		\node [style=none] (144) at (-1, 0.75) {};
		\node [style=black] (145) at (-0.5, 0.75) {};
		\node [style=none] (146) at (0, 1.125) {};
		\node [style=none] (147) at (0, 0.375) {};
		\node [style=none] (148) at (-1, 0) {};
		\node [style=none] (150) at (3.5, 0) {};
		\node [style=none] (151) at (1.5, -1.125) {};
		\node [style=none] (152) at (1.5, 1.125) {};
		\node [style=none] (153) at (1.25, -0.375) {};
		\node [style=none] (154) at (1.25, 0.375) {};
		\node [style=none] (157) at (2, 1.125) {};
		\node [style=none] (158) at (2, -1.125) {};
		\node [style=none] (161) at (3.5, 0.75) {};
		\node [style=none] (162) at (3.5, -0.75) {};
		\node [style=label] (163) at (4, -0.75) {$\Delta$};
		\node [style=label] (164) at (4, 0.75) {$\Gamma$};
		\node [style=label] (165) at (4, 0) {$A$};
		\node [style=none] (166) at (-1, -0.75) {};
		\node [style=black] (167) at (-1.5, -0.75) {};
		\node [style=none] (168) at (-2, -0.375) {};
		\node [style=none] (169) at (-2, -1.125) {};
		\node [style=bboxOp, minimum height=1.8em] (170) at (-3.25, 0) {$\encoding{e}$};
		\node [style=none] (171) at (-1, 0.75) {};
		\node [style=black] (172) at (-1.5, 0.75) {};
		\node [style=none] (173) at (-2, 1.125) {};
		\node [style=none] (174) at (-2, 0.375) {};
		\node [style=none] (175) at (-1, 0) {};
		\node [style=none] (176) at (-6.25, 0) {};
		\node [style=none] (177) at (-3.5, -1.125) {};
		\node [style=none] (178) at (-3.5, 1.125) {};
		\node [style=none] (179) at (-3.25, -0.375) {};
		\node [style=none] (180) at (-3.25, 0.375) {};
		\node [style=none] (181) at (-4, 1.125) {};
		\node [style=none] (182) at (-4, -1.125) {};
		\node [style=none] (183) at (-5.5, 0.5) {};
		\node [style=none] (184) at (-5.5, -0.5) {};
		\node [style=bboxOp, minimum height=1.8em] (185) at (-6.5, 0) {$\encoding{P}$};
		\node [style=none] (186) at (-6.5, 0.5) {};
		\node [style=none] (187) at (-6.5, -0.5) {};
	\end{pgfonlayer}
	\begin{pgfonlayer}{edgelayer}
		\draw [tape] (129.center)
			 to (128.center)
			 to (127.center)
			 to (130.center)
			 to cycle;
		\draw (131.center) to (133);
		\draw [bend left] (133) to (138.center);
		\draw [bend right] (133) to (139.center);
		\draw (144.center) to (145);
		\draw [bend left] (145) to (146.center);
		\draw [bend right] (145) to (147.center);
		\draw (148.center) to (142);
		\draw (142) to (150.center);
		\draw (147.center) to (154.center);
		\draw (138.center) to (153.center);
		\draw (139.center) to (151.center);
		\draw (146.center) to (152.center);
		\draw [in=180, out=0, looseness=1.25] (151.center) to (158.center);
		\draw [in=0, out=180, looseness=1.25] (157.center) to (152.center);
		\draw [in=360, out=-180, looseness=0.75] (161.center) to (157.center);
		\draw [in=0, out=-180, looseness=0.75] (162.center) to (158.center);
		\draw (166.center) to (167);
		\draw [bend right] (167) to (168.center);
		\draw [bend left] (167) to (169.center);
		\draw (171.center) to (172);
		\draw [bend right] (172) to (173.center);
		\draw [bend left] (172) to (174.center);
		\draw (175.center) to (170);
		\draw (170) to (176.center);
		\draw (174.center) to (180.center);
		\draw (168.center) to (179.center);
		\draw (169.center) to (177.center);
		\draw (173.center) to (178.center);
		\draw [in=0, out=180, looseness=1.25] (177.center) to (182.center);
		\draw [in=180, out=0, looseness=1.25] (181.center) to (178.center);
		\draw [in=-180, out=0, looseness=0.75] (183.center) to (181.center);
		\draw [in=180, out=0, looseness=0.75] (184.center) to (182.center);
		\draw (186.center) to (183.center);
		\draw (187.center) to (184.center);
	\end{pgfonlayer}
\end{tikzpicture}
}
 \stackrel{\eqref{ax:copieradj1}}{\leq} 
    \begin{tikzpicture}
	\begin{pgfonlayer}{nodelayer}
		\node [style=none] (127) at (-6.75, 1.5) {};
		\node [style=none] (128) at (2.75, 1.5) {};
		\node [style=none] (129) at (2.75, -1.5) {};
		\node [style=none] (130) at (-6.75, -1.5) {};
		\node [style=none] (138) at (0.5, -0.45) {};
		\node [style=none] (139) at (0.5, -1.125) {};
		\node [style=bbox, minimum height=1.8em] (142) at (0.5, 0) {$\encoding{e}$};
		\node [style=none] (146) at (0.5, 1.125) {};
		\node [style=none] (147) at (0.5, 0.45) {};
		\node [style=none] (150) at (2.75, 0) {};
		\node [style=none] (151) at (0.75, -1.125) {};
		\node [style=none] (152) at (0.75, 1.125) {};
		\node [style=none] (153) at (0.5, -0.45) {};
		\node [style=none] (154) at (0.5, 0.45) {};
		\node [style=none] (157) at (1.25, 1.125) {};
		\node [style=none] (158) at (1.25, -1.125) {};
		\node [style=none] (161) at (2.75, 0.75) {};
		\node [style=none] (162) at (2.75, -0.75) {};
		\node [style=label] (163) at (3.25, -0.75) {$\Delta$};
		\node [style=label] (164) at (3.25, 0.75) {$\Gamma$};
		\node [style=label] (165) at (3.25, 0) {$A$};
		\node [style=none] (168) at (-2, -0.45) {};
		\node [style=none] (169) at (-2, -1.125) {};
		\node [style=bboxOp, minimum height=1.8em] (170) at (-2, 0) {$\encoding{e}$};
		\node [style=none] (173) at (-2, 1.125) {};
		\node [style=none] (174) at (-2, 0.45) {};
		\node [style=none] (176) at (-5, 0) {};
		\node [style=none] (177) at (-2.25, -1.125) {};
		\node [style=none] (178) at (-2.25, 1.125) {};
		\node [style=none] (179) at (-2, -0.45) {};
		\node [style=none] (180) at (-2, 0.45) {};
		\node [style=none] (181) at (-2.75, 1.125) {};
		\node [style=none] (182) at (-2.75, -1.125) {};
		\node [style=none] (183) at (-4.25, 0.5) {};
		\node [style=none] (184) at (-4.25, -0.5) {};
		\node [style=bboxOp, minimum height=1.8em] (185) at (-5.25, 0) {$\encoding{P}$};
		\node [style=none] (186) at (-5.25, 0.5) {};
		\node [style=none] (187) at (-5.25, -0.5) {};
	\end{pgfonlayer}
	\begin{pgfonlayer}{edgelayer}
		\draw [tape] (129.center)
			 to (128.center)
			 to (127.center)
			 to (130.center)
			 to cycle;
		\draw (142) to (150.center);
		\draw (147.center) to (154.center);
		\draw (138.center) to (153.center);
		\draw (139.center) to (151.center);
		\draw (146.center) to (152.center);
		\draw [in=180, out=0, looseness=1.25] (151.center) to (158.center);
		\draw [in=0, out=180, looseness=1.25] (157.center) to (152.center);
		\draw [in=360, out=-180, looseness=0.75] (161.center) to (157.center);
		\draw [in=0, out=-180, looseness=0.75] (162.center) to (158.center);
		\draw (170) to (176.center);
		\draw (174.center) to (180.center);
		\draw (168.center) to (179.center);
		\draw (169.center) to (177.center);
		\draw (173.center) to (178.center);
		\draw [in=0, out=180, looseness=1.25] (177.center) to (182.center);
		\draw [in=180, out=0, looseness=1.25] (181.center) to (178.center);
		\draw [in=-180, out=0, looseness=0.75] (183.center) to (181.center);
		\draw [in=180, out=0, looseness=0.75] (184.center) to (182.center);
		\draw (186.center) to (183.center);
		\draw (187.center) to (184.center);
		\draw (173.center) to (146.center);
		\draw (147.center) to (174.center);
		\draw (168.center) to (138.center);
		\draw (139.center) to (169.center);
		\draw (170) to (142);
	\end{pgfonlayer}
\end{tikzpicture}
}
 \stackrel{\eqref{eq:cb:adj:sv}}{\leq} 
    \begin{tikzpicture}
	\begin{pgfonlayer}{nodelayer}
		\node [style=none] (127) at (0, 1.5) {};
		\node [style=none] (128) at (3.5, 1.5) {};
		\node [style=none] (129) at (3.5, -1.5) {};
		\node [style=none] (130) at (0, -1.5) {};
		\node [style=none] (150) at (3.5, 0) {};
		\node [style=none] (161) at (3.5, 0.575) {};
		\node [style=none] (162) at (3.5, -0.575) {};
		\node [style=label] (163) at (4, -0.75) {$\Delta$};
		\node [style=label] (164) at (4, 0.75) {$\Gamma$};
		\node [style=label] (165) at (4, 0) {$A$};
		\node [style=none] (176) at (1.75, 0) {};
		\node [style=bboxOp, minimum height=1.8em] (185) at (1.5, 0) {$\encoding{P}$};
		\node [style=none] (186) at (1.5, 0.575) {};
		\node [style=none] (187) at (1.5, -0.575) {};
	\end{pgfonlayer}
	\begin{pgfonlayer}{edgelayer}
		\draw [tape] (129.center)
			 to (128.center)
			 to (127.center)
			 to (130.center)
			 to cycle;
		\draw (186.center) to (161.center);
		\draw (150.center) to (176.center);
		\draw (187.center) to (162.center);
	\end{pgfonlayer}
\end{tikzpicture}
}
. \] 

  \medskip

  \noindent\((\subseteq)\). \; \(\encoding{P_{1}} \dcomp \encoding{C} \stackrel{(P_1 \subseteq P_2)}{\leq} \encoding{P_{2}} \dcomp  \encoding{C} \stackrel{\text{(Ind. hyp.)}}{\leq} \encoding{Q_{2}} \stackrel{(Q_1 \subseteq Q_2)}{\leq} \encoding{P_{2}} \encoding{Q_{1}}\).\\

  \medskip

  \noindent(SEQ). \; \(\encoding{P} \dcomp \encoding{C ; D} \stackrel{\text{(Def. of $\encoding{\cdot}$)}}{=} \encoding{P} \dcomp \encoding{C} \dcomp \encoding{D} \stackrel{\text{(Ind. hyp.)}}{\leq} \encoding{Q} \dcomp \encoding{D} \stackrel{\text{(Ind. hyp.)}}{\leq} \encoding{R}\).\\

  \medskip

  \noindent\((\mathsf{if-else})\). \;\\
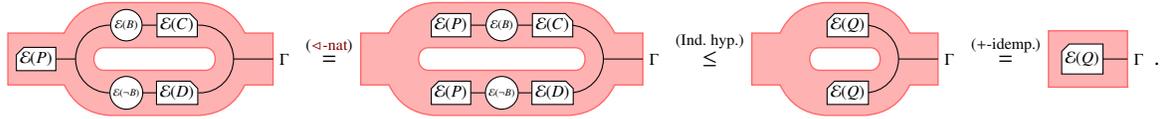
  \begin{align*}
    

}
.
  \end{align*}

  \medskip

  \noindent\((\mathsf{while})\). \; First, observe that the following holds:
    \[ 
    %
}
. \]
    Then, by~\eqref{ax:trace:tape:AU2}, the inequality below holds and the derivation concludes the proof.
    \[ 
    %
}
. \]
\end{proof}

\subsection{Other Program Logics}
While the calculus of relations, Kleene Algebra and Kleene algebra with tests allow to express binary relations, tape diagrams are able to express relations with arbitrary source and target. For instance a tape $\t\colon U \to 1$ represents a predicate over $U$, typically intended as the set of all memories. This ability allows tape diagrams to easy express triples occurring in different programs logics, such as incorrectness logic \cite{o2019incorrectness}, sufficient incorrectness \cite{DBLP:journals/corr/abs-2310-18156} and necessary \cite{DBLP:conf/vmcai/CousotCFL13}. The correspondence between the triples of this logic and inequality in tapes is illustrated in \Cref{tab:triples-inequalities}.
\begin{table}[H]
  \centering
  \renewcommand{\arraystretch}{1.7}
  \begin{tabular}{ c @{\qquad} c @{\qquad} c }
    \toprule
    Logic & Triple & Inequality \\
    \midrule
    Hoare & \(\{P\}C\{Q\}\) & $
    \InputIfFileExists{hoare/inclusions/hoare/lhs.tikz}{}{\input{./tikz/hoare/inclusions/hoare/lhs.tikz}}
 \leq 
    \InputIfFileExists{hoare/inclusions/hoare/rhs.tikz}{}{\input{./tikz/hoare/inclusions/hoare/rhs.tikz}}
$\\ %
    Incorrectness & \([P]C[Q]\) & $
    \InputIfFileExists{hoare/inclusions/hoare/lhs.tikz}{}{\input{./tikz/hoare/inclusions/hoare/lhs.tikz}}
 \geq 
    \InputIfFileExists{hoare/inclusions/hoare/rhs.tikz}{}{\input{./tikz/hoare/inclusions/hoare/rhs.tikz}}
$\\ %
    Sufficient incorrectness & \(\lAngle P \rAngle C \lAngle Q \rAngle\) & $
    \InputIfFileExists{hoare/inclusions/suffInc/lhs.tikz}{}{\input{./tikz/hoare/inclusions/suffInc/lhs.tikz}}
 \leq 
    \InputIfFileExists{hoare/inclusions/suffInc/rhs.tikz}{}{\input{./tikz/hoare/inclusions/suffInc/rhs.tikz}}
$\\ %
    Necessary & \((P)C(Q)\) & $
    \InputIfFileExists{hoare/inclusions/suffInc/lhs.tikz}{}{\input{./tikz/hoare/inclusions/suffInc/lhs.tikz}}
 \geq 
    \InputIfFileExists{hoare/inclusions/suffInc/rhs.tikz}{}{\input{./tikz/hoare/inclusions/suffInc/rhs.tikz}}
$\\[4pt] %
    \bottomrule
  \end{tabular}
  \caption{Correspondence between triples and inequalities.\label{tab:triples-inequalities}}
\end{table}

\bibliography{references}

\appendix

\section{Coherence Axioms}\label{app:coherence axioms}

In this Appendix we collect together various Figures, listing the coherence axioms required by the definition of the algebraic structures we consider in the article.

 \begin{figure}[H]
     \begin{equation} \label{ax:monoidaltriangle}
         \input{tikz-cd/monoidal_triangle.tikz} \tag{M1}
     \end{equation}
     \begin{equation}\label{ax:monoidalpentagone}
         \input{tikz-cd/monoidal_pentagon.tikz} \tag{M2}
     \end{equation}
     \caption{Coherence axioms of monoidal categories}
     \label{fig:moncatax}
 \end{figure}

\begin{figure}[H]
   \begin{minipage}[t]{0.48\textwidth}
       \begin{equation}\label{eq:symmax1}
           \input{tikz-cd/symmax1.tikz} \tag{S1}
       \end{equation}
   \end{minipage}
   \hfill
   \begin{minipage}[t]{0.48\textwidth}
        \begin{equation}\label{eq:symmax2}
            \input{tikz-cd/symmax2.tikz} \tag{S2}
        \end{equation}    
   \end{minipage}
    \begin{equation}\label{eq:symmax3}
        \input{tikz-cd/symmax3.tikz} \tag{S3}
    \end{equation}
    \caption{Coherence axioms of symmetric monoidal categories}
    \label{fig:symmmoncatax}
\end{figure}

\begin{figure}[H]
    \begin{equation}\label{eq:coherence diag}\tag{FP1}
        \begin{tikzcd}[column sep=4.5em,baseline=(current  bounding  box.center)]
            X \perG Y \ar[r,"\diag{X \perG Y}"] \ar[d,"\diag X \perG \diag Y"'] & (X \perG Y) \perG (X \perG Y) \\
            (X \perG Y) \perG (Y \perG Y) \ar[d,"\assoc X X {Y \perG Y}"'] \\
            X \perG (X \perG (Y \perG Y)) \ar[d,"\id X \perG \Iassoc X Y Y"'] & X \perG (Y \perG (X \perG Y)) \ar[uu,"\Iassoc X Y {X \perG Y}"'] \\
            X \perG ((X \perG Y) \perG Y) \ar[r,"\id X \perG (\symm{X}{Y}^{\perG} \perG \id Y)"] & X \perG ((Y \perG X) \perG Y) \ar[u,"\id X \perG \assoc Y X Y"'] 
        \end{tikzcd}
    \end{equation}
    \\
	\begin{minipage}[b]{0.33\textwidth}
	 	\begin{equation}\label{eq:coherence bang}\tag{FP2}
            \begin{tikzcd}[baseline=(current  bounding  box.center)]
                X \perG Y \ar[r,"\bang{X \perG Y}"] \ar[d,"\bang X \perG \bang Y"'] & \unoG \\
                \unoG \perG \unoG \ar[ur,"\lunit \unoG"']
            \end{tikzcd}
	    \end{equation} 
	\end{minipage}
	\hfill
	\begin{minipage}[b]{0.26\textwidth}
		\begin{equation}\tag{FP3}
		\begin{tikzcd}
		I \ar[r,shift left=2,"\diag I"] \ar[r,shift right=2,"\Ilunit I"'] & I \perG I
		\end{tikzcd}
		\end{equation}
		\end{minipage}
	\hfill
	\begin{minipage}[b]{0.26\textwidth}
		\begin{equation}\label{eq:bang I = id I}\tag{FP4}
            \begin{tikzcd}
            I \ar[r,shift left=2,"\bang I"] \ar[r,shift right=2,"\id I"'] & I
            \end{tikzcd}
		\end{equation}
	\end{minipage}

    \begin{equation}\label{eq:coherence codiag}\tag{FC1}
        \begin{tikzcd}[column sep=4.5em,baseline=(current  bounding  box.center)]
        (X \perG Y) \perG (X \perG Y) \ar[dd,"\assoc X Y {X \perG Y}"'] \ar[r,"\codiag{X \perG Y}"] & X \perG Y \\
        & (X \perG X) \perG (Y \perG Y) \ar[u,"\codiag X \perG \codiag Y"']\\
            X \perG (Y \perG (X \perG Y)) \ar[d,"\id X \perG \Iassoc Y X Y"'] & X \perG (X \perG (Y \perG Y)) \ar[u,"\Iassoc X X {Y \perG Y}"']  \\
        X \perG ((Y \perG X) \perG Y) \ar[r,"\id X \perG ( \symm{Y}{X}^{\perG} \perG \id Y)"] & X \perG (( X \perG Y) \perG Y) \ar[u,"\id X \perG \assoc X Y Y"']
        \end{tikzcd}
    \end{equation}
    \\
\begin{minipage}[b]{0.33\textwidth}
	\begin{equation}\label{eq:coherence cobang}\tag{FC2}
        \begin{tikzcd}[baseline=(current  bounding  box.center)]
        \unoG \ar[r,"\cobang{X \perG Y}"] \ar[d,"\Ilunit \unoG"']  & X \perG Y \\
        \unoG \perG \unoG \ar[ur,"\cobang X \perG \cobang Y"']  
        \end{tikzcd}
	\end{equation}
\end{minipage}
\hfill
\begin{minipage}[b]{0.26\textwidth}
	\begin{equation}\tag{FC3}
        \begin{tikzcd}
        I \perG I \ar[r,shift left=2,"\codiag I"] \ar[r,shift right=2,"\lunit I"'] &  I
        \end{tikzcd}
	\end{equation}
\end{minipage}
\hfill
\begin{minipage}[b]{0.26\textwidth}
	\begin{equation}\label{eq:cobang I = id I}\tag{FC4}
        \begin{tikzcd}
        I \ar[r,shift left=2,"\cobang I"] \ar[r,shift right=2,"\id I"'] & I
        \end{tikzcd}
	\end{equation}
\end{minipage}
\caption{Coherence axioms for (co)commutative (co)monoids}
\label{fig:fbcoherence}
\end{figure}

\begin{figure}[H]
    \begin{minipage}[t]{0.45\textwidth}
        \begin{equation}
            \label{eq:rigax1}\tag{R1}
            \scalebox{0.9}{$\input{tikz-cd/rigax1.tikz}$}
        \end{equation}    
    \end{minipage}
    \hfill
    \begin{minipage}[t]{0.45\textwidth}
        \begin{equation}
            \label{eq:rigax2}\tag{R2}
            \scalebox{0.9}{$\input{tikz-cd/rigax2.tikz}$}
        \end{equation}
    \end{minipage}
    
    \begin{equation}
        \label{eq:rigax3}\tag{R3}
        \scalebox{0.9}{$\input{tikz-cd/rigax3.tikz}$}
    \end{equation}
    \begin{equation}
        \label{eq:rigax4}\tag{R4}
        \scalebox{0.9}{$\input{tikz-cd/rigax4.tikz}$}
    \end{equation}
    \begin{equation}
        \label{eq:rigax5}\tag{R5}
        \scalebox{0.9}{$\input{tikz-cd/rigax5.tikz}$}
    \end{equation}

    \begin{minipage}[t]{0.25\textwidth}
        \begin{equation}
            \label{eq:rigax6}\tag{R6}
            \scalebox{0.9}{$\input{tikz-cd/rigax6.tikz}$}
        \end{equation}
    \end{minipage}
    \hfill
    \begin{minipage}[t]{0.45\textwidth}
        \begin{equation}
            \label{eq:rigax7}\tag{R7}
            \scalebox{0.9}{$\input{tikz-cd/rigax7.tikz}$}
        \end{equation}
    \end{minipage}
    \hfill
    \begin{minipage}[t]{0.25\textwidth}
        \begin{equation}
            \label{eq:rigax8}\tag{R8}
            \scalebox{0.9}{$\input{tikz-cd/rigax8.tikz}$}
        \end{equation}
    \end{minipage}
    \\
    \begin{minipage}[t]{0.48\textwidth}
        \begin{equation}
            \label{eq:rigax9}\tag{R9}
            \scalebox{0.9}{$\input{tikz-cd/rigax9.tikz}$}
        \end{equation}    
    \end{minipage}
    \hfill
    \begin{minipage}[t]{0.48\textwidth}
        \begin{equation}
            \label{eq:rigax10}\tag{R10}
            \scalebox{0.9}{$\input{tikz-cd/rigax10.tikz}$}
        \end{equation}
    \end{minipage}
    \\
    \begin{minipage}[t]{0.50\textwidth}
        \begin{equation}
            \label{eq:rigax11}\tag{R11}
            \scalebox{0.9}{$\input{tikz-cd/rigax11.tikz}$}
        \end{equation}    
    \end{minipage}
    \hfill
    \begin{minipage}[t]{0.46\textwidth}
        \begin{equation}
            \label{eq:rigax12}\tag{R12}
            \scalebox{0.9}{$\input{tikz-cd/rigax12.tikz}$}
        \end{equation}
    \end{minipage}
    \caption{Coherence Axioms of symmetric rig categories}
    \label{fig:rigax}
\end{figure}

\begin{figure}[H]
    \begin{equation}
        \label{eq:dl1}
        \input{tikz-cd/dl1.tikz}
    \end{equation}
    \begin{equation}
        \label{eq:dl2}
        \input{tikz-cd/dl2.tikz}
    \end{equation}
    \begin{equation}
        \label{eq:dl3}
        \input{tikz-cd/dl3.tikz}
    \end{equation}
    \begin{minipage}[t]{0.44\textwidth}
    \begin{equation}
        \label{eq:dl4}
        \input{tikz-cd/dl4.tikz}
    \end{equation}
    \end{minipage}
    \hfill
    \begin{minipage}[t]{0.54\textwidth}
    \begin{equation}
        \label{eq:dl5}
        \input{tikz-cd/dl5.tikz}
    \end{equation}
    \end{minipage}
    \\
    \begin{minipage}[t]{0.48\textwidth}
        \begin{equation}
            \label{eq:dl7}
            \input{tikz-cd/dl7.tikz}
        \end{equation}
        \end{minipage}
        \hfill
        \begin{minipage}[t]{0.50\textwidth}
        \begin{equation}
            \label{eq:dl8}
            \input{tikz-cd/dl8.tikz}
        \end{equation}
    \end{minipage}
    \begin{equation}
        \label{eq:dl6}
        \input{tikz-cd/dl6.tikz}
    \end{equation}
    \caption{Derived laws of symmetric rig categories}
    \label{fig:dlaw}
\end{figure}

\section{Appendix to Section~\ref{sec:traced}}

\subsection{Proofs of Section \ref{sec:UTConstruction-products}}

\begin{proof}[Proof of Corollary \ref{cor:free-uniform-trace-fb}]
  By \Cref{prop:free-uniform-trace-cartesian}, the adjunction in \Cref{th:free-uniform-trace} restricts to finite product categories.

  Suppose that \(\Cat{B}\) is a finite coproduct category.
  Then, \(\opcat{\Cat{B}}\) is a finite product category and so is \(\UTr(\opcat{\Cat{B}})\) by \Cref{prop:free-uniform-trace-cartesian}.
  This means that \(\opcat{\UTr(\opcat{\Cat{B}})} = \UTr(\Cat{B})\) is a finite coproduct category because the uniformity relation is symmetric.

  Then, if \(\Cat{B}\) has biproducts, \(\UTr(\Cat{B})\) also has them.
  The unit and counit restrict to finite biproduct categories for the same reason they restrict to finite product categories.
\end{proof}

\subsection{Proofs of Section \ref{ssec:tracerig}}

\begin{proof}[Proof of Proposition \ref{lemma:utr-whisk-algebra}] 
  
  \textsc{Equation}~\eqref{eq:whisk:id}.2.
    \begin{align*}
      \RW{X}{\id{Y} \mid \zero}
      &=  \st{\RW{X}{\id{Y}}}{\zero \per X} \tag{Definition~\ref{def:utr-whisk}} \\
      &=  \st{\id{Y \per X}}{\zero \per X} \tag{\ref{eq:whisk:id} in $\Cat{C}$} \\
      &=  \st{\id{Y \per X}}{\zero} \tag{Table~\ref{tab:equationsonobject}}
    \end{align*}

  \textsc{Equation}~\eqref{eq:whisk:id}.1.
    \begin{align*}
      \LW{X}{\id{Y} \mid \zero}
      &=  \st{\symmt{X}{Y}}{\zero} ; \RW{X}{\id{Y} \mid \zero} ; \st{\symmt{Y}{X}}{\zero} \tag{Definition~\ref{def:utr-whisk}} \\
      &=  \st{\symmt{X}{Y}}{\zero} ; \st{\id{Y \per X}}{\zero} ; \st{\symmt{Y}{X}}{\zero} \tag{\ref{eq:whisk:id}.2} \\
      &=  \st{\symmt{X}{Y} ; \id{Y \per X} ; \symmt{Y}{X}}{\zero} \tag{\ref{eq:seq in Utr}} \\
      &=  \st{\symmt{X}{Y} ; \symmt{Y}{X}}{\zero} \tag{Table~\ref{fig:freestricmmoncatax}} \\
      &=  \st{\id{X \per Y}}{\zero} \tag{Table~\ref{fig:freestricmmoncatax}}
    \end{align*}

  \textsc{Equation}~\eqref{eq:whisk:funct}.2. Let $(f \mid S) \colon Y \to Z, (g \mid T) \colon Z \to W$, then the following holds:
  \begin{align*}
      & \RW{X}{\st{f}{S} ; \st{g}{T}} \\
    = \; & \RW{X}{(\symmp{S}{T} \piu \id{Y}) ; (\id{T} \piu f) ; (\symmp{T}{S} \piu \id{Z}) ; (\id{S} \piu g) \mid S \piu T} \tag{\ref{eq:seq in Utr}} \\
    = \; & \st{\RW{X}{(\symmp{S}{T} \piu \id{Y}) ; (\id{T} \piu f) ; (\symmp{T}{S} \piu \id{Z}) ; (\id{S} \piu g)}}{(S \piu T) \per X} \tag{Definition~\ref{def:utr-whisk}} \\
    = \; & \st{\RW{X}{(\symmp{S}{T} \piu \id{Y}) ; (\id{T} \piu f) ; (\symmp{T}{S} \piu \id{Z}) ; (\id{S} \piu g)}}{(S \per X) \piu (T \per X)} \tag{Table~\ref{tab:equationsonobject}} \\
    = \; & \st{\RW{X}{\symmp{S}{T} \piu \id{Y}} ; \RW{X}{\id{T} \piu f} ; \RW{X}{\symmp{T}{S} \piu \id{Z}} ; \RW{X}{\id{S} \piu g}}{(S \per X) \piu (T \per X)} \tag{\ref{eq:whisk:funct} in $\Cat{C}$} \\
    = \; & \st{(\symmp{S\per X}{T \per X} \piu \id{Y \per X}) ; (\id{T \per X} \piu \RW{X}{f}) ; (\symmp{T \per X}{S \per X} \piu \id{Z \per X}) ; (\id{S \per X} \piu \RW{X}{g})}{(S \per X) \piu (T \per X)} \tag{\ref{eq:whisk:id}, \ref{eq:whisk:funct piu}, \ref{eq:whisk:symmp} in $\Cat{C}$} \\
    = \; & \st{\RW{X}{f}}{S \per X} ; \st{\RW{X}{g}}{T \per X} \tag{\ref{eq:seq in Utr}} \\
    = \; & \RW{X}{f \mid S} ; \RW{X}{g \mid T} \tag{Definition~\ref{def:utr-whisk}}
  \end{align*}

  \textsc{Equation}~\eqref{eq:whisk:funct}.1. Let $(f \mid S) \colon Y \to Z, (g \mid T) \colon Z \to W$, then the following holds:
  \begin{align*}
    \LW{X}{\st{f}{S} ; \st{g}{T}}
    &= \st{\symmt{X}{Y}}{\zero} ; \RW{X}{\st{f}{S} ; \st{g}{T}} ; \st{\symmt{W}{X}}{\zero} \tag{Definition~\ref{def:utr-whisk}} \\
    &= \st{\symmt{X}{Y}}{\zero} ; \RW{X}{f \mid S} ; \RW{X}{g \mid T} ; \st{\symmt{W}{X}}{\zero} \tag{\ref{eq:whisk:funct}.2} \\
    &= \st{\symmt{X}{Y}}{\zero} ; \RW{X}{f \mid S} ; \st{\symmt{Z}{X}}{\zero} ; \st{\symmt{X}{Z}}{\zero} ; \RW{X}{g \mid T} ; \st{\symmt{W}{X}}{\zero} \tag{Table~\ref{fig:freestricmmoncatax}} \\
    &= \LW{X}{f \mid S} ; \LW{X}{g \mid T} \tag{Definition~\ref{def:utr-whisk}}
  \end{align*}

  \textsc{Equation}~\eqref{eq:whisk:uno}.2. Let $(f \mid S) \colon X \to Y$, then the following holds:
  \begin{align*}
    \RW{\uno}{f \mid S} 
    &= \st{\RW{\uno}{f}}{S \per \uno} \tag{Definition~\ref{def:utr-whisk}} \\
    &= \st{f}{S \per \uno} \tag{\ref{eq:whisk:uno}.2 in $\Cat{C}$} \\
    &= \st{f}{S} \tag{Table~\ref{tab:equationsonobject}}
  \end{align*}

  \textsc{Equation}~\eqref{eq:whisk:uno}.1. Let $(f \mid S) \colon X \to Y$, then the following holds:
  \begin{align*}
    \LW{\uno}{f \mid S} 
    &= \st{\symmt{\uno}{X}}{\zero} ; \LW{\uno}{f \mid S}  ; \st{\symmt{Y}{\uno}}{\zero} \tag{Definition~\ref{def:utr-whisk}} \\
    &= \st{\symmt{\uno}{X}}{\zero} ; \st{f}{S}  ; \st{\symmt{Y}{\uno}}{\zero} \tag{\ref{eq:whisk:uno}.2} \\
    &= \st{\id{X}}{\zero} ; \st{f}{S}  ; \st{\id{Y}}{\zero} \tag{\ref{eq:symmax2}} \\
    &= \st{f}{S} \tag{Table~\ref{fig:freestricmmoncatax}}
  \end{align*}

  \textsc{Equation}~\eqref{eq:whisk:zero}.2. Let $(f \mid S) \colon X \to Y$, then the following holds:
  \begin{align*}
    \RW{\zero}{f \mid S} 
    &= \st{\RW{\zero}{f}}{S \per \zero} \tag{Definition~\ref{def:utr-whisk}} \\
    &= \st{\id{\zero}}{S \per \zero} \tag{\ref{eq:whisk:zero}.2 in $\Cat{C}$} \\
    &= \st{\id{\zero}}{\zero} \tag{Table~\ref{tab:equationsonobject}}
  \end{align*}

  \textsc{Equation}~\eqref{eq:whisk:zero}.1. Let $(f \mid S) \colon X \to Y$, then the following holds:
  \begin{align*}
    \LW{\zero}{f \mid S} 
    &= \st{\symmt{X}{\zero}}{\zero} ; \RW{\zero}{f \mid S} ; \st{\symmt{\zero}{Y}}{\zero} \tag{Definition~\ref{def:utr-whisk}} \\
    &= \st{\symmt{X}{\zero}}{\zero} ; \st{\id{\zero}}{\zero} ; \st{\symmt{\zero}{Y}}{\zero} \tag{\ref{eq:whisk:zero}.2} \\
    &= \st{\id{\zero}}{\zero} ; \st{\id{\zero}}{\zero} ; \st{\id{\zero}}{\zero} \tag{\ref{eq:rigax9}} \\
    &= \st{\id{\zero}}{\zero} \tag{Table~\ref{fig:freestricmmoncatax}}
  \end{align*}

  \textsc{Equation}~\eqref{eq:whisk:funct piu}.2. Let $(f_1 \mid S_1) \colon X_1 \to Y_1, (f_2 \mid S_2) \colon X_2 \to Y_2$, then the following holds:
  \begin{align*}
     & \RW{X}{\st{f_1}{S_1} \piu \st{f_2}{S_2}} \\
    = \; & \RW{X}{ (\id{S_1} \piu \symmp{S_2}{X_1} \piu \id{X_2}) ; (f_1 \piu f_2) ; (\id{S_1} \piu \symmp{Y_1}{S_2} \piu \id{Y_2}) \mid S_1 \piu S_2 } \tag{\ref{eq:per in Utr}} \\
    = \; & \st{\RW{X}{(\id{S_1} \piu \symmp{S_2}{X_1} \piu \id{X_2}) ; (f_1 \piu f_2) ; (\id{S_1} \piu \symmp{Y_1}{S_2} \piu \id{Y_2})}}{(S_1 \piu S_2) \per X} \tag{Definition~\ref{def:utr-whisk}} \\
    = \; & \st{(\id{S_1 \per X} \piu \symmp{S_2 \per X}{X_1 \per X} \piu \id{X_2 \per X}) ; (\RW{X}{f_1} \piu \RW{X}{f_2}) ; (\id{S_1 \per X} \piu \symmp{Y_1 \per X}{S_2 \per X} \piu \id{Y_2 \per X})}{(S_1 \piu S_2) \per X} \tag{\ref{eq:whisk:id}, \ref{eq:whisk:funct}, \ref{eq:whisk:funct piu}, \ref{eq:whisk:symmp} in $\Cat{C}$} \\
    = \; & \st{(\id{S_1 \per X} \piu \symmp{S_2 \per X}{X_1 \per X} \piu \id{X_2 \per X}) ; (\RW{X}{f_1} \piu \RW{X}{f_2}) ; (\id{S_1 \per X} \piu \symmp{Y_1 \per X}{S_2 \per X} \piu \id{Y_2 \per X})}{(S_1 \per X) \piu (S_2 \per X)} \tag{Table~\ref{tab:equationsonobject}} \\
    = \; & \RW{X}{f_1 \mid S_1} \piu \RW{X}{f_2 \mid S_2} \tag{\ref{eq:per in Utr}}
  \end{align*}

  \textsc{Equation}~\eqref{eq:whisk:funct piu}.1. Let $(f_1 \mid S_1) \colon X_1 \to Y_1, (f_2 \mid S_2) \colon X_2 \to Y_2$, then the following holds:
  \begin{align*}
     & \LW{X}{\st{f_1}{S_1} \piu \st{f_2}{S_2}} \\
    = \; & \st{\symmt{X}{X_1 \piu X_2}}{\zero} ; \RW{X}{\st{f_1}{S_1} \piu \st{f_2}{S_2}} ; \st{\symmt{Y_1 \piu Y_2}{X}}{\zero} \tag{Definition~\ref{def:utr-whisk}} \\
    = \; & \st{\symmt{X}{X_1 \piu X_2}}{\zero} ; (\RW{X}{f_1 \mid S_1} \piu \RW{X}{f_2 \mid S_2}) ; \st{\symmt{Y_1 \piu Y_2}{X}}{\zero} \tag{\ref{eq:whisk:funct piu}.2} \\
    = \; & \st{\dl{X}{X_1}{X_2} ; (\symmt{X}{X_1} \piu \symmt{X}{X_2})}{\zero} ; (\RW{X}{f_1 \mid S_1} \piu \RW{X}{f_2 \mid S_2}) ; \st{ (\symmt{Y_1}{X} \piu \symmt{Y_2}{X}) ; \Idl{X}{Y_1}{Y_2} }{\zero} \tag{\ref{eq:rigax1}} \\
    = \; & \st{\dl{X}{X_1}{X_2}}{\zero} ; (\st{\symmt{X}{X_1}}{\zero} \piu \st{\symmt{X}{X_2}}{\zero}) ; (\RW{X}{f_1 \mid S_1} \piu \RW{X}{f_2 \mid S_2}) ; (\st{\symmt{Y_1}{X}}{\zero} \piu \st{\symmt{Y_2}{X}}{\zero}) ; \st{\Idl{X}{Y_1}{Y_2}}{\zero} \tag{\ref{eq:seq in Utr}, \ref{eq:per in Utr}} \\
    = \; & \st{\dl{X}{X_1}{X_2}}{\zero} ; ( ( \st{\symmt{X}{X_1}}{\zero} ; \RW{X}{f_1 \mid S_1} ; \st{\symmt{Y_1}{X}}{\zero} ) \piu (\st{\symmt{X}{X_2}}{\zero} ; \RW{X}{f_2 \mid S_2} ; \st{\symmt{Y_2}{X}}{\zero} ) ) ; \st{\Idl{X}{Y_1}{Y_2}}{\zero} \tag{Table~\ref{fig:freestricmmoncatax}} \\
    = \; & \st{\dl{X}{X_1}{X_2}}{\zero} ; ( \LW{X}{f_1 \mid S_1} \piu \LW{X}{f_2 \mid S_2} ) ; \st{\Idl{X}{Y_1}{Y_2}}{\zero} \tag{Definition~\ref{def:utr-whisk}} \\
 \end{align*}

 \textsc{Equation}~\eqref{eq:whisk:sum}.2. Let $(f \mid S) \colon Z \to W$, then the following holds:
 \begin{align*}
   & \RW{X \piu Y}{f \mid S} \\
  = \; & \st{\RW{X \piu Y}{f}}{S \per (X \piu Y)} \tag{Definition~\ref{def:utr-whisk}} \\
  = \; & \st{ \dl{S \piu Z}{X}{Y} ; (\RW{X}{f} \piu \RW{Y}{f}); \Idl{S \piu W}{X}{Y} }{S \per (X \piu Y)} \tag{\ref{eq:whisk:sum}.2 in $\Cat{C}$} \\
  = \; & { \scriptstyle \st{ (\dl{S}{X}{Y} \piu \dl{Z}{X}{Y}) ; (\id{S \per X} \piu \symmp{S \per Y}{Z \per X} \piu \id{Z \per Y}) ; (\RW{X}{f} \piu \RW{Y}{f}); (\id{S \per X} \piu \symmp{S \per Y}{W \per X} \piu \id{W \per Y}) ; (\Idl{S}{X}{Y} \piu \Idl{W}{X}{Y})}{S \per (X \piu Y)} } \tag{\ref{eq:rigax5}}  \\
  = \; & { \scriptstyle \st{ (\id{S \per X} \piu \id{S \per Y} \piu \dl{Z}{X}{Y}) ; (\id{S \per X} \piu \symmp{S \per Y}{Z \per X} \piu \id{Z \per Y}) ; (\RW{X}{f} \piu \RW{Y}{f}); (\id{S \per X} \piu \symmp{S \per Y}{W \per X} \piu \id{W \per Y}) ; (\id{S \per X} \piu \id{S \per Y} \piu \Idl{W}{X}{Y})}{(S \per X) \piu (S \per Y)} } \tag{\ref{ax:trace:sliding}} \\
  = \; & { \scriptstyle \st{\dl{Z}{X}{Y}}{\zero} ; \st{ (\id{S \per X} \piu \symmp{S \per Y}{Z \per X} \piu \id{Z \per Y}) ; (\RW{X}{f} \piu \RW{Y}{f}); (\id{S \per X} \piu \symmp{S \per Y}{W \per X} \piu \id{W \per Y})}{(S \per X) \piu (S \per Y)} ; \st{\Idl{W}{X}{Y}}{\zero} } \tag{\ref{eq:seq in Utr}} \\
  = \; & { \st{\dl{Z}{X}{Y}}{\zero} ;  ( \st{\RW{X}{f}}{S \per X} \piu \st{\RW{Y}{f}}{S \per Y} )  ; \st{\Idl{W}{X}{Y}}{\zero} } \tag{\ref{eq:per in Utr}} \\
  = \; & { \st{\dl{Z}{X}{Y}}{\zero} ;  ( \RW{X}{f \mid S} \piu \RW{Y}{f \mid S} )  ; \st{\Idl{W}{X}{Y}}{\zero} } \tag{Definition~\ref{def:utr-whisk}} \\
\end{align*}

\textsc{Equation}~\eqref{eq:whisk:sum}.1. Let $(f \mid S) \colon Z \to W$, then the following holds:
\begin{align*}
   & \LW{X \piu Y}{f \mid S} \\
  = \; & \st{\symmt{X \piu Y}{Z}}{\zero} ; \RW{X \piu Y}{f \mid S} ; \st{\symmt{W}{X \piu Y}}{\zero} \tag{Definition~\ref{def:utr-whisk}} \\
  = \; & \st{\symmt{X \piu Y}{Z}}{\zero} ; ( \st{\dl{Z}{X}{Y}}{\zero} ;  ( \RW{X}{f \mid S} \piu \RW{Y}{f \mid S} )  ; \st{\Idl{W}{X}{Y}}{\zero} ) ; \st{\symmt{W}{X \piu Y}}{\zero} \tag{\ref{eq:whisk:sum}.2} \\
  = \; & \st{ (\symmt{X}{Z} \piu \symmt{Y}{Z}) ; \Idl{Z}{X}{Y} }{\zero} ; ( \st{\dl{Z}{X}{Y}}{\zero} ;  ( \RW{X}{f \mid S} \piu \RW{Y}{f \mid S} )  ; \st{\Idl{W}{X}{Y}}{\zero} ) ; \st{\dl{W}{X}{Y} ; (\symmt{W}{X} \piu \symmt{W}{Y})}{\zero} \tag{\ref{eq:rigax1}} \\
  = \; & \st{ (\symmt{X}{Z} \piu \symmt{Y}{Z}) ; \Idl{Z}{X}{Y} ; \dl{Z}{X}{Y} }{\zero} ; ( ( \RW{X}{f \mid S} \piu \RW{Y}{f \mid S} ) ) ; \st{\Idl{W}{X}{Y} ; \dl{W}{X}{Y} ; (\symmt{W}{X} \piu \symmt{W}{Y})}{\zero} \tag{\ref{eq:seq in Utr}} \\
  = \; & \st{ \symmt{X}{Z} \piu \symmt{Y}{Z}}{\zero} ; ( \RW{X}{f \mid S} \piu \RW{Y}{f \mid S} ) ; \st{\symmt{W}{X} \piu \symmt{W}{Y}}{\zero} \tag{$\delta^l$ iso} \\
  = \; & ( \st{\symmt{X}{Z}}{\zero} \piu \st{\symmt{Y}{Z}}{\zero} ) ; ( \RW{X}{f \mid S} \piu \RW{Y}{f \mid S} ) ; ( \st{\symmt{W}{X}}{\zero} \piu \st{\symmt{W}{Y}}{\zero} ) \tag{\ref{eq:per in Utr}} \\
  = \; & ( \st{\symmt{X}{Z}}{\zero} ; \RW{X}{f \mid S} ; \st{\symmt{W}{X}}{\zero} ) \piu ( \st{\symmt{Y}{Z}}{\zero} ; \RW{Y}{f \mid S} ; \st{\symmt{W}{Y}}{\zero} ) \tag{Table~\ref{fig:freestricmmoncatax}} \\
  = \; & \LW{X}{f \mid S} \piu \LW{Y}{f \mid S} \tag{Definition~\ref{def:utr-whisk}}
\end{align*}

\textsc{Equation}~\eqref{eq:tape:LexchangeR}. 
Let $f_1 \colon S_1 \piu X_1 \to S_1 \piu Y_1$ and $f_2 \colon S_2 \piu X_2 \to S_2 \piu Y_2$ be morphisms in $\Cat{C}$ and observe that the following holds by \eqref{eq:tape:LexchangeR} in $\Cat{C}$:
\[ \st{\LW{S_1 \piu X_1}{f_2} ; \RW{S_2 \piu Y_2}{f_1}}{\zero} = \st{\RW{S_2 \piu X_2}{f_1} ; \LW{S_1 \piu Y_1}{f_2}}{\zero}. \]
Using string diagrams for $(\Cat{C}, \piu, \zero)$, the equality above translates into the equality between diagrams below:
\begin{align*}
  \st{\scalebox{0.9}{

}
}}{\zero}.
\end{align*}
By precomposing and postcomposing with appropriate distributors, the following holds:
\[
  \st{\scalebox{0.9}{
    %
}
}}{\zero}$ as strict morphism, we obtain the following equality by~\eqref{ax:trace:uniformity}:
\[%
\]
By naturality of symmetry, the following holds:
\[%
}
}}{S_1 \per S_2}$ as strict morphism, we obtain the following equality by~\eqref{ax:trace:uniformity}:
\[%
\]
which, by \eqref{eq:seq in Utr} and Definition~\ref{def:utr-whisk}, corresponds to the equality below:
\[
  \st{\delta^l_{X_1, S_2, X_2} ; \LW{X_1}{f_2} ; \delta^l_{X_1,S_2,Y_2}}{X_1 \per S_2} ; \RW{Y_2}{f_1 \mid S_1} = \RW{X_2}{f_1 \mid S_1} ; \st{\delta^{-l}_{Y_1, S_2, X_2} ; \LW{Y_1}{f_2} ; \delta^l_{Y_1,S_2,Y_2}}{Y_1 \per S_2}.
\]
To conclude the proof, observe that for every $(f \mid S) \colon Y \to Z$ it holds that
\begin{equation}\label{eq:LRaux}
  \LW{X}{f \mid S} = \st{\Idl{X}{S}{Y} ; \LW{X}{f} ; \dl{X}{S}{Z}}{X \per S}
\end{equation}
as shown below:
\begin{align*}
  \LW{X}{f \mid S}
  &= \st{\symmt{X}{Y}}{\zero} ; \RW{X}{f \mid S} ; \st{\symmt{Z}{X}}{\zero} \tag{Definition~\ref{def:utr-whisk}} \\
  &= \st{\symmt{X}{Y}}{\zero} ; \st{\RW{X}{f}}{S \per X} ; \st{\symmt{Z}{X}}{\zero} \tag{Definition~\ref{def:utr-whisk}} \\
  &= \st{ (\id{S \per X} \piu \symmt{X}{Y}) ; \RW{X}{f} ; (\id{S \per X} \piu \symmt{Z}{X}) }{S \per X}  \tag{\ref{eq:seq in Utr}} \\
  &= \st{ (\symmt{X}{S} \piu \symmt{X}{Y}) ; \RW{X}{f} ; (\symmt{S}{X} \piu \symmt{Z}{X}) }{X \per S}  \tag{\ref{ax:trace:sliding}} \\
  &= \st{ (\symmt{X}{S} \piu \symmt{X}{Y}) ; \symmt{S \piu Y}{X} ; \LW{X}{f} ; \symmt{X}{S \piu Z} ; (\symmt{S}{X} \piu \symmt{Z}{X}) }{X \per S}  \tag{\ref{eq:LRnatsym} in $\Cat{C}$} \\
  &= \st{\Idl{X}{S}{Y} ; \LW{X}{f} ; \dl{X}{S}{Z}}{X \per S} \tag{\ref{eq:rigax1}}
\end{align*}

\textsc{Equation}~\eqref{eq:whisk:symmp}.
\begin{align*}
  \RW{X}{\symmp{Y}{Z} \mid \zero} 
  &= \st{\RW{X}{\symmp{Y}{Z}}}{\zero \per X} \tag{Definition~\ref{def:utr-whisk}} \\
  &= \st{ \symmp{Y \per X}{Z \per X} }{\zero \per X} \tag{\ref{eq:whisk:symmp} in $\Cat{C}$} \\
  &= \st{ \symmp{Y \per X}{Z \per X} }{\zero} \tag{Table~\ref{tab:equationsonobject}}
\end{align*}

\textsc{Equation}~\eqref{eq:symmper}. 
\begin{align*}
  & \LW{X}{\symmp{Y}{Z} \mid \zero} ; \RW{Y}{\symmp{X}{Z} \mid \zero} \\
  = \; & \st{\symmt{X}{Y \per Z}}{\zero} ; \RW{X}{\symmp{Y}{Z} \mid \zero} ; \st{\symmt{Y \per Z}{X}}{\zero} ; \RW{Y}{\symmp{X}{Z} \mid \zero} \tag{Definition~\ref{def:utr-whisk}} \\
  = \; & \st{\symmt{X}{Y \per Z}}{\zero} ; \st{\RW{X}{\symmp{Y}{Z}}}{\zero \per X} ; \st{\symmt{Y \per Z}{X}}{\zero} ; \st{\RW{Y}{\symmp{X}{Z}}}{\zero \per Y} \tag{Definition~\ref{def:utr-whisk}} \\
  = \; & \st{\symmt{X}{Y \per Z}}{\zero} ; \st{\symmp{Y}{Z} \per \id{X}}{\zero \per X} ; \st{\symmt{Y \per Z}{X}}{\zero} ; \st{\symmp{X}{Z} \per \id{Y}}{\zero \per Y} \tag{Proposition~\ref{prop:whisk}} \\
  = \; & \st{\symmt{X}{Y \per Z}}{\zero} ; \st{\symmp{Y}{Z} \per \id{X}}{\zero} ; \st{\symmt{Y \per Z}{X}}{\zero} ; \st{\symmp{X}{Z} \per \id{Y}}{\zero} \tag{Table~\ref{tab:equationsonobject}} \\
  = \; & \st{\symmt{X}{Y \per Z} ; (\symmp{Y}{Z} \per \id{X}) ; \symmt{Y \per Z}{X} ; (\symmp{X}{Z} \per \id{Y})}{\zero} \tag{\ref{eq:seq in Utr}} \\
  = \; & \st{\symmt{X \per Y}{Z}}{\zero} \tag{Table~\ref{tab:equationsonobject}}
\end{align*}

\textsc{Equation}~\eqref{eq:LRnatsym}. Let $(f \mid S) \colon Y \to Z$, then the following holds:
\begin{align*}
  \RW{X}{f \mid S} ; \st{\symmt{Z}{X}}{\zero} 
  &= \st{\symmt{Y}{X}}{\zero} ; \st{\symmt{X}{Y}}{\zero} ; \RW{X}{f \mid S} ; \st{\symmt{Z}{X}}{\zero} \tag{Table~\ref{fig:freestricmmoncatax}} \\
  &= \st{\symmt{Y}{X}}{\zero} ; \LW{X}{f \mid S} \tag{Definition~\ref{def:utr-whisk}}
\end{align*}

\textsc{Equation}~\eqref{eq:tape:RR}. Let $(f \mid S) \colon Z \to W$, then the following holds:
\begin{align*}
  \RW{X}{\RW{Y}{f \mid S}}
  &= \RW{X}{\st{\RW{Y}{f}}{S \per Y}} \tag{Definition~\ref{def:utr-whisk}} \\
  &= \st{\RW{X}{\RW{Y}{f}}}{S \per Y \per X} \tag{Definition~\ref{def:utr-whisk}} \\
  &= \st{\RW{Y \per X}{f}}{S \per Y \per X} \tag{\ref{eq:tape:RR} in $\Cat{C}$}
\end{align*}

\textsc{Equation}~\eqref{eq:tape:LR}. Let $(f \mid S) \colon Z \to W$, then the following holds:
\begin{align*}
   &  \LW{X}{\RW{Y}{f \mid S}} \\
  =\; & \st{\symmt{X}{Z \per Y}}{\zero} ; \RW{X}{\RW{Y}{f \mid S}} ; \st{\symmt{W \per Y}{X}}{\zero} \tag{Definition~\ref{def:utr-whisk}} \\
  =\; & \st{\symmt{X}{Z \per Y}}{\zero} ; \RW{X}{\RW{Y}{f} \mid S \per Y} ; \st{\symmt{W \per Y}{X}}{\zero} \tag{Definition~\ref{def:utr-whisk}} \\
  =\; & \st{\symmt{X}{Z \per Y}}{\zero} ; \st{\RW{X}{\RW{Y}{f}}}{S \per Y \per X} ; \st{\symmt{W \per Y}{X}}{\zero} \tag{Definition~\ref{def:utr-whisk}} \\
  =\; & \st{\symmt{X}{Z \per Y}}{\zero} ; \st{ \symmt{(S \piu Z)\per Y}{X} ; \LW{X}{\RW{Y}{f}} ; \symmt{X}{(S \piu W) \per Y} }{S \per Y \per X} ; \st{\symmt{W \per Y}{X}}{\zero} \tag{\ref{eq:LRnatsym} in $\Cat{C}$} \\
  =\; & \st{\symmt{X}{Z \per Y}}{\zero} ; \st{ \symmt{(S \piu Z)\per Y}{X} ; \RW{Y}{\LW{X}{f}} ; \symmt{X}{(S \piu W) \per Y} }{S \per Y \per X} ; \st{\symmt{W \per Y}{X}}{\zero} \tag{\ref{eq:tape:LR} in $\Cat{C}$} \\
  =\; & { \scriptstyle \st{\symmt{X}{Z \per Y}}{\zero} ; \st{ (\symmt{S \per Y}{X} \piu \symmt{Z \per Y}{X}) ; \Idl{X}{S \per Y}{Z \per Y} ; \RW{Y}{\LW{X}{f}} ; \dl{X}{S \per Y}{W \per Y} ; (\symmt{X}{S \per Y} \piu \symmt{X}{W \per Y}) }{S \per Y \per X} ; \st{\symmt{W \per Y}{X}}{\zero} } \tag{\ref{eq:rigax1}} \\
  =\; & { \scriptstyle \st{ (\id{S \per Y \per X} \piu \symmt{X}{Z \per Y}) ; (\symmt{S \per Y}{X} \piu \symmt{Z \per Y}{X}) ; \Idl{X}{S \per Y}{Z \per Y} ; \RW{Y}{\LW{X}{f}} ; \dl{X}{S \per Y}{W \per Y} ; (\symmt{X}{S \per Y} \piu \symmt{X}{W \per Y}) ; (\id{S \per Y \per X} \piu \symmt{W \per Y}{X}) }{S \per Y \per X} } \tag{\ref{eq:seq in Utr}} \\
  =\; & { \st{ (\symmt{S \per Y}{X} \piu \id{X \per Z \per Y}) ; \Idl{X}{S \per Y}{Z \per Y} ; \RW{Y}{\LW{X}{f}} ; \dl{X}{S \per Y}{W \per Y} ; ( \symmt{X}{S \per Y} \piu \id{X \per W \per Y} ) }{S \per Y \per X} } \tag{Table~\ref{fig:freestricmmoncatax}} \\
  =\; & { \st{ \Idl{X}{S \per Y}{Z \per Y} ; \RW{Y}{\LW{X}{f}} ; \dl{X}{S \per Y}{W \per Y} }{X \per S \per Y} } \tag{\ref{ax:trace:sliding}} \\
  =\; & { \st{ \RW{Y}{\Idl{X}{S}{Z}} ; \RW{Y}{\LW{X}{f}} ; \RW{Y}{\dl{X}{S}{W}} }{X \per S \per Y} } \tag{\ref{eq:whisk:dl} in $\Cat{C}$} \\
  =\; & { \st{ \RW{Y}{\Idl{X}{S}{Z} ; \LW{X}{f} ; \dl{X}{S}{W}} }{X \per S \per Y} } \tag{\ref{eq:whisk:funct} in $\Cat{C}$} \\
  =\; & \RW{Y}{\Idl{X}{S}{Z} ; \LW{X}{f} ; \dl{X}{S}{W} \mid X \per S} \tag{Definition~\ref{def:utr-whisk}} \\
  =\; & \RW{Y}{\LW{X}{f \mid S}} \tag{\ref{eq:LRaux}}
\end{align*}

\textsc{Equation}~\eqref{eq:tape:LL}. Let $(f \mid S) \colon Z \to W$, then the following holds:
\begin{align*}
  \LW{X}{\LW{Y}{f \mid S}}
  &=  \st{\symmt{X}{Y \per Z}}{\zero} ; \RW{X}{\LW{Y}{f \mid S}} ; \st{\symmt{Y \per W}{X}}{\zero} \tag{Definition~\ref{def:utr-whisk}} \\
  &=  \st{\symmt{X}{Y \per Z}}{\zero} ; \LW{Y}{\RW{X}{f \mid S}} ; \st{\symmt{Y \per W}{X}}{\zero} \tag{\ref{eq:tape:LR}} \\
  &=  \st{\symmt{X}{Y \per Z}}{\zero} ; \st{\symmt{Y}{Z \per X}}{\zero} ; \RW{Y}{\RW{X}{f \mid S}} ; \st{\symmt{W \per X}{Y}}{\zero} ; \st{\symmt{Y \per W}{X}}{\zero} \tag{Definition~\ref{def:utr-whisk}} \\
  &=  \st{\symmt{X}{Y \per Z}}{\zero} ; \st{\symmt{Y}{Z \per X}}{\zero} ; \RW{X \per Y}{f \mid S} ; \st{\symmt{W \per X}{Y}}{\zero} ; \st{\symmt{Y \per W}{X}}{\zero} \tag{\ref{eq:tape:RR}} \\
  &=  \st{\symmt{X \per Y}{Z}}{\zero} ; \RW{X \per Y}{f \mid S} ; \st{\symmt{W}{X \per Y}}{\zero} \tag{Table~\ref{fig:freestricmmoncatax}} \\
  &=  \LW{X \per Y}{f \mid S} \tag{Definition~\ref{def:utr-whisk}} \\
\end{align*}

\textsc{Equation}~\eqref{eq:whisk:dl}. 
\begin{align*}
  \RW{X}{\dl{Y}{Z}{W} \mid \zero}
  &= \st{\RW{X}{\dl{Y}{Z}{W}}}{\zero \per X} \tag{Definition~\ref{def:utr-whisk}} \\
  &= \st{\dl{Y}{Z \per X}{W \per X}}{\zero \per X} \tag{\ref{eq:whisk:dl} in $\Cat{C}$} \\
  &= \st{\dl{Y}{Z \per X}{W \per X}}{\zero} \tag{Table~\ref{tab:equationsonobject}}
\end{align*}

\textsc{Equation}~\eqref{eq:whisk:Ldl}.
\begin{align*}
  \LW{X}{\dl{Y}{Z}{W} \mid \zero}
  &= \st{\symmt{X}{Y \per (Z \piu W)}}{\zero} ; \RW{X}{\dl{Y}{Z}{W} \mid \zero} ; \st{\symmt{(Y  \per Z) \piu (Y \per W)}{X}}{\zero} \tag{Definition~\ref{def:utr-whisk}} \\
  &= \st{\symmt{X}{Y \per (Z \piu W)}}{\zero} ; \st{\RW{X}{\dl{Y}{Z}{W}}}{\zero \per X} ; \st{\symmt{(Y  \per Z) \piu (Y \per W)}{X}}{\zero} \tag{Definition~\ref{def:utr-whisk}} \\
  &= \st{\symmt{X}{Y \per (Z \piu W)}}{\zero} ; \st{\RW{X}{\dl{Y}{Z \per X}{W \per X}}}{\zero} ; \st{\symmt{(Y  \per Z) \piu (Y \per W)}{X}}{\zero} \tag{Table~\ref{tab:equationsonobject}} \\
  &= \st{\symmt{X}{Y \per (Z \piu W)} ; \RW{X}{\dl{Y}{Z \per X}{W \per X}} ; \symmt{(Y  \per Z) \piu (Y \per W)}{X}}{\zero} \tag{\ref{eq:seq in Utr}} \\
  &= \st{\LW{X}{\dl{Y}{Z \per X}{W \per X}} ; \symmt{X}{(Y  \per Z) \piu (Y \per W)} ; \symmt{(Y  \per Z) \piu (Y \per W)}{X} }{\zero} \tag{\ref{eq:LRnatsym} in $\Cat{C}$} \\
  &= \st{\LW{X}{\dl{Y}{Z \per X}{W \per X}} }{\zero} \tag{Table~\ref{fig:freestricmmoncatax}} \\
  &= \st{\dl{X \per Y}{Z}{W} ; \Idl{Y}{X \per Z}{X \per W}}{\zero} \tag{\ref{eq:whisk:Ldl} in $\Cat{C}$} \\
  &= \st{\dl{X \per Y}{Z}{W}}{\zero} ; \st{\Idl{Y}{X \per Z}{X \per W}}{\zero} \tag{\ref{eq:seq in Utr}}
\end{align*}

\end{proof}

\begin{lemma}\label{lemma:utr-per-monoidal}
    $(\UTr(\Cat C), \per, \uno, \symmt)$ is a strict symmetric monoidal category.
\end{lemma}
\begin{proof}
    First we show that $\per$, as defined in~\eqref{eq:utr-per}, is a functor, i.e. that it preserves identities and composition.

    \medskip

    \textsc{$\per$-functoriality}.
    \begin{align*}
      \st{\id{X}}{\zero} \per \st{\id{Y}}{\zero} 
      &= \LW{X}{\id{Y} \mid \zero} ; \RW{Y}{\id{X} \mid \zero} \tag{\ref{eq:utr-per}} \\
      &= \st{\id{X \per Y}}{\zero} ; \st{\id{X \per Y}}{\zero} \tag{\ref{eq:whisk:id}} \\
      &= \st{\id{X \per Y}}{\zero} \tag{Table~\ref{fig:freestricmmoncatax}}
    \end{align*}

    \begin{align*}
      (\st{f_1}{S_1} ; \st{f_2}{S_2}) \per (\st{f_3}{S_3} ; \st{f_4}{S_4})
      &= \LW{X}{\st{f_3}{S_3} ; \st{f_4}{S_4}} ; \RW{Z'}{\st{f_1}{S_1} ; \st{f_2}{S_2}} \tag{\ref{eq:utr-per}} \\
      &= \LW{X}{f_3 \mid S_3} ; \LW{X}{f_4 \mid S_4} ; \RW{Z'}{f_1 \mid S_1} ; \RW{Z'}{f_2 \mid S_2} \tag{\ref{eq:whisk:funct}} \\
      &= \LW{X}{f_3 \mid S_3} ; \RW{Y'}{f_1 \mid S_1} ; \LW{Y}{f_4 \mid S_4} ; \RW{Z'}{f_2 \mid S_2} \tag{\ref{eq:tape:LexchangeR}} \\
      &= (\st{f_1}{S_1} \per \st{f_3}{S_3}) ; (\st{f_2}{S_2} \per \st{f_4}{S_4}) \tag{\ref{eq:utr-per}}
    \end{align*}

    Now we define the associator, left and right unitors and symmetries as the obvious lifitng of those in $\Cat{C}$, i.e. as $\st{\alpha^\per_{X,Y,Z}}{\zero}, \st{\lambda^\per_{X}}{\zero}, \st{\rho^\per_{X}}{\zero}, \st{\symmt{X}{Y}}{\zero}$. Observe that since $\Cat{C}$ is right strict,  $\st{\alpha^\per_{X,Y,Z}}{\zero}, \st{\lambda^\per_{X}}{\zero}$ and $\st{\rho^\per_{X}}{\zero}$ are identities. What is left to prove is that they are components of natural transformations.

    \textsc{Left-unitality}.
    \begin{align*}
      \st{\id{\uno}}{\zero} \per \st{f}{S} 
      &= \LW{\uno}{f \mid S} ; \RW{Y}{\id{\uno} \mid \zero} \tag{\ref{eq:utr-per}} \\
      &= \st{f}{S} ; \st{\id{\uno \per Y}}{\zero} \tag{\ref{eq:whisk:uno}, \ref{eq:whisk:id}} \\
      &= \st{f}{S} ; \st{\id{Y}}{\zero} \tag{Table~\ref{tab:equationsonobject}} \\
      &= \st{f}{S} \tag{Table~\ref{fig:freestricmmoncatax}}
    \end{align*}

    \textsc{Right-unitality}.
    \begin{align*}
      \st{f}{S} \per \st{\id{\uno}}{\zero}
      &= \LW{X}{\id{\uno} \mid \zero} ; \RW{\uno}{f \mid S} \tag{\ref{eq:utr-per}} \\
      &= \st{\id{X \per \uno}}{\zero} ; \st{f}{S} \tag{\ref{eq:whisk:uno}, \ref{eq:whisk:id}} \\
      &= \st{\id{X}}{\zero} ; \st{f}{S} \tag{Table~\ref{tab:equationsonobject}} \\
      &= \st{f}{S} \tag{Table~\ref{fig:freestricmmoncatax}}
    \end{align*}

    \textsc{Associativity}.
    \begin{align*}
      (\st{f_1}{S_1} \per \st{f_2}{S_2}) \per \st{f_3}{S_3}
      &= \LW{X_1 \per X_2}{f_3 \mid S_3} ; \RW{Y_3}{\st{f_1}{S_1} \per \st{f_2}{S_2}} \tag{\ref{eq:utr-per}} \\
      &= \LW{X_1 \per X_2}{f_3 \mid S_3} ; \RW{Y_3}{\LW{X_1}{f_2 \mid S_2} ; \RW{Y_2}{f_1 \mid S_1}} \tag{\ref{eq:utr-per}} \\
      &= \LW{X_1 \per X_2}{f_3 \mid S_3} ; \RW{Y_3}{\LW{X_1}{f_2 \mid S_2}} ; \RW{Y_3}{\RW{Y_2}{f_1 \mid S_1}} \tag{\ref{eq:whisk:funct}} \\
      &= \LW{X_1}{\LW{X_2}{f_3 \mid S_3}} ; \LW{X_1}{\RW{Y_3}{f_2 \mid S_2}} ; \RW{Y_2 \per Y_3}{f_1 \mid S_1} \tag{\ref{eq:tape:LL}, \ref{eq:tape:LR}, \ref{eq:tape:RR}} \\
      &= \LW{X_1}{\LW{X_2}{f_3 \mid S_3} ; \RW{Y_3}{f_2 \mid S_2}} ; \RW{Y_2 \per Y_3}{f_1 \mid S_1} \tag{\ref{eq:whisk:funct}} \\
      &= \LW{X_1}{\st{f_2}{S_2} \per \st{f_3}{f_3}} ; \RW{Y_2 \per Y_3}{f_1 \mid S_1} \tag{\ref{eq:utr-per}} \\
      &= \st{f_1}{S_1} \per (\st{f_2}{S_2} \per \st{f_3}{S_3}) \tag{\ref{eq:utr-per}}
    \end{align*}

    \textsc{$\symmt$-naturality}.
    \begin{align*}
      (\st{f}{S} \per \st{\id{Z}}{\zero}) ; \st{\symmt{Y}{Z}}{\zero}
      &= \LW{X}{\id{Z} \mid \zero} ; \RW{Z}{f \mid S} ; \st{\symmt{Y}{Z}}{\zero} \tag{\ref{eq:utr-per}} \\
      &= \st{\id{X \per Z}}{\zero} ; \RW{Z}{f \mid S} ; \st{\symmt{Y}{Z}}{\zero} \tag{\ref{eq:whisk:id}} \\
      &= \RW{Z}{f \mid S} ; \st{\symmt{Y}{Z}}{\zero} \tag{Table~\ref{fig:freestricmmoncatax}} \\
      &=  \st{\symmt{X}{Z}}{\zero} ; \LW{Z}{f \mid S} \tag{\ref{eq:LRnatsym}} \\
      &=  \st{\symmt{X}{Z}}{\zero} ; \LW{Z}{f \mid S} ; \st{\id{Z \per Y}}{\zero} \tag{Table~\ref{fig:freestricmmoncatax}} \\
      &=  \st{\symmt{X}{Z}}{\zero} ; \LW{Z}{f \mid S} ; \RW{Y}{\id{Z} \mid \zero} \tag{\ref{eq:whisk:id}} \\
      &=  \st{\symmt{X}{Z}}{\zero} ; (\st{\id{Z}}{\zero} \per \st{f}{S}) \tag{\ref{eq:utr-per}}
    \end{align*}

    \textsc{$\symmt$-inverses}.
    \begin{align*}
      \st{\symmt{X}{Y}}{\zero} ; \st{\symmt{Y}{X}}{\zero}
      &= \st{\symmt{X}{Y} ; \symmt{Y}{X}}{\zero} \tag{\ref{eq:seq in Utr}} \\
      &= \st{\id{X \per Y}}{\zero} \tag{Table~\ref{fig:freestricmmoncatax}}
    \end{align*}

    \textsc{Coherence}.
    The coherence axioms~\eqref{ax:monoidaltriangle}, \eqref{ax:monoidalpentagone}, \eqref{eq:symmax1}, \eqref{eq:symmax2} and \eqref{eq:symmax3} hold because they hold in $\Cat{C}$. This is due to the fact that these axioms involve only morphisms of the form $\st{f}{\zero}$, thus their composition in $\UTr{(\Cat{C})}$ amounts to their composition in $\Cat{C}$.
\end{proof}

\begin{proof}[Proof of Theorem \ref{th:free-uniform-trace-rig}]
    By Lemma~\ref{lemma:utr-per-monoidal} we know that there is another symmetric monoidal structure on $\UTr(\Cat C)$. To prove that $\UTr(\Cat C)$ is also a rig category, we need to show that:
    \begin{enumerate}
        \item there are left and right natural distributors and annihilators, and;
        \item the coherence axioms in Figure~\ref{fig:rigax} are satisfied.
    \end{enumerate}
    We define the left and right distributors and annihilators as $\st{\dl{X}{Y}{Z}}{\zero}, \st{\dr{X}{Y}{Z}}{\zero}, \st{\annl{X}}{\zero}$ and $\st{\annr{X}}{\zero}$. 
    To prove that these are natural it is convenient to prove coherence first.

    \textsc{Coherence}. The structural morphisms in $\UTr(\Cat{C})$ are defined as the obvious lifting of the corresponding morphisms in $\Cat{C}$. Thus, the axioms  \eqref{eq:rigax1}-\eqref{eq:rigax12} in Figure~\ref{fig:rigax} hold in $\UTr(\Cat{C})$ because they hold in $\Cat{C}$.

    \textsc{Naturality of $\lambda^\bullet$}.
    \begin{align*}
      \st{\id{\zero}}{\zero} \per \st{f}{S}
      &= \LW{\zero}{f \mid S} ; \RW{Y}{\id{\zero} \mid \zero} \tag{\ref{eq:utr-per}} \\
      &= \st{\id{\zero}}{\zero} ; \RW{Y}{\id{\zero} \mid \zero} \tag{\ref{eq:whisk:zero}} \\
      &= \st{\id{\zero}}{\zero} ; \st{\id{\zero \per Y}}{\zero \per Y} \tag{\ref{eq:whisk:id}} \\
      &= \st{\id{\zero}}{\zero} ; \st{\id{\zero}}{\zero} \tag{Table~\ref{tab:equationsonobject}} \\
      &= \st{\id{\zero}}{\zero} \tag{Table~\ref{fig:freestricmmoncatax}}
    \end{align*}

    \textsc{Naturality of $\rho^\bullet$}.
    \begin{align*}
      \st{f}{S} \per \st{\id{\zero}}{\zero}
      &= \LW{X}{\id{\zero} \mid \zero} ; \RW{\zero}{f \mid S} \tag{\ref{eq:utr-per}} \\
      &= \LW{X}{\id{\zero} \mid \zero} ; \st{\id{\zero}}{\zero} \tag{\ref{eq:whisk:zero}} \\
      &= \st{\id{Y \per \zero}}{Y \per \zero} ; \st{\id{\zero}}{\zero} \tag{\ref{eq:whisk:id}} \\
      &= \st{\id{\zero}}{\zero} ; \st{\id{\zero}}{\zero} \tag{Table~\ref{tab:equationsonobject}} \\
      &= \st{\id{\zero}}{\zero} \tag{Table~\ref{fig:freestricmmoncatax}}
    \end{align*}

    \textsc{Naturality of $\delta^r$}.
    \begin{align*}
      (\st{f_1}{S_1} \piu \st{f_2}{S_2}) \per \st{f_3}{S_3}
      &= \LW{X_1 \piu X_2}{f_3 \mid S_3} ; \RW{Y_3}{f_1 \piu f_2} \tag{\ref{eq:utr-per}} \\
      &= (\LW{X_1}{f_3 \mid S_3} \piu \LW{X_2}{f_3 \mid S_3}) ; (\RW{Y_3}{f_1 \mid S_1} \piu \RW{Y_3}{f_2 \mid S_2}) \tag{\ref{eq:whisk:funct piu}, \ref{eq:whisk:sum}} \\
      &= (\LW{X_1}{f_3 \mid S_3} ; \RW{Y_3}{f_1 \mid S_1}) \piu (\LW{X_2}{f_3 \mid S_3};\RW{Y_3}{f_2 \mid S_2}) \tag{Table~\ref{fig:freestricmmoncatax}} \\
      &= (\st{f_1}{S_1} \per \st{f_3}{S_3}) \piu (\st{f_2}{S_2} \per \st{f_3}{S_3}) \tag{\ref{eq:utr-per}}
    \end{align*}

    \textsc{Naturality of $\delta^l$}. It follows from the fact that $\delta^r$ and $\sigma^\per$ are natural and axiom \eqref{eq:rigax1}.

    This proves that \(\UTr(\Cat{C})\) is a rig category.
    For a rig functor \(\fun{F} \colon \Cat{B} \to \Cat{C}\) (see~\cite[Definition 5.1.1]{johnsonYau}), define the rig coherence structure for \(\UTr(\fun{F}) \colon \UTr(\Cat{B}) \to \UTr(\Cat{C})\) via the identity-on-objects functor \(\eta_{\Cat{C}} \colon \Cat{C} \to \UTr(\Cat{C})\).
    This makes \(\UTr(\fun{F})\) a symmetric monoidal functor for the multiplicative structure and a rig functor because all the coherence structure is lifted from that of \(\fun{F}\).
    The components of the unit \(\eta_{\Cat{C}} \colon \Cat{C} \to \fun{U}(\UTr(\Cat{C}))\) and counit \(\epsilon_{\Cat{D}} \colon \UTr(\fun{U}(\Cat{D})) \to \Cat{D}\) of the adjunction are rig functors because they are both identity-on-objects.
    This concludes the proof that both the left and right adjoints restrict and both the unit and counit restrict.
\end{proof}

\subsection{Proofs of Section \ref{ssec:rigbiut}}

\begin{proof}[Proof of Proposition \ref{prop:free-uniform-trace-fb-rig}]
  By \Cref{cor:free-uniform-trace-fb}, the adjunction in the statement restricts to finite biproduct categories and, by \Cref{th:free-uniform-trace-rig}, to rig categories.
  This means that, \emph{(i)} for a finite biproduct rig category $\Cat{C}$, $\UTr(\Cat{C})$ is both a uniformly traced finite biproduct category and a rig category;
  \emph{(ii)} for a finite biproduct rig functor \(\fun{F} \colon \Cat{B} \to \Cat{C}\), \(\UTr(\fun{F}) \colon \UTr(\Cat{B}) \to \UTr(\Cat{C})\) is a uniformly traced finite biproduct rig functor;
  \emph{(iii)} the components of the unit \(\eta_{\Cat{C}}\) are finite biproduct rig functors and the components of the counit \(\epsilon_{\Cat{D}}\) are uniformly traced finite biproduct rig functors.
\end{proof}

\section{Appendix to Section \ref{sec:tapes}}

\begin{proof}[Proof of Lemma \ref{lemma:tr-adj}]
    Given a strict ut-fb category $\Cat{D}$ and a functor $H \colon \Cat{C} \to U_3(\Cat{D})$, one can define the ut-fb functor $H^\sharp \colon F_3(\Cat{C}) \to \Cat{D}$ inductively on objects of $F_3(\Cat{C})$ as 
    \[ H^\sharp(I) \defeq I \qquad \qquad H^\sharp (AP) \defeq H(A) \perG H^\sharp(P) \]
    and on arrows as 
    \[H^{\sharp}(\id\unoG) \defeq{\id{\unoG}} \qquad H^{\sharp}(\id{A})\defeq H(\id{A}) \qquad H^{\sharp}(\tape{c})\defeq H(c) \]
	\[H^{\sharp}(f; g)\defeq H^{\sharp}(f); H^{\sharp}(g) \qquad  H^{\sharp}(f \perG g)\defeq H^{\sharp}(f) \perG H^{\sharp}(g)  \qquad  H^{\sharp}(\sigma^\perG_{A,B})\defeq\sigma^\perG_{H(A),H(B)} \]
	\[ H^{\sharp}(\bang{A})\defeq\bang{H(A)} \qquad  H^{\sharp}(\diag{A})\defeq\diag{H(A)} \qquad  H^{\sharp}(\cobang{A})\defeq\cobang{H(A)} \qquad H^{\sharp}(\codiag{A})\defeq\codiag{H(A)} \qquad H^\sharp(\trace_A f) \defeq \trace_{H(A)} H^\sharp(f)\]
    Observe that $H^{\sharp}$ is well-defined:
    
	\begin{minipage}{0.45\linewidth}
		\begin{align*}
			H^{\sharp}(\tape{\id{A}}) & =\tag{Def. $H^{\sharp}$}  H(\id{A}) \\
& =\tag{Def. $H^{\sharp}$} H^{\sharp}(\id{A})
		\end{align*}
	\end{minipage}
	\hfill
	\begin{minipage}{0.45\linewidth}
		\begin{align*}
			H^{\sharp}(\tape{c;d}) & =\tag{Def. $H^{\sharp}$}  H(c;d) \\
			& =\tag{Fun. $H$} H(c); H(d)\\
			& =\tag{Def. $H^{\sharp}$} H^{\sharp}(\tape{c}); H^{\sharp}(\tape{d}) \\
			& =\tag{Fun. $H^{\sharp}$} H^{\sharp}(\tape{c}; \tape{d}) \\
		\end{align*}
	\end{minipage}
	
	The axioms in Tables ~\ref{fig:freestricmmoncatax},~\ref{fig:freestrictfbcat},~\ref{tab:trace-axioms} and~\ref{tab:uniformity} are preserved by $H^{\sharp}$, since they hold in $\Cat{D}$.
	By definition, $G ; H^{\sharp}= H$. Moreover $ H^{\sharp}$ is the unique strict ut-fb functor satisfying this equation. Thus indeed $F_3 \dashv U_3$.
\end{proof}

\subsection{The rig structure of $\CatTrTape$}\label{app:rigtrTape} %
In the proof of Theorem \ref{thm:freeut-fb}, in particular in \eqref{eq:rigstruct}, we have defined the rig structure of $\CatTrTape$ by means of the isomorphisms $H \colon \UTr(\CatTape) \to \CatTrTape$ and $K \colon \CatTrTape \to \UTr(\CatTape)$.
Since this definition is a bit uncomfortable, we illustrate in this appendix that such construction coincides with the one in   Tables \ref{tab:producttape}, \ref{tab:wisktraces} and Table \ref{table:def dl symmt},

One can readily check that both $H$ and $K$ are identity on objects; $H$ maps an arbitrary arrow $(\t \mid P)$ into $\trace_{P}(\t)$, while $K$ is defined inductively as prescribed by the proof of  Lemma \ref{lemma:tr-adj}:
    \[K(\id\zero) \defeq{(\id{\zero}\mid \zero)} \qquad K(\id{A})\defeq (\id{A}\mid \zero) \qquad K(\tape{c})\defeq (\tape{c} \mid \zero) \]
	\[K(f; g)\defeq K(f); K(g) \qquad  K(f \piu g)\defeq K(f) \piu K(g)  \qquad  K(\sigma^\piu_{A,B})\defeq (\sigma^\piu_{A,B} \mid \zero) \]
	\[ K(\bang{A})\defeq (\bang{A} \mid \zero) \qquad  K(\diag{A})\defeq (\diag{A}\mid \zero) \qquad  K(\cobang{A})\defeq (\cobang{A}\mid \zero) \qquad K(\codiag{A})\defeq (\codiag{A}\mid \zero) \qquad K(\trace_A f) \defeq \trace_{A} K(f)\]
With these definitions is immediate to check that $\dl{P}{Q}{R}$ and $\symmt{P}{Q}$ in $\CatTrTape$ are defined exactly as in $\CatTape$: see Table \ref{table:def dl symmt}.
For $\t_1\colon P_1 \to Q_1$ and $\t_2\colon P_2 \to Q_2$, $\t_1 \per \t_2$ in $\CatTrTape$ can be characterised as
\begin{equation}\label{eq:blabla}\t_1 \per \t_2 = \LWI {P_1}{\t_1}  ; \RWI {Q_2}{\t_2} \end{equation}
where $\LWI{P_1}{\cdot}$ and $\RWI{P_2}{\cdot}$ are the left and right whiskerings defined by extending those in $\CatTape$ (Table \ref{tab:producttape}) with the cases for traces illustrated in Table \ref{tab:wisktraces}.
In order to prove \eqref{eq:blabla}, we rely on the following three lemmas.

\begin{lemma}\label{lemma:functRWI} For all traced tapes $\t, \t_1, \t_2$, monomials $U$ and polynomials $P,Q,S$, the following holds: %
    \begin{enumerate}
        \item $\RWI{U}{\t_1 ; \t_2} = \RWI{U}{\t_1} ; \RWI{U}{\t_2}$
        \item $\RWI{U}{\t_1 \piu \t_2} = \RWI{U}{t_1} \piu \RWI{U}{\t_2}$
        \item $\RWI{U}{\trace_Q \t} = \trace_{Q \per U} \RWI{U}{\t}$
        \item $\RWI{P}{\t} ; \symmt{Q}{P} = \symmt{S}{P} ; \LWI{S}{\t}$
    \end{enumerate}    
\end{lemma}
\begin{proof}Observe that $1.$, $2.$ and $4.$ correspond to the laws \eqref{eq:whisk:funct}, \eqref{eq:whisk:funct piu} and \eqref{eq:LRnatsym} in Table~\ref{table:whisk}. These are proved exactly as in~\cite[Lemma 5.9]{bonchi2023deconstructing}.

The proof of $3.$ proceeds by induction on $Q$. For the base case $\zero$, observe that both sides of the equation amount to $\RWI{U}{\t}$ by means of axiom \eqref{ax:trace:vanishing}. For the inductive case $V \piu Q$, the following holds
    \begin{align*}
        \RWI{U}{\trace_{V \piu Q}\t}
        &= \RWI{U}{\trace_{Q}\trace_V\t} \tag{Table~\ref{tab:inddefutfb}} \\
        &= \trace_{Q \per U}\RWI{U}{\trace_V\t} \tag{Induction hypothesis} \\
        &= \trace_{Q \per U}\trace_{V \per U}{\RWI{U}{\t}} \tag{Table~\ref{tab:wisktraces}} \\
        &= \trace_{(V \per U) \piu (Q \per U)}{\RWI{U}{\t}} \tag{Table~\ref{tab:inddefutfb}} \\
        &= \trace_{(V \piu Q) \per U}{\RWI{U}{\t}} \tag{Table~\ref{tab:equationsonobject}}
    \end{align*}
\end{proof}

\begin{lemma}\label{lemma:monomialwiskeringprime} 
For all traced tapes $\t$ and monomials $U$, $H(\RW {U}{K\t})= \RWI {U}{\t}$.
\end{lemma}
\begin{proof}
The proof goes by induction on $\t$. For the base cases $b\in \{\id{\zero}, \id{A},\tape{c},\sigma^\piu_{A,B},\bang{A} , \diag{A}, \cobang{A}, \codiag{A} \}$:
\begin{align*}
H(\RW {U}{Kb}) & = H(\RW {U}{(b \mid \zero)}) \tag{def. of $K$}\\
& = H({(\RW {U}{b} \mid \zero \per U)}) \tag{Definition \ref{def:utr-whisk}}\\
& = H({(\RW {U}{b} \mid \zero )})  \tag{Table~\ref{tab:equationsonobject}}\\
& = \trace_{\zero}{(\RW {U}{b} )} \tag{Definition of $H$}\\
& = \RW {U}{b}  \tag{\ref{ax:trace:vanishing}}\\
& = \RWI {U}{b}  \tag{Def. of $\RWI{U}{\cdot}$}\\
\end{align*}

For the inductive case $\t_1;\t_2$, one simply exploits functoriality and the induction hypothesis. 
\begin{align*}
H(\RW {U}{K (\t_1 ;\t_2)}) & = H(\RW {U}{K (\t_1) ; K(\t_2)}) \tag{funct. of $K$}\\
& = H(\RW {U}{K (\t_1)} ; \RW {U}{K(\t_2)}) \tag{funct. of $\RW{U}{-}$}\\
& = H(\RW {U}{K (\t_1)}) ; H(\RW {U}{K\t_2}) \tag{funct. of $H$}\\
& = \RWI {U}{\t_1} ; \RWI{U}{\t_2} \tag{Induction Hypothesis}\\
& = \RWI {U}{\t_1; \t_2}   \tag{Lemma \ref{lemma:functRWI}.1}\\
\end{align*}

The case  $\t_1 \piu \t_2$ is analogous, but one uses Lemma \ref{lemma:functRWI}.2. For $\trace_A \t$, one uses Lemma \ref{lemma:functRWI}.3. 
\end{proof}

\begin{lemma}\label{lemma:wiskeringprime} 
For all tapes $\t$ and polynomial $P$, $H(\LW {P}{K\t})= \LWI {P}{\t}$ and $H(\RW {P}{K\t})= \RWI {P}{\t}$.
\end{lemma}
\begin{proof}
We prove by induction on $P$ that $H(\RW {P}{K\t})= \RWI {P}{\t}$.

We consider the base case of $\zero$: %
\begin{align*}
H(\RW {\zero}{K\t}) &= H(\id{\zero} \mid \zero) \tag{\ref{eq:whisk:zero}.2}\\
& = \id{\zero} \tag{def. of $H$}\\
& =  \RWI{\zero}{\t} \tag{def. of $\RWI{\zero}{-}$}
\end{align*}

We consider the inductive case $U \piu P$.
\begin{align*}
H(\RW {U \piu P}{K\t}) &= H(\, (\dl{Z}{A}{U}\mid \zero) ; (\RW{U}{K\t} \piu \RW{P}{K\t}) ; (\Idl{W}{A}{U}\mid \zero) \,) \tag{\ref{eq:whisk:sum}.2}\\
& = H (\dl{Z}{A}{U}\mid \zero) ; (H\RW{U}{K\t} \piu H\RW{P}{K\t}) ; H(\Idl{W}{A}{U}\mid \zero) \tag{funct. of $H$}\\
& = \dl{Z}{A}{U} ; (H\RW{U}{K\t} \piu H\RW{P}{K\t}) ; \Idl{W}{A}{U} \tag{def. of $H$}\\
& = \dl{Z}{A}{U} ; (\RWI{U}{\t} \piu H\RW{P}{K\t}) ; \Idl{W}{A}{U} \tag{Lemma \ref{lemma:monomialwiskeringprime}}\\
& = \dl{Z}{A}{U} ; (\RWI{U}{\t} \piu \RWI{P}{\t}) ; \Idl{W}{A}{U} \tag{Induction hypothesis}\\
& = \RWI{U\piu P}{\t}  \tag{def. of $\RWI{\zero}{-}$ (Table \ref{tab:producttape})}\\
\end{align*}

We now face $\LW {P}{\cdot}$:
\begin{align*}
H(\LW {P}{K\t}) &= H(\, (\symmt{P}{Y} \mid \zero) ; \RW {P}{K\t} ; (\symmt{Z}{P} \mid \zero) \,)  \tag{Definition \ref{def:utr-whisk}}\\
&= H (\symmt{P}{Y} \mid \zero) ; H \RW {P}{K\t} ; H(\symmt{Z}{P} \mid \zero)   \tag{funct of $H$}\\
&= \symmt{P}{Y}  ; H \RW {P}{K\t} ; \symmt{Z}{P}    \tag{def. of $H$}\\
&= \symmt{P}{Y}  ; \RWI {P}{\t} ; \symmt{Z}{P}    \tag{previous point}\\
&= \LWI {P}{\t}     \tag{Lemma \ref{lemma:functRWI}.4}\\
\end{align*} 
\end{proof}

The proof of \eqref{eq:blabla} is concluded by the following derivation.
\begin{align*}
  \t_1 \per \t_2 & =   H(K(\t_1) \per K(\t_2)) \tag{\ref{eq:rigstruct}}\\
  &= H( \, \LW {P_1}{K\t_1}  ; \RW {Q_2}{K\t_2} \,) \tag{def. $\per$}\\
  & = H( \LW {P_1}{K\t_1})  ; H( \RW {Q_2}{K\t_2} ) \tag{funct. of $H$}  \\
  & = \LWI {P_1}{\t_1}; \RWI {Q_2}{\t_2} \tag{Lemma \ref{lemma:wiskeringprime}}
\end{align*}

\section{Appendix to Section \ref{sec:kleene}}\label{app:Kleene} %

\begin{proof}[Proof of Lemma \ref{lemma:order-adjointness}]
  For the $(\implies)$ direction we assume $f \leq g$ and prove separtely the following two inclusions.

  \noindent\begin{minipage}{0.48\textwidth}
    \begin{align*}
      
    \begin{tikzpicture}
	\begin{pgfonlayer}{nodelayer}
		\node [style=label] (105) at (-2.75, 0) {$X$};
		\node [style=black] (107) at (1.5, 0) {};
		\node [style=none] (108) at (0.5, -1) {};
		\node [style=none] (109) at (0.5, 1) {};
		\node [style=label] (120) at (2.75, 0) {$Y$};
		\node [style=black] (121) at (-1.5, 0) {};
		\node [style=none] (122) at (-0.5, -1) {};
		\node [style=none] (123) at (-0.5, 1) {};
		\node [style=none] (124) at (2.25, 0) {};
		\node [style=none] (125) at (-2.25, 0) {};
		\node [style=stringbox] (126) at (0, 1) {$f$};
		\node [style=stringbox] (127) at (0, -1) {$g$};
	\end{pgfonlayer}
	\begin{pgfonlayer}{edgelayer}
		\draw [bend left=45] (109.center) to (107);
		\draw [bend left=45] (107) to (108.center);
		\draw [bend right=45] (123.center) to (121);
		\draw [bend right=45] (121) to (122.center);
		\draw (125.center) to (121);
		\draw (124.center) to (107);
		\draw (123.center) to (126);
		\draw (109.center) to (126);
		\draw (122.center) to (127);
		\draw (127) to (108.center);
	\end{pgfonlayer}
\end{tikzpicture}
}

      &\leq 
    \begin{tikzpicture}
	\begin{pgfonlayer}{nodelayer}
		\node [style=label] (105) at (-2.75, 0) {$X$};
		\node [style=black] (107) at (1.5, 0) {};
		\node [style=none] (108) at (0.5, -1) {};
		\node [style=none] (109) at (0.5, 1) {};
		\node [style=label] (120) at (2.75, 0) {$Y$};
		\node [style=black] (121) at (-1.5, 0) {};
		\node [style=none] (122) at (-0.5, -1) {};
		\node [style=none] (123) at (-0.5, 1) {};
		\node [style=none] (124) at (2.25, 0) {};
		\node [style=none] (125) at (-2.25, 0) {};
		\node [style=stringbox] (126) at (0, 1) {$g$};
		\node [style=stringbox] (127) at (0, -1) {$g$};
	\end{pgfonlayer}
	\begin{pgfonlayer}{edgelayer}
		\draw [bend left=45] (109.center) to (107);
		\draw [bend left=45] (107) to (108.center);
		\draw [bend right=45] (123.center) to (121);
		\draw [bend right=45] (121) to (122.center);
		\draw (125.center) to (121);
		\draw (124.center) to (107);
		\draw (123.center) to (126);
		\draw (109.center) to (126);
		\draw (122.center) to (127);
		\draw (127) to (108.center);
	\end{pgfonlayer}
\end{tikzpicture}
}
 \tag{Hyp.}  \\
      &= 
    \begin{tikzpicture}
	\begin{pgfonlayer}{nodelayer}
		\node [style=label] (105) at (-5, 0) {$X$};
		\node [style=black] (107) at (0.75, 0) {};
		\node [style=none] (108) at (-0.25, -1) {};
		\node [style=none] (109) at (-0.25, 1) {};
		\node [style=label] (120) at (2, 0) {$Y$};
		\node [style=black] (121) at (-1.5, 0) {};
		\node [style=none] (122) at (-0.5, -1) {};
		\node [style=none] (123) at (-0.5, 1) {};
		\node [style=none] (124) at (1.5, 0) {};
		\node [style=stringbox] (126) at (-3, 0) {$g$};
		\node [style=none] (127) at (-4.5, 0) {};
	\end{pgfonlayer}
	\begin{pgfonlayer}{edgelayer}
		\draw [bend left=45] (109.center) to (107);
		\draw [bend left=45] (107) to (108.center);
		\draw [bend right=45] (123.center) to (121);
		\draw [bend right=45] (121) to (122.center);
		\draw (124.center) to (107);
		\draw (123.center) to (109.center);
		\draw (127.center) to (126);
		\draw (122.center) to (108.center);
		\draw (126) to (121);
	\end{pgfonlayer}
\end{tikzpicture}
}
  \tag{\ref{ax:comonoid:nat:copy}}\\
      &\leq 
    \begin{tikzpicture}
	\begin{pgfonlayer}{nodelayer}
		\node [style=label] (105) at (-2, 0) {$X$};
		\node [style=label] (120) at (2, 0) {$Y$};
		\node [style=none] (124) at (1.5, 0) {};
		\node [style=stringbox] (126) at (0, 0) {$g$};
		\node [style=none] (127) at (-1.5, 0) {};
	\end{pgfonlayer}
	\begin{pgfonlayer}{edgelayer}
		\draw (127.center) to (126);
		\draw (126) to (124.center);
	\end{pgfonlayer}
\end{tikzpicture}
}
 \tag{\ref{ax:adjbiprod:3}}
    \end{align*}  
  \end{minipage}
  \quad
  \vline
  \begin{minipage}{0.48\textwidth}
    \begin{align*}
      
    \begin{tikzpicture}
	\begin{pgfonlayer}{nodelayer}
		\node [style=label] (105) at (-2, 0) {$X$};
		\node [style=label] (120) at (2, 0) {$Y$};
		\node [style=none] (124) at (1.5, 0) {};
		\node [style=stringbox] (126) at (0, 0) {$g$};
		\node [style=none] (127) at (-1.5, 0) {};
	\end{pgfonlayer}
	\begin{pgfonlayer}{edgelayer}
		\draw (127.center) to (126);
		\draw (126) to (124.center);
	\end{pgfonlayer}
\end{tikzpicture}
}
 
      &= 
    \begin{tikzpicture}
	\begin{pgfonlayer}{nodelayer}
		\node [style=label] (105) at (-2.75, 0) {$X$};
		\node [style=black] (107) at (1.5, 0) {};
		\node [style=none] (108) at (0.5, -1) {};
		\node [style=none] (109) at (0.5, 1) {};
		\node [style=label] (120) at (2.75, 0) {$Y$};
		\node [style=black] (121) at (-1.5, 0) {};
		\node [style=none] (122) at (-0.5, -1) {};
		\node [style=none] (123) at (-0.5, 1) {};
		\node [style=none] (124) at (2.25, 0) {};
		\node [style=none] (125) at (-2.25, 0) {};
		\node [style=stringbox] (127) at (0, -1) {$g$};
		\node [style=black] (128) at (-0.5, 1) {};
		\node [style=black] (129) at (0.5, 1) {};
	\end{pgfonlayer}
	\begin{pgfonlayer}{edgelayer}
		\draw [bend left=45] (109.center) to (107);
		\draw [bend left=45] (107) to (108.center);
		\draw [bend right=45] (123.center) to (121);
		\draw [bend right=45] (121) to (122.center);
		\draw (125.center) to (121);
		\draw (124.center) to (107);
		\draw (122.center) to (127);
		\draw (127) to (108.center);
		\draw (128) to (123.center);
		\draw (129) to (109.center);
	\end{pgfonlayer}
\end{tikzpicture}
}
 \tag{\ref{ax:comonoid:unit}, \ref{ax:monoid:unit}}\\
      &= 
    \begin{tikzpicture}
	\begin{pgfonlayer}{nodelayer}
		\node [style=label] (105) at (-2.75, 0) {$X$};
		\node [style=black] (107) at (3, 0) {};
		\node [style=none] (108) at (2, -1) {};
		\node [style=none] (109) at (2, 1) {};
		\node [style=label] (120) at (4.25, 0) {$Y$};
		\node [style=black] (121) at (-1.5, 0) {};
		\node [style=none] (122) at (-0.5, -1) {};
		\node [style=none] (123) at (-0.5, 1) {};
		\node [style=none] (124) at (3.75, 0) {};
		\node [style=none] (125) at (-2.25, 0) {};
		\node [style=stringbox] (126) at (1.5, 1) {$f$};
		\node [style=stringbox] (127) at (1.5, -1) {$g$};
		\node [style=black] (128) at (-0.5, 1) {};
		\node [style=black] (129) at (0.25, 1) {};
	\end{pgfonlayer}
	\begin{pgfonlayer}{edgelayer}
		\draw [bend left=45] (109.center) to (107);
		\draw [bend left=45] (107) to (108.center);
		\draw [bend right=45] (123.center) to (121);
		\draw [bend right=45] (121) to (122.center);
		\draw (125.center) to (121);
		\draw (124.center) to (107);
		\draw (109.center) to (126);
		\draw (122.center) to (127);
		\draw (127) to (108.center);
		\draw (129) to (126);
		\draw (128) to (123.center);
	\end{pgfonlayer}
\end{tikzpicture}
}
 \tag{\ref{ax:monoid:nat:discard}}\\
      &\leq 
    \begin{tikzpicture}
	\begin{pgfonlayer}{nodelayer}
		\node [style=label] (105) at (-2.75, 0) {$X$};
		\node [style=black] (107) at (1.5, 0) {};
		\node [style=none] (108) at (0.5, -1) {};
		\node [style=none] (109) at (0.5, 1) {};
		\node [style=label] (120) at (2.75, 0) {$Y$};
		\node [style=black] (121) at (-1.5, 0) {};
		\node [style=none] (122) at (-0.5, -1) {};
		\node [style=none] (123) at (-0.5, 1) {};
		\node [style=none] (124) at (2.25, 0) {};
		\node [style=none] (125) at (-2.25, 0) {};
		\node [style=stringbox] (126) at (0, 1) {$f$};
		\node [style=stringbox] (127) at (0, -1) {$g$};
	\end{pgfonlayer}
	\begin{pgfonlayer}{edgelayer}
		\draw [bend left=45] (109.center) to (107);
		\draw [bend left=45] (107) to (108.center);
		\draw [bend right=45] (123.center) to (121);
		\draw [bend right=45] (121) to (122.center);
		\draw (125.center) to (121);
		\draw (124.center) to (107);
		\draw (123.center) to (126);
		\draw (109.center) to (126);
		\draw (122.center) to (127);
		\draw (127) to (108.center);
	\end{pgfonlayer}
\end{tikzpicture}
}
 \tag{\ref{ax:adjbiprod:4}}
    \end{align*}
  \end{minipage}

  \medskip

  For the $(\impliedby)$ direction we assume $f+g = g$ and prove the following inclusion.
  \begin{equation*}
    
    \begin{tikzpicture}
	\begin{pgfonlayer}{nodelayer}
		\node [style=label] (105) at (-2, 0) {$X$};
		\node [style=label] (120) at (2, 0) {$Y$};
		\node [style=none] (124) at (1.5, 0) {};
		\node [style=stringbox] (126) at (0, 0) {$f$};
		\node [style=none] (127) at (-1.5, 0) {};
	\end{pgfonlayer}
	\begin{pgfonlayer}{edgelayer}
		\draw (127.center) to (126);
		\draw (126) to (124.center);
	\end{pgfonlayer}
\end{tikzpicture}
}
 
    \axeq{\text{\ref{ax:comonoid:unit}, \ref{ax:monoid:unit}}} 
    
    \begin{tikzpicture}
	\begin{pgfonlayer}{nodelayer}
		\node [style=label] (105) at (-2.75, 0) {$X$};
		\node [style=black] (107) at (1.5, 0) {};
		\node [style=none] (108) at (0.5, 1) {};
		\node [style=none] (109) at (0.5, -1) {};
		\node [style=label] (120) at (2.75, 0) {$Y$};
		\node [style=black] (121) at (-1.5, 0) {};
		\node [style=none] (122) at (-0.5, 1) {};
		\node [style=none] (123) at (-0.5, -1) {};
		\node [style=none] (124) at (2.25, 0) {};
		\node [style=none] (125) at (-2.25, 0) {};
		\node [style=stringbox] (127) at (0, 1) {$f$};
		\node [style=black] (128) at (-0.5, -1) {};
		\node [style=black] (129) at (0.5, -1) {};
	\end{pgfonlayer}
	\begin{pgfonlayer}{edgelayer}
		\draw [bend right=45] (109.center) to (107);
		\draw [bend right=45] (107) to (108.center);
		\draw [bend left=45] (123.center) to (121);
		\draw [bend left=45] (121) to (122.center);
		\draw (125.center) to (121);
		\draw (124.center) to (107);
		\draw (122.center) to (127);
		\draw (127) to (108.center);
		\draw (128) to (123.center);
		\draw (129) to (109.center);
	\end{pgfonlayer}
\end{tikzpicture}
}
 
    \axeq{\text{\ref{ax:monoid:nat:discard}}} 
    
    \begin{tikzpicture}
	\begin{pgfonlayer}{nodelayer}
		\node [style=label] (105) at (-2.75, 0) {$X$};
		\node [style=black] (107) at (3, 0) {};
		\node [style=none] (108) at (2, 1) {};
		\node [style=none] (109) at (2, -1) {};
		\node [style=label] (120) at (4.25, 0) {$Y$};
		\node [style=black] (121) at (-1.5, 0) {};
		\node [style=none] (122) at (-0.5, 1) {};
		\node [style=none] (123) at (-0.5, -1) {};
		\node [style=none] (124) at (3.75, 0) {};
		\node [style=none] (125) at (-2.25, 0) {};
		\node [style=stringbox] (126) at (1.5, -1) {$g$};
		\node [style=stringbox] (127) at (1.5, 1) {$f$};
		\node [style=black] (128) at (-0.5, -1) {};
		\node [style=black] (129) at (0.25, -1) {};
	\end{pgfonlayer}
	\begin{pgfonlayer}{edgelayer}
		\draw [bend right=45] (109.center) to (107);
		\draw [bend right=45] (107) to (108.center);
		\draw [bend left=45] (123.center) to (121);
		\draw [bend left=45] (121) to (122.center);
		\draw (125.center) to (121);
		\draw (124.center) to (107);
		\draw (109.center) to (126);
		\draw (122.center) to (127);
		\draw (127) to (108.center);
		\draw (129) to (126);
		\draw (128) to (123.center);
	\end{pgfonlayer}
\end{tikzpicture}
}
 
    \axsubeq{\ref{ax:adjbiprod:4}}
    
    \begin{tikzpicture}
	\begin{pgfonlayer}{nodelayer}
		\node [style=label] (105) at (-2.75, 0) {$X$};
		\node [style=black] (107) at (1.5, 0) {};
		\node [style=none] (108) at (0.5, -1) {};
		\node [style=none] (109) at (0.5, 1) {};
		\node [style=label] (120) at (2.75, 0) {$Y$};
		\node [style=black] (121) at (-1.5, 0) {};
		\node [style=none] (122) at (-0.5, -1) {};
		\node [style=none] (123) at (-0.5, 1) {};
		\node [style=none] (124) at (2.25, 0) {};
		\node [style=none] (125) at (-2.25, 0) {};
		\node [style=stringbox] (126) at (0, 1) {$f$};
		\node [style=stringbox] (127) at (0, -1) {$g$};
	\end{pgfonlayer}
	\begin{pgfonlayer}{edgelayer}
		\draw [bend left=45] (109.center) to (107);
		\draw [bend left=45] (107) to (108.center);
		\draw [bend right=45] (123.center) to (121);
		\draw [bend right=45] (121) to (122.center);
		\draw (125.center) to (121);
		\draw (124.center) to (107);
		\draw (123.center) to (126);
		\draw (109.center) to (126);
		\draw (122.center) to (127);
		\draw (127) to (108.center);
	\end{pgfonlayer}
\end{tikzpicture}
}
 
    \axeq{\text{Hyp.}}
    
    \begin{tikzpicture}
	\begin{pgfonlayer}{nodelayer}
		\node [style=label] (105) at (-2, 0) {$X$};
		\node [style=label] (120) at (2, 0) {$Y$};
		\node [style=none] (124) at (1.5, 0) {};
		\node [style=stringbox] (126) at (0, 0) {$g$};
		\node [style=none] (127) at (-1.5, 0) {};
	\end{pgfonlayer}
	\begin{pgfonlayer}{edgelayer}
		\draw (127.center) to (126);
		\draw (126) to (124.center);
	\end{pgfonlayer}
\end{tikzpicture}
}

  \end{equation*}
\end{proof}

\begin{proof}[Proof of Proposition \ref{prop:matrixform}]
The normal form of fb category is well known: see e.g. \cite[Proposition 2.7]{harding2008orthomodularity}.

For the ordering, observe that if $f \leq g$ then, since $\Cat{C}$ is poset enriched,  
\[(\id{S} \piu \cobang{X}); f ; (\id{T} \piu \cobang{Y}) \leq  (\id{S} \piu \cobang{X}); g ; (\id{T} \piu \cobang{Y})\]
that is $f_{ST} \leq g_{ST}$. Similarly for the others.

Viceversa from $f_{ST} \leq g_{ST}$ $f_{SY} \leq g_{SY}$
$f_{XT} \leq g_{XT}$ and $ f_{XY} \leq g_{XY}$, one can use the formal form to deduce immediately that $f\leq g$.
\end{proof}

\begin{lemma}\label{lemma:equivalentUnif}
The following hold:
\begin{enumerate}
\item $(AU1)$ iff $(AU1')$; 
\item   $(AU2)$ iff $(AU2')$.
\end{enumerate}
\end{lemma}
\begin{proof}
We prove the first point. The second is analogous,

Since the conclusion of  $(AU1)$ and $(AU1')$ are identical, its enough to prove the equivalence of the premises of the two laws.
\begin{itemize}
\item We prove that the premises of $(AU1')$ entail $(AU1)$. Assume that $\exists r_1,r_2$ such that (a) $r_2 \leq r_1$ and (b) $f ; (r_1 \piu \id{}) \leq (r_2 \piu \id{}) ; g$. Thus:
\[ f ; (r_2 \piu \id{}) \stackrel{(a)}{\leq} f ; (r_1 \piu \id{})\stackrel{(b)}{\leq} (r_2 \piu \id{}) ; g\]
Observe that by replacing $r_2$ by $r$ in the above, one obtains exactly the premise of $(AU1)$.
\item We prove that $(AU1)$ entails $(AU1')$. Assume that $(AU1)$ holds. Then $(AU1)'$ holds by taking both $r_1$ and $r_2$ to be $r$.
\end{itemize}
\end{proof}

\subsection{Proof of Proposition \ref{prop:trace-star}}

In order to prove Proposition \ref{prop:trace-star}, we  rely on a result from \cite{cuazuanescu1994feedback}, see also \cite{selinger2011survey}.

\begin{proposition}[From \cite{cuazuanescu1994feedback}]\label{prop:stef} In a category $\Cat{C}$ with finite biproducts, giving a trace is equivalent to giving a repetition operation, i.e., a family of operators \(\kstar{(\cdot)} \colon \Cat{C}(S,S) \to  \Cat{C}(S,S)\) satisfying the following three laws.
\begin{equation}\label{eq:fromstefanescu}
\kstar{f}=\id{} +f;\kstar{f} \qquad \kstar{(f+g)}= \kstar{(\kstar{f};g)};\kstar{f} \qquad \kstar{(f;g)};f = f;\kstar{(g;f)}
\end{equation}
\end{proposition}
The correspondence between traces and $\kstar{(\cdot)}$ is illustrated in Figure \ref{fig:star-trace}.

\begin{proposition}\label{prop:star-fixpoint}
  Let \(\Cat{C}\) be a monoidal category with finite biproducts and trace. For each $f\colon X \to X$ define $\kstar{f}$ as in \Cref{fig:star-trace}. %
  Then, 
  \[\kstar{f}=\id{X} +\kstar{f};f\]
\end{proposition}
\begin{proof} 
  \begingroup
  \allowdisplaybreaks
  \begin{align*}
   \id{} +\kstar{f};f
    &= 
    \InputIfFileExists{propStarFix/step1.tikz}{}{\input{./tikz/propStarFix/step1.tikz}}
 \tag{\eqref{eq:covolution}} \\
    &= 
    \InputIfFileExists{propStarFix/step1slide.tikz}{}{\input{./tikz/propStarFix/step1slide.tikz}}
 \tag{\Cref{fig:star-trace}}  \\
    &= 
    \InputIfFileExists{propStarFix/step2.tikz}{}{\input{./tikz/propStarFix/step2.tikz}}
 \tag{\ref{ax:trace:sliding}}  \\
    &= 
    \InputIfFileExists{propStarFix/step3.tikz}{}{\input{./tikz/propStarFix/step3.tikz}}
 \tag{\ref{ax:monoid:nat:copy}} \\
    &= 
    \InputIfFileExists{propStarFix/step4.tikz}{}{\input{./tikz/propStarFix/step4.tikz}}
 \tag{\ref{ax:trace:sliding}} \\
    &= 
    \InputIfFileExists{propStarFix/step5.tikz}{}{\input{./tikz/propStarFix/step5.tikz}}
 \tag{\ref{ax:trace:yanking}} \\
    &= 
    \InputIfFileExists{propStarFix/step6.tikz}{}{\input{./tikz/propStarFix/step6.tikz}}
 \tag{\ref{ax:comonoid:nat:copy}} \\
    &= \kstar{f} \tag{\Cref{fig:star-trace}}
  \end{align*}
  \endgroup
\end{proof}

\begin{corollary}\label{cor:bounds-star}
   Let \(\Cat{C}\) be a fb category with idempotent convolution and trace. %
  The following inequalities hold for all $f\colon X \to X$:
\[      \begin{array}{c}
      \id{X} + f \dcomp \kstar{f} \leq \kstar{f}  \\
      \id{X} + \kstar{f} \dcomp f \leq \kstar{f} 
    \end{array}
  \]
In particular,   $f \dcomp \kstar{f} \leq \kstar{f} $, $ \kstar{f}; f \leq \kstar{f} $ and $\id{X} \leq \kstar{f}$.
\end{corollary}
\begin{proof}
  By \Cref{lemma:order-adjointness}, \Cref{prop:stef} and \Cref{prop:star-fixpoint}.
\end{proof}

\begin{lemma}\label{lemma:uniform-star}
  Let $\Cat{C}$ be a poset enriched monoidal category with finite biproducts and trace. $\Cat{C}$ satisfies the axioms in \Cref{fig:ineq-uniformity} iff it satisfies those in \Cref{fig:uniform-star}.
  \begin{figure}[ht!]
    \centering
    \[
    \begin{array}{c}
      
    \InputIfFileExists{equivunif/EU1_lhs.tikz}{}{\input{./tikz/equivunif/EU1_lhs.tikz}}
 \leq 
    \InputIfFileExists{equivunif/EU1_rhs.tikz}{}{\input{./tikz/equivunif/EU1_rhs.tikz}}
 \implies 
    \InputIfFileExists{equivunif/EU1TR_lhs.tikz}{}{\input{./tikz/equivunif/EU1TR_lhs.tikz}}
 \leq 
    \InputIfFileExists{equivunif/EU1TR_rhs.tikz}{}{\input{./tikz/equivunif/EU1TR_rhs.tikz}}

      \\[20pt]
      
    \InputIfFileExists{equivunif/EU1_lhs.tikz}{}{\input{./tikz/equivunif/EU1_lhs.tikz}}
 \geq 
    \InputIfFileExists{equivunif/EU1_rhs.tikz}{}{\input{./tikz/equivunif/EU1_rhs.tikz}}
 \implies 
    \InputIfFileExists{equivunif/EU1TR_lhs.tikz}{}{\input{./tikz/equivunif/EU1TR_lhs.tikz}}
 \geq 
    \InputIfFileExists{equivunif/EU1TR_rhs.tikz}{}{\input{./tikz/equivunif/EU1TR_rhs.tikz}}

    \end{array}
    \]
    \caption{Equivalent uniformity axioms.\label{fig:uniform-star}}
  \end{figure}
\end{lemma}
\begin{proof}
  The poset enirched monoidal category obtained by inverting the 2-cells also has biproducts and trace.
  Thus, we show the first of the implications in \Cref{fig:uniform-star} and \Cref{fig:ineq-uniformity}, while the other ones follow by this observation.

  For one direction, suppose that the trace satisfy the axioms in \Cref{fig:ineq-uniformity}, and consider \(f \colon X \to X\), \(g \colon Y \to Y\) and \(r \colon X \to Y\) in \(\Cat{C}\) such that \(f \dcomp r \leq r \dcomp g\).
Observe that %
  \begin{align*}
    

}
 \tag{Hypothesis}
  \end{align*}
  By the first implication in \Cref{fig:ineq-uniformity}, we obtain
  \[ 
    %
}
 .\]

  For the other direction, suppose that the axioms in \Cref{fig:uniform-star} hold, and consider \(f \colon S \piu X \to S \piu Y\), \(g \colon T \piu X \to T \piu Y\) and \(r \colon S \to T\) in \(\Cat{C}\) such that \(f \dcomp (r \piu \id{}) \leq (r \piu \id{}) \dcomp g\).
  Since \(\Cat{C}\) has biproducts, both \(f\) and \(g\) can be written in matrix normal form to obtain element-by-element inequalities.
  \begingroup
  \allowdisplaybreaks
  \begin{align*}
    
    %
 \right. \tag{Hypothesis} \\
  \end{align*}
  \endgroup
  With these inequalities, we show the inequality between the traces.
  \begingroup
  \allowdisplaybreaks
  \begin{align*}
    
    \InputIfFileExists{lemmaUnifEquiv/LTR/traceEquiv/step1.tikz}{}{\input{./tikz/lemmaUnifEquiv/LTR/traceEquiv/step1.tikz}}
 
    &= 
    \InputIfFileExists{lemmaUnifEquiv/LTR/traceEquiv/step2.tikz}{}{\input{./tikz/lemmaUnifEquiv/LTR/traceEquiv/step2.tikz}}
 \tag{Proposition \ref{prop:matrixform}} \\
    &= 
    \InputIfFileExists{lemmaUnifEquiv/LTR/traceEquiv/step3.tikz}{}{\input{./tikz/lemmaUnifEquiv/LTR/traceEquiv/step3.tikz}}
 \tag{trace axioms} \\
    &\leq 
    \InputIfFileExists{lemmaUnifEquiv/LTR/traceEquiv/step4.tikz}{}{\input{./tikz/lemmaUnifEquiv/LTR/traceEquiv/step4.tikz}}
 \tag{ii} \\
    &\leq 
    \InputIfFileExists{lemmaUnifEquiv/LTR/traceEquiv/step5.tikz}{}{\input{./tikz/lemmaUnifEquiv/LTR/traceEquiv/step5.tikz}}
 \tag{i} \\
    &\leq 
    \InputIfFileExists{lemmaUnifEquiv/LTR/traceEquiv/step6.tikz}{}{\input{./tikz/lemmaUnifEquiv/LTR/traceEquiv/step6.tikz}}
 \tag{iii} \\
    &\leq 
    \InputIfFileExists{lemmaUnifEquiv/LTR/traceEquiv/step6.tikz}{}{\input{./tikz/lemmaUnifEquiv/LTR/traceEquiv/step6.tikz}}
 \tag{iv} \\
    &= 
    \InputIfFileExists{lemmaUnifEquiv/LTR/traceEquiv/step8.tikz}{}{\input{./tikz/lemmaUnifEquiv/LTR/traceEquiv/step8.tikz}}
 \tag{trace axioms} \\
    &= 
    \InputIfFileExists{lemmaUnifEquiv/LTR/traceEquiv/step9.tikz}{}{\input{./tikz/lemmaUnifEquiv/LTR/traceEquiv/step9.tikz}}
 \tag{Proposition \ref{prop:matrixform}}
  \end{align*}
  \endgroup
\end{proof}

The above results can be rephrased in terms of $\kstar{(\cdot)}$ as defined in Figure \ref{fig:star-trace}: 
$\Cat{C}$ satisfies the axioms in \Cref{fig:ineq-uniformity} iff $\kstar{(\cdot)}$ satisfies
\begin{equation}\label{eq:unifstar}
\begin{array}{r@{}c@{}l }
     f \dcomp r \leq r \dcomp g & \implies & \kstar{f} \dcomp r \leq r \dcomp \kstar{g} \\
    f \dcomp r \geq r \dcomp g & \implies & \kstar{f} \dcomp r \geq r \dcomp \kstar{g}
\end{array}
\end{equation}

\begin{remark}
It is worth to remark that in \cite{cuazuanescu1994feedback}, it was proved that the implications obtained by replacing $\leq$ by $=$ in \eqref{eq:unifstar} are equivalent to the standard uniformity axioms in Figure \ref{fig:uniformity}.
\end{remark}

It is also immediate to see that the axiom in \Cref{fig:happy-trace} is equivalent to the following.
\begin{equation}\label{eq:happystar}
\kstar{\id{}} \leq \id{}
\end{equation}

\begin{lemma}\label{lemma:uniform-kozen}
 Let $\Cat{C}$ be a fb category with idempotent comvolution and trace. 
 $\Cat{C}$ satisfies the axioms in \Cref{fig:ineq-uniformity} and in \Cref{fig:happy-trace} iff $\kstar{(\cdot)}$ as defined in Figure \ref{fig:star-trace} satisfies the following:
\begin{equation}\label{eq:stroingstar}     
\begin{array}{r@{}c@{}l}
 f \dcomp r \leq r & \implies & \kstar{f} \dcomp r \leq r \\
 l \dcomp f \leq l & \implies & l \dcomp \kstar{f}  \leq l  
    \end{array}
\end{equation}

\end{lemma}

\begin{proof}

We prove that \eqref{eq:unifstar} and \eqref{eq:happystar} hold iff \eqref{eq:stroingstar} holds.

  For one direction, assume that \eqref{eq:unifstar} and \eqref{eq:happystar} hold. To prove that the first implication in \eqref{eq:stroingstar} holds, 
 consider \(f \colon X \to X\) and \(r \colon X \to Y\) such that \(f \dcomp r \leq r\).
  Then, \(f \dcomp r \leq r \dcomp \id{Y}\) and,
  \begin{align*}
    & \kstar{f} \dcomp r &&\\
    & \leq r \dcomp \kstar{\id{Y}} &&\text{\eqref{eq:unifstar}}\\
    & \leq r \dcomp \id{Y} && \text{\eqref{eq:happystar}}\\
    & = r
  \end{align*}
  The second implication follows the symmetric argument. %
  
  For the other direction, assume that \eqref{eq:stroingstar} holds.  To prove \eqref{eq:happystar}, observe that \(\id{} \dcomp \id{} \leq \id{}\).   By \eqref{eq:stroingstar}, \(\kstar{\id{}} = \kstar{\id{}} \dcomp \id{} \leq \id{}\).
  
  To prove the first implication in  \eqref{eq:unifstar}, consider \(f \colon X \to X\), \(g \colon Y \to Y\) and \(r \colon X \to Y\) such that \(f \dcomp r \leq r \dcomp g\).
  Then \(f \dcomp r \dcomp \kstar{g} \leq r \dcomp g \dcomp \kstar{g} \leq r \dcomp \kstar{g}\), where the latter inequality holds by \Cref{cor:bounds-star}.
  By \eqref{eq:stroingstar}, \(\kstar{f} \dcomp r \dcomp \kstar{g} \leq r \dcomp \kstar{g}\).
  By \Cref{cor:bounds-star}, \(\kstar{f} \dcomp r \leq \kstar{f} \dcomp r \dcomp \kstar{g}\), which gives  \(\kstar{f} \dcomp r \leq r \dcomp \kstar{g}\).
  
  To prove the second implication in  \eqref{eq:unifstar}, we proceed similarly: assume that  \(  r \dcomp g \leq f \dcomp r \).
  Then \(\kstar{f} \dcomp r \dcomp g \leq   \kstar{f} \dcomp f \dcomp r \leq  \kstar{f} \dcomp r\), where the latter inequality holds by \Cref{cor:bounds-star}.
  By the second implication in \eqref{eq:stroingstar}, \(\kstar{f} \dcomp r \dcomp \kstar{g} \leq \kstar{f}\dcomp r\).
  By \Cref{cor:bounds-star}, \(  r \dcomp \kstar{g} \leq \kstar{f} \dcomp r \dcomp \kstar{g}\), which gives  \( r \dcomp \kstar{g} \leq \kstar{f} \dcomp r \).
  
\end{proof}

We have now all the ingredients to prove Proposition \ref{prop:trace-star}.

\begin{proof}[Proof of Proposition \ref{prop:trace-star}]
  Suppose that \(\Cat{C}\) is a Kleene bicategory. Then one can define a $\kstar{(\cdot)}$ as in \Cref{fig:star-trace}. By Corollary \ref{cor:bounds-star} and
 \Cref{lemma:uniform-kozen}, $\kstar{(\cdot)}$ satisfies the laws in \eqref{eq:Kllenelaw}. Thus, it is a Kleene star.

  Conversely, suppose that \(\Cat{C}\) has a Kleene star operator $\kstar{(\cdot)}$. One can easily show (e.g., by using completeness of Kozen axiomatisation in \cite{Kozen94acompleteness}) that the laws of Kleene star in \eqref{eq:Kllenelaw} entail those in \eqref{eq:fromstefanescu}. Thus, by Proposition \ref{prop:stef}, $\kstar{(\cdot)}$ gives us a trace as defined in the right of  \Cref{fig:star-trace}.
 By \Cref{lemma:uniform-kozen}, this trace satisfies the laws in \Cref{fig:ineq-uniformity} and in \Cref{fig:happy-trace}. Thus $\Cat{C}$ is a Kleene bicategory.
\end{proof}

\subsection{Proofs of Section \ref{ssec:matrix}}

\begin{proof}[Proof of Proposition \ref{prop:matrices-kleene-bicategory}] %
  By \eqref{eq:matrixadj}, \(\Mat{\Cat{K}}\) has finite biproducts.
  The posetal structure is defined element-wise from the posetal structure of \(\Cat{K}\). 
  We check that it gives adjoint biproducts. The following two derivations prove $\codiag{} \dashv \diag{}$.

  \noindent\begin{minipage}{0.44\textwidth}
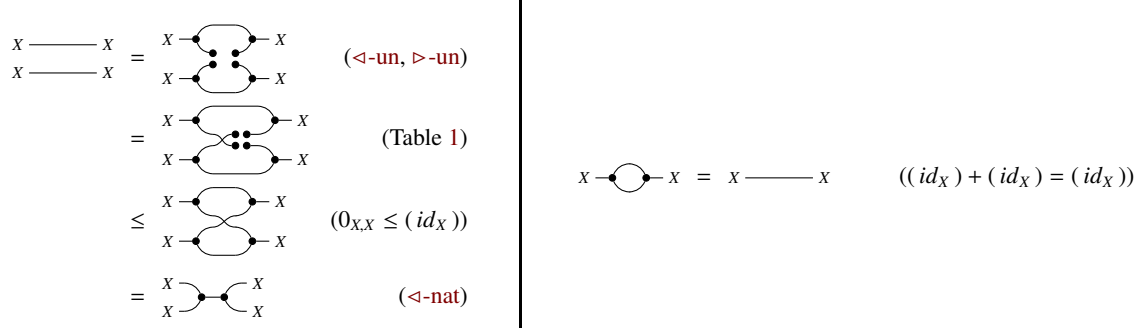
    \begin{align*}
      
    \begin{tikzpicture}
	\begin{pgfonlayer}{nodelayer}
		\node [style=label] (105) at (-1.5, -0.625) {$X$};
		\node [style=label] (110) at (-1.5, 0.625) {$X$};
		\node [style=none] (117) at (-1, -0.625) {};
		\node [style=none] (118) at (-1, 0.625) {};
		\node [style=label] (120) at (2.5, -0.625) {$X$};
		\node [style=label] (124) at (2.5, 0.625) {$X$};
		\node [style=none] (125) at (2, -0.625) {};
		\node [style=none] (126) at (2, 0.625) {};
	\end{pgfonlayer}
	\begin{pgfonlayer}{edgelayer}
		\draw (118.center) to (126.center);
		\draw (125.center) to (117.center);
	\end{pgfonlayer}
\end{tikzpicture}
}
 
      &= 
    \begin{tikzpicture}
	\begin{pgfonlayer}{nodelayer}
		\node [style=label] (105) at (-2.25, -0.875) {$X$};
		\node [style=label] (110) at (-2.25, 0.875) {$X$};
		\node [style=none] (117) at (-1.75, -0.875) {};
		\node [style=none] (118) at (-1.75, 0.875) {};
		\node [style=label] (120) at (2.75, -0.875) {$X$};
		\node [style=label] (124) at (2.75, 0.875) {$X$};
		\node [style=none] (125) at (2.25, -0.875) {};
		\node [style=none] (126) at (2.25, 0.875) {};
		\node [style=black] (127) at (1.5, 0.875) {};
		\node [style=none] (128) at (0.75, 0.25) {};
		\node [style=none] (129) at (0.75, 1.5) {};
		\node [style=black] (130) at (-1, 0.875) {};
		\node [style=black] (131) at (-0.25, 0.25) {};
		\node [style=none] (132) at (-0.25, 1.5) {};
		\node [style=black] (134) at (1.5, -0.875) {};
		\node [style=none] (135) at (0.75, -1.5) {};
		\node [style=none] (136) at (0.75, -0.25) {};
		\node [style=black] (137) at (-1, -0.875) {};
		\node [style=none] (138) at (-0.25, -1.5) {};
		\node [style=black] (139) at (-0.25, -0.25) {};
		\node [style=black] (142) at (0.75, 0.25) {};
		\node [style=black] (143) at (0.75, -0.25) {};
	\end{pgfonlayer}
	\begin{pgfonlayer}{edgelayer}
		\draw [bend left] (129.center) to (127);
		\draw [bend left] (127) to (128.center);
		\draw [bend right] (132.center) to (130);
		\draw [bend right] (130) to (131);
		\draw (118.center) to (130);
		\draw (127) to (126.center);
		\draw [bend left] (136.center) to (134);
		\draw [bend left] (134) to (135.center);
		\draw [bend right] (139) to (137);
		\draw [bend right] (137) to (138.center);
		\draw (134) to (125.center);
		\draw (137) to (117.center);
		\draw (132.center) to (129.center);
		\draw (135.center) to (138.center);
		\draw (128.center) to (142);
		\draw (143) to (136.center);
	\end{pgfonlayer}
\end{tikzpicture}
}
 \tag{\ref{ax:comonoid:unit}, \ref{ax:monoid:unit}} \\
      &= 
    \begin{tikzpicture}
	\begin{pgfonlayer}{nodelayer}
		\node [style=label] (105) at (-2.25, -0.875) {$X$};
		\node [style=label] (110) at (-2.25, 0.875) {$X$};
		\node [style=none] (117) at (-1.75, -0.875) {};
		\node [style=none] (118) at (-1.75, 0.875) {};
		\node [style=label] (120) at (3.75, -0.875) {$X$};
		\node [style=label] (124) at (3.75, 0.875) {$X$};
		\node [style=none] (125) at (3.25, -0.875) {};
		\node [style=none] (126) at (3.25, 0.875) {};
		\node [style=black] (127) at (2.5, 0.875) {};
		\node [style=none] (128) at (1.75, 0.25) {};
		\node [style=none] (129) at (1.75, 1.5) {};
		\node [style=black] (130) at (-1, 0.875) {};
		\node [style=none] (131) at (-0.25, 0.25) {};
		\node [style=none] (132) at (-0.25, 1.5) {};
		\node [style=black] (134) at (2.5, -0.875) {};
		\node [style=none] (135) at (1.75, -1.5) {};
		\node [style=none] (136) at (1.75, -0.25) {};
		\node [style=black] (137) at (-1, -0.875) {};
		\node [style=none] (138) at (-0.25, -1.5) {};
		\node [style=none] (139) at (-0.25, -0.25) {};
		\node [style=black] (140) at (1.25, 0.25) {};
		\node [style=black] (141) at (1.25, -0.25) {};
		\node [style=black] (142) at (0.75, 0.25) {};
		\node [style=black] (143) at (0.75, -0.25) {};
	\end{pgfonlayer}
	\begin{pgfonlayer}{edgelayer}
		\draw [bend left] (129.center) to (127);
		\draw [bend left] (127) to (128.center);
		\draw [bend right] (132.center) to (130);
		\draw [bend right] (130) to (131.center);
		\draw (118.center) to (130);
		\draw (127) to (126.center);
		\draw [bend left] (136.center) to (134);
		\draw [bend left] (134) to (135.center);
		\draw [bend right] (139.center) to (137);
		\draw [bend right] (137) to (138.center);
		\draw (134) to (125.center);
		\draw (137) to (117.center);
		\draw (132.center) to (129.center);
		\draw (135.center) to (138.center);
		\draw (136.center) to (141);
		\draw [in=-180, out=0, looseness=1.25] (131.center) to (143);
		\draw [in=180, out=0] (139.center) to (142);
		\draw (140) to (128.center);
	\end{pgfonlayer}
\end{tikzpicture}
}
 \tag{\Cref{fig:freestricmmoncatax}} \\
      &\leq 
    \begin{tikzpicture}
	\begin{pgfonlayer}{nodelayer}
		\node [style=label] (105) at (-2.25, -0.875) {$X$};
		\node [style=label] (110) at (-2.25, 0.875) {$X$};
		\node [style=none] (117) at (-1.75, -0.875) {};
		\node [style=none] (118) at (-1.75, 0.875) {};
		\node [style=label] (120) at (2.75, -0.875) {$X$};
		\node [style=label] (124) at (2.75, 0.875) {$X$};
		\node [style=none] (125) at (2.25, -0.875) {};
		\node [style=none] (126) at (2.25, 0.875) {};
		\node [style=black] (127) at (1.5, 0.875) {};
		\node [style=none] (128) at (0.75, 0.25) {};
		\node [style=none] (129) at (0.75, 1.5) {};
		\node [style=black] (130) at (-1, 0.875) {};
		\node [style=none] (131) at (-0.25, 0.25) {};
		\node [style=none] (132) at (-0.25, 1.5) {};
		\node [style=black] (134) at (1.5, -0.875) {};
		\node [style=none] (135) at (0.75, -1.5) {};
		\node [style=none] (136) at (0.75, -0.25) {};
		\node [style=black] (137) at (-1, -0.875) {};
		\node [style=none] (138) at (-0.25, -1.5) {};
		\node [style=none] (139) at (-0.25, -0.25) {};
		\node [style=none] (142) at (0.75, 0.25) {};
		\node [style=none] (143) at (0.75, -0.25) {};
	\end{pgfonlayer}
	\begin{pgfonlayer}{edgelayer}
		\draw [bend left] (129.center) to (127);
		\draw [bend left] (127) to (128.center);
		\draw [bend right] (132.center) to (130);
		\draw [bend right] (130) to (131.center);
		\draw (118.center) to (130);
		\draw (127) to (126.center);
		\draw [bend left] (136.center) to (134);
		\draw [bend left] (134) to (135.center);
		\draw [bend right] (139.center) to (137);
		\draw [bend right] (137) to (138.center);
		\draw (134) to (125.center);
		\draw (137) to (117.center);
		\draw (132.center) to (129.center);
		\draw (135.center) to (138.center);
		\draw [in=-180, out=0, looseness=1.25] (131.center) to (143.center);
		\draw [in=180, out=0] (139.center) to (142.center);
		\draw (128.center) to (142.center);
		\draw (143.center) to (136.center);
	\end{pgfonlayer}
\end{tikzpicture}
}
 \tag{$0_{X,X} \leq ( \, \id{X} \, )$} \\
      &= 
    \begin{tikzpicture}
	\begin{pgfonlayer}{nodelayer}
		\node [style=label] (105) at (-1.5, -0.625) {$X$};
		\node [style=black] (107) at (0, 0) {};
		\node [style=none] (108) at (-0.75, -0.625) {};
		\node [style=none] (109) at (-0.75, 0.625) {};
		\node [style=label] (110) at (-1.5, 0.625) {$X$};
		\node [style=none] (117) at (-1, -0.625) {};
		\node [style=none] (118) at (-1, 0.625) {};
		\node [style=label] (120) at (2.5, -0.625) {$X$};
		\node [style=black] (121) at (1, 0) {};
		\node [style=none] (122) at (1.75, -0.625) {};
		\node [style=none] (123) at (1.75, 0.625) {};
		\node [style=label] (124) at (2.5, 0.625) {$X$};
		\node [style=none] (125) at (2, -0.625) {};
		\node [style=none] (126) at (2, 0.625) {};
	\end{pgfonlayer}
	\begin{pgfonlayer}{edgelayer}
		\draw [bend left] (109.center) to (107);
		\draw [bend left] (107) to (108.center);
		\draw [bend right] (123.center) to (121);
		\draw [bend right] (121) to (122.center);
		\draw (107) to (121);
		\draw (123.center) to (126.center);
		\draw (125.center) to (122.center);
		\draw (108.center) to (117.center);
		\draw (118.center) to (109.center);
	\end{pgfonlayer}
\end{tikzpicture}
}
 \tag{\ref{ax:comonoid:nat:copy}}
    \end{align*}
  \end{minipage}
  \qquad
  \vline
  \!\!\!\begin{minipage}{0.55\textwidth}
    \begin{align*}
      
    \begin{tikzpicture}
	\begin{pgfonlayer}{nodelayer}
		\node [style=label] (105) at (-2, 0) {$X$};
		\node [style=black] (107) at (0.75, 0) {};
		\node [style=none] (108) at (0, -0.625) {};
		\node [style=none] (109) at (0, 0.625) {};
		\node [style=label] (120) at (2, 0) {$X$};
		\node [style=black] (121) at (-0.75, 0) {};
		\node [style=none] (122) at (0, -0.625) {};
		\node [style=none] (123) at (0, 0.625) {};
		\node [style=none] (124) at (1.5, 0) {};
		\node [style=none] (125) at (-1.5, 0) {};
	\end{pgfonlayer}
	\begin{pgfonlayer}{edgelayer}
		\draw [bend left] (109.center) to (107);
		\draw [bend left] (107) to (108.center);
		\draw [bend right] (123.center) to (121);
		\draw [bend right] (121) to (122.center);
		\draw (123.center) to (109.center);
		\draw (108.center) to (122.center);
		\draw (125.center) to (121);
		\draw (124.center) to (107);
	\end{pgfonlayer}
\end{tikzpicture}
}
 
      = 
    \begin{tikzpicture}
	\begin{pgfonlayer}{nodelayer}
		\node [style=label] (110) at (-1.5, 0) {$X$};
		\node [style=none] (118) at (-1, 0) {};
		\node [style=label] (124) at (2.5, 0) {$X$};
		\node [style=none] (126) at (2, 0) {};
	\end{pgfonlayer}
	\begin{pgfonlayer}{edgelayer}
		\draw (118.center) to (126.center);
	\end{pgfonlayer}
\end{tikzpicture}}
 \tag{$(\, \id{X} \,) + (\, \id{X} \,) = (\, \id{X} \,)$}
    \end{align*}
  \end{minipage}

  \bigskip

  The following two derivations prove $\cobang{} \dashv \bang{}$.

  \noindent\begin{minipage}{0.44\textwidth}
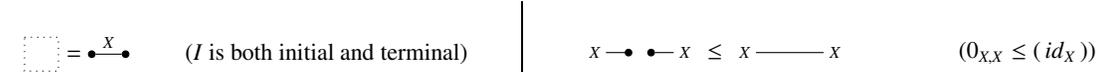
    \begin{align*}
      
    \begin{tikzpicture}
	\begin{pgfonlayer}{nodelayer}
		\node [style=none] (113) at (0.75, 0.75) {};
		\node [style=none] (114) at (-0.75, 0.75) {};
		\node [style=none] (115) at (-0.75, -0.75) {};
		\node [style=none] (116) at (0.75, -0.75) {};
	\end{pgfonlayer}
	\begin{pgfonlayer}{edgelayer}
		\draw [dotted] (114.center)
			 to (113.center)
			 to (116.center)
			 to (115.center)
			 to cycle;
	\end{pgfonlayer}
\end{tikzpicture}
}
 
      = 
    \begin{tikzpicture}
	\begin{pgfonlayer}{nodelayer}
		\node [style=black] (107) at (0.75, 0) {};
		\node [style=black] (119) at (-0.75, 0) {};
		\node [style=label] (120) at (0, 0.5) {$X$};
	\end{pgfonlayer}
	\begin{pgfonlayer}{edgelayer}
		\draw (119) to (107);
	\end{pgfonlayer}
\end{tikzpicture}
}
 \tag{$I$ is both initial and terminal}
    \end{align*}
  \end{minipage}
  \qquad
  \vline
  \!\!\!\!\!\!\!\!\!\!\!\!\!\begin{minipage}{0.55\textwidth}
    \begin{align*}
      
    \begin{tikzpicture}
	\begin{pgfonlayer}{nodelayer}
		\node [style=label] (105) at (-2, 0) {$X$};
		\node [style=black] (107) at (0.5, 0) {};
		\node [style=label] (120) at (2, 0) {$X$};
		\node [style=black] (121) at (-0.5, 0) {};
		\node [style=none] (124) at (1.5, 0) {};
		\node [style=none] (125) at (-1.5, 0) {};
	\end{pgfonlayer}
	\begin{pgfonlayer}{edgelayer}
		\draw (125.center) to (121);
		\draw (124.center) to (107);
	\end{pgfonlayer}
\end{tikzpicture}
}
 
      \leq 
    \begin{tikzpicture}
	\begin{pgfonlayer}{nodelayer}
		\node [style=label] (105) at (-2, 0) {$X$};
		\node [style=none] (117) at (-1.5, 0) {};
		\node [style=label] (120) at (2, 0) {$X$};
		\node [style=none] (125) at (1.5, 0) {};
	\end{pgfonlayer}
	\begin{pgfonlayer}{edgelayer}
		\draw (117.center) to (125.center);
	\end{pgfonlayer}
\end{tikzpicture}
}
 \tag{$0_{X,X} \leq ( \, \id{X} \, )$}
    \end{align*}
  \end{minipage}

  \bigskip
  
  By Lemma 3.3 in \cite{Kozen94acompleteness}, \(\Mat{\Cat{K}}\) has a Kleene star operator. Thus, by \Cref{prop:trace-star}, $\Mat{\Cat{K}}$ is a Kleene bicategory.
\end{proof}

\section{Appendix to Section \ref{sec:traced}}

\subsection{Appendix to Section \ref{sec:2trREL}}

Given a category traced monoidal category $\Cat{C}$, we call a \emph{well typed relation} a set of pairs $(f,g)$ of arrows of $\Cat{C}$ with the same domain and codomain. We write $\WTrelC$ for the set of all well typed relations over $\Cat{C}$. Observe that $\WTrelC$ is a complete lattice with the ordering given by set inclusion. 

\begin{definition}
Let $\Cat{C}$ be a traced monoidal category and $\basicR$ a well typed relation. The functions \[\basicR,r,t,s,;,\perG,ut, ut1,ut2 \colon \WTrelC \to \WTrelC\] are defined for all $\mathbb{R}\in \WTrelC$ as follows:
\begin{itemize}
\item $r(\mathbb{R}) \defeq \{(f,f) \mid f\in \Cat{C}[X,Y]\}$; %
\item $t(\mathbb{R}) \defeq \{(f,h) \mid \exists g\in \Cat{C}[X,Y] \text{ such that } (f,g)\in \mathbb{R} \text{ and } (g,h)\in \mathbb{R}\}$; %
\item $s(\mathbb{R}) \defeq \{(g,f) \mid  (f,g)\in \mathbb{R} \}$; %
\item $\basicR(\mathbb{R}) \defeq \basicR$ 
\item $;(\mathbb{R}) \defeq \{(f_1;g_1\,,\, f_2;g_2) \mid (f_1,g_1)\in \mathbb{R} \text{ and } (f_2,g_2)\in \mathbb{R} \}$; 
\item $\perG(\mathbb{R}) \defeq \{(f_1 \perG g_1\,,\, f_2 \perG g_2) \mid (f_1,g_1)\in \mathbb{R} \text{ and } (f_2,g_2)\in \mathbb{R} \}$; 
\item $ut(\mathbb{R}) \defeq \{(  \trace_{S}f \,,\, \trace_{T}g ) \mid \exists r_1,r_2 \text{ such that } (r_1,r_2) \in  \mathbb{R} \text{ and }(\,f ; (r_1 \piu \id{Y}) ,\; (r_2 \piu \id{X}) ; g\,) \in \mathbb{R}\}$;
\item $ut1(\mathbb{R}) \defeq \{(  \trace_{S}f \,,\, \trace_{T}g ) \mid \exists r_1,r_2 \text{ such that } (r_2,r_1) \in  \mathbb{R} \text{ and }(\,f ; (r_1 \piu \id{Y}) ,\; (r_2 \piu \id{X}) ; g\,) \in \mathbb{R}\}$;
\item $ut2(\mathbb{R}) \defeq \{(  \trace_{S}f \,,\, \trace_{T}g ) \mid \exists r_1,r_2 \text{ such that } (r_2,r_1) \in  \mathbb{R} \text{ and }(\, (r_1 \piu \id{X}) ; f  ,\;   g ; (r_2 \piu \id{Y})\,) \in \mathbb{R}\}$;
\end{itemize}
\end{definition}
The reader should not care for the time being to the rules $ut1$ and $ut2$: their relevance will become clear in Appendix \ref{app:kleene-tapes}.

We define  $\mathsf{uc} \colon \WTrelC \to \WTrelC$ as
\begin{equation*}
\mathsf{uc} \defeq ( r \cup t \cup s \cup  ; \cup  \perG \cup ut)
\end{equation*}
Note that each of the functions in the above correspond to one rule in \eqref{eq:uniformprecong}. More precisely, it holds that
\begin{equation}\label{eq:blablablaUC}
\congB = (\mathsf{uc} \cup \basicR)^\omega
\end{equation}
where, as usual, $f^\omega$ stands for $\bigcup_n f^n$. 

We call a well typed relation $\mathbb{R}$ a \emph{uniform congruence} iff $\mathsf{upc}(\mathbb{R}) \subseteq \mathbb{R}$.

\begin{lemma}\label{lemma:ScottUC}
$\mathsf{uc} \colon \WTrelC \to \WTrelC$ is Scott-continuous.
\end{lemma}
\begin{proof}
One can proceed modularly and prove separately that $id$, $r$, $t$ $;$, $\piu$, $ \per$ , $ut1$, and $ut2$ are Scott-continuous, to then deduce (from standard modularity results) that $\mathsf{uc}$ is Scott-continuous.
The fact that $r$, $t$, $s$ $;$, $\perG$ are Scott continuous is well known. We illustrate below the proof for $ut$. 

Monotonicity of $ut$ is obvious. 
Thus, we only need to prove that \[ut(\bigcup_n\mathbb{R}_n) \subseteq \bigcup_n ut (\mathbb{R}_n) \]
for all directed families $\{\mathbb{R}_n\}_{n\in \N}$ of well typed relations.

Let $(f,g)\in ut(\bigcup_n\mathbb{R}_n) $. Then there exists $f'$ and $g'$ such that $f=\trace{f'}$ and $g=\trace{f'}$. Moreover, there exist $n,m\in \N$ such that 
\[(r_1,r_2) \in  \mathbb{R}_n \text{ and }(\, (r_1 \piu \id{X}) ; f'  ,\;   g' ; (r_2 \piu \id{Y})\,) \in \mathbb{R}_m\]
Since $\{\mathbb{R}_n\}_{n\in \N}$ is directed, there exists some $o\in \N$ such that $\mathbb{R}_n \subseteq \mathbb{R}_o \supseteq \mathbb{R}_m$. Thus  
\[(r_1,r_2) \in  \mathbb{R}_o \text{ and }(\, (r_1 \piu \id{X}) ; f'  ,\;   g' ; (r_2 \piu \id{Y})\,) \in \mathbb{R}_o\]
and thus, by definition of $ut$,
\[(f,g)\in ut(\mathbb{R}_o)\subseteq \bigcup_n ut(\mathbb{R}_n) \text{.} \]
\end{proof}

\begin{lemma}\label{lemma:uniformpreUC}
$\congB$ is a uniform congruence.
\end{lemma}
\begin{proof}
By Lemma \ref{lemma:ScottUC}, one can use the Kleene fixed point theorem to deduce that the least fixed point of $\mathsf{upc} \cup \basicR$ is $\bigcup_n (\mathsf{uc} \cup \basicR)^n$ that by \eqref{eq:blablablaUC} is 
exactly $\congB$.
\end{proof}

\begin{lemma}\label{lemma:smallestUC}
$\congB$ is the smallest uniform congruence including $\basicR$. That is, if $\mathbb{R}$ is a uniform congruence and $\basicR \subseteq \mathbb{R}$, then $\congB \subseteq \mathbb{R}$.
\end{lemma}
\begin{proof}
As observed in the proof above $\congB$ is the least fixed point of $\mathsf{uc} \cup \basicR$. By Knaster-Tarski fixed point theorem, if $\mathbb{R}$ is a well-typed relation such that 
\[\mathsf{uc} \cup \basicR (\mathbb{R}) \subseteq \mathbb{R}\text{,}\]
namely, a uniform congruence including $\basicR$, then $\congB \subseteq \mathbb{R}$.
\end{proof}

\begin{corollary}\label{cor:smallestUC}
Let $\Cat{C}$ be a traced monoidal category, then $\approx_\Cat{C}$ is the smallest uniform congruence on the arrow of $\Cat{C}$.
\end{corollary}
\begin{proof}
In the above lemma replace $\basicR$ by the empty set $\emptyset$.
\end{proof}

\begin{proof}[Proof of Proposition \ref{prop:uniformCat}]
Recall that $\Unif(\Cat{C})$ has the same objects of $\Cat{C}$ and that arrows are $\approx$-equivalence classes $[f]\colon X \to Y$ of arrows of $\Cat{C}$.
Composition and monoidal product are defined as expected: $[f];[g]\defeq [f;g]$ and $[f] \perG [g] \defeq [f \perG g]$. Observe that these operations are well defined by the rules $(;)$ and $(\perG)$: if $f \approx f'$ and $g \approx g'$, then $f;g \approx f';g'$ and  $f \perG g \approx f'\perG g'$. Similarly $ \trace_{S}[f] \defeq [\trace_{S}f] $ is well defined by $(ut)$. Since $\Cat{C}$ is a traced monoidal category, one can deduce immediately that also $\Unif(\Cat{C})$ is trace monoidal one.

We need to show that $\Unif(\Cat{C})$ also respect uniformity. Assume that there exists an arrow $[r]$ in $\Unif(\Cat{C})$ such that 
\[ [f] ; ([r] \piu [\id{Y}]) = ( [r] \piu [\id{X}]) ; [g] \text{.}\]
By definition of composition and monoidal product, the above means that there are arrows in $\Cat{C}$, $i_1, f_1,r_1,r_2,g_2,i_2$ such that 
\[i_1 \approx \id{Y} \quad f_1 \approx f \quad r_1 \approx r \approx r_2 \qquad g_2 \approx g \quad i_2 \approx \id{X}\]
and
\[ f_1 ; (r_1 \piu i_1) \approx ( r_2 \piu i_2) ; g_2 \text{.}\]
By transitivity
\[ f ; (r_1 \piu \id{Y}) \approx ( r_2 \piu \id{X}) ; g \]
 and thus, by (ut),
 \[\trace_{S}f \approx \trace_T{g}\]
that is $\trace_{S}[f] = \trace_{T}[g]$.
\end{proof}

\begin{lemma}\label{lemma:traced-monoidal-functors-preserve-equivalence}
  Traced monoidal functors preserve uniformity equivalence.
  Explicitly, for a traced monoidal functor \(\fun{F} \colon \Cat{B} \to \Cat{C}\) and two morphisms \(f,g \colon X \to Y\) in \(\Cat{B}\), if \(f \unifeq_{\Cat{B}} g\), then \(\fun{F}(f) \unifeq_{\Cat{C}} \fun{F}(g)\).
\end{lemma}
\begin{proof}%
We define $\approx'$ on $\Cat{B}$ as $f \approx' g$ iff $Ff \approx_{\Cat{C}}Fg$.

We prove that $\approx'$ is a uniform congruence, namely that $\mathsf{uc}(\approx') \subseteq \approx'$. 
Since $\approx_{\Cat{C}}$ is an equivalence relation, then $r(\approx') \subseteq \approx'$, $t(\approx') \subseteq \approx'$ and $s(\approx') \subseteq \approx'$.
For the monotone maps $;$ and $\perG$ one uses the fact that $F$ is a  morphism of traced monoidal categories
and that $\approx_\Cat{C}$ is closed by these operations. For instance to prove $;(\approx') \subseteq \approx'$,

\begin{align*}
f_1 \approx' f_2 \text{ and } g_1 \approx' g_2 & \Longleftrightarrow Ff_1 \approx_{\Cat{C}} Ff_2 \text{ and } Fg_1 \approx_{\Cat{C}} Fg_2 \tag{def} \\
&  \Longrightarrow Ff_1; Fg_1 \approx_{\Cat{C}} Ff_2; Fg_2 \tag{$\approx_\Cat{C}$ is closed by $;$}\\
& \Longrightarrow F (f_1; g_1) \approx_{\Cat{C}} F(f_2; g_2) \tag{$F$ functor}\\
& \Longleftrightarrow f_1; g_1 \approx' f_2; g_2 \tag{def} \\
\end{align*}

We illustrate below the proof for $ut(\approx') \subseteq \approx'$. %
\begin{align*}
 & \exists r_1,r_2\colon S \to T\text{  such that  }r_1 \approx' r_2\text{  and }f ; (r_1 \piu \id{Y}) \approx' (r_2 \piu \id{X}) ; g \\ 
 \Longrightarrow & Fr_1  \approx_{\Cat{C}} Fr_2\text{  and }F(\,f ; (r_1 \piu \id{Y})\,)  \approx_{\Cat{C}} F(\,(r_2 \piu \id{X}) ; g) \tag{def}\\
 \Longleftrightarrow & Fr_1  \approx_{\Cat{C}} Fr_2\text{  and }Ff ; (Fr_1 \piu F\id{Y})  \approx_{\Cat{C}} (Fr_2 \piu F\id{X}) ; Fg \tag{$F$ functor}\\
  \Longrightarrow & \trace_{S}Ff \leq_{\Cat{C}} \trace_{T}Fg \tag{$\approx_\Cat{C}$ is closed by $ut$} \\
    \Longleftrightarrow & F(\trace_{S}f) \approx_{\Cat{C}} F(\trace_{T}g) \tag{$F$ preserves traces} \\
    \Longleftrightarrow & \trace_{S}f \approx' \trace_{T}g \tag{def} 
\end{align*}

This concludes the proof that $\mathsf{uc}(\approx') \subseteq \approx'$, namely that $\approx'$ is a uniform congruence. By Corollary \ref{cor:smallestUC}, we have that $\approx_{\Cat{B}} \subseteq \approx'$.
This means that if $f \approx_\Cat{B} g$ then $Ff \approx_{\Cat{C}}Fg$.
\end{proof}

\begin{proof}[Proof of Lemma \ref{lemma:uniformity-quotient-functor}]
  By Proposition \ref{prop:uniformCat}, thus \(\Unif(\Cat{B})\) is a uniformly traced monoidal category.
  This gives the action of \(\Unif\) on objects.

  On morphisms, \(\Unif\) assigns to a traced monoidal functor \(\fun{F} \colon \Cat{B} \to \Cat{C}\) the corresponding functor on equivalence classes:
  for \(f \in \Cat{B}(X,Y)\), \(\Unif(\fun{F})([f]) \defeq [\fun{F}(f)]\), where \([h]\) denotes the $\approx$-equivalence class of \(h\).
  By \Cref{lemma:traced-monoidal-functors-preserve-equivalence}, if \([f] = [g]\), then \([\fun{F}(f)] = [\fun{F}(g)]\), so \(\Unif(\fun{F})\) is well-defined on equivalence classes.
  Finally, \(\Unif(\fun{F})\) inherits the monoidal structure from \(\fun{F}\) and it preserves the trace.
  \[ \Unif(\fun{F})(\trace_{S}[f]) = \Unif(\fun{F})([\trace_{S}f]) = [\fun{F}(\trace_{S}f)] = [\trace_{\fun{F}S}(\fun{F}f)] = \trace_{\fun{F}S}[(\fun{F}f)] = \trace_{\fun{F}S}(\Unif(\fun{F})[f])\]
\end{proof}

\begin{lemma}\label{rem:uniformity-quotient-idempotent}
Let $\Cat{C}$ be a monoidal category. If \(\Cat{C}\) is uniformly traced, then for all arrows $f,g$, it holds that if $f\approx_{\Cat{C}}g$, then $f=g$.
\end{lemma}
\begin{proof}
Let $\mathbb{ID}\defeq\{(f,f) \mid f \in Cat{C}[X,Y]\}$ be the well typed identity relation on the arrow of $\Cat{C}$.

Observe that if $\Cat{C}$ is uniformly traced, then $ut(\mathbb{ID}) \subseteq \mathbb{ID}$.

Moreover one can immeditaely check that, for any $\Cat{C}$, the followings hold:
\[r(\mathbb{ID}) \subseteq \mathbb{ID} \qquad t(\mathbb{ID})\subseteq \mathbb{ID} \quad s(\mathbb{ID})\subseteq \mathbb{ID} \quad ;(\mathbb{ID})\subseteq \mathbb{ID} \quad \perG(\mathbb{ID})\subseteq \mathbb{ID} \]
Thus $\mathsf{uc}(\emptyset) \subseteq \mathsf{uc}(\mathbb{ID}) \subseteq \mathbb{ID}$ and thus for all $n\in \N$,
\[\mathsf{uc}^n(\emptyset)\subseteq \mathbb{ID}\text{,}\]
namely $\approx_\Cat{C}\subseteq \mathbb{ID}$.
\end{proof}

\begin{proof}[Proof of Proposition \ref{prop:free-uniform}]
  By \Cref{lemma:uniformity-quotient-functor}, \(\Unif\) is a functor.
  We show that it is a left adjoint by defining the unit of the adjunction and checking the universal property.
  The components of the unit are traced monoidal functors \(\eta_{\Cat{B}} \colon \Cat{B} \to \fun{U}(\Unif(\Cat{B}))\).
  We define them to be identity-on-objects, \(\eta_{\Cat{B}}(X) \defeq X\), and to map a morphism to its uniformity equivalence class \(\eta_{\Cat{B}}(f) \defeq [f]\).
  By \Cref{prop:uniformCat}, uniformity equivalence classes respect compositions, monoidal products and trace, which makes \(\eta_{\Cat{B}}\) a functor.
  Naturality follows from the definitions of \(\eta\) and \(\Unif\).
  \begin{align*}
    & \fun{U}(\Unif(\fun{F}))(\eta_{\Cat{B}}(X)) && \fun{U}(\Unif(\fun{F}))(\eta_{\Cat{B}}(f)) \\
    & = \fun{U}(\Unif(\fun{F}))(X) &    & = \fun{U}(\Unif(\fun{F}))([f]) \\
    & = \fun{F}(X) &    & = [\fun{F}(f)] \\
    & = \eta_{\Cat{C}}(\fun{F}(X)) &    & = \eta_{\Cat{C}}(\fun{F}(f))
  \end{align*}

  Let \(\fun{G} \colon \Cat{B} \to \fun{U}(\Cat{D})\) be a traced monoidal functor and define \(\hat{\fun{G}} \colon \Unif(\Cat{B}) \to \Cat{C}\) as \(\hat{\fun{G}}(X) \defeq \fun{G}(X)\) and \(\hat{\fun{G}}([f]) \defeq \fun{G}(f)\).
  By \Cref{lemma:traced-monoidal-functors-preserve-equivalence}, if \(f \unifeq g\), then \(\fun{G}(f) \unifeq \fun{G}(g)\).
  Since \(\Cat{C}\) is uniformly traced, by \Cref{rem:uniformity-quotient-idempotent} this shows that \(\fun{G}(f) = \fun{G}(g)\) and that \(\hat{\fun{G}}\) is well-defined.
  Since \(\fun{G}\) is a traced monoidal functor, so is \(\hat{\fun{G}}\).
  Finally, \(\hat{\fun{G}}\) is the only possible functor satisfying \(\fun{U}(\hat{\fun{G}})(\eta_{\Cat{B}}(f)) = \fun{G}(f)\).
\end{proof}

\begin{proof}[Proof of Theorem \ref{th:free-uniform-trace}]
  The results in~\cite{katis2002feedback} construct an adjunction \(\freeTr \colon \SMC \leftrightarrows \TSMC \co{\colon} \fun{U}\) that gives the free traced monoidal category over a symmetric monoidal category.
  By \Cref{prop:free-uniform}, there is an adjunction that quotients by uniformity, \(\Unif \colon \TSMC \leftrightarrows \UTSMC \co{\colon} \fun{U}\).
  By composing these two adjunctions, we obtain the desired adjunction.

  The unit and counit of this adjunction are compositions of the units and counits of the smaller adjunctions.
  We describe them explicitly.
  The components of the unit are identity-on-objects symmetric monoidal functors \(\eta_{\Cat{B}} \colon \Cat{B} \to \fun{U}(\UTr(\Cat{B}))\).
  A morphism \(f \colon X \to Y\) in \(\Cat{B}\) is mapped to the uniformity equivalence class of \(f\) with monoidal unit state space, \(\eta_{\Cat{B}}(f) = [(f \mid I)]\).
  The components of the counit are identity-on-objects traced monoidal functors \(\epsilon_{\Cat{C}} \colon \UTr(\fun{U}(\Cat{C})) \to \Cat{C}\).
  A morphism \([(f \mid S)] \colon X \to Y\) in \(\UTr(\fun{U}(\Cat{C}))\) is mapped to the trace of \(f\) on \(S\), \(\epsilon_{\Cat{C}}([(f \mid S)]) = \trace_{S}f\).
\end{proof}

\section{Appendix to Section~\ref{sec:kleene-tapes}}\label{app:kleene-tapes}

\begin{proof}[Proof of Lemma \ref{lemma:ditributivityper}]
The laws holds in an fb-rig category by Proposition 6.1 in \cite{bonchi2023deconstructing}. Thus, in particular, they hold in any Kleene rig category.
\end{proof}

\subsection{Proof of Theorem \ref{thm:Kleenetapesfree}} %
Theorem \ref{thm:Kleenetapesfree} follows almost trivially by freeness of $\CatTrTape$ (Theorem \ref{thm:freeut-fb}). However, since the definition of $\precongK$ in \eqref{eq:uniformprecong} involves the uniformity laws ($ut$-1) and ($ut$-2), the proof of Theorem \ref{thm:Kleenetapesfree} requires some extra care. To stay on the safe side, we are going to be a little pedantic and illustrate all details.

\smallskip

Given a category $\Cat{C}$, we call a \emph{well typed relation} a set of pairs $(f,g)$ of arrows of $\Cat{C}$ with the same domain and codomain. We write $\WTrelC$ for the set of all well typed relations over $\Cat{C}$. Observe that $\WTrelC$ is a complete lattice with the ordering given by set inclusion. 

Whenever $\Cat{C}$ has enough structure, one can define the maps $\mathbb{I},r,t,;,\piu,\per,ut1,ut2 \colon \WTrelC \to \WTrelC$ as follows: for all $\mathbb{R}\in \WTrelC$
\begin{itemize}
\item $\basicR(\mathbb{R}) \defeq \basicR$ ($\basicR$ is some element in $\WTrelC$); %
\item $r(\mathbb{R}) \defeq \{(f,f) \mid f\in \Cat{C}[X,Y]\}$; %
\item $t(\mathbb{R}) \defeq \{(f,h) \mid \exists g\in \Cat{C}[X,Y] \text{ such that } (f,g)\in \mathbb{R} \text{ and } (g,h)\in \mathbb{R}\}$; %
\item $;(\mathbb{R}) \defeq \{(f_1;g_1\,,\, f_2;g_2) \mid (f_1,g_1)\in \mathbb{R} \text{ and } (f_2,g_2)\in \mathbb{R} \}$; 
\item $\piu(\mathbb{R}) \defeq \{(f_1 \piu g_1\,,\, f_2 \piu g_2) \mid (f_1,g_1)\in \mathbb{R} \text{ and } (f_2,g_2)\in \mathbb{R} \}$; 
\item $\per(\mathbb{R}) \defeq \{(f_1 \per g_1\,,\, f_2 \per g_2) \mid (f_1,g_1)\in \mathbb{R} \text{ and } (f_2,g_2)\in \mathbb{R} \}$; 
\item $ut1(\mathbb{R}) \defeq \{(  \trace_{S}f \,,\, \trace_{T}g ) \mid \exists r_1,r_2 \text{ such that } (r_2,r_1) \in  \mathbb{R} \text{ and }(\,f ; (r_1 \piu \id{Y}) ,\; (r_2 \piu \id{X}) ; g\,) \in \mathbb{R}\}$;
\item $ut2(\mathbb{R}) \defeq \{(  \trace_{S}f \,,\, \trace_{T}g ) \mid \exists r_1,r_2 \text{ such that } (r_2,r_1) \in  \mathbb{R} \text{ and }(\, (r_1 \piu \id{X}) ; f  ,\;   g ; (r_2 \piu \id{Y})\,) \in \mathbb{R}\}$;
\end{itemize}
and $\mathsf{upc} \colon \WTrelC \to \WTrelC$ as
\begin{equation*}
\mathsf{upc} \defeq ( r \cup t \cup ; \cup \piu \cup \per \cup ut1 \cup ut2)
\end{equation*}
Note that each of the function defined above correspond to a rule in \eqref{eq:uniformprecong}. More precisely, it holds that
\begin{equation}\label{eq:blablabla}
\precongB = (\mathsf{upc} \cup \basicR)^\omega
\end{equation}
where, as expected, $f^\omega$ stands for $\bigcup_n f^n$.

\begin{remark}
It is worth to be precise and explain that in \eqref{eq:uniformprecong} we took $\basicR$ to be a well typed relation over $\CatTrTape$, while in \eqref{eq:blablabla} $\basicR$ is defined for an arbitrary category $\Cat{C}$ with enough structure. Below, we will first illustrate some result at this higher level of generality and then we will focus on $\basicK$ over  $\CatTrTape$.
\end{remark}

\begin{lemma}\label{lemma:Scott}
$\mathsf{upc} \colon \WTrelC \to \WTrelC$ is Scott-continuous.
\end{lemma}
\begin{proof}
One can proceed modularly and prove separetly that $id$, $r$, $t$ $;$, $\piu$, $ \per$ , $ut1$, and $ut2$ are Scott-continuous, to then deduce (from standard modularity results) that $\mathsf{upc}$ is Scott-continuous.
The fact that $id$, $r$, $t$ $;$, $\piu$, $ \per$ are Scott continuous is well known. We illustrate below the proof for $ut1$. The one for $ut2$ is similar.

Monotonicity of $ut1$ is obvious. 
Thus, we only need to prove that \[ut1(\bigcup_n\mathbb{R}_n) \subseteq \bigcup_n ut1 (\mathbb{R}_n) \]
for all directed families $\{\mathbb{R}_n\}_{n\in \N}$ of well typed relations.

Let $(f,g)\in ut1(\bigcup_n\mathbb{R}_n) $. Then there exists $f'$ and $g'$ such that $f=\trace{f'}$ and $g=\trace{f'}$. Moreover, there exist $n,m\in \N$ such that 
\[(r_2,r_1) \in  \mathbb{R}_n \text{ and }(\, (r_1 \piu \id{X}) ; f'  ,\;   g' ; (r_2 \piu \id{Y})\,) \in \mathbb{R}_m\]
Since $\{\mathbb{R}_n\}_{n\in \N}$ is directed, there exists some $o\in \N$ such that $\mathbb{R}_n \subseteq \mathbb{R}_o \supseteq \mathbb{R}_m$. Thus  
\[(r_2,r_1) \in  \mathbb{R}_o \text{ and }(\, (r_1 \piu \id{X}) ; f'  ,\;   g' ; (r_2 \piu \id{Y})\,) \in \mathbb{R}_o\]
and thus, by definition of $ut1$,
\[(f,g)\in ut1(\mathbb{R}_o)\subseteq \bigcup_n ut1(\mathbb{R}_n) \text{.} \]
\end{proof}

Hereafter we call a well typed relation $\mathbb{R}$ a \emph{uniform precongruence} iff $\mathsf{upc}(\mathbb{R}) \subseteq \mathbb{R}$.

\begin{lemma}\label{lemma:uniformpre}
$\precongB$ is a uniform precongruence.
\end{lemma}
\begin{proof}
By Lemma \ref{lemma:Scott}, one can use the Kleene fixed point theorem to deduce that the least fixed point of $\mathsf{upc} \cup \basicR$ is $\bigcup_n (\mathsf{upc} \cup \basicR)^n$ that by \eqref{eq:blablabla} is 
exactly $\precongB$.
\end{proof}

\begin{lemma}\label{lemma:smallest}
$\precongB$ is the smallest uniform precongruence including $\basicR$. That is, if $\mathbb{R}$ is a uniform precongruence and $\basicR \subseteq \mathbb{R}$, then $\precongB \subseteq \mathbb{R}$.
\end{lemma}
\begin{proof}
As observed in the proof above $\precongB$ is the least fixed point of $\mathsf{upc} \cup \basicR$. By Knaster-Tarski fixed point theorem, if $\mathbb{R}$ is a well-typed relation such that 
\[\mathsf{upc} \cup \basicR (\mathbb{R}) \subseteq \mathbb{R}\text{,}\]
namely, a uniform precongruence including $\basicR$, then $\precongB \subseteq \mathbb{R}$.
\end{proof}

\begin{proposition}\label{prop:tapeisKleenerig}
$\CatKTape$ is a $\sort$-sesquistrict Kleene rig category.
\end{proposition}
\begin{proof}
By Propositions \ref{prop:iso} and \ref{prop:tracedtacesutfb}, $\CatTrTape$ is a $\sort$-sesquistrict traced fb rig category. Since $\CatKTape$ is obtained by quotienting $\CatTrTape$, then $\CatKTape$  is a traced fb rig category. By definition of $\basicK$, the axioms in Figure \ref{fig:adjoint-biproducts} hold and thus $\CatKTape$  is a fb category with idempotent convolution. By definition of $\basicK$ also the axioms in Figure \ref{fig:happy-trace} hold. To conclude that it is a Kleene rig category is enough to show the laws in Figure \ref{fig:ineq-uniformity} or, equivalently, the laws in \eqref{eq:equivalentuni1} and \eqref{eq:equivalentuni2}. 

We illustrate the proof for \eqref{eq:equivalentuni1}. The one for \eqref{eq:equivalentuni2} is identical.
Recall that arrows of $\CatKTape$ are equivalence classes of arrows of $\CatTrTape$ w.r.t. $\sim_{\basicK}  \defeq  \leq_{\basicK} \cap \geq_{\basicK}$. All the operations, such as compostion and monoidal products, are defined on equivalence classes in the expected way, e.g. $[f];[g]=[f;g]$. The ordering is the expected one: $[f] \precongK [g]$ iff $f \precongK g$.
We have to prove that 
\[
\text{If }\exists [r_1],[r_2]\colon S \to T\text{ such that } [r_2] \leq_{\basicK} [r_1]\text{ and }[f] ; ([r_1] \piu [\id{Y}]) \leq_{\basicK} ([r_2] \piu [\id{X}]) ;[ g]\text{, then }\trace_{S}[f] \leq_{\basicK} \trace_{T}[g]
\]
which, by definition of the operations, is equivalent to 
\[
\text{If }\exists [r_1],[r_2]\colon S \to T\text{ such that } [r_2] \leq_{\basicK} [r_1]\text{ and }[f ; (r_1 \piu \id{Y})] \leq_{\basicK} [(r_2 \piu \id{X}) ; g]\text{, then } [\trace_{S}f] \leq_{\basicK} [\trace_{T}g]
\]
which, by definition of the ordering is equivalent to
\[
\text{If }\exists r_1,r_2\colon S \to T\text{ such that } r_2 \leq_{\basicK} r_1\text{ and }f ; (r_1 \piu \id{Y}) \leq_{\basicK} (r_2 \piu \id{X}) ; g\text{, then } \trace_{S}f \leq_{\basicK} \trace_{T}g\text{;}
\]
The latter holds, since by Lemma \ref{lemma:uniformpre}, $\precongK$ is a uniform precongruence.

\end{proof}

\begin{lemma}\label{lemma:orderingfunctors}
Let $\Cat{C}$ be a $\sort$-sesquistrict Kleene rig category with ordering $\leq_{\Cat{C}}$. Let $F \colon \CatTrTape \to \Cat{C}$ be a ut-fb rig morphism. For all traced tapes $\t_1,\t_2\colon P \to Q$, if $\t_1 \precongK \t_2$ then
$F\t_1 \leq_{\Cat{C}}F\t_2$.
\end{lemma}
\begin{proof}
Define $\leq'$ on $\CatTrTape$ as $\t_1 \leq' \t_2$ iff $F\t_1 \leq_{\Cat{C}}F\t_2$.

We first prove that $\leq'$ is a uniform precongruence, namely that $\mathsf{upc}(\leq') \subseteq \leq'$. 
Since $\leq_{\Cat{C}}$ is a poset, then $r(\leq') \subseteq \leq'$ and $t(\leq') \subseteq \leq'$.
For the monotone maps $;$, $\piu$ and $\per$ one uses the fact that $F$ is a morphism
and that $\Cat{C}$ is poset enriched. For instance to prove $;(\leq') \subseteq \leq'$,
\begin{align*}
f_1 \leq' f_2 \text{ and } g_1 \leq' g_2 & \Longleftrightarrow Ff_1 \leq_{\Cat{C}} Ff_2 \text{ and } Fg_1 \leq_{\Cat{C}} Fg_2 \tag{def} \\
&  \Longrightarrow Ff_1; Fg_1 \leq_{\Cat{C}} Ff_2; Fg_2 \tag{$\Cat{C}$ is poset enrichmed}\\
& \Longrightarrow F (f_1; g_1) \leq_{\Cat{C}} F(f_2; g_2) \tag{$F$ functor}\\
& \Longleftrightarrow f_1; g_1 \leq' f_2; g_2 \tag{def} \\
\end{align*}
We illustrate below the proof for $ut1(\leq') \subseteq \leq'$. The one for $ut2$ is similar.
\begin{align*}
 & \exists r_1,r_2\colon S \to T\text{  such that  }r_2 \leq' r_1\text{  and }f ; (r_1 \piu \id{Y}) \leq' (r_2 \piu \id{X}) ; g \\ 
 \Longrightarrow & Fr_2  \leq_{\Cat{C}} Fr_1\text{  and }F(\,f ; (r_1 \piu \id{Y})\,)  \leq_{\Cat{C}} F(\,(r_2 \piu \id{X}) ; g) \tag{def}\\
 \Longleftrightarrow & Fr_2  \leq_{\Cat{C}} Fr_1\text{  and }Ff ; (Fr_1 \piu F\id{Y})  \leq_{\Cat{C}} (Fr_2 \piu F\id{X}) ; Fg \tag{$F$ functor}\\
  \Longrightarrow & \trace_{S}Ff \leq_{\Cat{C}} \trace_{T}Fg \tag{$\Cat{C}$ is a Kleene rig category and \eqref{eq:equivalentuni1}} \\
    \Longleftrightarrow & F(\trace_{S}f) \leq_{\Cat{C}} F(\trace_{T}g) \tag{$F$ preserves traces} \\
    \Longleftrightarrow & \trace_{S}f \leq' \trace_{T}g \tag{def} 
\end{align*}
Next, we observe that 
\begin{align*}
 \basicK\defeq & \{(\id{P\piu P}  \, , \, \codiag{P};\diag{P} \mid P \in Ob(\CatTrTape) \} \cup \{( \diag{P}; \codiag{P} \,,\, \id{P} ) \mid P \in Ob(\CatTrTape)\} \cup \\
 & \{(\id{\zero}  \, , \, \cobang{P};\bang{P} \mid P \in Ob(\CatTrTape) \} \cup \{( \bang{P}; \cobang{P} \,,\, \id{P} ) \mid P \in Ob(\CatTrTape)\} \cup \\
 & \{(\trace_{P}(\codiag{P};\diag{P}) \,,\, \id{P})\mid P \in Ob(\CatTrTape)\}
\end{align*}
is included into $\leq'$. The proof proceeds by cases and again it relies on the fact that $\Cat{C}$ has the structure of a Kleene bicategory and that $F$ preserves such structure. For instance, to prove that 
 $\{(\id{P\piu P}  \, , \, \codiag{P};\diag{P} \mid P \in Ob(\CatTrTape) \} \subseteq \leq' $, one shows that
\begin{align*}
\id{P\piu P} \leq' \codiag{P};\diag{P} & \Longleftrightarrow F\id{P\piu P} \leq_{\Cat{C}} F(\codiag{P};\diag{P}) \tag{def}\\
& \Longleftrightarrow \id{FP\piu FP} \leq_{\Cat{C}} \codiag{FP};\diag{FP} \tag{$F$ morphism fb-categories} 
\end{align*}
and conclude by observing that the latter holds since $\Cat{C}$ is a Kleene bicategory.

\medskip

Now, since $\leq'$ is a uniform precongruence and since $\basicK\subseteq \leq'$ then, by Lemma \ref{lemma:smallest}, one has that  $\precongK\subseteq \leq'$. This means that if 
$\t_1 \precongK \t_2$ then $F\t_1 \leq_{\Cat{C}}F\t_2$.
\end{proof}

Now, the proof of Theorem \ref{thm:Kleenetapesfree}  amounts to properly use the above result and Theorem \ref{thm:freeut-fb}.

\begin{proof}[Proof of Theorem \ref{thm:Kleenetapesfree}]
Recall that by Theorem \ref{thm:freeut-fb}, $\sort \to \CatTrTape$ is a free $\sort$-sesquistrict ut-fb rig category generated by $(\sort, \sign)$. This means that (Definition \ref{def:freesesqui}) there exists an interpretation $(\alpha_\sort, \alpha_\sign)$ with $\alpha_\sort\colon \sort \to \sort$ and $\alpha_\sign\colon \sort \to Ar(\CatTrTape)$ such that for any $\sort$-sesquistrict ut-fb rig category $\Cat{S} \to \Cat{C}$ and any interpretation $(\alpha_\sort', \alpha_\sign')$ with $\alpha_\sort'\colon \sort \to Ob(\Cat{S})$ and $\alpha_\sign'\colon \sort \to Ar(\Cat{C})$, there exists a unique sesquistrict rig functor $(\alpha, \beta)$ with $\alpha \colon \sort \to \Cat{S}$ and $\beta\colon \CatTrTape \to \Cat{C} $ such that $\alpha_\sort ;  \alpha = \alpha_\sort '$ and  $\alpha_\sign ;  \beta = \alpha_\sign '$.

We need to show that the same property hold for $\sort \to \CatKTape$ when replacing ut-fb rig category by Kleene rig category.

First, observe that there is a ut-fb morphism \[\eta \colon \CatTrTape \to \CatKTape\] that is the identity on object, i.e., $\eta(P)\defeq P$, and maps an arrows $\t\colon P \to Q$ into the $\sim$-equivalence classes $[\t]\colon P \to Q$.
We can thus fix as interpretation $(\tilde{\alpha_\sort}, \tilde{\alpha_\sign})$ as (a) $\tilde{\alpha_\sort} \defeq \alpha_\sort$ and (b) $\tilde{\alpha_\sign}\defeq \alpha_\sign ; \eta$.

\medskip

Now take $\Cat{S} \to \Cat{C}$ to be any $\sort$-sesquistrict Kleene rig category with an interpretation $(\alpha_\sort', \alpha_\sign')$. Since it is a Kleene rig category, it is in particular a ut-fb rig category and thus, by the freeness of $\CatTrTape$ there exists a unique sesquistrict ut-fb rig functor $(\alpha, \beta)$ with \[\alpha \colon \sort \to \Cat{S} \text{ and }\beta\colon \CatTrTape \to \Cat{C} \] such that  (c) $\alpha_\sort ;  \alpha = \alpha_\sort '$ and  (d) $\alpha_\sign ;  \beta = \alpha_\sign '$.

Now define $\tilde{\beta} \colon \CatKTape \to \Cat{C}$ as 
\[\tilde{\beta}(P)\defeq \beta(P) \text{ and } \tilde{\beta}(\,[f]\,) \defeq \beta(f) \]
for all objects $P$ and arrows $[f]$ of $\CatKTape$. Observe that this is well defined thanks to Lemma \ref{lemma:orderingfunctors}: if $f\sim_\basicK g$, namely $f \precongK g$ and $g \precongK f$, then $\beta(f)=\beta(g)$.
Moreover $\tilde{\beta}$ preserves the ordering $\precongK$ of $\CatKTape$, again thanks to Lemma \ref{lemma:orderingfunctors}. Thus  $\tilde{\beta}$ is a Kleene rig morphism.

Observe that, by definition, (e) $\eta; \tilde{\beta} = \beta$. Thus, $\tilde{\alpha_\sign} ;  \tilde{\beta} \stackrel{(b)}{=} \alpha_\sign ; \eta ;  \tilde{\beta} \stackrel{(e)}{=} \alpha_\sign ; \beta \stackrel{(d)}{=} \alpha_\sign '$.

Next define $\tilde{\alpha}\colon \sort \to \Cat{S}$ as $\alpha$. Thus $\tilde{\alpha_\sort} ; \tilde{\alpha}  \stackrel{(a)}{=} \alpha_\sort ; \alpha  \stackrel{(c)}{=}  \alpha_\sort '$.

Finally, the fact that $(\tilde{\alpha},\tilde{\beta})$ is a $\sort$-sesquistricty Kleene rig morphism from $\sort \to \CatTrTape$ to $\sort \to \CatKTape$ follows from the fact that t $(\alpha,\beta)$ is a $\sort$-sesquistrict ut-fb rig morphism.
\end{proof}

\section{Appendix to Section~\ref{sec:cb}}\label{app:kcb}

\begin{proof}[Proof of \Cref{prop:kc rig laws}]
     The laws in the top-left group can be seen to hold via the completeness theorem for fb-cb rig categories in~\cite{bonchi2023deconstructing}. 
     
     The first two laws in the top-right group hold in any cartesian bicategory. The remaining laws are proved below. When it is convenient we use string diagrams to depict the $\piu$ monoidal structure of the rig category.

	$[ \kstar{(f \sqcap g)} \leq \kstar{f} \sqcap \kstar{g} ]$. The following holds for all $f,g \colon X \to Y$:
	\begin{align*}
		\kstar{f} \sqcap \kstar{g} &= \;\; \copier{X} ; (\kstar{f} \per \kstar{g}) ; \cocopier{X} \tag{\ref{eq:cb:covolution}} \\
		&\geq \;\; \copier{X} ; \kstar{(f \per g)} ; \cocopier{X} \tag{\Cref{prop:star-per}} \\
		&= 
    \begin{tikzpicture}
	\begin{pgfonlayer}{nodelayer}
		\node [style=label] (105) at (-3.75, -1.225) {$X$};
		\node [style=black] (107) at (0.5, -0.25) {};
		\node [style=none] (108) at (-0.75, -1.225) {};
		\node [style=none] (109) at (-0.75, 0.725) {};
		\node [style=none] (117) at (-3.25, -1.225) {};
		\node [style=none] (118) at (-2.75, 0.725) {};
		\node [style=label] (120) at (5.25, -1.225) {$X$};
		\node [style=black] (121) at (1.5, -0.25) {};
		\node [style=none] (122) at (2.75, -1.225) {};
		\node [style=none] (123) at (2.75, 0.725) {};
		\node [style=none] (125) at (4.75, -1.225) {};
		\node [style=none] (127) at (-2.75, 2.225) {};
		\node [style=none] (128) at (2.75, 2.225) {};
		\node [style=stringbox] (129) at (-1.75, 0.725) {$\scriptstyle{f \per g}$};
		\node [style=stringbox] (130) at (3.25, -1.225) {$\scriptstyle \cocopier{X}$};
		\node [style=none] (131) at (3.25, 0.725) {};
		\node [style=none] (132) at (3.25, 0.725) {};
		\node [style=none] (133) at (3.25, 2.225) {};
		\node [style=stringbox] (134) at (-1.75, -1.225) {$\scriptstyle \copier{X}$};
	\end{pgfonlayer}
	\begin{pgfonlayer}{edgelayer}
		\draw [bend left] (109.center) to (107);
		\draw [bend left] (107) to (108.center);
		\draw [bend right] (123.center) to (121);
		\draw [bend right] (121) to (122.center);
		\draw (107) to (121);
		\draw (108.center) to (117.center);
		\draw [bend right=90, looseness=1.75] (127.center) to (118.center);
		\draw (127.center) to (128.center);
		\draw (118.center) to (129);
		\draw (129) to (109.center);
		\draw (125.center) to (130);
		\draw (130) to (122.center);
		\draw (131.center) to (132.center);
		\draw (128.center) to (133.center);
		\draw [bend right=90, looseness=1.75] (132.center) to (133.center);
		\draw (123.center) to (132.center);
	\end{pgfonlayer}
\end{tikzpicture}
}
 \\
		&\geq 
    \begin{tikzpicture}
	\begin{pgfonlayer}{nodelayer}
		\node [style=label] (105) at (-3.75, -1.225) {$X$};
		\node [style=black] (107) at (0.5, -0.25) {};
		\node [style=none] (108) at (-0.75, -1.225) {};
		\node [style=none] (109) at (-0.75, 0.725) {};
		\node [style=none] (117) at (-3.25, -1.225) {};
		\node [style=none] (118) at (-2.75, 0.725) {};
		\node [style=label] (120) at (7.25, -1.225) {$X$};
		\node [style=black] (121) at (1.5, -0.25) {};
		\node [style=none] (122) at (2.75, -1.225) {};
		\node [style=none] (123) at (2.75, 0.725) {};
		\node [style=none] (125) at (6.75, -1.225) {};
		\node [style=none] (127) at (-2.75, 2.225) {};
		\node [style=none] (128) at (2.75, 2.225) {};
		\node [style=stringbox] (129) at (-1.75, 0.725) {$\scriptstyle{f \per g}$};
		\node [style=stringbox] (130) at (3.25, -1.225) {$\scriptstyle \cocopier{X}$};
		\node [style=none] (131) at (6.25, 0.725) {};
		\node [style=none] (132) at (6.25, 0.725) {};
		\node [style=none] (133) at (6.25, 2.225) {};
		\node [style=stringbox] (134) at (-1.75, -1.225) {$\scriptstyle \copier{X}$};
		\node [style=stringbox] (135) at (3.25, 0.725) {$\scriptstyle \cocopier{X}$};
		\node [style=stringbox] (136) at (5.5, 0.725) {$\scriptstyle \copier{X}$};
	\end{pgfonlayer}
	\begin{pgfonlayer}{edgelayer}
		\draw [bend left] (109.center) to (107);
		\draw [bend left] (107) to (108.center);
		\draw [bend right] (123.center) to (121);
		\draw [bend right] (121) to (122.center);
		\draw (107) to (121);
		\draw (108.center) to (117.center);
		\draw [bend right=90, looseness=1.75] (127.center) to (118.center);
		\draw (127.center) to (128.center);
		\draw (118.center) to (129);
		\draw (129) to (109.center);
		\draw (125.center) to (130);
		\draw (130) to (122.center);
		\draw (131.center) to (132.center);
		\draw (128.center) to (133.center);
		\draw [bend right=90, looseness=1.75] (132.center) to (133.center);
		\draw (123.center) to (132.center);
	\end{pgfonlayer}
\end{tikzpicture}
}
 \tag{\ref{ax:copieradj1}} \\
		&= 
    \begin{tikzpicture}
	\begin{pgfonlayer}{nodelayer}
		\node [style=label] (105) at (-6, -1.225) {$X$};
		\node [style=black] (107) at (0.5, -0.25) {};
		\node [style=none] (108) at (-0.75, -1.225) {};
		\node [style=none] (109) at (-0.75, 0.725) {};
		\node [style=none] (117) at (-5.5, -1.225) {};
		\node [style=none] (118) at (-5, 0.725) {};
		\node [style=label] (120) at (5.5, -1.225) {$X$};
		\node [style=black] (121) at (1.5, -0.25) {};
		\node [style=none] (122) at (2.75, -1.225) {};
		\node [style=none] (123) at (2.75, 0.725) {};
		\node [style=none] (125) at (5, -1.225) {};
		\node [style=none] (127) at (-5, 2.225) {};
		\node [style=none] (128) at (2.75, 2.225) {};
		\node [style=stringbox] (129) at (-1.75, 0.725) {$\scriptstyle{f \per g}$};
		\node [style=stringbox] (130) at (3.25, -1.225) {$\scriptstyle \cocopier{X}$};
		\node [style=none] (131) at (4.5, 0.725) {};
		\node [style=none] (132) at (4.5, 0.725) {};
		\node [style=none] (133) at (4.5, 2.225) {};
		\node [style=stringbox] (134) at (-1.75, -1.225) {$\scriptstyle \copier{X}$};
		\node [style=stringbox] (135) at (3.25, 0.725) {$\scriptstyle \cocopier{X}$};
		\node [style=stringbox] (136) at (-4.25, 0.725) {$\scriptstyle \copier{X}$};
	\end{pgfonlayer}
	\begin{pgfonlayer}{edgelayer}
		\draw [bend left] (109.center) to (107);
		\draw [bend left] (107) to (108.center);
		\draw [bend right] (123.center) to (121);
		\draw [bend right] (121) to (122.center);
		\draw (107) to (121);
		\draw (108.center) to (117.center);
		\draw [bend right=90, looseness=1.75] (127.center) to (118.center);
		\draw (127.center) to (128.center);
		\draw (118.center) to (129);
		\draw (129) to (109.center);
		\draw (125.center) to (130);
		\draw (130) to (122.center);
		\draw (131.center) to (132.center);
		\draw (128.center) to (133.center);
		\draw [bend right=90, looseness=1.75] (132.center) to (133.center);
		\draw (123.center) to (132.center);
	\end{pgfonlayer}
\end{tikzpicture}
}
 \tag{\ref{ax:trace:sliding}} \\
		&= 
    \begin{tikzpicture}
	\begin{pgfonlayer}{nodelayer}
		\node [style=label] (105) at (-6, -1.225) {$X$};
		\node [style=black] (107) at (2.75, -0.25) {};
		\node [style=none] (108) at (1.5, -1.225) {};
		\node [style=none] (109) at (1.5, 0.725) {};
		\node [style=none] (117) at (-5.5, -1.225) {};
		\node [style=none] (118) at (-5, 0.725) {};
		\node [style=label] (120) at (6, -1.225) {$X$};
		\node [style=black] (121) at (3.75, -0.25) {};
		\node [style=none] (122) at (5, -1.225) {};
		\node [style=none] (123) at (5, 0.725) {};
		\node [style=none] (125) at (5.5, -1.225) {};
		\node [style=none] (127) at (-5, 2.225) {};
		\node [style=none] (128) at (5, 2.225) {};
		\node [style=stringbox] (129) at (-1.75, 0.725) {$\scriptstyle{f \per g}$};
		\node [style=none] (131) at (5, 0.725) {};
		\node [style=none] (132) at (5, 0.725) {};
		\node [style=none] (133) at (5, 2.225) {};
		\node [style=stringbox] (134) at (-1.75, -1.225) {$\scriptstyle \copier{X}$};
		\node [style=stringbox] (136) at (-4.25, 0.725) {$\scriptstyle \copier{X}$};
		\node [style=stringbox] (137) at (0.75, -1.225) {$\scriptstyle \cocopier{X}$};
		\node [style=stringbox] (138) at (0.75, 0.725) {$\scriptstyle \cocopier{X}$};
	\end{pgfonlayer}
	\begin{pgfonlayer}{edgelayer}
		\draw [bend left] (109.center) to (107);
		\draw [bend left] (107) to (108.center);
		\draw [bend right] (123.center) to (121);
		\draw [bend right] (121) to (122.center);
		\draw (107) to (121);
		\draw (108.center) to (117.center);
		\draw [bend right=90, looseness=1.75] (127.center) to (118.center);
		\draw (127.center) to (128.center);
		\draw (118.center) to (129);
		\draw (129) to (109.center);
		\draw (131.center) to (132.center);
		\draw (128.center) to (133.center);
		\draw [bend right=90, looseness=1.75] (132.center) to (133.center);
		\draw (123.center) to (132.center);
		\draw (122.center) to (125.center);
	\end{pgfonlayer}
\end{tikzpicture}
}
 \tag{\ref{ax:specfrob}} \\
		&= 
    \begin{tikzpicture}
	\begin{pgfonlayer}{nodelayer}
		\node [style=label] (105) at (-6, -1.225) {$X$};
		\node [style=black] (107) at (2.75, -0.25) {};
		\node [style=none] (108) at (1.5, -1.225) {};
		\node [style=none] (109) at (1.5, 0.725) {};
		\node [style=none] (117) at (-5.5, -1.225) {};
		\node [style=none] (118) at (-5, 0.725) {};
		\node [style=label] (120) at (6, -1.225) {$X$};
		\node [style=black] (121) at (3.75, -0.25) {};
		\node [style=none] (122) at (5, -1.225) {};
		\node [style=none] (123) at (5, 0.725) {};
		\node [style=none] (125) at (5.5, -1.225) {};
		\node [style=none] (127) at (-5, 2.225) {};
		\node [style=none] (128) at (5, 2.225) {};
		\node [style=stringbox] (129) at (-1.75, 0.725) {$\scriptstyle{f \per g}$};
		\node [style=none] (131) at (5, 0.725) {};
		\node [style=none] (132) at (5, 0.725) {};
		\node [style=none] (133) at (5, 2.225) {};
		\node [style=stringbox] (136) at (-4.25, 0.725) {$\scriptstyle \copier{X}$};
		\node [style=stringbox] (138) at (0.75, 0.725) {$\scriptstyle \cocopier{X}$};
	\end{pgfonlayer}
	\begin{pgfonlayer}{edgelayer}
		\draw [bend left] (109.center) to (107);
		\draw [bend left] (107) to (108.center);
		\draw [bend right] (123.center) to (121);
		\draw [bend right] (121) to (122.center);
		\draw (107) to (121);
		\draw (108.center) to (117.center);
		\draw [bend right=90, looseness=1.75] (127.center) to (118.center);
		\draw (127.center) to (128.center);
		\draw (118.center) to (129);
		\draw (129) to (109.center);
		\draw (131.center) to (132.center);
		\draw (128.center) to (133.center);
		\draw [bend right=90, looseness=1.75] (132.center) to (133.center);
		\draw (123.center) to (132.center);
		\draw (122.center) to (125.center);
	\end{pgfonlayer}
\end{tikzpicture}
}
 \tag{\ref{eq:cb:covolution}} \\
		&= \kstar{(f \sqcap g)}.
	\end{align*}

	\bigskip

	$[\kstar{\top} = \top]$. First we prove that $\top ; \kstar{\top} = \top$. The left-to-right inclusion trivially holds since $\top$ is the top element of the meet-semilattice $\Cat{C}[X,X]$. For the other inclusion the following holds:
	\begin{align*}
		\top ; \kstar{\top} &= 
    \begin{tikzpicture}
	\begin{pgfonlayer}{nodelayer}
		\node [style=label] (105) at (-3.5, -1.225) {$X$};
		\node [style=black] (107) at (2.75, -0.25) {};
		\node [style=none] (108) at (1.5, -1.225) {};
		\node [style=none] (109) at (1.5, 0.725) {};
		\node [style=none] (117) at (-3, -1.225) {};
		\node [style=none] (118) at (-2.5, 0.725) {};
		\node [style=label] (120) at (6, -1.225) {$X$};
		\node [style=black] (121) at (3.75, -0.25) {};
		\node [style=none] (122) at (5, -1.225) {};
		\node [style=none] (123) at (5, 0.725) {};
		\node [style=none] (125) at (5.5, -1.225) {};
		\node [style=none] (127) at (-2.5, 2.225) {};
		\node [style=none] (128) at (5, 2.225) {};
		\node [style=none] (131) at (5, 0.725) {};
		\node [style=none] (132) at (5, 0.725) {};
		\node [style=none] (133) at (5, 2.225) {};
		\node [style=stringbox] (134) at (-1.75, -1.225) {$\scriptstyle \discharger{X}$};
		\node [style=stringbox] (136) at (-1.75, 0.725) {$\scriptstyle \discharger{X}$};
		\node [style=stringbox] (137) at (0.75, -1.225) {$\scriptstyle \codischarger{X}$};
		\node [style=stringbox] (138) at (0.75, 0.725) {$\scriptstyle \codischarger{X}$};
	\end{pgfonlayer}
	\begin{pgfonlayer}{edgelayer}
		\draw [bend left] (109.center) to (107);
		\draw [bend left] (107) to (108.center);
		\draw [bend right] (123.center) to (121);
		\draw [bend right] (121) to (122.center);
		\draw (107) to (121);
		\draw (108.center) to (117.center);
		\draw [bend right=90, looseness=1.75] (127.center) to (118.center);
		\draw (127.center) to (128.center);
		\draw (131.center) to (132.center);
		\draw (128.center) to (133.center);
		\draw [bend right=90, looseness=1.75] (132.center) to (133.center);
		\draw (123.center) to (132.center);
		\draw (122.center) to (125.center);
		\draw (118.center) to (109.center);
	\end{pgfonlayer}
\end{tikzpicture}
}
 \tag{\ref{eq:cb:covolution}} \\
		&= 
    \begin{tikzpicture}
	\begin{pgfonlayer}{nodelayer}
		\node [style=label] (105) at (-4.75, -1.225) {$X$};
		\node [style=black] (107) at (-0.5, -0.25) {};
		\node [style=none] (108) at (-1.75, -1.225) {};
		\node [style=none] (109) at (-1.75, 0.725) {};
		\node [style=none] (117) at (-4.25, -1.225) {};
		\node [style=none] (118) at (-3.75, 0.725) {};
		\node [style=label] (120) at (4.75, -1.225) {$X$};
		\node [style=black] (121) at (0.5, -0.25) {};
		\node [style=none] (122) at (1.75, -1.225) {};
		\node [style=none] (123) at (1.75, 0.725) {};
		\node [style=none] (125) at (4.25, -1.225) {};
		\node [style=none] (127) at (-3.75, 2.225) {};
		\node [style=none] (128) at (3.75, 2.225) {};
		\node [style=none] (131) at (1.75, 0.725) {};
		\node [style=none] (132) at (3.75, 0.725) {};
		\node [style=none] (133) at (3.75, 2.225) {};
		\node [style=stringbox] (134) at (-3, -1.225) {$\scriptstyle \discharger{X}$};
		\node [style=stringbox] (136) at (-3, 0.725) {$\scriptstyle \discharger{X}$};
		\node [style=stringbox] (137) at (3, -1.225) {$\scriptstyle \codischarger{X}$};
		\node [style=stringbox] (138) at (3, 0.725) {$\scriptstyle \codischarger{X}$};
	\end{pgfonlayer}
	\begin{pgfonlayer}{edgelayer}
		\draw [bend left] (109.center) to (107);
		\draw [bend left] (107) to (108.center);
		\draw [bend right] (123.center) to (121);
		\draw [bend right] (121) to (122.center);
		\draw (107) to (121);
		\draw (108.center) to (117.center);
		\draw [bend right=90, looseness=1.75] (127.center) to (118.center);
		\draw (127.center) to (128.center);
		\draw (131.center) to (132.center);
		\draw (128.center) to (133.center);
		\draw [bend right=90, looseness=1.75] (132.center) to (133.center);
		\draw (123.center) to (132.center);
		\draw (122.center) to (125.center);
		\draw (118.center) to (109.center);
	\end{pgfonlayer}
\end{tikzpicture}
}
 \tag{\eqref{ax:comonoid:nat:copy}, \eqref{ax:monoid:nat:copy}} \\
		&= 
    \begin{tikzpicture}
	\begin{pgfonlayer}{nodelayer}
		\node [style=label] (105) at (-7.25, -1.225) {$X$};
		\node [style=black] (107) at (-0.5, -0.25) {};
		\node [style=none] (108) at (-1.75, -1.225) {};
		\node [style=none] (109) at (-1.75, 0.725) {};
		\node [style=none] (117) at (-6.75, -1.225) {};
		\node [style=none] (118) at (-6.25, 0.725) {};
		\node [style=label] (120) at (4.75, -1.225) {$X$};
		\node [style=black] (121) at (0.5, -0.25) {};
		\node [style=none] (122) at (1.75, -1.225) {};
		\node [style=none] (123) at (1.75, 0.725) {};
		\node [style=none] (125) at (4.25, -1.225) {};
		\node [style=none] (127) at (-6.25, 2.225) {};
		\node [style=none] (128) at (3.75, 2.225) {};
		\node [style=none] (131) at (1.75, 0.725) {};
		\node [style=none] (132) at (3.75, 0.725) {};
		\node [style=none] (133) at (3.75, 2.225) {};
		\node [style=stringbox] (134) at (-3, -1.225) {$\scriptstyle \discharger{X}$};
		\node [style=stringbox] (136) at (-3, 0.725) {$\scriptstyle \discharger{X}$};
		\node [style=stringbox] (137) at (3, -1.225) {$\scriptstyle \codischarger{X}$};
		\node [style=stringbox] (138) at (-5, 0.725) {$\scriptstyle \codischarger{X}$};
	\end{pgfonlayer}
	\begin{pgfonlayer}{edgelayer}
		\draw [bend left] (109.center) to (107);
		\draw [bend left] (107) to (108.center);
		\draw [bend right] (123.center) to (121);
		\draw [bend right] (121) to (122.center);
		\draw (107) to (121);
		\draw (108.center) to (117.center);
		\draw [bend right=90, looseness=1.75] (127.center) to (118.center);
		\draw (127.center) to (128.center);
		\draw (131.center) to (132.center);
		\draw (128.center) to (133.center);
		\draw [bend right=90, looseness=1.75] (132.center) to (133.center);
		\draw (123.center) to (132.center);
		\draw (122.center) to (125.center);
		\draw (118.center) to (109.center);
	\end{pgfonlayer}
\end{tikzpicture}
}
 \tag{\ref{ax:trace:sliding}} \\
		&\geq 
    \begin{tikzpicture}
	\begin{pgfonlayer}{nodelayer}
		\node [style=label] (105) at (-4.75, -1.225) {$X$};
		\node [style=black] (107) at (-0.5, -0.25) {};
		\node [style=none] (108) at (-1.75, -1.225) {};
		\node [style=none] (109) at (-1.75, 0.725) {};
		\node [style=none] (117) at (-4.25, -1.225) {};
		\node [style=none] (118) at (-2, 0.725) {};
		\node [style=label] (120) at (4.75, -1.225) {$X$};
		\node [style=black] (121) at (0.5, -0.25) {};
		\node [style=none] (122) at (1.75, -1.225) {};
		\node [style=none] (123) at (1.75, 0.725) {};
		\node [style=none] (125) at (4.25, -1.225) {};
		\node [style=none] (127) at (-2, 2.225) {};
		\node [style=none] (128) at (2, 2.225) {};
		\node [style=none] (131) at (1.75, 0.725) {};
		\node [style=none] (132) at (2, 0.725) {};
		\node [style=none] (133) at (2, 2.225) {};
		\node [style=stringbox] (134) at (-3, -1.225) {$\scriptstyle \discharger{X}$};
		\node [style=stringbox] (137) at (3, -1.225) {$\scriptstyle \codischarger{X}$};
	\end{pgfonlayer}
	\begin{pgfonlayer}{edgelayer}
		\draw [bend left] (109.center) to (107);
		\draw [bend left] (107) to (108.center);
		\draw [bend right] (123.center) to (121);
		\draw [bend right] (121) to (122.center);
		\draw (107) to (121);
		\draw (108.center) to (117.center);
		\draw [bend right=90, looseness=1.75] (127.center) to (118.center);
		\draw (127.center) to (128.center);
		\draw (131.center) to (132.center);
		\draw (128.center) to (133.center);
		\draw [bend right=90, looseness=1.75] (132.center) to (133.center);
		\draw (123.center) to (132.center);
		\draw (122.center) to (125.center);
		\draw (118.center) to (109.center);
	\end{pgfonlayer}
\end{tikzpicture}
}
 \tag{\ref{ax:dischargeradj1}} \\
		&= 
    \begin{tikzpicture}
	\begin{pgfonlayer}{nodelayer}
		\node [style=label] (105) at (-2.75, 0.025) {$X$};
		\node [style=none] (108) at (0, 0.025) {};
		\node [style=none] (117) at (-2.25, 0.025) {};
		\node [style=label] (120) at (2.75, 0.025) {$X$};
		\node [style=none] (122) at (0, 0.025) {};
		\node [style=none] (125) at (2.25, 0.025) {};
		\node [style=stringbox] (134) at (-1, 0.025) {$\scriptstyle \discharger{X}$};
		\node [style=stringbox] (137) at (1, 0.025) {$\scriptstyle \codischarger{X}$};
	\end{pgfonlayer}
	\begin{pgfonlayer}{edgelayer}
		\draw (108.center) to (117.center);
		\draw (122.center) to (125.center);
	\end{pgfonlayer}
\end{tikzpicture}
}
 \tag{\ref{ax:kb:traceid}} \\
		&= \top. \tag{\ref{eq:cb:covolution}}
	\end{align*}
	To conclude, observe that:
	\[ \kstar{\top} \stackrel{\text{(\Cref{prop:stef})}}{=} \id{X} + \top ; \kstar{\top} = \id{X} + \top = \top. \]

	\bigskip

	$[\op{(f + g)} = \op{f} + \op{g}]$. The following hold for all $f,g \colon X \to Y$:
	\begin{align*}
		\op{(f + g)} &= \op{(\diag{X} ; (f \piu g) ; \codiag{Y})} \tag{\ref{eq:covolution}} \\
				   &= \op{\codiag{Y}} ; \op{(f \piu g)} ; \op{\diag{X}} \tag{\Cref{table:re:daggerproperties}} \\
				   &= \op{\codiag{Y}} ; (\op{f} \piu \op{g}) ; \op{\diag{X}} \tag{\Cref{table:kc rig derived laws}} \\
				   &= \diag{Y} ; (\op{f} \piu \op{g}) ; \codiag{X} \tag{\Cref{table:kc rig derived laws}} \\
				   &= \op{f} + \op{g}. \tag{\ref{eq:covolution}}
	\end{align*}

	\bigskip

	$[\op{\bot} = \bot]$. \begin{align*}
		\op{\bot} \stackrel{\eqref{eq:covolution}}{=} \op{(\bang{X}; \cobang{Y})} \stackrel{\text{(\Cref{table:re:daggerproperties})}}{=} \op{\cobang{Y}} ; \op{\bang{X}} \stackrel{\text{(\Cref{table:kc rig derived laws})}}{=} \bang{Y} ; \cobang{X} \stackrel{\eqref{eq:covolution}}{=} \bot.
	\end{align*}

	\bigskip

	$[\kstar{f} + \kstar{g} \leq \kstar{(f + g)}]$. The following holds for all $f,g \colon X \to X$: 
	\begin{align*}
		\kstar{f} + \kstar{g} &= 

}
 \tag{\eqref{ax:monoid:unit}, \eqref{ax:comonoid:unit}} \\
		&= \;\; \id{X}.
	\end{align*}

	\bigskip

	$[\kstar{(\op{f})} = \op{(\kstar{f})}]$. First, note that the following law is equivalent to the first law in~\eqref{eq:stroingstar} (see e.g. \cite{Kozen94acompleteness}):
	\begin{equation}\label{eq:kozen-equiv-axiom}
		g + f;r \leq r \implies \kstar{f};g \leq r.
	\end{equation}
	Then observe that the following holds for all $f \colon X \to X$:
	\[ \id{X} + \op{f} ; \op{(\kstar{f})} \stackrel{(\text{\Cref{table:re:daggerproperties}})}{=} \id{X} + \op{(\kstar{f} ; f)} = \op{(\id{X} + \kstar{f} ; f)} \stackrel{(\text{\Cref{prop:star-fixpoint}})}{=} \op{(\kstar{f})}. \]
	Thus, by \eqref{eq:kozen-equiv-axiom} the inequality below holds:
	\begin{equation}\label{eq:star-dagger-inclusion}
		\kstar{(\op{f})} = \kstar{(\op{f})} ; \id{X} \leq \op{(\kstar{f})}.
	\end{equation}
	For the other inclusion we exploit~\eqref{eq:star-dagger-inclusion} and the fact that $\op{(\cdot)}$ is involutive:
	\[ \op{(\kstar{f})} \stackrel{(\text{\Cref{table:re:daggerproperties}})}{=} \op{(\kstar{(f^{\dag \dag})})} \stackrel{\eqref{eq:star-dagger-inclusion}}{\leq} (\kstar{(\op{f})})^{\dag \dag} \stackrel{(\text{\Cref{table:re:daggerproperties}})}{=} \kstar{(\op{f})}. \]

	Finally, we prove the laws of distributive lattices at the bottom of the table.

	$[f \sqcap (g + h) = (\, f \sqcap g \, ) + (\, f \sqcap h  \,)]$. The following holds for all $f,g,h \colon X \to Y$:
	\begin{align*}
		f \sqcap (g + h) &= \copier{X} ; (\, f \otimes (g + h) \,) ; \cocopier{Y}  \tag{\ref{eq:cb:covolution}} \\
										   &= \copier{X} ; (\, (\,(f \otimes g) + (f \otimes h)\,) \,) ; \cocopier{Y}  \tag{\Cref{lemma:ditributivityper}} \\
										   &=  (\, (\, \copier{X} ; (f \otimes g) \, ) + (\,\copier{X} ; (f \otimes h)\,) \,) ; \cocopier{Y} \tag{\ref{eq:cmon enrichment}} \\
										   &=  (\, \copier{X} ; (f \otimes g) ; \cocopier{Y} \, ) + (\,\copier{X} ; (f \otimes h) ; \cocopier{Y}\,) \tag{\ref{eq:cmon enrichment}} \\
										   &=  (\, f \sqcap g \, ) + (\, f \sqcap h  \,).  \tag{\ref{eq:cb:covolution}}
	\end{align*}

	\bigskip
				
	$[f \sqcap \bot = \bot]$. The following holds for all $f \colon X \to Y$:
	\begin{align*}
		f \sqcap \bot &\leq \top \sqcap \bot \\
					  &= \bot.
	\end{align*}
	As usual, the other inclusion holds since $\bot$ is the bottom element.

	\bigskip
				
	$[f + (g \sqcap h) = (\, f + g \, ) \sqcap (\, f + h  \,)]$ and $[f + \top = \top]$. These two equations hold in every lattice satisfying the dual equations proved above (see e.g. \cite{birkhoff1940lattice}).
\end{proof}

\subsection{Proofs for Theorem \ref{thm:KleeneCartesiantapesfree}}

\begin{lemma}\label{lemma:orderingkc}
Let $\Cat{C}$ be a $\sort$-sesquistrict kc rig category with ordering $\leq_{\Cat{C}}$. Let $F \colon \CatKTapeC \to \Cat{C}$ be a morphism of Kleene rig category such that
 \begin{equation}\label{eq:monoidF}
F (\tape{\cocopier{P}})  =   \cocopier{F(P)} \qquad
F (\tape{\codischarger{P}})  =   \codischarger{F(P)} \qquad
F (\tape{\copier{P}})  =   \copier{{F(P)}} \qquad
F (\tape{\discharger{P}})  =   \discharger{{F(P)}} 
\end{equation}
for all $P\in Ob(\CatKTapeC)$.
For all Kleene tapes $\t_1,\t_2\colon P \to Q$, if $\t_1 \precongKC \t_2$ then
$F\t_1 \leq_{\Cat{C}}F\t_2$.
\end{lemma}
\begin{proof}
Define $\leq'$ on $\CatTrTape$ as $\t_1 \leq' \t_2$ iff $F\t_1 \leq_{\Cat{C}}F\t_2$.
By using exactly the same proof of Lemma \ref{lemma:orderingfunctors}, one can show that $\leq'$ is a uniform precongruence, namely that $\mathsf{upc}(\leq') \subseteq \leq'$, and that $\basicK\subseteq \leq'$.

Next, we observe that $\basicCB \subseteq \leq'$.
The proof proceeds by cases and  it relies on the fact that $\Cat{C}$ has the structure of a cartesian bicategory and that, thanks to \eqref{eq:monoidF}, $F$ preserves such structure. For instance, to prove that 
 $\{( \cocopier{P};\copier{P}    \, , \, \id{P\per P} \mid P \in Ob(\CatTrTape) \} \subseteq \leq' $, one shows that
\begin{align*}
 \cocopier{P};\copier{P}  \leq' \id{P\per P}& \Longleftrightarrow   F(\cocopier{P};\copier{P})  \leq_{\Cat{C}}  F\id{P\per P} \tag{def}\\
& \Longleftrightarrow   \cocopier{FP};\copier{FP} \leq_{\Cat{C}} \id{FP\per FP} \tag{\eqref{eq:monoidF} and $F$ rig functor} 
\end{align*}
and conclude by observing that the latter holds since $\Cat{C}$ is a cartesian bicategory.

\medskip

Now, since $\leq'$ is a uniform precongruence and since $\basicK \cup \basicCB\subseteq \leq'$ then, by Lemma \ref{lemma:smallest}, one has that  $\precongKC\subseteq \leq'$. This means that if 
$\t_1 \precongKC \t_2$ then $F\t_1 \leq_{\Cat{C}}F\t_2$.
\end{proof}

\begin{proof}[Proof of Theorem \ref{thm:KleeneCartesiantapesfree}]
Recall that, by Theorem \ref{thm:Kleenetapesfree}, $I \colon \sort \to \CatKTapeC$ is a free $\sort$-sesquistrict Kleene rig category generated by $(\sort, \sign+\Gamma)$. 
Here, $I$ is the obvious embedding mapping each sort $A\in \sort$ into the polynomial $A\in (\sort^*)^*$.
The intepretation of $(\sort, \sign+\Gamma)$ into $I \colon \sort \to \CatKTapeC$ consists of the functions
\[\alpha_{\sort}\colon \sort \to \sort \qquad \alpha_{\sign}\colon \sign \to Ar(\CatKTapeC) \quad \alpha_{\Gamma}\colon \Gamma \to Ar(\CatKTapeC) \]
defined as expected: $\alpha_{\sort}(A)\defeq A$, $\alpha_\sign(\sigma)=\tape{\sigma}$ and $\alpha_{\Gamma}(\gamma)=\tape{\gamma}$ for all $A\in \sort$, $\sigma\in \sign$ and $\gamma\in \Gamma$.

Now, observe that there is a Kleene rig morphism \[\eta \colon \CatKTapeC \to \KTCB\] that is the identity on object, i.e., $\eta(P)\defeq P$, and maps an arrows $\t\colon P \to Q$ into the $\sim_{\basicKC}$-equivalence classes $[\t]\colon P \to Q$.
We can thus fix as interpretation $(\tilde{\alpha_\sort}, \tilde{\alpha_\sign})$ of $(\sort, \sign)$ into $\KTCB$ as follows \[\text{(a) }\tilde{\alpha_\sort} \defeq \alpha_\sort \text{ and (b) }\tilde{\alpha_\sign}\defeq \alpha_\sign ; \eta\text{.}\]

Let $H\colon \Cat{S} \to \Cat{C}$ be a $\sort$-sesquistrict kc rig category and $(\alpha_\sort', \alpha_{\sign}')$ be an interpretation of $(\sort, \sign)$ into $\Cat{S} \to \Cat{C}$.
Recall that, by definition of interpretation $\alpha_\sort'\colon \sort \to \Cat{S}$ and $\alpha_{\sign}' \colon \sign \to Ar(\Cat{C})$.

Since $\Cat{C}$ is a cartesian bicategory, there are (co)monoids for each object in $Oc(\Cat{C})$. Thus one can define $\alpha_\Gamma'\colon \Gamma \to Ar(\Cat{C})$ as
\[
\alpha_{\Gamma}' (\cocopier{A})  \defeq   \cocopier{H{\alpha_{\sort}'(A)}}{} \qquad
\alpha_{\Gamma}' (\codischarger{A})  \defeq   \codischarger{H{\alpha_{\sort}'(A)}}{} \qquad
\alpha_{\Gamma}' (\copier{A})  \defeq   \copier{H{\alpha_{\sort}'(A)}}{} \qquad
\alpha_{\Gamma}' (\discharger{A})  \defeq   \discharger{H{\alpha_{\sort}'(A)}}{}
\]
We can thus take the copairing of $\alpha_{\sign}'$ and $\alpha_{\Gamma}'$, hereafter denoted as $[\alpha_{\sign}',\alpha_{\Gamma}'] \colon \sign+\Gamma \to Ar(\Cat{C})$ to have an interpretation $(\alpha_\sort', [\alpha_{\sign}',\alpha_{\Gamma}'])$ of $(\sort,\sign+\Gamma)$ into $H\colon \Cat{S} \to \Cat{C}$.

By freeness of  $\sort \to \CatKTapeC$, one has a unique sesquistrict Kleene rig functor
$(\alpha, \beta)$ with $\alpha \colon \sort \to \Cat{S}$ and $\beta\colon \CatKTapeC \to \Cat{C} $ such that \[\text{(c) }\alpha_\sort ;  \alpha = \alpha_\sort '  \qquad \text{(d) }\alpha_{\sign} ; \beta = \alpha_{\sign}' \qquad \text{(e) }\alpha_\Gamma; \beta =\alpha_{\Gamma}' \]
Since $(\alpha, \beta)$ is a $\sort$-sesquistrict functor from $I \colon \sort \to \CatKTapeC$ to $H\colon \Cat{S} \to \Cat{C}$ then, by definition,
(f) $\alpha ; H = I ; \beta$ and thus  $\alpha_{\sort}' ; H \stackrel{(c)}{=} \alpha_\sort ;  \alpha ; H \stackrel{(f)}{=} \alpha_\sort ; I;\beta$.
The latter, together with (e) and the definition of $\alpha_\sort$ and $\alpha_\Gamma$, gives us the following facts:
\begin{equation*}
\beta (\tape{\cocopier{A}})  =   \cocopier{\beta(A)} \qquad
\beta (\tape{\codischarger{A}})  =   \codischarger{\beta(A)} \qquad
\beta (\tape{\copier{A}})  =   \copier{{\beta(A)}} \qquad
\beta (\tape{\discharger{A}})  =   \discharger{{\beta(A)}} 
\end{equation*}
A simple inductive arguments, exploiting in the base case the above equivalences, and in the inductive case the inductive definitions in \eqref{eq:copierind} and \eqref{eq:cocopierind} and the coherences conditions in \eqref{eq:fbcbcoherence}, confirms that the followings hold for all $P\in Ob(\CatKTapeC)$.
 \begin{equation}\label{eq:monoidbeta}
\beta (\tape{\cocopier{P}})  =   \cocopier{\beta(P)} \qquad
\beta (\tape{\codischarger{P}})  =   \codischarger{\beta(P)} \qquad
\beta (\tape{\copier{P}})  =   \copier{{\beta(P)}} \qquad
\beta (\tape{\discharger{P}})  =   \discharger{{\beta(P)}} 
\end{equation}

Now define $\tilde{\beta} \colon \KTCB \to \Cat{C}$ as 
\[\tilde{\beta}(P)\defeq \beta(P) \text{ and } \tilde{\beta}(\,[f]\,) \defeq \beta(f) \]
for all objects $P$ and arrows $[f]$ of $\KTCB$. Observe that this is well defined thanks to Lemma \ref{lemma:orderingkc}: if $f\sim_\basicKC g$, namely $f \precongKC g$ and $g \precongKC f$, then $\beta(f)=\beta(g)$.
Moreover $\tilde{\beta}$ preserves the ordering $\precongKC$ of $\KTCB$, again thanks to Lemma \ref{lemma:orderingkc}. To conclude that $\tilde{\beta}$  is a morphism of kc rig categories, it only remains to show that it is a morphism of 
Cartesian bicategories but this is trivial by \eqref{eq:monoidbeta}.

Observe that, by definition, (g) $\eta; \tilde{\beta} = \beta$. Thus, $\tilde{\alpha_\sign} ;  \tilde{\beta} \stackrel{(b)}{=} \alpha_\sign ; \eta ;  \tilde{\beta} \stackrel{(g)}{=} \alpha_\sign ; \beta \stackrel{(d)}{=} \alpha_\sign '$.

Next take $\tilde{\alpha}\colon \sort \to \Cat{S}$ as $\alpha$. Thus $\tilde{\alpha_\sort} ; \tilde{\alpha}  \stackrel{(a)}{=} \alpha_\sort ; \alpha  \stackrel{(c)}{=}  \alpha_\sort '$.

Finally, the fact that $(\tilde{\alpha},\tilde{\beta})$ is a morphism of $\sort$-sesquistrict kc rig categories from $I \colon \sort \to \KTCB$ to $H \colon \sort \to \Cat{C}$, namely that $I;\tilde{\beta}=\tilde{\alpha};H$ follows  immediately from (f).
\end{proof}

\subsection{Proofs of other results}

\begin{proof}[Proof of \Cref{funct:sem}]
    Observe that there exists a kc rig morphism $\eta \colon \KTCB \to \KTCBI$ defined as the identity on objects and mapping tapes $\t\colon P \to Q$ into $\sim_{\basicR}$-equivalence classes $[\t]\colon P \to Q$.
    
    Let $\mathcal{I}=(\alpha_\sort, \alpha_\sign)$ be a model of $(\sort, \sign)$ in $\Cat{S} \to \Cat{C}$ and let $\alpha_\sign^\sharp\colon \KTCB \to \Cat{C}$ be the morphism induced by freeness of $\KTCB$. 
    Define $\tilde{\alpha}_\sign^\sharp \colon \KTCBI \to \Cat{C}$ as $\tilde{\alpha}_\sign^\sharp(P)\defeq\alpha_\sign(P)$ for all objects $P$ and  $\tilde{\alpha}_\sign^\sharp([\t])\defeq\alpha_\sign(\t)$  for $\sim_{\basicR}$-equivalence classes $[\t]\colon P \to Q$. 
    Since $(\alpha_\sort, \alpha_\sign)$, then $\alpha_\sign^\sharp$ preserves $\precongB$ and thus $\tilde{\alpha}_\sign^\sharp$ is well defined. Checking that it is a kc rig morphism it is immediate from the fact that $\alpha_\sign^\sharp\colon \KTCB \to \Cat{C}$ is a kc rig morphism.
    
    Viceversa, from a morphism $\beta \colon \KTCBI \to \Cat{C}$ one can construct an interpretation  $\mathcal{I}$ of $(\sort, \sign)$ in $\Cat{S} \to \Cat{C}$ by precomposing first with $\eta$ and then with the trivial interpretation of  $(\sort, \sign)$ in $\sort \to \KTCB$. The unique sesquistrict kc rig morphism induced by $\mathcal{I}$ is exactly $\eta; \beta$. Since $\eta;\beta$ factors through $\KTCBI$, it obviously preserves $\precongB$ and thus $\mathcal{I}$ is a model of $(\sign, \basicR)$.
    
    To conclude that the correspondene is bijective, it is enough to observe that $\alpha_\sign^\sharp = \eta ; \tilde{\alpha}_\sign^\sharp$
\end{proof}

\end{document}